\documentclass[reqno]{amsbook}
\usepackage{hyperref}
\usepackage[T1]{fontenc}
\usepackage[latin1]{inputenc}  
\usepackage{ae,aecompl}

\usepackage{amssymb,amscd}
\usepackage{url}
\usepackage{amsmath}
\usepackage{amsthm,txfonts,amsxtra}
\usepackage{bbm}
\usepackage{yfonts}
\usepackage{enumitem}
\usepackage{nomencl}
  
\usepackage{graphicx}

\numberwithin{equation}{chapter}
\numberwithin{table}{chapter}
\numberwithin{figure}{chapter}
\numberwithin{section}{chapter}

\newtheoremstyle{bold}
{.5\baselineskip}{.5\baselineskip}{\itshape}{}{\bfseries}{.}{.5em}{}
\newtheoremstyle{shy}
{.5\baselineskip}{.5\baselineskip}{}{}{\bfseries}{.}{.5em}{}

\makeatletter
\def\@captionfont{\small}

\renewenvironment{proof}[1][\proofname]{\par
  \pushQED{\qed}%
  \normalfont \topsep6\p@\@plus6\p@\relax
  \trivlist
  \item[\hskip\labelsep
        \scshape
    #1\@addpunct{.}]\ignorespaces
}{%
  \popQED\endtrivlist\@endpefalse
}

\def\chapter{%
  \if@openright\cleardoublepage\else\clearpage\fi
  \thispagestyle{empty}\global\@topnum\z@
  \@afterindenttrue \secdef\@chapter\@schapter}

\def\section{\@startsection{section}{1}%
  \z@{.9\linespacing\@plus\linespacing}{.5\linespacing}%
  {\large\bfseries\boldmath\centering}}

\def\subsection{\@startsection{subsection}{2}%
  \z@{.7\linespacing\@plus\linespacing}{.5\linespacing}%
  {\normalfont\scshape\centering}}

\def\theindex{\@restonecoltrue\if@twocolumn\@restonecolfalse\fi
  \columnseprule\z@ \columnsep 35\p@
  \@indextitlestyle
  \thispagestyle{empty}%
  \let\item\@idxitem
  \parindent\z@  \parskip\z@\@plus.3\p@\relax
  \raggedright
  \hyphenpenalty\@M
  \footnotesize}

\renewcommand{\@bibtitlestyle}{%
  \@xp\section\@xp*\@xp{\bibname}%
}
\renewcommand{\tocchapter}[3]{%
  \indentlabel{\@ifnotempty{#2}{\ignorespaces#1 #2.\quad}}#3}

\renewcommand{\tocsection}[3]{%
  \indentlabel{\@ifnotempty{#2}{\makebox[3.2em][l]{\ignorespaces#1 #2.}}}#3}

\makeatother

\renewcommand{\bibname}{References}

\theoremstyle{bold}
\newtheorem{theorem}{Theorem}[chapter]
\newtheorem{proposition}[theorem]{Proposition}
\newtheorem{lemma}[theorem]{Lemma}
\newtheorem{corollary}[theorem]{Corollary}
\newtheorem{conjecture}[theorem]{Conjecture}

  \newtheorem{assumption}[theorem]{Assumption}
\theoremstyle{shy}
\newtheorem{definition}[theorem]{Definition}
\newtheorem{remark}[theorem]{Remark}

\def\sumint{\int\hspace{-1.3em}\sum}

\newcommand{\dd}{\ts\mathrm{d}\ts}
\newcommand{\ee}{\ts\mathrm{e}\ts}

\newcommand{\ts}{\hspace{0.5pt}}

\makeindex

\def\a{\alpha}
\def\b{\beta}
\def\d{\delta}
\def\e{\epsilon}
\def\ve{\varepsilon}
\def\f{\phi}
\def\g{\gamma}

\def\l{\lambda}

\def\s{\sigma}
\def\t{\tau}
\def\th{\theta}

\def\o{\omega}
\def\D{\Delta}
\def\L{\Lambda}
\def\G{\Gamma}

\def\S{\Sigma}

\def\del #1{\frac{\partial^{#1}}{\partial\l^{#1}}}

\def\eee{\mathrm{e}}

\def\eu{{1\kern-.25em\rm{I}}}
\def\f1{{1\kern-.25em\rm{I}}}
\def\R{{\mathbb R}}  
\def\N{{\mathbb N}}  
\def\P{{\mathbb P}}  
\def\E{{\mathbb E}}  
\def\M{{\mathbb M}}  

\def\bx {\mathbf{x}}  

\def\bz{\mathbf{z}}

\newcommand{\DD}{\mathbb{D}}

\def\del{\partial}

\def\FF{{\mathfrak F}}

\def\HH{{\mathcal H}}

\def\LL{{\mathcal L}}
\def\MM{{\mathcal M}}
\def\NN{{\mathcal N}}
\def\OO{{\mathcal O}}
\def\PP{{\mathcal P}}

\def\UU{{\mathcal U}}

\def\XX{{\mathcal X}}

\newcommand{\cG}{\mathcal{G}}

\newcommand{\cL}{\mathcal{L}}

\def\wt{\widetilde}
\def\wh{\widehat}

\def\TH(#1){\label{#1}}\def\thv(#1){\ref{#1}}
\def\Eq(#1){\label{#1}}\def\eqv(#1){(\ref{#1})}

\def\var{\hbox{\rm var}}
\def\sfrac#1#2{{\textstyle{#1\over #2}}}

\def\1{\mathbbm{1}}

\def\be{\begin{equation}}
\def\ee{\end{equation}}
\newcommand{\bea}{\begin{eqnarray}}
\newcommand{\eea}{\end{eqnarray}}

\def\supp{\hbox{\rm supp}}

\def\as{\hbox{\rm a.s.}}

\makeatletter
\newenvironment{chapquote}[2][1em]
  {\setlength{\@tempdima}{#1}%
   \def\chapquote@author{#2}%
   \parshape 1 \@tempdima \dimexpr\textwidth-2\@tempdima\relax%
   \itshape}
  {\par\normalfont\hfill--\ \chapquote@author\hspace*{\@tempdima}\par\bigskip}
\makeatother

\begin{document}
\hypersetup{linkbordercolor=1 1 1} 
\frontmatter
\title{Stochastic individual-based models}

\author{Anton Bovier} 
\address{University of Bonn}\email {bovier@uni-bonn.de}
\author{Anna Kraut}

\address{St. Olaf College}\email{kraut1@stolaf.edu}
\maketitle

\tableofcontents
\chapter*{Preface}\label{chap0}
These are lecture notes for an advanced topics course in the master's programme at Bonn University.
It aims to give a concise review of some of the work that has been done around the 
topic of adaptive dynamics from a rigorous stochastic point of view over the last 25 years and to make this
area accessible to students with a good knowledge of probability in general, and the theory of Markov 
processes in particular. 

Our emphasis is on the issue of emerging scaling limits, where scaling parameters are time, population size, mutation rates, and mutation step size. These allow to exhibit within a fairly simple class of models a variety of biologically relevant phenomena. These notes are organised in three parts. Part 1 presents some historical background as well as the mathematical setting and main tools. 
Part 2 is the core of the notes and discusses the various scaling regimes and scaling limits. Part 3 looks at some extensions on the basic model, notably diploid models, phenotypic plasticity, and effects of environmental changes over time.

We are deeply grateful to my collaborators on the matters of these lectures, Martina Baar,
Nicolas Champagnat,  Loren Coquille, Manuel Esser, Hannah Mayer, Charline Smadi, Rebecca Ströfer, Shi-Dong Wang,
and colleagues from the life sciences, Nicole Glodde, Michael Hölzel, Thomas Tüting, and Meri Rogova. Finally, we want to express our appreciation for Sylvie Méléard for animating this field of research relentlessly for more than twenty years. The present notes have some overlap with her books \cite{BansayeMel2015} and \cite{Meleard2016}.

\chapter*{Acknowledgements}\label{chap00}

This work was funded by the Deutsche Forschungsgemeinschaft (DFG, German Research Foundation) under Germany's Excellence Strategy - GZ 2047/1, Project-ID 390685813,  GZ 2151 - Project-ID 390873048,
 Project-ID 211504053 - SFB 1060,
and through the Priority Programme 1590 \emph{Probabilistic Structures in Evolution}.

\mainmatter
\part{Background and mathematical tools}
\chapter{Background and history}\label{chapter1}

\begin{chapquote}
{Erasmus Darwin, \emph{The Temple of Nature: Or, The Origin of Society}}
{\frakfamily\fraklines {By firm immutable immortal laws:\\
Impress'd on Nature by the Great First Cause,\\
Say, Muse! how rose from elemental strife\\
Organic forms:, and kindled into life;\\
How Love and Sympathy with potent charm\\
Warm the cold heart, the lifted hand disarm;\\
Allure with pleasures:, and alarm with pains:,\\
And bind Society in golden chains:.}  }  
\end{chapquote}

\def\Evol{\emph{Evolution} }

\index{evolution}
The  topic of these lecture notes is the mathematical understanding of the principal mechanisms 
of \emph{Evolution} based on elementary principles. It is the 
understanding of modern 
science that the bafflingly complex world of biology that we witness on our 
planet has emerged through 
an essentially chemical process which is termed \emph{Evolution}, 
essentially because certain 
molecules based  on carbon strings have the possibility to replicate and 
modify themselves. It 
appears out of  contemporary reach to understand in any detail how 
this could have happened, and how this 
could produce  amazing structures, such as even a bacterium, let alone 
a human being. 
Our ambition here is far more limited. We ask whether we can understand, starting from 
a set of elementary hypotheses
in  a simplified mathematical  model that contains what are believed to be 
the fundamental mechanisms 
characterising biological systems (or populations, as we will call them), 
how key features of \Evol 
emerge and how something like a complex, structured population can 
come about? 

The theory of \emph{Evolution} is usually traced back to Charles Darwin 
and his groundbreaking treatise
\emph{The Origin of Species} \cite{Darwin1859}, which was published in 
1859. There certainly have been
precursors of his theory (among others, by Darwin's grandfather, Erasmus 
Darwin, who wrote a long poem \emph{The Temple of Nature: Or, The 
Origin of Society.} \cite{erasmus1804}), see, e.g.
 the nice account in 
Siddhartha Mukherjee's brilliant book \emph{The Gene} \cite{mukharjee2016}.  In any case, it is probably fair to assert 
that it was Charles Darwin's work that laid the foundations of the 
modern scientific approach to 
the theory of \emph{Evolution}. 

The key features inherent in any biological system that are seen as the driving 
forces of \emph{Evolution}
 identified already by
Darwin are:
\begin{itemize}
\item \emph{birth and death}: individuals die and reproduce \index{birth}\index{death}
\item \emph{heredity}: the offspring of individuals inherit properties \index{heredity}
(\emph{traits}) of their ancestors
\item \emph{variation}: heredity is not perfect, and sometimes the traits of  \index{mutation} the offspring differ from those of their
ancestors
\item \emph{selection}, or \emph{the survival of the fittest}, a concept inspired by the  1798  essay \emph{An Essay on the Principle of Population} by Thomas Malthus \cite{malthus1798}.  \index{selection}
\end{itemize}

Selection results from the interaction between individuals, notably the competition for resources,
but also many other effects (predation, symbiosis, parasitism, etc.). \index{competition}

How heredity actually works and how variations arise have long remained mysterious, and conflicting opinions have existed. 
Jean-Baptiste  Lamarck \cite{lamarck}  thought that individuals change their traits 
gradually to adapt to changing environmental conditions, and that these changed traits could be inherited by their offspring. In fact, that animals change their traits if conditions change appears
a priori quite plausible and is observable. How these changes can be inherited appears much 
less plausible, but at the time may have seemed a necessary hypothesis. 
The molecular mechanism of heredity was discovered only much later to be rooted in a single molecule, 
the DNA, by Francis Crick and  James Watson \cite{crickwatson1953},
  Maurice Wilkins \cite{wilkins1953}, and Rosalind Franklin \cite{franklin1953}. This allowed a much better understanding of the details of how heredity and mutation come 
about; however, for the larger picture of how \emph{Evolution} works, this should not be too important, and for the larger part 
(but not all) of these lectures, we will ignore this. Our goal is to build mathematical models of 
populations that exhibit these key features and try to derive the fate of these populations.

In the remainder of this chapter, we give a brief overview of the classical (mathematical) theories of 
 \emph{Evolution}. They can be broadly \index{population dynamics}
classified into two branches: \emph{population dynamics}, which focuses on \emph{ecology}, i.e.\ on aspects 
of competition and other interactions between different species, and \emph{population genetics},
which focuses on heredity and the genealogical structure of populations.  \index{population genetics}

\section{Population dynamics}
Population dynamics can indeed be traced back to Malthus' essay \cite{malthus1798}, where 
he lays out that an unrestrained population will grow exponentially, but that in all real \emph{states}
this growth must be restrained by the limited amounts of food that are available.  In modern terms, this 
leads to a simple differential equation describing the time evolution of  the size $n(t)$ of a (monomorphic) 
population
\be\Eq(lv.1)
\frac d{dt}n(t) = r n(t) - cn(t)^2,
\ee
where $r=b-d$  is the difference between the birth-rate $b$ and the death-rate $d$, and $c$ is a measure 
of the competitive pressure two individuals exert on each other. 
Note that here $n$ is not the number of individuals (which would need to be an integer), but 
rather a rescaled measure for the mass of the population when the mass of an individual tends to zero 
while the number of individuals tends to infinity at the same rate. We will come to understand these
equations as limits of discrete, finite population models later on.

 For positive $r$ and positive $c$, 
this equation has two fixpoints, $0$ and $r/c$, where $0$ is unstable and $r/c$ is stable. Any population 
that at time $t=0$ has a positive size $r/c>n_0>0$ will grow eventually towards the size $r/c$, while 
a population with $n_0>r/c$ will shrink towards this same value.  

\begin{figure}[h]
\begin{center}
\includegraphics[width=0.6\textwidth]{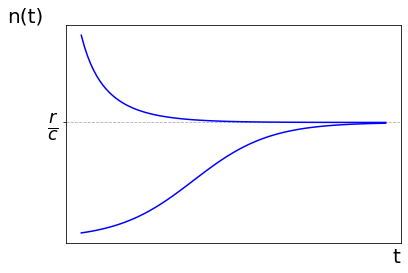}
\end{center}
\caption{Solutions $n(t)$ approaching $r/c$ from above or below depending on their initial population size $n_0$.}
\end{figure}

The analysis of differential equations of this type goes back to the work of Alfred Lotka 
\cite{lotka1912}
and Vito Volterra \cite{volterra1928}.  One calls systems of differential equations of the form 
\be\Eq(lv.2)
\frac d{dt}n_i(t) = n_i(t) \left(r_i - \sum_{j=1}^kc_{ij}n_j(t)\right),  \quad i=1,\dots,k,
\ee
\emph{competitive Lotka-Volterra equations} if all coefficients $c_{ij}$ are non-negative, and
simply Lotka-Volterra equations in the general case.  \index{Lotka-Volterra equations}
 \index{Lotka-Volterra equations!competitive}
 
A particularly famous example is the \emph{predator-prey system}, where $k=2$, $r_1>0$, $r_2<0$, $c_{12}>0$, $c_{2,1}<0$, and $c_{11}=c_{22}=0$. 
The interesting feature of the predator-prey system is that it admits periodic solutions.

\begin{figure}[h]
\begin{center}
\includegraphics[width=0.6\textwidth]{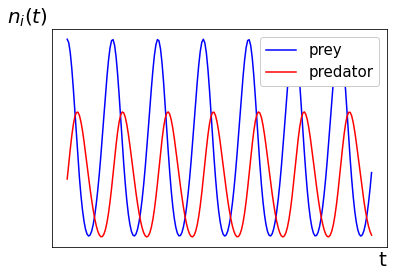}
\end{center}
\caption{Exemplary periodic solution of a two-dimensional predator-prey system.}
\end{figure}

It is interesting to note that, as shown by Smale \cite{smale76},
without further assumptions on the coefficients, the system \eqv(lv.2) can have arbitrarily complex long-time behaviour if $k\geq 5$.  In fact, for $k\geq 4$, the behaviour of Lotka-Volterra systems has not been fully analysed and seems out of reach. However, a complete analysis of three-dimensional \emph{competitive} Lotka-Volterra systems has been given by Zeeman \cite{Zee1993}.

Population dynamics has developed into a huge area. For a modern treatise, see Hofbauer and Sigmund \cite{hofbauersigmund}.  
Lotka-Volterra equations play an important role throughout this book, and we will come back to them in 
more detail in Chapter \ref{chapter3}.

\section{Population genetics} 

The roots of population genetics can be traced to the work of Gregor Mendel in  1865 and 
1869 \cite{mendel1865,mendel1869}. Published in an obscure Moravian journal, Mendel's 
path-breaking work remained unnoticed until 1900, when it was discovered by, among others, Hugo de Vries \cite{deVries1901}.
Mendel introduced the concept of \emph{genes} and \emph{alleles}, which make up the \emph{genotype} of an 
individual. \emph{Heredity} is due to the transfer of parental alleles to their offspring by some random 
mechanism. \Evol is the resulting change of the allele composition, or genotype, of the 
individuals in the population. The genotype of an individual then determines its appearance, or 
\emph{phenotype}.  \index{genotype}\index{phenotype} \index{heredity}

A key step towards a mathematisation of population genetics came about around 1920 with the work of 
Ronald Aylmer Fisher \cite{fisher18},  John Burdon Sanderson Haldane \cite{haldane24a,haldane24b}, 
and Sewall Green Wright
\cite{wright31}. A still widely used mathematical model in population genetics is the Wright-Fisher model.  In its simplest version, it goes as follows. A population consists of a fixed number, $N$, of 
individuals. For simplicity, we look only at a single gene to characterise its genotype.
We also assume that we are in a \emph{homozygotous} case, i.e.\ there is only a single \emph{allele} 
for this gene. Assume further that this allele can be of one of two types, $a$ or $b$. 
Such a population is assumed to evolve in discrete time (e.g. daily or annual cycles). When a 
new generation appears at time $t+1$, $t\in \N$, \index{gene} \index{allele} \index{homozygotous}
each individual picks one of the individuals present at time $t$ uniformly at 
random and adopts its allele type.

\begin{figure}[h]
\begin{center}
\includegraphics[width=0.5\textwidth]{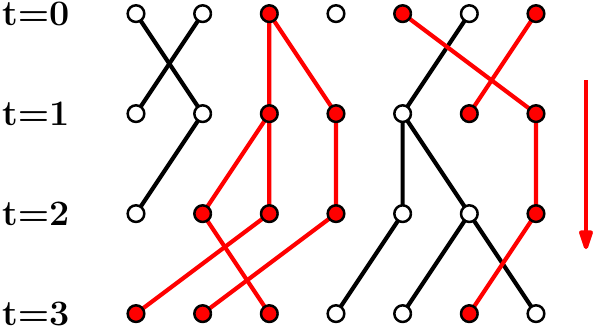}
\end{center}
\caption{Forward-in-time propagation of the $a$-allele (red).}
\end{figure}

Clearly, the number $N_a$ of $a$-type alleles in this model 
is  \index{Wright-Fischer model}
a discrete time Markov chain with transition rates
\be
\Eq(fisher.1)
\P(N_a(t+1)=n|N_a(t)=m) = { N\choose n} \left(\frac mN\right)^n\left(1-\frac mN\right)^{N-n}.
\ee
This allows for the computation of all kinds of interesting quantities, such as the distribution of the 
time until the allele $a$ is eliminated from the population. Even though the model is fully 
neutral, i.e.\ there is no selective advantage for any of the alleles, the distribution of the number of $a$-alleles
changes due to purely random fluctuations. This phenomenon is termed \emph{genetic drift} in 
population genetics, although it ought to be called genetic diffusion, and its discovery was a big deal. 
In fact, the number of $a$-alleles, $N_a(t)$, in this model is a martingale, i.e.  \index{genetic drift}
\be\Eq(fisher.2)
\E[N_a(t+1)|N_a(t)]=N_a(t).
\ee
This reflects the famous Hardy-Weinberg theorem of population genetics \cite{Hardy1908,Weinberg1908}, which states that in \index{Hardy-Weinberg theorem}
 infinite populations with no selection, the allele frequencies remain constant. 
 
 Also, the variance of one time step is 
 given by 
 \be\Eq(fisher.3)
 \E\left[(N_a(t+1)-N_a(t))^2\big |N_a(t)\right] = N \frac{N_a(t)}{N}\left(1-\frac{N_a(t)}{N}\right),
 \ee
 which means that the variance of the \emph{frequencies}, $N_a/N$, tends to zero like $1/N$. This suggests the existence of a non-trivial \emph{scaling limit}, namely 
 \be\Eq(fisher.4)
 \lim_{N\uparrow \infty} \frac {N_a([t N])-N_a(0)}{N} = X(t),
 \ee
 where now $t\in \R_+$ and $[tN]$ denotes the integer part of $tN$.
  $X$ is a diffusion process, called \emph{Wright-Fisher diffusion}, 
 which is a solution of the stochastic differential equation \index{diffusion!Wright-Fisher}
 \be
\Eq(fisher.5)
dX(t) =\sqrt {X(t)(1-X(t))} dB(t),
\ee 
where $B$ is standard Brownian motion. This model is the basic \emph{diffusion approximation } model 
in population genetics. It was first introduced by Motoo Kimura \cite{Kimura55} in 1955 and derived 
by Ethier and Norman \cite{EthNor77} in 1977.

\index{selection}
\emph{Selection} can be incorporated in the Wright-Fisher model by introducing a bias in the 
selection of the parent according to its allele type. This can be achieved by increasing the
probability to choose an $a$-allele as parent at time $t+1$ to $p(N_a(t))\equiv(1+s)(N_a(t))/(sN_a(t)+N)$, 
for some $s>0$, so that the transition rates become 
\be
\Eq(fisher.6)
\P(N_a(t+1)=n|N_a(t)=m) = { N\choose n} \left({p(m)}\right)^n\left(1-{p(m)}\right)^{N-n}.
\ee
Then 
\be\Eq(fisher.7)
\E[N_a(t+1)|N_a(t)]= p(N_a(t))N=N_a(t)\frac {1+s}{1+sN_a(t)/N},
\ee
which is larger than $N_a(t)$, as long as $0<N_a(t)<N$.
To obtain a non-trivial diffusion limit, one must assume a small selective effect, namely $s=r/N$. In that 
case, the diffusion approximation gives the stochastic differential equation
 \be
\Eq(fisher.8)
dX(t) = rX(t)(1-X(t))dt+\sqrt {X(t)(1-X(t))} dB(t).
\ee 
Notice that, if the selective advantage is larger than $O(1/N)$, then one obtains as a scaling limit 
(under appropriate time-rescaling) the deterministic equation
 \be
\Eq(fisher.8.1)
\frac {d}{dt}X(t) = rX(t)(1-X(t)).
\ee

 The Wright-Fisher model is, however, much richer than just the allele number process. The fact that 
 each individual alive at time $t$ chooses a parent at time $t-1$ implies a \emph{genealogical structure} 
 of the population in the past.
 {Tracing the genealogy of a subset of individuals at time $t$ back in time corresponds to the study of coalescing random walks. This yields a duality relation between the forward-in-time process $N_a(t)$ and the backwards-in-time coalescent that is very useful in the analysis of the system.}
 
\begin{figure}[h]
\begin{center}
\includegraphics[width=0.5\textwidth]{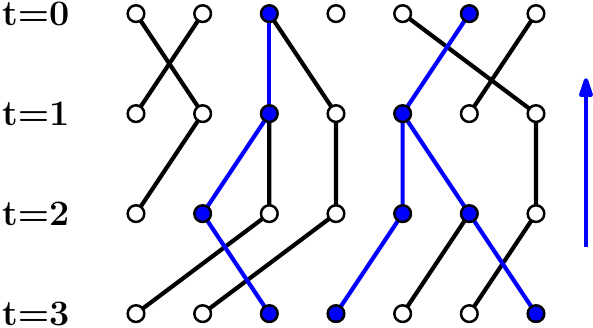}
\end{center}
\caption{Backwards-in-time tracing of (common) ancestries (blue).}
\end{figure}
 
 The Wright-Fisher models are discrete-time models with synchronous updates of generations. 
 The analogous continuous time model is called the \emph{Moran model} \cite{Moran1958}.
 Here, each individual chooses, after an exponential time, another individual from the current population uniformly at random and adopts the allele type of that individual. \index{Moran model}
The process $N_a(t)$ is then a continuous time Markov chain with transition rates 
\be\Eq(moran.1)
c(n,n\pm 1) = \frac nN\frac {N-n}N,\quad c(n,n)=-2 \frac nN\frac {N-n}N,
\ee
and all other rates are zero. Clearly, selection can be incorporated as in the Wright-Fisher model.

For a modern exposition of population genetics, see the monograph by Ewens \cite{ewens04}.

\section {The F-KPP equation} \index{F-KPP equation}
In 1937, Fisher  \cite{fisher1937} considered the problem of how an advantageous allele spreads in a spatially structured 
population. Taking as space  $\R$, he considered $u(t,x)$ as the frequency of the advantageous allele 
$a$ in the population at the location $x$ at time $t$. Assuming that the allele spreads in space in a diffusive 
way, he arrived at the partial differential equation 
\be\Eq(fisher.9)
\frac \del{\del t}u(t,x)= \frac {\s^2}{2}\frac {\del^2}{\del x^2} u(t,x) +r u(t,x)(1-u(t,x)).
\ee
The standardised version of this equation has $\s=r=1$.
The same equation was also proposed by Kolmogorov, Petrovsky, and Piscounov in the same year 
\cite{kpp} and is now known as the 
Fisher-Kolmogorov-Petrovsky-Piscounov, short F-KPP, equation. 
The equation can also be given a population dynamic interpretation. Consider a 
one-dimensional habitat and a species that has the same a priori positive fitness at any locus. 
Assume further
that individuals compete with each other only if they are in the same location. If these individuals move
diffusively in space, $u(t,x)$ can be interpreted as the mass of the population 
 at time $t$ at $x$. It is quite interesting to see that the same equation describes two very different 
 scenarios, with the term $u(1-u)$ being the selective drift in the population genetic interpretation and
 the logistic growth in the population dynamic one. 
 
 While in both interpretations, the scenario modelled does not appear very realistic, 
 Equation  \eqv(fisher.9) has the interesting feature that it admits \emph{travelling wave solutions} that describe
 the progression of a favourable allele into a population, or the advance of a population into accessible 
 territory.  
A travelling wave moving with speed $\l$ would be a solution, $u$, of the 
F-KPP equation such that
\be\Eq(fisher.10)
\frac d{dt} u(t,x+\l t) =0.
\ee
A simple computation shows that this implies that $u(t,x+\l t)=w_\l(x)$, where
\be\Eq(fisher.11)
\frac 12\frac{\del}{ \del {x}^2} w_\l (x) +\l \frac{\del}{\del x} w_\l(x) +w_\l(x)(1-w_\l(x))=0.
\ee
Fisher  \cite{fisher1937}, as well 
Kolmogorov and coauthors \cite{kpp},  and  later Uchiyama \cite{uchiyama} have shown that
such solutions exist for $\l\geq \sqrt 2$ and that they are unique up to translation. 
Bramson, in his seminal monograph \cite{bramson_monograph}, has shown, among other things,
that for a class of initial conditions that includes the Heaviside initial condition,\be\Eq(fisher.12)
u(t, x+m(t)) \to w_{\sqrt 2}(x), \quad t\uparrow\infty
\ee
uniformly in $x$, where
\be\Eq(fisher.13) 
m(t) = \sqrt{2} t -\frac{3}{2\sqrt{2}} \ln t.
\ee
This equation has sparked continued interest both in the pde and the stochastic community,
also due to its connection with branching Brownian motion (which is, however, not the subject of
 this book).
 
 Note that the deterministic F-KPP equation is the same in the population dynamics and population genetics 
 settings. However, this changes when the selection strength is lowered, and stochastic effects remain. In the 
 population genetic setting, this leads to the stochastic F-KPP equation in the form (see 
 Doering \cite{Doering2003})
 \be\Eq(fisher.14)
du(t,x)= \frac {\s^2}{2}\frac {\del^2}{\del x^2} u(t,x) dt+
r u(t,x)(1-u(t,x))dt +\sqrt {X(t)(1-X(t))} dB(t).
\ee
In the population dynamic setting, the equation has a different noise term 
and reads
 \be\Eq(fisher.15)
du(t,x)= \frac {\s^2}{2}\frac {\del^2}{\del x^2} u(t,x) dt+r u(t,x)(1-u(t,x))dt +
\sqrt {X(t)} dB(t).
\ee
The derivation of the noise terms from a microscopic model is, however, not rigorous and 
has been criticised, e.g.\ by Houchmandzadeh and Vallade \cite{houchmandzadeh}.
\index{F-KPP equation!stochastic}

\section{Adaptive dynamics} \index{adaptive dynamics}
\emph{Adaptive dynamics} is somewhat of an outgrowth of both population dynamics and population 
genetics \cite{MG96}.  Johan A. J. (Hans) Metz, one of its prominent protagonists, describes it in his essay \emph{Adaptive 
Dynamics} \cite{metz2012}
as ``a simplified theoretical approach to \emph{meso-evolution}, defined here as evolutionary changes in 
the values of traits of representative individuals and concomitant patterns of taxonomic diversification".
Further, ``Trait changes result
from the micro-evolutionary process of mutant substitutions taking place against the
backdrop of a genetic architecture and developmental system as deliverers of mutational
variation, internal selection keeping the machinery of a body in concert, and ecological
selection due to the interactions of individuals with their conspecifics, resources,
predators, parasites and diseases. Adaptive dynamics focuses on these encompassing mechanisms". 

An important assumption of adaptive dynamics, which we will encounter in the analysis of the mathematical models later 
on, is the separation of the time scales of ecology and evolution. This means, in essence, that the population reaches an equilibrium state due to ecological interaction before a new mutation arises. As Metz points out, ``In reality, this 
assumption holds good only rather rarely. The idea is that yet
arguments based on it may often lead to outcomes that are fair approximations, provided
one applies them selectively and takes a sufficiently gross look at reality". 
Note that this assumption implies an effectively low rate of (trait-changing, advantageous) mutations.

Adaptive dynamics thus deals with our fundamental objective, describing how populations (rather than individuals\index{trait}\index{trait space}\index{fitness}), characterised by some homogeneous traits, evolve in time into multifaceted families of 
populations exhibiting a wide variety of traits. A fundamental concept characterising 
this \emph{trait space} is \emph{fitness}. Fitness has different meanings. 
In population dynamics, fitness is the initial exponential growth rate of a population; in population 
genetics, it is often referred to as the probability of an individual to reach maturity, i.e.\ to produce 
offspring. In adaptive dynamics, fitness is viewed in a more \emph{dynamic} fashion. The fitness of a 
population with a specific trait does not only depend on this trait, but also on the state of the 
entire population, as effects of competition or other interactions play a significant role. 
{As a result, the fitness of a certain trait varies over time, as the composition of the population changes to adapt to its environment.}
In fact, if a population is in ecological equilibrium, then all co-existing traits have zero fitness (i.e.\ they 
do not grow or shrink). Under the assumption of separation of time scales, the 
only relevant fitness parameter is then the so-called \emph{invasion fitness}, \index{invasion fitness}
which is the initial exponential growth rate of a mutant born with a given trait in the presence of 
the current equilibrium population. Mutation rates and invasion fitness (and chance)
 then decide the further evolution of the population.
 
 A resulting fundamental concept of adaptive dynamics is  that of an \index{evolutionary stable condition}
 \emph{evolutionary stable
 condition} (ESC). This is a population in ecological equilibrium such that 
 all mutant traits have negative 
 invasion fitness. The fate of evolution is reached if a population is in 
 such a condition. 
 Clearly, for a given environment, there may be many ESCs. The task of 
 adaptive dynamics would 
 then be to identify all ESCs and to decide the ways to reach such ESCs. 
 On the way towards 
 ESCs, adaptive dynamics identifies two key mechanisms: the \index{canonical equation (CEAD)}
 \emph{canonical equation of 
 adaptive dynamics} (CEAD), which describes how a (monomorphic) 
 population moves in trait space 
 (under the further assumption of small mutation steps), and \index{adaptive speciation}
 \emph{adaptive speciation}, which describes how polymorphism emerges 
 as a population arrives
 at a state where several directions of mutations are viable and a bifurcation can happen.
 
 Many of the conclusions of the theory of adaptive dynamics have been derived in a mathematically rigorous way in the context of stochastic, individual-based models.
This was initiated in the works of Bolker and Pacala \cite{BolPac1,BolPac2}, Dieckman
and Law \cite{DieLaw,DL96}, and Fournier and Méléard \cite{FM04}. Another approach, 
coming from evolutionary game theory, is due to Hofbauer and Sigmund \cite{HS90}.

 \section{Adaptive walks}

 A certain variant of adaptive dynamics goes by the name of \emph{adaptive walks}. The concept was  \index{adaptive walks}
 introduced initially by  John 
 Maynard Smith \cite{Maynard62,Maynard70}. He emphasised the fact that mutations take place in the \index{protein space}
 space of genotypes, or, as he put it, \emph{protein space}. He pointed out that because mutations 
 are rare,  from a given gene encoding a protein, only relatively few neighbouring proteins are reachable, and 
 that
 among these, there must be a \emph{functional protein} if evolution is to continue from this point. Indeed, 
 such a protein should be \emph{fitter} than the original one, if the population is to shift to this mutant. 
 In his words ``It follows that if evolution by natural selection is to
occur, functional proteins must form a continuous network
which can be traversed by unit mutational steps without
passing through nonfunctional intermediates."

The ideas of Maynard Smith were largely ignored until the late
1980s when Kauffman and Levine \cite{KaLe87,Kau92} started a systematic study of what was then called 
\emph{adaptive walks}. This went along with the introduction of various random fitness landscapes, notably
Kauffman's NK models, which provided a framework for analytic and numerical studies of key properties of such 
walks, which is still ongoing. 
The general assumption of adaptive walks is that steps are only possible to nearest neighbours and only 
to sites of at least equal fitness. Different choices are considered when it comes to the specifics of
the probabilities with which possible neighbours of higher fitness are selected. Two popular and important examples are \textit{natural} adaptive walks, where transition probabilities are proportional to the increase in fitness, and \textit{greedy} adaptive walks, which always transition to the fittest neighbouring trait.

    Important objects of study are the accessibility of certain traits \cite{SchKr14,BeBruShi16,BeBruShi17}, the distribution of local fitness maxima \cite{NoKr15}, and the average length of paths before reaching such a local maximum \cite{KrugKarl03,Orr03}. These depend on two sources of randomness - the random fitness landscape and the random transitions between traits.

An adaptive walk comes to a halt when a site is reached whose fitness is higher than all of its neighbours, i.e.\ that constitutes a local fitness maximum.
However, already Maynard Smith asked ``How often, if ever, has evolution passed through a
non-functional sequence? If so, has this been achieved
by the random walk of genes rendered redundant by
duplication, or by the chance concurrence of two or more
mutations?" \cite{Maynard70}. More specifically, Gillespie \cite{gillespie1984molecular}
considered the problem of the crossing of a fitness valley: ``Obviously, given enough time, the double mutant will eventually
appear and will become fixed in the population.
We require an estimate of the
time required for this event to occur in order to fully appreciate the extent of the stagnation."
 He then presented estimates of the mean time for this to happen. 

An extension of the classical adaptive walks, in which transition probabilities are determined by balancing a high increase in fitness and a low number of mutation steps, allows for longer jumps that can skip over such fitness valleys \cite{JainKrug05,JainKrug07,Jain07}. More recently, the concept of \textit{adaptive flights} has been introduced, where the random walk transitions directly between local fitness maxima \cite{NeidKrug11}. An extensive overview of current results on adaptive walks and flight
can be found in Joachim Krug's article \cite{Krug2021}.

\chapter{Mathematical preliminaries}\label{chapter2}

\begin{chapquote}
{Thomas Malthus, \emph{Essay on the Principle of Population}}
{\frakfamily\fraklines {
Assuming then, my postulate as: granted,\\
I say, that the power of population is:\\
indefinitely greater than the power in the\\
earth to produce subsistence for man.}}
\end{chapquote}

\begin{chapquote}
{Charles Darwin, \emph{The Origin of Species}}
{Although some species may be now increasing, more or less rapidly, \\ in numbers, all cannot do so, for the world would not hold them.}
\end{chapquote}

This chapter provides the basic mathematical tools needed in this book. Although we try to be reasonably self-contained,
we must assume some background in probability theory and, in particular, the 
theory of Markov processes. For this, we recommend the textbooks by Rogers and Williamson \cite{Rogers1, Rogers2} and 
Ethier and Kurtz \cite{EthKur1986}.

\section{The individual based model of adaptive dynamics} \index{individual based models}
In this section, we introduce the basic stochastic, individual-based model of adaptive dynamics. This model was 
introduced by Bolker and Pacala \cite{BolPac1,BolPac2} and Dieckmann and Law \cite{DieLaw}. 
The rigorous analysis of the model was initiated by Méléard and collaborators, see, e.g. \cite{FM04,Cha06,C_ME,CM11}.      
Note that the notion of \emph{individual-based models} in a wider sense was introduced in ecology
    and is widely used there. For a historic overview, see the book by Grimm and Railsback \cite{Grimm2005} or the more recent review
by Grimm and DeAngelis \cite{DeAngelisGrimm}.

We consider a population of a single asexual species that is composed of a finite number of individuals, each of them characterised by a phenotypic trait that takes values in some space $ \XX$, called the  \textit{trait space}. $ \XX$ is, 
in the most general setting, an arbitrary Polish space.  \index{trait space}
\emph{Evolution} will take place in this space. All of biology is then encoded in the following functions:

\begin{enumerate}[label=(\roman*)]
\setlength{\itemsep}{6pt}
\item{$b(x)\in\mathbb R_+$ is the \textit{rate of birth} of an individual with trait $x\in \XX$.}
\item{$d(x)\in\mathbb R_+$ is the \textit{rate of natural death} of an individual with trait $x\in \XX$.}
\item{$c(x,y)\in\mathbb R_+$ is the \textit{competition kernel} \index{competition kernel}
		which models the competitive pressure felt by an individual with trait $x\in \XX$ from an individual with trait $y\in \XX$.}
\item{ $ p(x)\in [0,1]$ is the \textit{probability that a mutation occurs at birth} from an individual with trait $x\in \XX$.}
    \item{$m(x,dy)$ is the \textit{mutation law}. If the mutant is born from an individual with trait $x$, then the mutant trait is given by $y\in  \XX$, where $y$ is a random variable with law $m(x,dy)$.  }
\end{enumerate}
\nomenclature{$b(x)$}{birth rate}%
\nomenclature{$d(x)$}{death rate}%
\nomenclature{$c(x,y)$}{competition kernel}%
\nomenclature{$p(x)$}{mutation probability}%
\nomenclature{$m(x,dy)$}{mutation law}%
\nomenclature{$ \XX$}{trait space}%
\nomenclature{$\mathbb R_+$}{non-negative real numbers}%
\begin{remark} We consider the events listed above as the minimal features to obtain a model that  
can reflect the main features of an evolving population. We will stick with this through the bulk of this book. One may, if desired, 
add further features to the model, and, in Part 3, we will look at some examples.
\end{remark}

At any time $t$, we consider a finite number, $N_t$, of individuals, each of them having  a trait value 
$x_i(t)\in  \XX$. 
\nomenclature{$N_t$}{population size at time $t$}%
It is convenient to represent the population state at time $t$ by a point measure,  
\index{point measure}
$\nu_t$,
\be
 	\nu _t=  \sum_{i=1}^{N_t}\delta_{x_i(t)}.
 \ee 
 \nomenclature{$x_i(t)$}{trait of an individual $i$ at time $t$}%
 \nomenclature{$\nu_t$}{point measure representing the population at time $t$}%
Let $\langle \nu , f \rangle$ denote the integral of a measurable function $f$ with respect to the measure $\nu $. 
Then $\langle\nu _t,\1 \rangle=N_t $ and for any $x\in \XX$, 
the non-negative number $\langle\nu _t,\1_{\{x\}}\rangle$ is called the \textit{density of trait $x$ at time $t$}. 
With this notation, an individual with trait $x$ in the population $\nu _t$ dies due to age or competition with rate
\nomenclature{$\langle \nu,f\rangle$}{integral of the function $f$ with respect to the measure $\nu$}%
\nomenclature{$\d_x$}{Dirac measure at $x$}%
\be
 d(x)+\int_{ \XX}\ c(x,y)\nu_t(dy).
\ee
Let 
\be
	 \MM( \XX)\equiv\left\{ \sum_{i=1}^{n}\delta_{x_i}\,:\, n\geq 0,\; x_1,...,x_n\in  \XX\right\},
 \ee
 \nomenclature{$ \MM( \XX)$}{space of point measures on $ \XX$}%
 denote the set of finite, nonnegative point measures on $ \XX$, equipped with the vague topology, 
 The population process, $(\nu_t)_{t\geq0}$, is a 
$ \MM( \XX)$-valued Markov process with infinitesimal generator, $\LL$, 
defined, for any bounded measurable function $f$ from $ \MM( \XX)$ to $\mathbb R$ and for all $\nu\in  \MM( \XX)$, by \index{generator} \index{Markov process}
\nomenclature{$\LL$}{generator of a Markopv process}%
\bea\Eq(generator.1)
 	\nonumber
(\LL f)(\nu) &=& \int_{ \XX}\biggl(f\Bigl(\nu + {\delta_x}\Bigr)-f(\nu )\biggr)\bigl(1- p(x)\bigr)b(x) \nu (dx)\\ \nonumber
						&&+\int_{ \XX}\int_{\XX}\biggl(f\Bigl(\nu+{\delta_y} \Bigr)-f(\nu )\biggr) 	
							p(x)b(x) m(x,dy)  \nu (dx)\\
						&&+\int_{ \XX}\biggl(f\Bigl(\nu -{\delta_x} \Bigr)-f(\nu )\biggr)\Bigl(d(x)+
						\int_{ \XX}c(x,y)\nu (dy)\Bigr) \nu (dx).
\eea
The first and second terms are linear (in $\nu$) and describe the births (without and with mutation), 
but the third term is non-linear and describes the deaths due to age or competition. 
The density-dependent non-linearity of the third term models the competition in the population, 
and hence drives the selection process.
\begin{assumption}\label{ass} We will make the following assumptions on the parameters of the model:
\begin{enumerate}
\renewcommand{\labelenumi}{(\roman{enumi})}
\item $b$, $d$ and $c$ are measurable functions, and there exist $\overline b,\overline d,\overline c<\infty$ such that
\begin{center}$0\leq b(.)\leq \overline b,\quad 0\leq d(.)\leq \overline d\quad$ and $\quad 0\leq c(. ,.)\leq \overline c.$\end{center}
\item{
There exists $\underline c>0$ such that, 
for all $x\in \XX$,
$\underline c\leq c(x,x)$.}\vspace{0,5em}
\item{The support of $m(x,\:.\:)$ is a subset of $ \XX\setminus \{x\}$ and uniformly bounded for all $x\in  \XX$. }
%
\end{enumerate} 
\end{assumption}
Assumptions (i) and (iii) allow us to deduce the existence and uniqueness in law of a process on $\mathbb D(\mathbb R_+, \MM( \XX))$ with infinitesimal generator $ \LL$ (cf. \cite{FM04}). 
Assumption (ii) prevents the population from exploding or becoming extinct too fast. 

\section{Construction of the process}

The continuous-time process corresponding to the generator $\LL$ can be constructed fairly explicitly as a \index{jump process} \index{Markov process!construction}\index{jump process}
Markov jump process. To do this, we first construct a discrete-time Markov process  $\tilde \nu_n, n\in \N$ with stationary transition probabilities (see, e.g. \cite{EthKur1986}). These are given 
as follows. 
\index{transition probability}
If the process is in a state $\tilde \nu= \sum_{i=1}^n\d_{x_i}$, then it can move to the following states in one step:
\bea
\tilde \nu^{-i}&\equiv&\tilde \nu -\d_{x_i},\quad i=1,\dots n,\\
\tilde \nu^{+i}&\equiv&\tilde \nu+\d_{x_i}, \quad i=1,\dots n,\\
\tilde \nu^{i,z}&\equiv&\tilde \nu +\d_{z},\quad i=1,\dots n,\; z\in \supp(m(x_i,\cdot).
\eea
To compute the respective probabilities, we proceed as follows
\begin{itemize}
\item[(i)]  For  $i\in\{1,\dots,n\}$ compute 
\be
R(i)\equiv b(x_i)+ d(x_i)+\int_{ \XX}c(x_i,y)\tilde\nu (dy)
\ee
and $R(\tilde\nu)\equiv\sum_{i=1}^n R(i)$. 
\item[(ii)] Choose  the $i$ where a move is made with probability $R(i)/R(\tilde \nu)$. 
\item[(iii)] Conditional on having chosen $i$, $\tilde \nu$ goes to 
\bea
&&\tilde \nu^{-i},\; \text {with probability}\;  \frac{ d(x_i)+\int_{ \XX}c(x_i,y)\tilde\nu (dy)}{R(i)},\\
&&\tilde \nu^{+i},\; \text {with probability}\;  \frac{b(x_i)(1-p(x_i))}{R(i)},\\
&&\tilde \nu^{i,z},\; \text {with probability} \; \frac{b(x_i)p(x_i)}{R(i)}  m(x_i,dz).
\eea
\end{itemize}
Then the transition probabilities of all possible moves are given by
\bea
P(\tilde \nu, \tilde \nu^{-i})&=& \frac {R(i)}{R(\tilde\nu)} \frac{ d(x_i)+\int_{ \XX}c(x_i,y)\tilde\nu (dy)}{R(i)},\\
P(\tilde \nu, \tilde \nu^{+i})&=& \frac {R(i)}{R(\tilde\nu)}  \frac{b(x_i)(1-p(x_i))}{R(i)},\\
 P(\tilde \nu, \tilde \nu^{i,+z})&=& \frac {R(i)}{R(\tilde\nu)}\frac{b(x_i)p(x_i)}{R(i)}  m(x_i,dz),
 \eea
 for $i=1,\dots, n\equiv \tilde \nu(\1)$, and all other transitions have probability zero.
Together with the initial distribution of $\tilde \nu_0$,  this fixes the discrete-time process $\tilde \nu_k, k\in \N_0$.

Next, we define the clock process \index{clock process}
\be
S(k) \equiv \sum_{\ell=0}^{k-1} e_\ell /R(\tilde \nu_\ell),
\ee
where $e_\ell, \ell\in\N_0$ are iid standard exponential random variables.
Then define  the \emph{time change} of $\tilde\nu$ by $S$, namely,
\be
\nu_t=\tilde \nu_{S^{-1}(t)}.
\ee
Here $S^{-1}$ is the right-continuous inverse of $S$, i.e.
\be\Eq(inverse-rc)
S^{-1}(t)\equiv \sup\left(k:S(k)\leq t\right).
\ee
$\nu_t$ is then a continuous-time Markov process with generator \index{generator}
\be \Eq(disc-gen)
(\LL f) (\nu) =R(\nu) ((P-\1)f)(\nu).
\ee 
Existence of this process is obvious if $R$ is bounded from above. Under Assumptions \thv(ass), 
this is guaranteed, as long as the total population size remains finite. We show this later in Lemma 
\thv(total.1).

The construction of the process given above is also the recipe for the numerical simulation of the process
by what is called a Gillespie algorithm \cite{Gillespie1976}.  \index{Gillespie algorithm}
There is an alternative popular way to construct the process using Poisson point processes (see \cite{FM04}).

A useful tool to analyse this process is given by the martingale problem formulation. \index{martingale problem}
 For a detailed exposition of the theory, see the textbook by Ethier and Kurtz \cite{EthKur1986}. 
%

\begin{definition} \TH(martingale.1)
A stochastic process $\nu$ with state space $\MM(\XX)$ and continuous time is called a solution of the martingale problem associated to $\LL$ if, for any function  $f:\MM(X)\rightarrow \R$ in the domain of $\LL$,
\be\Eq(martingale.2)
M^{f,\nu}_t\equiv f(\nu_t)-f(\nu_0)-\int_0^t (\LL f)(\nu_s)ds
\ee
is a martingale. 
\end{definition}

\begin{remark} For our purposes, it is enough to know that 
the Markov process generated by $\LL$ that we constructed before is a solution to the martingale problem for $\LL$. 
This can be verified quite easily. 
\end{remark}

In the sequel, we will be interested in particular in linear functions of the form
\be
f(\nu)=\int\nu(dx)h(x)\equiv \langle\nu,h\rangle,
\ee
 and their martingale properties. 
To shorten the notation, we write 
\be\Eq(martingale.10)
(\LL f)(\nu) \equiv \int \nu(dx) \sumint c^\pm(x,dy,i;\nu) \left(f\left(\nu\pm \d_y\right)-f(\nu)\right),
\ee
where $\sum\hspace{-1em}\int$ is a sum of terms involving integrals over $dy$ and the coefficients may depend on
 $\nu$. Then 
 \be\Eq(martingale.11)
(\LL  \langle\nu,h\rangle) \equiv \int \nu(dx) \sumint c^\pm(x,dy,i;\nu) (\pm h (y)), 
\ee
and 
\be\Eq(martingale.12)
(\LL  \langle\nu,h\rangle^2) \equiv \int \nu(dx) \sumint c^\pm(x,dy,i;\nu)
\left( h(y)^2\pm 2h(y)\langle \nu,h\rangle\right). 
\ee
Associated to $\langle \nu_t,h\rangle$ is the martingale 
\be \Eq(martingale.13)
M_{1,t}^h\equiv \langle\nu_t,h\rangle-\langle\nu_0,h\rangle-
\int_0^tds \int \nu_s(dx) \sumint c^\pm(x,dy,i;\nu_s) (\pm h (y)),
\ee
which is the pure fluctuation part of this process.

\begin{remark} Note that the measure $\nu_t$ is determined by its values on all bounded measurable functions
of the form 
\be f(\nu)=\langle \nu,h_1\rangle\times\dots\times\langle \nu,h_n\rangle,
\ee
$n\in\N$ and $h_i$ as above. 
 From this, one can deduce that the martingale problem restricted to this set of functions is sufficient to uniquely determine the process $\nu$ (for details see, e.g., the lecture notes by Dawson \cite{Dawson2017}).
\end{remark}

\index{martingale}
\index{increasing process}
\index{bracket!of martingale}
To get control of this martingale, we compute its bracket. The bracket of a martingale $M_t$ is the unique increasing process 
$[M]_t$, such that $M_t^2-[M]_t$ is again a martingale (see \cite{EthKur1986}). 
 This means that we must decompose the square 
of $M_{1,t}^h$ into a martingale and an increasing process. 
 To do so, we note that there is a natural martingale 
\bea \Eq(martingale.14)
M_{2,t}^h&\equiv& \langle\nu_t,h\rangle^2-\langle\nu_0,h\rangle^2\\\nonumber
&&-\int_0^t ds
 \int \nu_s(dx) \sumint c^\pm(x,dy,i;\nu_s)
\left( h(y)^2\pm  2h(y)\langle \nu_s,h\rangle\right). 
\eea
Writing  out the square of $M_{1,t}^h$ gives
\bea\Eq(martingale.15)\nonumber
\left(M_{1,t}^h\right)^2&=& 
 \langle\nu_t,h\rangle^2+\langle\nu_0,h\rangle^2+
\left(\int_0^tds \int \nu_s(dx) \sumint c^\pm(x,dy,i;\nu_s) (\pm h (y))\right)^2\\
&&- 2\langle\nu_t,h\rangle\langle\nu_0,h\rangle-
2\langle\nu_t,h\rangle \int_0^tds \int \nu_s(dx) \sumint c^\pm(x,dy,i;\nu_s) (\pm h(y))
\nonumber\\
&&+2\langle\nu_0,h\rangle \int_0^tds \int \nu_s(dx) \sumint c^\pm(x,dy,i;\nu_s) (\pm h(y)).
\eea
Clearly, we hope to discover the two martingales that we already know in this expression. We see that, 
compared to $M_{2,t}^h$, the $\langle\nu_0,h\rangle^2$ has the wrong sign, so we add and subtract it twice
and slightly rearrange terms to get
\bea\Eq(martingale.16)\nonumber
&&\left(M_{1,t}^h\right)^2
=\langle\nu_t,h\rangle^2-\langle\nu_0,h\rangle^2\\\nonumber
&&+2 \langle\nu_0,h\rangle^2 -2\langle\nu_0,h\rangle\langle\nu_t,h\rangle
+2\langle\nu_0,h\rangle\int_0^tds \int \nu_s(dx) \sumint c^\pm(x,dy,i;\nu_s) (\pm h(y))\\\nonumber
&&-2\langle\nu_t,h\rangle\int_0^tds \int \nu_s(dx) \sumint c^\pm(x,dy,i;\nu_s) (\pm h(y))
\\
&&+\left(\int_0^tds \int \nu_s(dx) \sumint c^\pm(x,dy,i;\nu_s) (\pm h(y))\right)^2.
\eea
The terms in the second line are a martingale, $-2\langle\nu_0,h\rangle M_{1,t}^h$. Next, 
we want to write the first and third line as $M_{2,t}^h$ plus something. We see that this involves 
a term with $h(y)^2$ in \eqv(martingale.14), that has to be added, and the following term,  that has to be 
subtracted:
\bea\Eq(martingale.17)
&&2\int_0^tds   \left(\langle\nu_t,h\rangle-\langle\nu_s,h\rangle\right) \int   \nu_s(dx) \sumint 
c^\pm(x,dy,i;\nu_s) (\pm h(y))\nonumber\\
&&=\int_0^tds  \int_s^t   dr  \int   \nu_r(dx) \sumint c^\pm(x,dy,i;\nu_r) (\pm h(y))   \nonumber\\
&&\quad\times 2 \int   \nu_s(dx) \sumint 
c^\pm(x,dy,i;\nu_s) (\pm h(y))\nonumber\\
&&\quad+ 
2\int_0^tds    \left(M^h_{1,t}-M^h_{1,s}\right)   \int   \nu_s(dx) \sumint c^\pm(x,dy,i;\nu_s) (\pm h(y))\nonumber
\\ 
&&=  \left(\int_0^tds \int \nu_s(dx) \sumint c^\pm(x,dy,i;\nu_s) (\pm h(y))\right)^2\nonumber\\
&&\quad+ 2
\int_0^tds    \left(M^h_{1,t}-M^h_{1,s}\right)   \int   \nu_s(dx) \sumint c^\pm(x,dy,i;\nu_s) (\pm h(y)), 
\eea
where the last term is also a martingale. The square in the one-but-last line conveniently 
cancels the square in the last line of \eqv(martingale.16). 
Therefore, we arrive at 
\bea\Eq(martingale.18)\nonumber
\left(M_{1,t}^h\right)^2 &=& M_{2,t}^h+ <\nu_0,h> M^h_{1,t})\\\nonumber
&-& 
\int_0^tds    \left(M^h_{1,t}-M^h_{1,s}\right)   \int   \nu_s(dx) \sumint c^\pm(x,dy,i;\nu_s) (\pm h(y))
\\
&+&
\int_0^t ds
 \int \nu_s(dx) \sumint c^\pm(x,dy,i;\nu_s)h(y)^2. 
\eea
With the terms in the first two lines being martingales and the term in the last line being an increasing process,
we have achieved  our goal and obtained that
%
\be\Eq(martingale.20)
\left[M_1^h\right]_t=\int_0^t ds
 \int \nu_s(dx) \sumint c^\pm(x,dy,i;\nu_s) h(y)^2,
\ee 
or, explicitly,
\bea
\Eq(martingale.21)
\left[M_1^h\right]_t&=&\int_0^t ds
 \int \nu_s(dx)\Biggl( \left(b(x)(1-p(x))+d(x)+\int \nu_s(dy) c(x,y) \right)h(x)^2 \nonumber\\
 &&+b(x)p(x)\left(\int m(x,dy) 
 h(y)^2\right)\Biggr).
 \eea

We can use these observations to show that the total size of the population does not explode. 

\begin{lemma} \TH(total.1)
Assume that $\E \left[\langle \nu_0,\1\rangle^2\right] <\infty$. 
Then, for any $T<\infty$, 
\be\Eq(total.2)
\E \left[\sup_{t\leq T}\langle \nu_t,\1\rangle^2\right] <\infty.
\ee 
\end{lemma}

\begin{proof}
We have that
\be
\langle \nu_t,\1\rangle=\langle \nu_0,\1\rangle
+\int_0^t \LL\langle\nu_s,\1\rangle ds +M^{\1}_{1,t}.
\ee
Now
\be\Eq(total.3)
\LL\langle\nu,\1\rangle =
\int \nu(dx) \left(b(x)-d(x)- \int c(x,y)\nu(dy)\right)
\leq \max_x (b(x)-d(x)) \langle\nu,\1\rangle.
\ee
Thus, setting $\bar r\equiv \max_x (b(x)-d(x))$,
\be
\langle \nu_t,\1\rangle\leq \langle \nu_0,\1\rangle
+\bar r\int_0^t \langle\nu_s,\1\rangle ds +M^{\1}_{1,t}.
\ee
Since $\langle \nu_t,\1\rangle\geq 0$, it follows that
\bea\nonumber
\langle \nu_t,\1\rangle^2&\leq& 4\langle \nu_0,\1\rangle^2
+4\bar r^2\left(\int_0^t \langle\nu_s,\1\rangle ds\right)^2 +4(M^{\1}_{1,t})^2
\\&\leq&4\langle \nu_0,\1\rangle^2
+4t\bar r^2\int_0^t \langle\nu_s,\1\rangle^2ds +4(M^{\1}_{1,t})^2.
\eea
Next,
\be
\sup_{s\leq t}\langle \nu_s,\1\rangle^2\leq 4\langle \nu_0,\1\rangle^2
+4t\bar r^2\int_0^t\sup_{u\leq s} \langle\nu_u,\1\rangle^2ds +4 \sup_{s\leq t} (M^{\1}_{1,s})^2.
\ee
Thus, for $t\leq T$,
\be
\E\left[\sup_{s\leq t}\langle \nu_s,\1\rangle^2\right]\leq 4\E\left[\langle \nu_0,\1\rangle^2\right]
+4T\bar r^2\int_0^t \E\left[\sup_{u\leq s} \langle\nu_u,\1\rangle^2\right]ds +4
\E\left[ \sup_{s\leq t} (M^{\1}_{1,s})^2\right].
\ee
Finally, by Doob's $L^2$-inequality and \eqv(martingale.20)
\bea
\E\left[ \left( \sup_{s\leq t}M^{\1}_{1,s}\right)^2\right]
&\leq& 2 \sup_{s\leq t}
\E[M^{\1}_1]_s\leq C \int_0^t \E \langle \nu_s,\1\rangle^2ds\nonumber\\
&\leq& C \int_0^t \E\left[\sup_{u\leq s} \langle\nu_u,\1\rangle^2\right]ds,
\eea
for some constant $C<\infty$.
Thus, we get, for all $t\leq T$,
\be
\E\left[\sup_{s\leq t}\langle \nu_s,\1\rangle^2\right]\leq 4\E\left[\langle \nu_0,\1\rangle^2\right]
+(4T\bar r^2+4C) \int_0^t \E\left[\sup_{u\leq s} \langle\nu_u,\1\rangle^2\right]ds.
\ee
%
The assertion of the lemma now follows by applying Gronwall's inequality. \index{Gronwall's inequality}
\end{proof}

\begin{remark} 
The same result holds if $2$ is replaced by $p>2$, using the Burkholder-Davis-Gundy inequality (see \cite{Rogers1}) to control the
expectation of the $p$-th power of $M^{\1}_{1,t}$. 
\end{remark}

The bound on the total number of individuals given in Lemma \thv(total.1) implies that, for any time $T<\infty$, the rates $R(\nu_t)$ remain almost surely finite and hence the process can  be constructed as 
indicated above.

\section{Scaling} All the beauty of the model described above arises from the fact that the rates of the different 
elementary processes that take place can be scaled. 
The three scaling  parameters of the model are 
 the \emph{population size}, controlled by the scaling parameter $K$, \nomenclature{$K$}{carrying capacity}
the \emph{mutation probability}, controlled by the scaling parameter $\mu$, 
\nomenclature{$\mu$}{mutation rate}
the \emph{ mutation size}, controlled by the scaling parameter $\sigma$.
\nomenclature{$\s$}{mutation step size}
These scaling parameters are implemented by:
\begin{itemize}
\item[(i)] replacing $c(x,y)$ by $K^{-1} c(x,y)$. This implies that to change the death rate by an amount of 
order one requites the presence of $O(K)$ individuals, so the population size can reach of order $K$ 
individuals. This requires also consider the rescaled point measures $K^{-1} \sum_{i=1}^{n_t} \d_{x_i(t)}$,
and to consider processes with state spaces
\be\Eq(space.1)
	 \MM^K( \XX)\equiv\left\{ \frac 1K\sum_{i=1}^{n}\delta_{x_i}\,:\, n\geq 0,\; x_1,...,x_n\in  \XX\right\}.
 \ee
 $K$ is called the \emph{carrying capacity};
\item [(ii)] replacing the mutation rate $p(x)$ by $\mu p(x)$;
\item [(iii)] replacing the mutation term  $ \int m(x,dy) f(y)$ by a term
\be\Eq(mutation-scaled)
\int  m_\s(x,dh) f(x+h)=\int  m(x,dh)(f(x+\s h),
\ee
 for some probability kernel $m$.  Of course, this requires assuming some linear structure on the trait space.
 \end{itemize} 

It is often biologically reasonable to assume that $K$ is large, and both $\mu$ and $\d$ are small. We will 
see that a variety of limits can be constructed on different time scales. These limits will be seen to be 
related to some of the key mechanisms of adaptive dynamics. 

The generator with the scaling parameters takes the form
\bea\Eq(generator-K.1) 
\left(\LL^K f\right)(\nu^K )&=& \int_{ \XX}\biggl(f\Bigl(\nu^K +\sfrac {\delta_x}{K}\Bigr)-f(\nu^K )\biggr)\bigl(1- \mu p(x)\bigr)b(x) K\nu^K (dx)\\ \nonumber
						&+&\int_{ \XX}\int_{ \XX}\biggl(f\Bigl(\nu^K +\sfrac{\delta_{x+\s y}} K\Bigr)-f(\nu^K )\biggr)	
				\:	\mu		p(x)b(x)  m_\s(x,dy)  K\nu^K (dx)\\\nonumber
						&+&\int_{ \XX}\biggl(f\Bigl(\nu^K -\sfrac {\delta_x} K\Bigr)-f(\nu^K )\biggr)\left(d(x)+
						\int_{ \XX}c(x,y)\nu^K (dy)\right)K\nu^K (dx).
\eea
Notice that with the notation introduced in \eqv(martingale.10), this can be cast into the general form 
\be\Eq(generator-K.2)
(\LL^K f)(\nu^K) \equiv \int \nu^K(dx) \sumint  Kc^\pm(x,dy,i;\nu^K) \left(f\left(\nu^K\pm K^{-1}\d_{x+\s y}\right)-f(\nu^K)\right),
\ee
where $c^\pm$ depends on $u$ and $\s$, but not explicitly on $K$. Thus, the scaled processes are accelerated by a factor $K$ while the step-size is scaled down by a factor $1/K$. This suggests 
that a law of large numbers is to be expected.

\section {The law of large numbers}\label{known_results}
The first limiting result we discuss is the \emph{law of large numbers}, which emerges as
the size of the population tends to infinity while all other parameters, including the time horizon, are kept finite.
This fundamental result connects the stochastic model discussed here to the deterministic models, such as the 
Lotka-Volterra equations, which we have already seen in Chapter \ref{chapter1}. In this context, the law of large numbers goes back to
Tom Kurts \cite{Kurtz1970,Kurtz1971}, see also the exposition in Chapters 10 and 11 of  \cite{EthKur1986}. A different proof in a slightly more general setting was given by  N. Fournier and S. M\'el\'eard \cite{FM04}. We largely follow their approach.

\subsection{The general case}
We start with a theorem which describes the behaviour of the 
population process, for fixed $\mu$ and $\sigma$, when $K\uparrow\infty$. To simplify notation, we drop $\mu$ and $\s$ for the time being. We denote by $\MM_F(\XX)$ the space of finite measures on $\XX$.
\begin{theorem}[Theorem 5.3 in \cite{FM04}] \TH(lil.1)
Fix $\mu$ and $\sigma$.
Let Assumption \ref{ass} hold and assume in addition that 
the initial conditions $\nu_0^K$ converge,  as $K\uparrow\infty$, in law and for the weak topology on 
$ \MM_F( \XX)$ to some deterministic finite measure $\xi_0\in  \MM_FF( \XX)$ 
and that $\sup_K \mathbb E\left [ \langle \nu^{K}_0, \1\rangle^2\right] <\infty$. 
Then for all $T > 0$, the sequence $\nu^K$ converges, as $K\uparrow\infty$, in law, 
in  $\mathbb D([0,T ], \MM_F (  \XX))$, to a deterministic continuous function 
$\xi\!\in\! C([0,T ], \MM_F (  \XX))$.
This measure-valued function $\xi$ is the unique solution, 
satisfying $\sup_{t\in[0,T ]}\langle \xi_t, \1\rangle \!<\infty $, 
of the integro-differential equation written in its weak form:
for all bounded and measurable functions, $h: \XX \to\mathbb R$,
\bea\Eq(lln.1)
&&\int_{ \XX} \xi_t(dx)h(x)
	-\int_{ \XX} \xi_0(dx)h(x) \\\nonumber
		&&=\!\int_0^t ds \int_{ \XX }\xi_s(dx) p(x)b(x) \int_{\mathbb Z} m(x,dy) h(y)
						  \\ \nonumber
	&&+\!\int_0^t  ds \int_{ \XX }\xi_s(dx) h(x)\Big(\left(1\!-\! p(x)\right)b(x)
							\!-\!d(x)\!-\!\int_{ \XX}\xi_s(dy)c(x,y) \Big).
\eea
\end{theorem}

\begin{proof} 
First note that the equation \eqv(martingale.13)
in the rescaled model reads 
\be \Eq(martingale.30)
M_{1,t}^{h,K}\equiv \langle\nu^K_t,h\rangle-\langle\nu^K_0,h\rangle-
\int_0^tds \int \nu^K_s(dx) \sumint c^\pm(x,dy,i;\nu^K_s) (\pm h (y)).
\ee
Note that this is a function of $\nu^K$ and the acceleration factor $K$ that appears in the general form of the generator has cancelled against the $1/K$ in the difference $f(\nu^K\pm K^{-1}\d_{\cdot})-f(\nu^K)$.
Its bracket in the rescaled version takes the form
\be\Eq(martingale.31.1)
\left[M_1^{h,K}\right]_t=\frac 1K\int_0^t ds
 \int \nu^K_s(dx) \sumint c^\pm(x,dy,i;\nu^K_s) h(y)^2,
\ee 
where now a second $1/K$ appears, which is crucial for the LLN.  
One can easily check that the total mass in the rescaled models is bounded  \emph{uniformly} in $K$.

\begin{lemma} \TH(total.7)
Assume that $\sup_{K} \E \left[\langle  \nu^K_0,\1\rangle^2\right] <\infty$. 
Then, 
\be\Eq(total.2.1)
\sup_{K}\E \left[\sup_{t\leq T}\langle \nu^K_t,\1\rangle^2\right] <\infty.
\ee 
\end{lemma}
\begin{proof} 
The proof is a rerun of that of Lemma \eqv(total.1), taking care of the $K$-dependences.
\end{proof}

The crucial fact that leads to the proof of the LLN is that, in the rescaled model, the bracket of the martingale $M^{h,K}_1$ is given by \eqv(martingale.31.1), and thus is of order $1/K$. 
From this, it follows that any limit point of the sequence $\nu^K, K\in\N$  yields a solution of the 
deterministic equation \eqv(lln.1).

\begin{lemma} \TH(lln.2)
Assume that for some subsequence $K_\ell$, $\nu^{K_\ell} \rightarrow \xi$, weakly. Then 
$\xi$ solves \eqv(lln.1).
\end{lemma}

\begin{proof} 
First note that, with our notation, equation \eqv(lln.1)
reads 
\be\Eq(lln.3)
\langle  \xi_t,h\rangle
	-\langle  \xi_0,h\rangle 
	-\!\int_0^t ds \int_{ \XX }\xi_s(dx) 
		\sumint c^{\pm}(x,dy,i,\xi_s)(\pm h(y))\equiv F(\langle \xi,h\rangle)=0.
		\ee
However, we know that 
\be\Eq(lln.4)
F(\langle \nu^K,h\rangle) =M^{h,K}_{1,t}.
\ee
But by Doob's $L^2$-inequality,  
\be
\Eq(lln.5)
\E \left[\sup_{t\leq T} \left(M^{h,K}_{1,t}\right)^2\right] \leq 4 \E \left[M_{1}^{h,K}\right]_T.
\ee
Lemma \thv(total.7) and some straightforward estimates show that 
\be\Eq(lln.6)
\E \left[M_1^{h,K}\right]_T\leq \frac 1K C_T,
\ee
for some finite $C_T$, independent of $K$. Hence, whenever $K_\ell\uparrow\infty$ and 
$\nu^{K_\ell} \rightarrow \xi$,  $F(\langle \nu^{K_\ell},h\rangle)$ converges to zero in $L^2$, uniformly in $t\leq T$. 
In particular, $\xi$ is a solution to \eqv(lln.1). 
\end{proof}

To conclude the proof of the LLN, we need two more facts.  First, we need to show that Equation \eqv(lln.1) has a unique 
solution. Second, we must show that limit points exist, 
i.e. that the sequence $\nu^K$ is tight.

\begin{lemma}
 \TH(lln.7)
Assume that $\langle\xi_0,\1\rangle<\infty$. Then Equation \eqv(lln.1) has a unique solution for the given initial condition $\xi_0$.
\end{lemma}

\begin{proof} The proof is a rather straightforward application of the Gronwall lemma.
First, we show that the total mass of the process $\xi$ remains finite. This follows (as earlier) 
from the inequality
\be\Eq(tot.1)
0\leq \langle \xi_t,\1\rangle
\leq  \langle \xi_0,\1\rangle+\bar r \int_0^t ds \langle \xi_s,\1\rangle,
\ee
hence by the Gronwall lemma
\be\Eq(tot.2)
0\leq \langle \xi_t,\1\rangle \leq \langle \xi_0,\1\rangle \eee^{\bar r t}\equiv C_t.
\ee
 Next, define the total variation norm on the space of signed measures
\be\Eq(lln.8)
\| \xi-\xi'\|_{TV} \equiv \sup_{h:\|h\|_\infty\leq1} |\langle (\xi-\xi'),h\rangle|.
\ee
Now for $\xi_0=\xi'_0$, 
\bea\Eq(tot.2.1)
\langle (\xi-\xi'),h\rangle &=&
\int_0^t ds\int  \left(\xi_s(dx)-\xi'_s(dx)\right)
\Bigl((b(x)(1-p(x))-d(x))h(x)\\
&&\quad+b(x)p(x) \int m(x,dy) h(y)\Bigr)\nonumber\\
&+&\nonumber
\int_0^t ds\int  \left(\xi_s(dx)-\xi'_s(dx)\right)\int\left( \xi_s'(y) c(x,y)h(x)
+\xi(y) c(y,x)h(y)\right).
\eea
Due to the assumptions on the coefficients, for $h$ s.t. $\|h\|_\infty \leq 1$,
we get that 
\be\Eq(tot.3)
|\langle (\xi-\xi'),h\rangle |
\leq \left (\bar b+\bar d +2\bar c C_T\right) \int_0^t ds \|\xi_s-\xi'_s\|_{TV},
\ee
and hence
\be\Eq(tot.4)
\|\xi-\xi'\|_{TV}
\leq \left (\bar b+\bar d +2\bar c C_T\right) \int_0^t ds \|\xi_s-\xi'_s\|_{TV},
\ee
Gronwall now implies that $\| \xi_t-\xi'_t\|_{TV}=0$ if $\| \xi_0-\xi'_0\|_{TV}=0$, and hence the claimed uniqueness.
\end{proof}

Finally,  prove tightness. To show the tightness of the sequence of measures with respect to the weak 
topology, we only need to show that for any bounded continuous  function $h:\XX\rightarrow \R$
the laws of sequences $\langle \nu^K_t,h\rangle_{t\in [0,T]}$ are tight. For this, we use 
the following \emph{ Aldous criterion} \cite{Aldous78}:

\begin{theorem}
Assume that the following hold for the sequence $(\nu^K_t)_{t\leq T}$:
\begin{itemize}
\item[(i)] For any $t\in [0,T]$,  the sequence of the laws of $\langle \nu^K_t,h\rangle$ is tight.
\item[(ii)] For all $\e>0$ and $\eta >0$, there exists $K_0\in \N$ and $\d>0$, such that
for any sequence of stopping times $\t_K\leq T$,
\be\Eq(aldous.1)
\sup_{K\geq K_0} \sup_{0\leq \th\leq \d} \P\left[\left|\langle \nu^K_{\t_K},h\rangle
-\langle \nu^K_{\t_K+\th},h\rangle\right| >\eta\right]\leq \e.
\ee
\end{itemize}
Then the sequence  $(\langle\nu^K_t,h\rangle)_{t\leq T}$ is tight in the Skorokhod topology.
\end{theorem}

Both conditions are easily verified by calculations very similar to those proving Lemma \thv(total.1).
\end{proof}

\subsection{Finite trait space}
We will see later that it is interesting to consider the case when only finitely many traits are alive.
So assume that $\XX=\{x_1,\dots,x_n\}$. 
In this case, the process can be realised effectively as a Markov process with state space 
$\R_+^n$ with generator  $L^K$ acting on functions $f:\R_+^n\rightarrow\R$ as 
\bea
(L^Kf)(\bz)&=&K\sum_{i=1}^n n_i \Bigl[b(x_i)(1-p(x_i)) \left(f(\bz +\eee_i/K)-f(\bz)\right)\nonumber\\
&&\quad+ \sum_{j\neq i}b(x_i)p(x_i)m(x_i,x_j) \left(f(\bz+\eee_j/K)-f(\bz)\right)\nonumber\\
&&\quad +\Bigl(d(x_i)+\sum_{j=1}^n c(x_i,x_j) n_j\Bigr)\left(f(\bz-\eee_i/K)-f(\bz)\right)\Bigr]
\eea
For such processes, a large deviation principle is known to hold (see in particular the book by Dupuis and Ellis 
\cite{DuEll}, Chapter 10).  It states that, under Lipshitz conditions on the transition rates, 
the laws of the
family of processes on compact time intervals $[0,T]$, as measures on the Skorokhod space $\DD([0,T),\R^n)$
satisfy an LDP with rate $K$ and good rate function (action functional)
\be\Eq(rate.1)
I(\phi)=\begin{cases} \int_0^T \LL(\phi(t),\phi'(t))dt, &\;\;\text{if $\phi$  a.c. },\\
+\infty,&\;  \text {else},\end{cases}
\ee
where $\LL$ is a Lagrangian and zero action curves are solutions to \eqv(LV.1). 
From this, one can deduce that the convergence in the LLN holds even almost surely. The
 test functions $h$ can be limited to 
the $n$ indicator functions $\1_{x_i}$. We write $\langle \xi_t,\1_{x_i}\rangle\equiv n_i(t)$.
The equations \eqv(lln.1) then become the $n$ ordinary differential equations
\bea
\Eq(LV.1)
\frac d{dt} n_i(t)&=& n_i(t)\left ( b(x_i)-d(x_i)-\sum_{j=1}^n c(x_i,x_j)n_j(t)\right)\nonumber\\
&& -b(x_i)p(x_i)n_i(t)+\sum_{j\neq i} b(x_j)p(x_j) m(x_j,x_i) n_j(t), 
\eea
for $i=1,\dots, n$. Here $m(x_j,x_i)$ is the probability that a mutant from type $x_j$ is of type $x_i$. 
These are Lotka-Volterra equations with mutations. We devote the next Chapter to Lotka-Volterra equations.


\section{Diffusion limits} \index{diffusion limits}

In probability theory, when you see a law of large numbers, you naturally expect this to be followed by a central limit theorem.
So, in our context, one would expect that with different scalings, one would see non-deterministic limit processes that refine the 
deterministic processes discussed in the previous section. We have seen this already in the context of the Fisher-Wright model in 
Chapter \ref{chapter1}.

There are several ways in which non-deterministic processes can be obtained in the $K\uparrow \infty$ limit. We discuss this for completeness, but this topic will not be returned to in the remainder of this book.

\subsection{An Ornstein-Uhlenbeck process} \index{Ornstein-Uhlenbeck process}
One is to construct the usual Ornstein-Uhlenbeck-type limit by considering the
blow-up of the difference between the processes $\nu^K$ and their limit $\xi$, i.e. to set
\be\Eq(ou.1)
Y^K_t\equiv \sqrt K\left(\nu^K_t-\xi_t\right).
\ee
A simple computation then shows that
\bea\Eq(ou.2)
\langle Y^K_t,h\rangle -\langle Y^K_0,h\rangle&=&
\int_0^t ds\int Y^K_s(dx)\, \Biggl(  (b(x)(1-p(x))-d(x)) h(x) \nonumber\\
&&+b(x)p(x) \int m(x,dy)h(y)\nonumber\\
&&-  \int \xi_s(dy)\left( h(x)c(x,y)+c(y,x)h(y)\right)\nonumber\\
&&-\frac 1{\sqrt K}  \int Y^K_s(dy) c(x,y) h(y)\Biggr)\nonumber\\
&&+\sqrt KM^{K,h}_{1,t}.
\eea
The term in the one-but-last line vanishes in the limit $K\uparrow \infty$. The bracket of the last line is
\be\Eq(ou.3)
\left[\sqrt KM^{K,h}_{1}\right]_t=\int_0^t ds
 \int \nu^K_s(dx) \sumint c^\pm(x,dy,i;\nu^K_s) h(y)^2.
\ee 
As $K\uparrow \infty$, this bracket converges to 
\be\Eq(ou.4)
\left[M^{\infty,h}_{1}\right]_t=\int_0^t ds
 \int \xi_s(dx) \sumint c^\pm(x,dy,i;\xi_s) h(y)^2,
\ee 
and hence the martingale converges to the continuous martingale
$M^{\infty,h}_{1,t}$ with this bracket, which is a time change of Brownian motion.
Note also that the co-brackets are given by
\be\Eq(ou-co.4)
\left[M^{\infty,h}_{1}, M^{\infty,g}_1\right]_t=\int_0^t ds
 \int \xi_s(dx) \sumint c^\pm(x,dy,i;\xi_s) h(y)g(y).
\ee 
In particular, the two martingales are independent if $g$ and $h$ have disjoint support.

Consequently, one obtains the following convergence result.

\begin{theorem}[\cite{kurtz1978,Wang-Theses}] \TH(ou.5)
Assume that the second moment condition for the initial distribution holds and that 
$Y_0^K\rightarrow \g$ for some finite measure $\g$. 
Then the process $Y^K$ converges in law with respect to the Skorokhod topology to a continuous process  $Y$ 
with 
values in the space of distributions on $\XX$ that satisfies, for any smooth function $h$, 
\bea\Eq(ou.6)
\langle Y_t,h\rangle -\langle \g,h\rangle&=&
\int_0^t ds\int Y_s(dx)\,\Bigl((b(x)(1-p(x))-d(x)) h(x) \nonumber\\
&&+b(x)p(x) \int m(x,dy)h(y)\nonumber\\
&&-  \int \xi_s(dy) \left( h(x)c(x,y)+c(y,x)h(y)\right)\Bigr)\nonumber\\
&&+M^{\infty,h}_{1,t}.
\eea
\end{theorem}

\begin{remark}  
In the case of finite (or discrete) state space $\XX=\{x_1,\dots, x_n\}$,
we may choose as test functions just the set of indicator functions $h(x)=\1_{x_i}(x)$. 
Setting $\langle Y_t,\1_{x_i}\rangle \equiv Y_i(t)$, 
Eq. \eqv(ou.6)
 reduces to the system of stochastic differential equations.
\bea
dY_i(t)&=& Y_i(t)(b(x_i)(1-p(x_i))-d(x_i)) dt \nonumber\\
&&+ \sum_{j\neq i} m(x_j,x_i))b(x_j)p(x_j) Y_j(t) dt\nonumber\\
&&-  \sum_{j=1}^nc(x_i,x_j)\left( Y_i(t)n_t(x_j)+Yjt(t)n_i(t)\right)dt\nonumber\\
&&+dM^{\infty,x_i}_{1,t},
\eea
where the martingales are independent and have brackets
\bea\Eq(bracket.101)
[M^{\infty,x_i}_1]_t&=& \int_0^t ds\Bigl( n_i(s)(b(x_i)(1-p(x_i))+d(x_i)+\sum_{j=1}^nc(x_i,x_j)n_j(s)\nonumber\\
&&+\sum_{j\neq i} n_j(s)b(x_j)p(x_j)m(x_j,x_i)\Bigr).
\eea
Hence, 
\bea\Eq(bracket.102)
dM^{\infty,x_i}_{1,t} =&\Bigl( n_i(t)(b(x_i)(1-p(x_i))+d(x_i)+\sum_{j=1}^nc(x_i,x_j)n_j(t))
\nonumber\\ &+\sum_{j\neq i} n_j(t)b(x_j)p(x_j)m(x_j,x_i)\Bigr)^{1/2}dB_{i,t},
\eea
where $B_{i}$ are independent standard Brownian motions. Note that it is easy to verify that 
bracket of $M^{\infty,x_i}_{1,t}$ defined by \eqv(bracket.102) is given by \eqv(bracket.101). The converse is a
deep result in stochastic analysis.

In the case $\XX$ is a singleton, this reduces to the sde
\be
dY(t)= Y(t)\, (b-d) dt 
-  2n(t) c(x,x)Y(t)dt\nonumber\\
+dM^{\infty}_{1,t},
\ee
where 
\be
[M^{\infty}_{1}]_t=\int_0^t ds n(s)(b+d+cn(s)),
\ee
so that $dM^\infty_{1,t} =\sqrt {n(t)(b+d+cn(t))}dB_t$. 
\end{remark}

\subsection{Feller-type diffusions} 

The convergence to a deterministic process in the law of large numbers hinges on the fact that the martingale 
$M^{K,h}_1$ converges to zero. This is ultimately because the total number of events that occur in 
time of order one is of order $K$. Fournier and Méléard \cite{FM04} constructed a superprocess limit by increasing the number of 
events to $K^2$ while making sure that the mean effect of these events is just of order one.
This can be achieved by adding equally large birth- and death rates of order $K$. From the calculations before, 
this does not change the drift term, but blows up the martingale term.
To be precise, consider the family of generators defined as
\bea\Eq(generator-bu.1) \nonumber
\left(\wt\LL^K f\right)(\nu^K )&=& \int_{ \XX}\biggl(f\Bigl(\nu^K +\frac {\delta_x}{K}\Bigr)-f(\nu^K )\biggr)\left(\bigl(1- p(x)\bigr)b(x)+K\rho(x)\right) K\nu^K (dx)\\ 
					&&	\hspace{-8mm}+\int_{ \XX}\int_{ \XX}\biggl(f\Bigl(\nu^K +\frac{\delta_{y}} K\Bigr)-f(\nu^K )\biggr)		
							p(x)b(x) m(x,dy)  K\nu^K (dx)\\\nonumber
				&&	\hspace{-8mm}	+\int_{ \XX}\biggl(f\Bigl(\nu^K -\frac {\delta_x} K\Bigr)-f(\nu^K )\biggr)\Bigl(d(x)+K\rho (x)
						+\int_{ \XX}c(x,y)\nu^K (dy)\Bigr) K\nu^K (dx).
\eea
for some function $\rho$. 
In this case, we also get Equation \eqv(martingale.30) with the same $c^\pm$, but the bracket of the martingale
is now 
\be\Eq(martingale.31)
\left[M_1^{h,K}\right]_t= \int_0^t ds
 \int \nu^K_s(dx) \left(2\rho(x) h(x)^2+
\frac 1K
  \sumint c^\pm(x,dy,i;\nu^K_s) h(y)^2\right).
\ee 
Therefore, if $\nu^K$ converges to some distribution $\nu$, then for any test function $h$,
\bea
\langle \nu_t,h\rangle -\langle \nu_0,h\rangle&=&
\int_0^t ds\int \nu_s(dx)\,\Bigl( (b(x)(1-p(x))-d(x)) h(x) \nonumber\\
&&+b(x)p(x) \int m(x,dy)h(x+y)\nonumber\\
&&-  \int \nu_s(dy)  h(x)c(x,y)\Bigr)+\wh M^{\infty,h}_{1,t},
\eea
where 
\be 
\left[\wh M^{\infty,h}_{1}\right]_t=\int_0^t ds
 \int \nu_s(dx) 2\rho(x) h(x)^2.
 \ee
The formal differential form of this reads then
\bea
dX_t(x)&=& X_t(x)\, (b(x)(1-p(x))-d(x)) dt \nonumber\\
&&+ \int m(x-y,dy)b(x-y)p(x-y) X_t(x-y) dt \nonumber\\
&&-  \int (X_t(x)  X_t(y) c(x,y)dydt
+d\wh M^{\infty,x}_{1,t}.
\eea
This limit is the superprocess \index{superprocess} introduced by Alison Etheridge in \cite{Etheridge2004}.

In the single-site case, this reduces to 
\be
dX_t= X_t\,(b-d 
-   cX_t)dt
+d\wh M^{\infty}_{1,t},
\ee
where 
\be
[\wh M^{\infty}_{1}]_t=\int_0^t 2\rho X_sds.
\ee
This implies that
\be
d\wh M^{\infty}_{1,t}= \sqrt {2\rho X_t}dB_t,
\ee
and so in this case, we recover a Feller diffusion equation with competition, \index{Feller diffusion}
\be\Eq(fell.1)
dX_t= X_t\, (b-d 
-   cX_t)dt+\sqrt {2\rho X_t}dB_t.
\ee
Note the difference between the two cases: here we have a true non-linear sde, whereas in the previous case, we had an
 Ornstein-Uhlenbeck type process. \index{Ornstein-Uhlenbeck process}

\section{Branching processes}

Branching processes are used as tools for the analysis of the more complicated models we consider here. \index{branching process}
In this section, we review some of their elementary properties. 
A classical source for branching process theory is the textbook by Athreya and Ney \cite{AN}, but see also 
\cite{EthKur1986} and \cite{Dawson2017}.

\subsection{Discrete time Galton-Watson processes}
The classical \emph{Galton-Watson} branching process describes the evolution of the population size \index{Galton-Watson process}
of reproducing individuals without any interaction or competition in discrete time.
It is easily constructed via a family $\{\xi_k^n\}_{n,k\in\N}$ of iid random variables taking values in 
$\N_0$. Let in addition $Z_0$ be a random variable with values in $\N$, independent of the $\xi$. 
Then define 
\be\Eq(br.1)
Z_n=\sum_{k=1}^{Z_{n-1}} \xi_k^n, \quad n\in \N.
\ee
Then $Z_n $ is the size of the population in generation $n$ when $Z_0$ is the size of the population at time $0$ and $\xi_k^n$ is the number of offspring of individual $k$ in generation $n-1$. 
The law of the  $\xi^n_k$ is called the \emph{offspring distribution}. We will always 
exclude the trivial case where $\P(\xi^k_l=1)=1$ where no branching happens.
A key parameter of the offspring distribution is the  \emph{mean offspring size}, $m\equiv \E \left[\xi_k^n\right]$.
According to whether $m<1, m=1, m>1$, the process is called \emph{subcritical, critical}, or \emph{supercritical}.
In fact, $m$ plays the role of the growth rate of the process, as shown in the following theorem.

\begin{theorem}[\cite{EthKur1986}] \TH(br.2)
The process $Z$ defined by \eqv(br.1)
satisfies, whenever $m<\infty$,  for all $n\in \N$,
\be\Eq(br.3)
\lim_{Z_0\uparrow\infty} \frac {Z_n}{Z_0} =m^n, \as.
\ee
Moreover, if $\var(\xi^n_k)=\s^2<\infty$,
the joint distributions of 
\be\Eq(br,4)
W_n\equiv Z_0^{-1/2}(Z_n-m^nZ_0),
\ee
converge to those of 
\be\Eq(br,5)
W^\infty_n\equiv \sum_{\ell=1}^n m^{n-(\ell+1)/2}V_\ell,
\ee
where $V_\ell $ are iid $\NN(0,\s^2)$ distributed.
\end{theorem}

\begin{proof}
The proof of the first assertion is by induction and the standard law of large numbers.
For $n=1$,  the LLN gives 
\be\Eq(br.6)
\lim_{Z_0\uparrow \infty} \frac {Z_1}{Z_0} =\lim_{Z_0\uparrow \infty} \frac {1}{Z_0} 
\sum_{k=1}^{Z_{0}} \xi_k^0=m,\quad\as.
\ee
Assuming \eqv(br.2)
for $n-1$, we get \eqv(br.2) for $n$ in just the same way.

To prove the second assertion, we write $W_n$ by telescopic expansion as
\bea\Eq(br.7)
\nonumber
W_n&=& Z_0^{-1/2}\sum_{\ell=1}^n\left(m^{n-\ell}Z_\ell -m^{n-\ell+1}Z_{\ell-1}\right)\\
&=&\sum_{\ell=1}^n m^{n-\ell}\left(\frac {Z_{\ell-1}}{Z_0}\right)^{1/2}Z_{\ell-1}^{-1/2}\sum_{k=1}^{Z_{\ell-1}}
(\xi^\ell_k-m).
\eea
We know that the coefficients $\left(\frac {Z_{\ell-1}}{Z_0}\right)^{1/2}$ converge a.s. to $m^{(\ell-1)/2}$.
The random variables
\be\Eq(br.8)
U_\ell\equiv Z_{\ell-1}^{-1/2}\sum_{k=1}^{Z_{\ell-1}}
(\xi^\ell_k-m),
\ee
converge to iid $\NN(0,\s^2)$ r.v.'s by the central limit theorem.
\end{proof}

We now turn this result around.

\begin{theorem} \TH(br.9)
Let $\var (Z_0)<\infty$ and $\E[ \xi^1_1]=m<\infty$. 
$W_n\equiv Z_n m^{-n} $ is a  $L^1$-bounded martingale and $\E [W_n|\FF_0]=Z_0$. 
If $m>1$ and $\s^2<\infty$, $W_n$ is uniformly integrable. Hence, $W_n$ converges a.s. to some r.v. $W_\infty$ 
and $W_\infty$ is non-zero with $\E[W_\infty]=\E [Z_0]$ if $m>1$. 
\end{theorem}

\begin{proof} Clearly, 
\be
\Eq(br.10)
\E \left[Z_n|\FF_{n-1}\right] =Z_{n-1} \E\left[\xi_1^n\right] =mZ_{n-1},
\ee
and iterating, for any $n>k$, 
\be\Eq(br.11)
\E \left[Z_n|\FF_{k}\right] =Z_{k} m^{n-k}.
\ee
Hence $W$ is $L^1$-bounded martingale. 
A simple calculation shows that (without loss, let $Z_0=1$), 
\be\Eq(br.12)
\var(Z_n)=m^2\var(Z_{n-1)}+m^{n-1}\s^2,
\ee
from which, by iteration, we get, for $m\neq 1$, 
\be\Eq(br.13)
\var(Z_n)=m^n\frac{m^n-1}{m^2-m}\s^2,
\ee
or 
\be\Eq(br.13.1)
\var(W_n)=\frac{1-m^{-n}}{m^2-m}\s^2,
\ee
and  for $m=1$,
\be\Eq(br.13.2)
\var(Z_n)=\var (W_n)=n\s^2.
\ee
Hence $W$ is bounded in  $L^2$ if and only if $m>1$, and hence uniformly integrable in that case.
\end{proof}
       
       \index{branching process!generating function}
 If $m\leq 1$, the process dies out almost surely. To see this, one uses the generating functions 
 $g_n(s)\equiv \E\left[s^{Z_n}\right]$. Henceforth, we assume $Z_0=1$. 
Note that $g_1(s)=g(s)=\E\left[s^{\xi_1^1}\right]$.  
We have the recursion 
\be
\Eq(br.14)
g_n(s)=g\circ g_{n-1}(s)=g_{n-1}\circ g(s).
\ee
We define the extinction probability $q\equiv \P\left(\exists_n Z_n=0\right)$ and 
$q_n= \P\left(Z_n=0\right)$.  
 \index{branching process!extinction probability}

\begin{theorem} \TH(br.15) $q$ is the smallest non-negative solution of $g(s)=s$. 
It is one if $m\leq 1$ and strictly smaller than one if $m>1$.
\end{theorem}

\begin{proof} Note that $q_n\leq 1$ and monotone increasing, so it converges to some $q\leq 1$. 
Also, $q_n=g_n(0)$ and so $q_n=g(q_{n-1})$, and so by continuity of $g$, $g(q)=q$. Moreover, 
$q$ is the smallest root of this equation. Namely, if $p$ is another non-negative root, then 
$q_1=g(0)\leq g(p)=p$, since $g$ is non-decreasing. By iteration ist follows that 
$q_n=g(q_{n-1})\leq g(p)=p$, and so $q\leq p$. 
Since $g$ is convex on $[0,1]$, it can have a most two roots on $[0,1]$. Clearly, $1$ is such a root,
so $q $ must be the other one. If $g'(1)\leq 1$, then by convexity, $g(s)$ lies above $s$  for all $s\in [0,1)$ and thus $q=1$, unless $g(s)=s$, which is only the case in the trivial situation when 
$\P(\xi^1_1=1)=1$. 
 In the opposite case, a second root $q<1$ must exist.  \end{proof}

\subsection{The continuous time tree.}
From the discrete-time Galton-Watson tree with synchronised generations, we can 
pass to a continuous time tree where each individual that is born dies after an exponential time 
and produces several children that have the same distribution as $\xi^1_1$. 
In the specific case when  the number of children can only be zero or two, the resulting size process is
a birth-and-death chain with transition rates
\be\Eq(yule.1)
		q(n,n+1)= n b,\qquad q(n,n-1) =n d,
\ee
that is a continuous-time Markov process with state space $\N_0$ and generator $L$ acting
as
\be
(Lf)(n)\equiv n b(f(n+1)-f(n))+nd(f(n-1)-f(n)).
\ee
Note that this process corresponds precisely to the case of a mutation and competition case of our
stochastic models with monomorphic initial conditions.
The nice thing here is that it is very easy to compute probabilities of the type
\be\Eq(yule.2)
\P(\t_N<\t_0|X(0)=n)\equiv p_N(n),
\ee
as solutions of the boundary value problem 
\bea\Eq(yule.3)\nonumber
&&nb p_N(n+1)+ndp_N(n-1)=n(d+b)p_N(n).\quad  n\in \{1,\dots, N-1\},\\
&& p_N(0)=0,\nonumber \\
&& p_N(N)=1.
\eea
The solution of this equation is easily found to be (if $d\neq b$)
\be
\Eq(yule.4)
p_N(n)= \frac {\left(\frac db\right)^n-1} {\left(\frac db\right)^N-1}.
\ee
In particular, we have 
\be
\Eq(yule.5)
p_N(1)= \frac {\left(\frac db\right)-1} {\left(\frac db\right)^N-1}\rightarrow
\begin{cases} \frac {b-d}{b},\;\text{if} \; b>d,\\
0,\;\text{if} \; b<d,
\end{cases},\quad N\uparrow\infty.
\ee
 In the critical case $d=b$ we get $p_N(n)=n/N$,
and again
\be\Eq(yule.6)
p_N(1)=1/N\rightarrow 0, \quad N\uparrow\infty.
\ee

\begin{remark}
The form of the fixation probability $\frac {b-d}{b}$ may look counterintuitive since $d$ is in
the denominator. But rewriting it as $r/ (r+d)$, it is natural that for fixed growth rate $r$, 
the probability to die out increases with the death rate!
\end{remark}

Finally, we are interested in the time it takes for a branching process to reach a prescribed level or to die out.

\begin{lemma}
  \label{prop:BP}
  Let $(Z_t,t\geq 0)$ be a (binary) branching process with birth rate $b$ and death rate $d$. 
  Let $\tau_i$ be
  the first hitting time of level $i$ by $Z$. Let also $\P_j$
  denote the law of $Z$ conditionally on $Z_0=j$, and $\E_j$ the
  corresponding expectation. Then
  \begin{itemize}
  \item[(i)] In the supercritical case $b>d$,
  \be
  \frac {\ln k}{b} -O(1)\leq \E_1(\tau_k\wedge\tau_0)\leq\frac{1+\ln
      k}{b}. \label{eq:BP-2}
  \ee
  \item[(ii)] In the subcritical case $b<d$, 
  \be\Eq(subtime.1)
   \E_1(\tau_k\wedge\tau_0)\leq \frac 1{d-b}.
   \ee
   \item[(iii)] In the critical case $b=d$, 
    \be\Eq(subtime.2)
   \frac {\ln k }b-O(1)\leq \E_1(\tau_k\wedge\tau_0)\leq \frac {\ln k}{b}.
   \ee
   \end{itemize}
\end{lemma}

\begin{proof}
  Set $e_j=\E_j(\tau_k\wedge\tau_0)$, and $e_0=e_k=0$.
  $e$ satisfies the equations
  \be
 j b(e_{j+1}-e_j) +jd(e_{j-1}-e_j)j=1, \quad \forall
  j\in\{1,\ldots,k-1\}.
  \ee
   From this, one finds that
  \be\Eq(ej.1)
  e_{j+1}-e_j=\left(\frac{d}{b}\right)^je_1
  -\frac{1}{b}\sum_{i=1}^j\frac{1}{i}\left(\frac{d}{b}\right)^{j-i},
  \ee
  from which it follows by summing over $j\in\{0,\ldots,k-1\}$ that
  \bea\Eq(eeee.1)
  e_1&=&\frac{\frac{1}{b}\sum_{i=1}^{k-1}\frac{1}{i}
    \sum_{j=i}^{k-1}\left(\frac{d}{b}\right)^{j-i}}
  {\sum_{j=0}^{k-1}\left(\frac{d}{b}\right)^j}\\\nonumber
&=&  \frac{\frac{1}{b}\sum_{i=1}^{k-1}\frac{1}{i}
    \left(1-    \left(\frac{d}{b}\right)^{k-i}\right) }
    {1-\left(\frac{d}{b}\right)^k},
    \eea
    if $b\neq d$, and 
    \be\Eq(critical.1)
      e_1=\frac 1{b}\sum_{i=1}^{k-1}\left(\frac{1}i-\frac 1k\right),
      \ee
      if $b=d$.
  In Case (i), we can use that $   \left(1-    \left(\frac{d}{b}\right)^{k-i}\right)\leq 1-\left(\frac{d}{b}\right)^k$
  to bound 
  \be
  e_1\leq\frac{1}{b}\sum_{i=1}^{k-1}\frac{1}{i}   \leq\frac{1+\log k}{b}.
  \ee
  For a lower bound, we can use 
 \bea
 e_1&\geq& \frac{1- \left(\frac{d}{b}\right)^{k/2}}{1-
  \left(\frac{d}{b}\right)^k} \frac1b \sum_{i=1}^{k/2} \frac 1i \geq  \frac 1{1+ \left(\frac{d}{b}\right)^{k/2}}
   \frac{\ln k-\ln 2}b\nonumber\\
   &\geq& (1-(d/b)^{k/2})     \frac{\ln k-\ln 2}b \approx    \frac{\ln k}b +O(1),
   \eea
   for large $k$.
  In Case (ii), we write 
  \bea
 \Eq(supoerdooper.1)
&& \frac{\frac{1}{b}\sum_{i=1}^{k-1}\frac{1}{i}
    \left(1-    \left(\frac{d}{b}\right)^{k-i}\right) }
    {1-\left(\frac{d}{b}\right)^k}= \frac{\frac{1}{b}\sum_{i=1}^{k-1}\frac{1}{i}\left(\frac bd \right)^i
    \left(1-    \left(\frac{b}{d}\right)^{k-i}\right) }
    {1-\left(\frac{b}{d}\right)^k}\nonumber\\
    &&\leq \frac{1}{b}\sum_{i=1}^{\infty}\left(\frac bd \right)^i=\frac 1{d-b}.
    \eea
    which is a good bound if $d>b$. If $d-b$ goes to zero, one can improve the bound to $\ln k$. 
  Finally, in the critical case,
  from \eqv(critical.1)
  we get 
  \be
  e_1=\frac 1b\left(\sum_{i=1}^{k-1} \frac 1i -1 +\frac1k\right),
  \ee
  from which the claimed bounds follow. This
  ends the proof.
\end{proof}

%
%
\index{branching process!extinction time}

We will also need the mean time to hit zero or $k$ when starting in $j$. From \eqv(ej.1) we get 
\bea\Eq(ej.2)
  e_{j}&=&\sum_{\ell=0}^{j-1}\left(\frac{d}{b}\right)^\ell e_1
  -\frac{1}{b} \sum_{\ell=0}^{j-1}\sum_{i=1}^{\ell}\frac{1}{i}\left(\frac{d}{b}\right)^{\ell-i}\nonumber\\
  &=&\sum_{\ell=0}^{j-1}\left(\frac{d}{b}\right)^\ell e_1
  -\frac{1}{b} \sum_{i=1}^{j-1}\frac{1}{i}\sum_{\ell=i}^{j-1}\left(\frac{d}{b}\right)^{\ell-i}.
  \eea
If we denote by $e_j^m\equiv  \E_j(\tau_m\wedge\tau_0)$, we can write this as
\be\Eq(ej.3)
e_j^k= \left(e_1^k-e_1^j\right)\sum_{\ell=0}^{j-1}\left(\frac{d}{b}\right)^\ell.
\ee
In the subcritical case  (with possibly $k$-dependent difference between birth and dearth rates) and   $j$  (possibly depending on $k$)  so large that $(b/d)^j\downarrow 0$,
one finds from this that 
\be
\Eq(ej.4)
e^k_j=\frac 1{d-b}\left(\sum_{i=1}^{j-1} \frac 1i+\sum_{i=j}^{k-1} \frac 1i \left(\frac bd\right)^{i-j} \right)\left(1+o(1)\right).
\ee
In the critical  case, one finds 
\be
e_j^k=\frac 1b\left(\sum_{i=j}^{k-1} \frac ji-\frac{k-j}{k}\right)\sim \frac 1b j\ln (k/j).
\ee


\chapter{Competitive Lotka-Volterra systems }\label{chapter3}

\begin{chapquote}
{John Maynard Smith, \emph{Natural Selection and the Concept of Protein Space}}
{I do not want to discuss the problem of
the origin of life, but only to point out that it is a quite
different problem from that of the mechanism of evolution. }
\end{chapquote}


 \index{Lotka-Volterra equations!competitive}
In the absence of mutations and if initially there exists a finite number of phenotypes in the population,  
the system of equations \eqv(LV.1) reduces to the competitive system of Lotka-Volterra equations.
They play a crucial role in the sequel. 
 In this chapter, we provide some necessary background on these classical models.
\begin{definition}\label{cor} A system of $k$ ordinary differential equations of the form 
\be\label{LV-System}
\frac {d n_i(t)}{dt}= n_i(t)\biggl(b(x_i)-d(x_i)-\sum_{j=1}^k c(x_i,x_j)n_j(t)
\biggr), \qquad 1\leq i\leq k.
\ee
is called a \emph{competitive system of Lotka-Volterra 
equations} for the $k$ traits $x_1,\dots,x_k$. We denote the system \eqv(LV-System) by $LV(k,(x_1,...,x_k))$.
\end{definition}

\nomenclature{$n_i(t)$}{population size of type $i$ at time $t$}
\section{Coexistence and invasion fitness}

\nomenclature{$f(y,x)$}{invasion fitness}

\nomenclature{$LV(k,\bx)$}{Lotka-Volterra system with $k$ types}

We begin by introducing the notions of coexisting traits and invasion fitness (see \cite{CM11}).
\index{invasion fitness}
\index{coexisting traits}
\begin{definition}
We say that the distinct traits $x$ and $y$ \emph{coexist} if the system $LV(2,(x,y))$ admits a non-trivial equilibrium, named 
$\bar n(x,y)\! \in\! (0,\infty)^2$, which is locally stable in the sense that the eigenvalues of the Jacobian matrix of the system 
$LV(2, (x,y))$ at $ \bar n(x,y)$ are all not positive. We call this equilibrium \emph{strictly stable} if all eigenvalues are striclty negative. 
\end{definition}
The invasion of a single mutant trait in a monomorphic population which is close to its equilibrium is governed by its initial growth rate. Therefore, it is convenient to define the fitness of a mutant trait by its initial growth rate.
\begin{definition} If the resident population has the trait  $x\in\mathcal X$, then we call the following function \emph{invasion fitness} of the mutant trait $y$:
\begin{align}
f(y, x)=b(y)-d(y)-c(y,x)\bar n(x).
\end{align}
\end{definition}
\begin{remark}
The unique strictly stable equilibrium of $LV(1,x)$ is $ \bar n(x)=\frac{b(x)-d(x)}{c(x,x)}$, and hence $f(x,x)=0$, for all $x\in\mathcal{X}$.
\end{remark}

 Coexistence and invasion fitness are closely related.
\begin{proposition}\label{pro_coex}
There is coexistence between traits $x$ and $y$ if and only if
\be\Eq(coex.1)
f(x,y)> 0\; \text{and}\; f (y,x) > 0,\; \text {or}\; f (x,y) =  f (y,x)=0.
\ee
\end{proposition}

\begin{proof}
Note first that if $n$ is a fixpoint, then 
\be 
c(x,x)n(x) = r(x)- c(x,y) n(y) \leq r(x),
\ee
so that $n(x) \leq \bar n(x)$, and also $n(y)\leq \bar n(y)$, where $\bar n(x), \bar n(y)$ denote the monomorphic
fixpoints. Hence we can 
parametrise the fixpoint as $\bar n = (\a\bar n(x) , \b \bar n(y))$ with $\a,\b \in[0,1]$.
The fixpoint equations can then be written as
\bea\Eq(coex.5)
(1 -\a) r(x) -c(x,y)\b\bar n(y) &=& 0,\\
(1 -\b) r(y)-c(y, x)\a\bar n(x) &=& 0.
\eea
These equations can be conveniently written as
\bea\Eq(coex.6)
(1 -\a) f (x, y) -(\b-(1-\a)) c(x, y)\bar n(y) &=& 0,\\
(1 -\b) f (y, x) - (\a - (1 -\b))c(y, x) \bar n(x)& =& 0. \Eq(coex.6.1)
\eea
This yields the solutions
\be\Eq(coex.7)
\a=\frac{\left(\frac{f(y,x)}{c(y,x)\bar n(x)}+1
\right)
\frac{f(x,y)}{c(x,y)\bar n(y)}}{\left(\frac{f(y,x)}{c(y,x)\bar n(x)}+1\right)
\left(1+\frac{f (x,y)}{c(x,y)\bar n(y)}
\right)-1},
\ee
and
\be\Eq(coex.8)
\b=\frac{\left(\frac{f(x,y)}{c(x,y)\bar n(y)} + 1\right)
\frac{ f (y,x)}{c(y,x) \bar n(x)}}
{\left(\frac {f(y,x)}{c(y,x) \bar n(x)} + 1\right)\left(1 +\frac{f (x,y)}{c(x,y) \bar n(y)}\right)-1},
\ee
 provided the denominators are  not zero, which is the case whenever $f(x,y)>0$ and $f(y,x)>0$,
 In this case, both $\a$ and $\b$ are strictly positive. 
If both  $f (x,y)$ and $f (y, x)$ are equal to zero, then the denominator vanishes,
but equations \eqv(coex.6) and \eqv(coex.6.1) are satisfied whenever $\b = (1 -\a )$, so there is a
continuum of solutions with $\a\in  [0, 1]$ and $\b= (1 -\a )$. If 
$f (x,y)$ 
or if $f (y, x)$ is negative, then the solutions are not in the positive quadrant
and thus not admissible. Hence, only the monomorphic solutions exist in these
cases.
 \end{proof}
 
 \begin{remark} 
 We see that in the case $f(x,y)=f(y,x)=0$, there is a line of fixpoints, so that the condition of strict local 
 stability cannot hold. Indeed, the Jacobi matrices at these fixpoints have a zero eigenvalue.
 \end{remark}
 
\section{Polymorphism}  \index{polymorphic equilibria}
The concepts above can be extended readily to polymorphic equilibria.
 Recall that the equations for an equilibrium are
 \be\label{equi.1}
0= n_i(t)\biggl(b(x_i)-d(x_i)-\sum_{j=1}^k c(x_i,x_j)n_j(t)
\biggr), \qquad 1\leq i\leq k.
\ee
First, $0$ is always an equilibrium, but this will only be stable if, for all $i$, 
$d(x_i)>b(x_i)$, which of course we will assume not to be the case.
Next, we have the \emph{monomorphic} equilibria 
\be
\Eq(equi.2)
\bar n^j_i\equiv \begin {cases} \frac {b(x_j)-d(x_j)}{c(x_j,x_j)},&\text{if}\; i=j,\\
0,&\;\text{else},
\end{cases}
\ee
provided $b(x_j)-d(x_j)>0$. 
The stability of these equilibria is easy to analyse. The Jacobi matrix at 
the equilibrium $\bar x^j$ has components 
\be\Eq(equi.3)
H(\bar n^j)_{i,\ell}=\begin{cases} d(x_\ell)-b(x_\ell)-c(x_\ell,x_j)\bar x^j_j, &\, \text{if }
\ell=i, i\neq j,\\
-\bar n^j_j c(x_j,x_\ell),&\, \text{if } i=j.
\end{cases}
\ee
In terms of the concept of invasion fitness, this reads
\be\Eq(equi.4)
H(\bar n^j)_{i,\ell}=\begin{cases} f(x_\ell,x_j), &\, \text{if }
k=i, i\neq j,\\
-\bar n^j_j c(x_j,x_\ell),&\, \text{if } i=j.
\end{cases}
\ee
That is, 
\be\Eq(equi.3.1)
H(\bar n^j)=\left(\begin{matrix} 
f(x_1,x_j)&0&0&0&\dots&0\\
0&f(x_2,x_j)&0&0&\dots&0\\
\dots&\dots&\dots&\dots&\dots&\\
\dots&\dots&\dots&\dots&\dots&\\
-\bar n^j_jc(x_j,x_1)&-\bar n^j_jc(x_j,x_2)&\dots&-\bar n^j_jc(x_j,x_j)&\dots&-\bar n^j_jc(x_j,x_k)\\
\dots&\dots&\dots&\dots&\dots&\\
\dots&\dots&\dots&\dots&\dots&\\
0&0&0&0&0&f(x_k,x_j)
\end{matrix}\right)
\ee
Since the eigenvalues of $H(\bar n_j) $ are just the diagonal terms, a monomorphic fixed point is stable if all 
invasion fitnesses $ f(x_\ell,x_j)$ are negative. 

Polymorphic fixed points with a subset of $\ell \leq k$ positive components 
(let us take w.r.g. the first $\ell $ components) lead to the \emph{linear} equations 
 \be\label{equi.5}
b(x_i)-d(x_i)=\sum_{j=1}^\ell c(x_i,x_j)\bar n_j, \qquad 1\leq i\leq \ell,
\ee
for the equilibrium values $\bar n_j$, which in addition must be all 
strictly positive. The existence of such equilibria requires conditions 
on the parameters that are more difficult to verify. We have already discussed this for two traits before.

As a further illustration, 
take the simplest setting of equal competition when $c(x_i,x_j)=c$ for all $i,j=1,\dots, k$. Then the equations for a $k$-morphic 
equilibria are
\be
r(x_i)= c( \bar n_1+\dots +\bar n_k), \quad i=1,\dots,k,
\ee
which have only a solution if $r(x_i)=r$, independent of $i$. In this case, there is a continuum of solutions
given by all $\bar n\in \R_+^n$ such that 
$ \bar n_1+\dots +\bar n_k=r/c$. 
Barring such very symmetric situations, the simplest cases when polymorphic equilibria exist trivially are when 
there are subsets of sites between which the competition kernels are all zero.

In general, if  a fixed point exists, then 
\be\Eq(equi.6)
H(\bar n)_{i,\ell}=-\bar n_i c(x_i,x_\ell).
\ee
Assume that $\bar n\in \R_+^k$ is a stable fixed point with 
only strictly positive components. If a mutation to type $x_{k+1}$ appears, 
we must analyse the stability of the fixed point $(\bar n,0)$ in the $k+1$-dimensional system. The Jacobi matrix now takes the form
\be\Eq(equi.7)
H((\bar n,0))_{i,\ell}=\begin{cases} -\bar n_i c(x_i,x_\ell), &\, \text{if }
 i\leq k, \\
0, &\, \text{if } i=k+1, \ell\leq k,\\
b(x_{k+1})-d(x_{k+1}) -\sum_{j=1}^k c(x_{k+1},x_j)\bar n_j, &\text{if }
i=\ell=k+1.
\end{cases}
\ee
\be\Eq(equi.7.1)
H(\bar n^j)=\left(\begin{matrix} 
-\bar n_1c(x_1,x_1)&-\bar n_1c(x_1,x_2)&\dots&\dots&\dots&-\bar n_1c(x_1,x_{k+1}\\
-\bar n_2c(x_2,x_1)&-\bar n_2c(x_2,x_2)&\dots&\dots&\dots&-\bar n_2c(x_2,x_{k+1}\\
\dots&\dots&\dots&\dots&\dots&\\
\dots&\dots&\dots&\dots&\dots&\\
\dots&\dots&\dots&\dots&\dots&\\
-\bar n_kjc(x_\ell,x_1)&-\bar n_kjc(x_k,x_2)&\dots&\dots&\dots&-\bar n_kc(x_\ell,x_{k+1})\\
0&0&\dots&\dots&\dots&f(x_{k+1}, \bx)
\end{matrix}\right),
\ee
where 
\be
\Eq(equi.7.2)
f(x_{k+1}, \bx)\equiv b(x_{k+1})-d(x_{k+1}) -\sum_{j=1}^k c(x_{k+1},x_j)\bar n_j
\ee
is the invasion fitness of the trait $x_{k+1}$ in the presence of the equilibrium $\bar n$.
Again, by elementary linear algebra, the eigenvalues of $H((\bar n,0))$ are those of $H(\bar n)$ and 
$f(x_{k+1}, \bx)$.  Hence, if the invasion fitness is positive, then $H((\bar n,0))$ has a positive eigenvalue
and the fixed point $(\bar n,0)$  is unstable. Thus, if the mutant population 
manages to reach a macroscopic level $\e>0$, then the dynamical system 
will move the population away from the old fixed point towards a new equilibrium.

The discussion above justifies the following formal definition.
\begin{definition} 
For any $k\geq 2$, we say that the distinct traits  $(x_1,\ldots ,x_k)$ 
\emph{coexist} if the system $LV(k,  \bx)$ has a unique non-trivial 
equilibrium $\bar n (\bx)\in (0,\infty)^k$ which is locally strictly stable,  
in the sense  that all eigenvalues of the Jacobian matrix of the system 
$LV(k,  \mathbf x)$ at $\bar{n} (\mathbf x)$ have strictly negative real 
parts.

If the traits $(x_1,\ldots ,x_k)$  coexist, then the \emph{invasion fitness} 
of a mutant trait $ y$ which appears in the resident population is given by the function
\be\label{invasion.1} 
f(y,\bx)\equiv b(y)-d(y) -\sum_{j=1}^k c(y,x_j)\bar n_j.
\ee
\end{definition}

Under relatively mild conditions, one can even assert that a new equilibrium is the unique
 attractor of the $ k+1$-dimensional system.

\begin {proposition} \TH(unique.1) Assume that the competition kernel $c$ 
is reversible with respect to some strictly positive measure $\nu$, and that it 
defines a strictly positive quadratic form with respect to this measure. 
Then there exists a unique vector $\bar n\in \R^{k+1}_+$ such that for any initial condition $n(0)$ with $n_i(0)>0$, for all $i=1,\dots, k+1$, the solution 
of the competitive Lotka-Volterra system converges to $\bar n$.
\end{proposition}

\begin{proof}
This result can be found in \cite{JabinRaoul}. The proof is based on the fact 
that under the assumption made, there exists a convex Lyapunov \index{Lyapunov function}
function, namely
\be\Eq(unique.2)
L(n)=\frac 12 \sum_{i,j=1}^{k+1} \nu(x_i)c(x_i,x_j) n_in_j-\sum_{i=1}^{k+1} \nu(x_i)r(x_i)n_i.
\ee 
\nomenclature{L(n)}{Lypounov function}
To see this, 
note that 
\bea
\frac d{dt} L(n(t))&=&
\sum_{\ell} \frac {\del}{\del n_\ell} L(z(t))\frac {d}{dt}n_\ell(t)\\
=&&\hspace{-3mm}\sum_{\ell} \left(\frac 12\sum_{j}\left( \nu(x_\ell)c(x_\ell,x_j)+\nu(x_j)c(x_j,x_\ell)\right)n_j(t)
-\nu(x_\ell)r(x_\ell)\right)\frac {d}{dt}n_\ell(t)\nonumber\\\nonumber
=&&\hspace{-3mm}\sum_{\ell} \nu(x_\ell)\left(\sum_{j}c(x_\ell,x_j)n_j(t) -r(x_\ell)\right)
n_\ell(t)\left(r(x_\ell)-\sum_j c(x_\ell,x_j)n_j(t)\right)\nonumber\\
=&&\hspace{-3mm}-\sum_{\ell} \nu(x_\ell)n_\ell(t)\left(\sum_{j}c(x_\ell,x_j)n_j(t) -r(x_\ell)\right)^2.
\eea
Here, we used reversibility in the passage to the third line. 
The last expression is clearly negative, and convexity is ensured by the requirement that $\nu c$ defines 
a strictly positive definite quadratic form. The unique critical point and hence the minimum of $L$ is taken
for $n$ such that $\sum_{j}c(x_\ell,x_j)n_j(t) -r(x_\ell)=0$, for all $\ell=1,\dots, k+1$. 
If this solution lies in 
$\R^{k+1}_+$, then this is the unique fixed point $\bar n$ of the LV system.  Otherwise, if we start with initial conditions $n_i(0)>0$ for all $i$, the solutions move towards this fixpoint until one (or more) components hit
zero, i.e. the boundary of the positive quadrant. The solutions then get stuck in this boundary, 
but continue to 
lower $L$ and move towards the unique minimum of $L$ within this hyperplane (the restriction of a convex 
function to a linear submanifold is still convex). This continues until the unique relative global minimum of $L$
 in 
$\R_+^{k+1}$ is reached. 
 In this case, the new equilibrium will 
have fewer than $k+1$ non-zero components. 
\end{proof}

\begin{remark}
For more detailed expositions of Lyapunov's stability theory, see \cite{arnold, khalil}.
\end{remark}

\part{Limit processes}

\chapter{Small mutation limit in the deterministic system} \label{chapter4}

\begin{chapquote}
{Earl Nightingale, \emph{}}
{Never give up on a dream just because of the time it will take to accomplish it. The time will pass anyway.}
\end{chapquote}

In this chapter, we start to discuss the effect of scaling the mutation rate to zero.  This is biologically reasonable 
if we think only of advantageous mutations that change the phenotype to a fitter one. 
There are, of course, also non-coding mutations that may be much more frequent, but do not change the 
phenotype. We will ignore those here. We begin with the deterministic Lotka-Volterra equations with mutation and take the mutation rates to zero. Surprisingly, if we rescale time to infinity appropriately, one can obtain a non-trivial limit. 

\section{Convergence to a deterministic jump process}\label{section40}
It is easy to see (using Gronwall's lemma) that on finite time intervals, solutions converge, as $\mu\downarrow 0$, to those of the system with $\mu=0$.  
\begin{lemma} \TH(uto0.1)
Assume that  $\xi_0=\sum_{i=1}^N n_i(0)\delta_{x_i}$. Then the solutions to the deterministic equation
with $\mu>0$, $\xi^\mu_t$, converges weakly, uniformly in $t\leq T<\infty$, to 
\be\Eq(uto0.2)
\xi_t=\sum_{i=1}^N n_i(t)\delta_{x_i},
\ee
where $n(t)$ is the unique solution of \eqv(LV-System) with initial condition $n(0)$.
\end{lemma}

\index{Lotka-Volterra system}

Note that the Lemma holds whatever the state space $\XX$ is. 
The same is \emph{not} true if time tends to 
infinity as $\mu\downarrow 0$, as we will see shortly. 

We have seen that for finite time horizons, the limit of the deterministic equations as $\mu\downarrow 0$
is a mutation-free ecological equation. The reason for this is that the growth of solutions is at most 
exponential in time, and so anything seeded by a mutation term is proportional to $\mu$ and will vanish in 
the limit. This is no longer true if we consider longer time scales. Namely, if we wait a time of order $\ln 1/\mu$, 
a seed of order $\mu$ can grow to the size of order one. This was first observed in  \cite{BovWang2013} and elaborated in \cite{kraut2018}. In view of the scalings we shall encounter later, then $\mu$ is taken to be a 
function of the carrying capacity $K$, this corresponds to a situation when the mutation rate is 
\emph{large} (even though it tends to zero). Still, the adaptive dynamics paradigm that is maintained in that 
ecology happens at a time scale $1$ while the "mutation" effects require the divergent time scale $\ln (1/\mu)$.
In this regime, we want to make contact with the theory of \emph{ adaptive walks} and therefore chose the \index{adaptive walks} \index{small mutation limit}
base space to be a (finite) graph, mostly the hypercube $\HH_N\equiv \{0,1\}^N$. In the spirit of Maynard Smith,
we assume that mutations can only go along edges of the graph, to wit, towards nearest neighbours on the
hypercube.  For $x,y\in\HH_N$, we will write $|x-y|$ for the graph distance between $x$ and $y$.

To understand this, consider an initial \index{Lotka-Volterra equations!with mutation}
condition that is monomorphic and the simplest case the $\XX$ is just the set $\{1,2\}$. Then the deterministic system can be reduced to the two-dimensional Lotka-Volterra system with mutation,  
\bea\Eq(simex.1)
\frac {dn_1(t)}{dt}&=& n_1(t)(r_1-c(1,1)n_1(t)-c(1,2)n_2(t))-\mu m(1,2)n_1(t)+\mu m(2,1)n_2(t),\nonumber
\\
\frac {dn_2(t)}{dt}&=& n_2(t)(r_2-c(2,2)n_2(t)-c(2,1)n_1(t))-\mu m(2,1)n_2(t)+\mu m(1,2)n_1(t).\nonumber\\
\eea
To simplify the notation we set $m(i,j)=p(x_i)m(x_i,x_j)$.
Assume that $n_1(0)=\bar n_1$is strictly positive   and $n_2(0)=0$. 
Assume further that the invasion fitness of type two is positive, i.e. $r_2-c(2,1)\bar n_1>0$.
Then, at time $0$, we have 
\bea\Eq(simex.2)
\frac {dn_1(0)}{dt}&=& \bar n_1(r_1-c(1,1)\bar n_1)-\mu m(1,2)\bar n_1= -\mu m(1,2) \bar n_1,\nonumber
\\
\frac {dn_2(0)}{dt}&=& +\mu m(1,2)\bar n_1.
\eea
For $\mu$ small, this implies that at time $t=1$, 
\be\Eq(simex.3)
n_1(1) \sim \bar n_1(1-\mu m(1,2)),\quad n_2(1) \sim \mu m(1,2)\bar n_1.
\ee
Hence,   as long as $n_2(t)$ is small compared to $\bar n_1$,
\be\Eq(simex.4)
\frac {dn_2(t)}{dt}\geq     n_2(t)(r_2-c(2,1)\bar n_1).
 \ee
 and hence exponential growth at rate $(r_2-c(2,1)\bar n_1)\equiv R>0$ will set in, i.e 
 for $t>1$ and as long as $n_2(t)$ remains small compared to one, 
 \be\Eq(simex.5)
 n_2(t)\sim \mu m(1,2)\bar n_1 \eee^{(t-1)R},
 \ee 
 and so by time $t\sim \frac  1R \ln (um(1,2)\bar n_1 )$, $n_2$ will have reached a level $O(1)$ 
 that is independent of $\mu$. Then, for vanishing $\mu$, the system will evolve over times of order one, like 
 the \index{Lotka-Volterra system!competitive}
 mutation-free ($\mu=0$) competitive Lotka-Volterra system and approach its unique fixed point 
 $(0,\bar n_2)$. Thus, defining 
 \be\Eq(simex.7)
 Z^\mu(t)\equiv (n^\mu_1(t |\ln \mu|), n_2^\mu(t|\ln \mu|)), 
 \ee
 we see that  
 \be
 \Eq(simex.6)
 \lim_{\mu\downarrow 0} Z^\mu(t)= 
 \bar n_1 \1_{0\leq t<1/R}+\bar n_2\1_{t\geq 1/R}.
 \ee
 So, interestingly, on the time scale $\ln (1/\mu)$, the solution of the deterministic Lotka-Volterra system with 
 mutations converges to a \emph{deterministic jump process}.  \index{jump process!deterministic}
%
\index{polymorphic evolution sequence}
 What we observed in this simple example is generic and gives rise to the first example of a \emph{polymorphic evolution sequence} (PES), by which we mean a jump process between \index{polymorphic evolution sequence}
 equilibria of a sequence of competitive Lotka-Volterra systems. This can be described as follows.
 
 Let  $\XX\equiv \HH_N$.
 The system of differential equations we consider reads (for simplicity, we assume a constant
 mutation rate $\mu$. i.e. $p(x)\equiv 1$)
\bea\Eq(DE)\nonumber
&&\frac{d{n}^\mu_x(t)}{dt}=\left[r(x)-\sum_{y\in\HH_N}c(x,y)n^\mu_y(t)\right]n^\mu_x(t)+\mu\sum_{y\sim x}b(y)m(y,x)n^\mu_y(t)-\mu b(x)n^\mu_x(t),\\ &&\; x\in \HH_N.
\eea

  Let $\bx\subset \HH_N $ be a  subset of cardinality $m$, such that  
 the mutation free sub-system of \eqv(DE),
 $LV(m,\bx)$, namely
\be\Eq(DE.0)
\frac {d{n}_x(t)}{dt}=\left[r(x)-\sum_{y\in\HH_N}c(x,y)n_y(t)\right]n_x(t), \; x\in \bx,
\ee
has an equilibrium  $\bar n$, such that for all component of $\bar n$ are strictly positive.
Note that in this case, the extension of $\bar n$ to $\HH_N$, which is obtained by setting all other components 
equal to zero is an equilibrium of the full system \eqv(DE) with $\mu\equiv 0$. 
 
 Let us assume that we start with such a (pseudo) equilibrium as an initial condition. Then the following will happen.

 \noindent\textbf{Step 1:} 
  At time $1$, all the populations in $\bx$ remain very close to their equilibrium values.
   At all points $x\not\in \bx$, exponential growth can not yet substantially increase the population size.
   Thus, all populations at these sites are either of size 
 zero or of order $\mu^{\a_x}$ with $\a_x\in \N$. The populations at $x\in \bx$ remain close to their equilibrium 
 values. This remains true as long as none of the resident populations has reached a level $1\gg\e>0$
 (independent of $\mu$). 
 
  \noindent\textbf{Step 2:}  The populations at the sites $x\not \in \bx$ grow exponentially with rate 
  given by their invasion fitness with respect to the resident equilibrium until a time $T_{\e,1}$,
  which is the first time that one of the non-resident populations reaches the value $\e$. 
  Population growth also takes into account mutations. The system is, however, well approximated by a linear system.
  $T_{\e,1}$ is of order $\ln (1/\mu)$. 
  
    \noindent\textbf{Step 3:} At time $T_{\e,1}$, assume that the set  $J$ of sites $y$ for which at this time 
    $\lim_{\mu\downarrow 0} n_y(T_{\e,1})\neq 0$ is finite. (typically, this will be $\bx$ plus one new site.) 
    Then, in time of order one, the system will approach the equilibrium of $LV(|J|,J)$. Let $\bx'\subset J$ 
    be the subset on which this equilibrium is strictly positive. All sites outside $\bx'$ have a population size
    of some order $\mu^\a$. 
    
     \noindent\textbf{Step 4:} Restart as in Step 2 and iterate. Note that the initial condition is now not precisely
     of the form we started with, i.e. there are no components that have strictly zero populations. But this is 
     not important.
       \nomenclature{$LVE(\bx)$}{Lotka-Volterra equilibrium with types $\bx$}
       
       We now give a precise formulation of the results.
     
\begin{definition}
For a subset $\mathbf{x}\subset\HH_N$ we define the set of \emph{Lotka-Volterra equilibria} by 
\be
\text{LVE}(\mathbf{x}):=\left\{z\in(\R_{\geq0})^\mathbf{x}:\forall\ x\in\mathbf{x}:\ \Big[r(x)-\sum_{y\in\mathbf{x}}c(x,y)n_y\Big]n_x=0\right\}.\label{equil}
\ee
Moreover, define
\be
\text{LVE}_+(\mathbf{x}):=\left\{z\in(\R_{>0})^\mathbf{x}:\forall\ x\in\mathbf{x}:\ \Big[r(x)-\sum_{y\in\mathbf{x}}c(x,y)n_y\Big]n_x=0\right\}.\label{equil+}
\ee

 If $\text{LVE}_+(\mathbf{x})$ contains exactly one element, we denote it by
$\bar n_\mathbf{x}$, the \textit{equilibrium} of a population of coexisting traits $\mathbf{x}$.
\end{definition}
\begin{remark}
If $\text{LVE}_+(\mathbf{x})=\{\bar n_\mathbf{x}\}$, this implies $r(x)>0$ for all $x\in\mathbf{x}$. In the case where $\mathbf{x}=\{x\}$, we obtain $\bar n_x=\frac{r(x)}{c(x,x)}$. 
\end{remark}

\index{Lotka-Volterra!equilibrium}
We want to think of the limiting process, as $\mu\downarrow 0$, as a jump process on the set of LVEs. However, as long as $\mu>0$, a population will essentially never be a $LVE(\bx)$ with no population outside 
$\bx$. Therefore, we need the notion of a neighbourhood of such an equilibrium. It is clear that what happens on the time-scale $\ln (1/\mu)$ depends only on the exponent of $\mu$, more precisely, for 
any $n^\mu_y$ that tends to zero  with $\mu$, we keep track only of 
$\l_y\equiv \lim_{\mu\downarrow 0} \frac 1{\ln \mu} \ln n^\mu_y$.

\begin{definition}
Let $\bx\subset \HH_N$ and $\text{LVE}_+(\bx)\equiv \{\bar n_\bx\}$.
For $\eta>0$ and $\bar c\equiv \{c_y,C_y\}_{y\in \HH_N}$,  a sequence $n^\mu\in \MM(\HH_N)$ is said to 
satisfy $\text{IC}(\bx,\eta,\bar{c})$, or $n^\mu\in \text{IC}(\bx,\eta,\bar{c})$, if there exists $\l_y\geq 0$, such that,
for all $\mu$ small enough
\begin{align}
n^\mu_y\in[c_y\mu^{\lambda_y},C_y\mu^{\lambda_y}],\label{initial}
\end{align}
where
\begin{align}
\forall~y\in\bx:&\ \lambda_y=0,\ \bar n_\bx(y)-\eta\frac{\bar{c}}{|\bx|}\leq c_y,C_y\leq\bar n_\bx(y)+\eta\frac{\bar{c}}{|\bx|},\label{initialx}\\
\forall~y\in\HH_N\backslash\bx:&\ \lambda_y>0,\ 0\leq c_y,C_y<\infty\quad\text{or}\\
&\ \lambda_y=0,\ 0\leq c_y,C_y\leq\frac{\eta}{3},\ f(y,\bx)<0.
\end{align}
If $n^\mu_y(0)\equiv0$, we set $\lambda_y>\max_{z\in\HH_N: n^\mu_0(z)>0}\lambda_z+N$.
\end{definition}

The point of this definition is the following. If $n_0^\mu\in  \text{IC}(\bx,\eta,\bar{c})$,
and we remove all populations that have size of order $o(1)$, then under the mutation free 
Lotka-Volterra 
system, the population would approach the equilibrium $\text{LVE}_+(\bx)$. As $\mu\downarrow 0$, 
this is still true for $\mu$-independent times, that is, the population will be arbitrarily close to this 
state before a new population outside $\bx$ can grow. On the other hand, after the first population at such a 
site reaches a level of $O(1)$, the entire system will reach in finite time the set $\text{IC}(\bx',\eta,\bar{c})$,
for a new $\bx'$. 

Given an initial condition $n^\mu(0)\in  \text{IC}(\bx^0,\eta,\bar{c})$, we define
\begin{align}\Eq(ro.1)
\rho^0_y:=\min_{x\in\HH_N}[\lambda_x+|x-y|]
\end{align}

The point is that, if initially the population verifies $n_x(0)\sim \mu^{\l_x}$, for all $x$,  then in time one, 
due to mutations (and recall that growth does not have a significant effect in time $1$), the
population at $y$ will be (up to constant factors)
\be
n_y(1) \sim n_y(0) + \sum_{x\neq y} \mu^{|x-y|} n_x(0)  =
\mu^{\l_y}+\sum_{x\neq y} \mu^{|x-y|} \mu^{\l_x} \sim \mu^{ \rho^0_y}.
\ee
Here, we use the simple fact that 
\be
\lim_{\mu\downarrow 0} \frac 1{\ln \mu} \ln \left(\sum_x C_x\mu^{\l_x}\right)=\min_x \l_x.
\ee
  After that, in a time of order 
$\ln (1/\mu)$, the population will evolve by exponential growth until the first new population reaches a finite level, 
i.e. an exponent $0$. 

We will describe the evolution of the population in terms of the exponents of $\mu$ at each type. 
We define the following recursion:
\begin{align}
y^i_*&:=\arg\hspace{-5pt}\min_{\substack{y\in\HH_N:\\f(y,\mathbf{x}^{i-1})>0}}\frac{\rho^{i-1}_y}{f(y,\mathbf{x}^{i-1})},\\
T_i&:=T_{i-1}+\min_{\substack{y\in\HH_N:\\f(y,\mathbf{x}^{i-1})>0}}\frac{\rho^{i-1}_y}{f(y,\mathbf{x}^{i-1})},\\
\rho^i_y&:=\min_{z\in\HH_N}[\rho^{i-1}_z+|z-y|-(T_i-T_{i-1})f(z,\mathbf{x}^{i-1})],\\
\bx^i&=\{y\in \bx^{i-1}\cup y^i_*: \bar n_{ \bx^{i-1}\cup y^i_*}(y)>0\}.
\end{align}
with $\rho^0_y$ given by \eqv(ro.1) and $T_0=0$.

 Let us discuss the meaning of the terms appearing.
$ \frac{\rho^{i-1}_y}{f(y,\mathbf{x}^{i-1})}$ is the time (measured in units of $\ln (1/\mu)$ it takes for the population at $y$ that has initial size $\mu^{\rho^{i-1}_y}$ that grows with positive  rate $f(y,\mathbf{x}^{i-1})$
to reach a size of order $1$. $y^i_*$ is then the site where this happens first, and $T_i$ is the absolute time when this happens.  The formula for the new initial conditions, $\rho^i_y$, is tricky.  It takes into account that
three possible sources could dominate the new initial condition at $y$:
First,  the population at $y$ could have just continued to grow. This gives 
$\rho^i_y=\rho^{i-1}_y -(T_i-T_{i-1})f(y,\bx^{i-1})$. 
Second, it could come from mutants from the large populations in $z\in \bx^{i-1}$. This gives
$\rho^i_y=0+|z-y|$.  Finally, it could come from the mutants that have grown at any other site $z$, which have grown 
over the last period. This gives
$\rho^i_y=\rho^{i-1}_z+ |z-y|-(T_{i}-T_{i-1})f(z,\bx^{i-1})$.  Finally, $\bx^i$ is the set of types 
that have a positive population in the (unique) equilibrium on $\bx^{i-1}\cup y^i_*$. 

This recursion is well defined as long as there is a unique $y^i_*$ in each step. The recursion will 
stop if $f(y,\bx^{i-1})\leq0$ for all $y\in\HH_N\backslash\bx^{i-1}$. In that case, the population has reached \index{evolutionary stable condition}
an evolutionary stable condition at step $i-1$. Otherwise, uniqueness is the generic situation. What happens when there are several minimisers is slightly tricky, and to avoid
technicalities, we exclude this possibility by making the assumption:

\begin{itemize}
\item[\textbf{(C)}]
For every $i\geq 1$, there is either a unique minimiser
\begin{align}
y^i_*=\arg\hspace{-5pt}\min_{\substack{y\in\HH_N\\f(y,\bx^{i-1})>0}}\frac{\rho^{i-1}_y}{f(y,\bx^{i-1})}\label{unqmin}
\end{align}
or $f(y,\bx^{i-1})\leq0$ for all $y\in\HH_N\backslash\bx^{i-1}$.
\end{itemize}

%
%
%
%
%
%

Moreover, to guarantee that the conditions of $\text{IC}(\bx^i,\eta,\bar{c})$ are satisfied, we assume that the former resident traits of $\bx^{i-1}$, which are not part of the new resident traits $\bx^i$, have a negative invasion fitness.
\begin{itemize}
\item[\textbf{(D)}]
For every $i\geq1$ and for every $y\in(\bx^{i-1}\cup y^i_*)\backslash\bx^i$, $f(y,\bx^i)<0$.
\end{itemize}

The final result can be formulated as follows. 

     \begin{theorem} [\cite{kraut2018}] \TH(det-pes.1) 
Consider the system of differential equations (\ref{DE}) and assume (A), (B), (C), and (D). Let $n^\mu_0\in\text{IC}(\mathbf{x}^0,\eta,\bar{c})$, for $\eta$ small enough.
Let $\mathbf{x}^i$ be the support of the equilibrium state of the Lotka-Volterra system involving $\mathbf{x}^{i-1}\cup y^i_*$ and set $T_i:=\infty$, as soon as there exists no $y\in\HH_N$ such that $f(y,\mathbf{x}^{i-1})>0$. Then, for every $t\notin\{T_i,i\geq0\}$,
\begin{align}
\lim_{\mu\to0}\xi^\mu(t\ln(1/\mu))=\sum_{i=0}^\infty \1_{T_i\leq t< T_{i+1}}\sum_{x\in\mathbf{x}^i}\delta_{x}\bar n_{\mathbf{x}^i}(x).
\end{align}
\end{theorem}

 This result should be plausible from the discussion above. A detailed proof is given in \cite{kraut2018}.

\section{Equal competition}\label{section41}
A special case arises if we assume equal competition between all types.
 In this case, one can simplify the description of the limit process.

We introduce the additional assumption
\begin{itemize}
\item[\textbf{(E)}] For every $x,y\in\HH_N$, $c(x,y)\equiv c>0$.
\end{itemize}

First, we see that coexistence is getting even more exceptional.
Namely, if the traits $\bx$ coexist, then we must have \index{coexistence}

\be\Eq(co.1)
r(y)-c\sum_{x\in \bx} \bar n_x =0, \quad \forall y\in \bx,
\ee
implying that $r$ is constant on $\bx$. Moreover, in that case, the equilibrium is degenerate, and any distribution of the population on these coexisting traits is an equilibrium. From a biological point of 
view, these coexisting traits are not distinguishable, so we might as well lump them together.
We will therefore assume that for all $x,y\in\HH_N$, $r(x)\neq r(y)$. This then implies that there are 
only monomorphic equilibria.

 The invasion fitness $f(x,y)$ then satisfies \index{invasion fitness}
\begin{align}
f(x,y)=r(x)-c(x,y)\bar n_y(y)=r(x)-r(y),
\end{align}
that is, the invasion fitness of $x$ is positive whenever $x$ has a higher a-priori fitness than $y$. 
This also implies 
\begin{align}
f(x,y)=-f(y,x)~\text{ and }~f(x,y)+f(y,z)=f(x,z).\label{trans}
\end{align}
%
Therefore, $r(x)$ now represents the absolute fitness of the trait $x$. Traits with larger $r$ are accessible from traits with smaller $r$. The competition kernel only regulates the total population size.

The case of constant competition leads to significant simplifications on the level of the Lotka-Volterra
systems. Namely, it becomes equivalent to a \emph{linear} system.  To see this, introduce the 
new functions
\be\Eq(lin.1)
w_x(t)\equiv n_x(t)\exp\left(c\int_0^t \sum_y n_y(s) ds\right). 
\ee
One readily shows that $w$ satisfies the linear equations
\be\Eq(lin.2)
\frac {d}{dt}w_x(t)=(r(x)-\mu b(x))w_x(t) +\mu\sum_{y\neq x}b(y)m(y,x) w_y(t)
\ee 
with the same initial conditions as $z$.
Note that the frequencies of both $n$ and $w$ are the same, since
\be\Eq(lin.4)
\frac{w_x(t)}{\sum_yw_y(t)}=\frac{n_x(t)}{\sum_yn_y(t)}.
\ee


As before, we assume (C) to ensure that there is always a unique mutant that reaches the threshold of order 1 first after an invasion. Starting with only a single trait at its equilibrium size, i.e.\ $\bx^0=\{x^0\}$, this also implies that we avoid coexistence and always maintain a monomorphic resident population. This is because an invading trait has to have a higher rate $r$ than the current resident trait, which prevents a polymorphic Lotka-Volterra equilibrium.

%

The limiting adaptive walk can be described in a simpler way.
\index{adaptive walk}

\begin{theorem}[\cite{kraut2018}]\label{EqComp}
Consider the system of differential equations (\ref{DE}) and assume (A), (C), and (E). Let $n^\mu_0\in\text{IC}(\{x^0\},\eta,\bar{c})$ such that $\lambda_y\geq |y-x^0|$, for all $y\in\HH_N$ and $\eta$ small enough. Define
\begin{align}
x^i&:=\arg\hspace{-5pt}\min_{y\in\HH_N: r(y)>r(x^{i-1})>0}\frac{|y-x^0|-|x^{i-1}-x^0|}{r(y)-r(x^{i-1})},\\
T_i&:=\frac{|x^i-x^0|-|x^{i-1}-x^0|}{r(x^i)-r(x^{i-1})}.\label{argmin}
\end{align}
Set $T_i:=\infty$, as soon as there exists no $y\in\HH_N$ such that $r(y)>r(x^{i-1})>0$.
Then, for every $t\notin\{T_i,i\geq0\}$,
\begin{align}
\lim_{\mu\to0}\xi^\mu({t\ln({1}/{\mu})})=\sum_{i=0}^\infty \1_{T_i\leq t< T_{i+1}}\delta_{x^i}\bar n_{x^i}(x^i).
\end{align}
%
%
\end{theorem}
\begin{remark}
(i) The process described in Theorem \thv(EqComp) can be called a \emph{Trait Substitution Sequence} \index{trait substitution sequence}
(TSS) or an \emph{Adaptive Walk}. 
(ii) The adaptive walk in Theorem \ref{EqComp} continues as long as there is a trait with higher individual fitness, i.e.\ a higher rate $r$. As a result, the walk can cross arbitrarily large valleys in the fitness landscape (defined by $r$) and eventually reaches the global fitness maximum, where it remains.
(If we do not assume that $r$ does not take the same value twice,  this global maximum does not have to be unique. The adaptive walk then reaches a maximum that is closest to $x^0$ in $\HH_N$).

(iii) Every invasion step increases the distance on $\HH_N$ between the resident trait and $x_0$. This can be seen inductively. Consider the $(i+1)^\text{st}$ invasion. $x^i$ was a minimiser of $(|y-x^0|-|x^{i-1}-x^0|)/f(y,x^{i-1})$. If $y$ satisfies $f(y,x^i)>0$, then
\begin{align}
\frac{|y-x^0|-|x^{i-1}-x^0|}{f(y,x^{i-1})}\geq\frac{|x^i-x^0|-|x^{i-1}-x^0|}{f(x^i,x^{i-1})},
\end{align}
and since $f(y,x^{i-1})=f(y,x^i)+f(x^i,x^{i-1})>f(x^i,x^{i-1})$ and $|x^i-x^0|>|x^{i-1}-x^0|$, $|y-x^0|-|x^{i-1}-x^0|>|x^i-x^0|-|x^{i-1}-x^0|$, and hence $|y-x^0|>|x^i-x^0|$.   
  \end{remark}
     
     Theorem \thv(det-pes.1) provides a limiting process that describes the
     evolution under the effects of ecology and mutation/migration as predicted by adaptive dynamics. \index{adaptive dynamics}
     While it captures some interesting features of evolution, several unrealistic features can be seen as drawbacks.
     
     \begin{enumerate}
     \item In time $1$, all sites that can be reached from an initial population through finitely many mutations become populated and immediately start to grow if their invasion fitness is positive. This is
     also true if such a site can only be reached through a sequence of mutations that are all deleterious.
     The time for invasion of such sites is then only  of order $k\ln (1/\mu)$, where $k$ is the minimal number of steps required.
   \item The derivation of this limiting process from the underlying stochastic model is somewhat inconsistent.  Convergence to the deterministic limit as $K\uparrow \infty$ is proven for 
     given $\mu$ for finite time intervals $[0,T]$.  
   \item Even arbitrarily small populations grow exponentially, even though for a fixed $K$ population, 
     if $\mu^\a K<1$, there will be no individual and thus no growth can occur. 
     \end{enumerate} 
     
     Effectively, this means taking the limit $K\uparrow\infty$ first and $\mu\downarrow 0$ later,
     we are still dealing with ``very large" mutation rates, and this leads to unwanted effects.
     A more realistic treatment thus requires taking a joint limit when $K\uparrow \infty$ and $\mu=\mu_K\downarrow 0$ are taken simultaneously.

\section{Convergence for a limited radius of mutation}\label{section42}

In the setting above, already after an arbitrarily small time, mutation has induced a positive population size for every possible trait. These mutant populations have size of order $\mu$ to the power of the distance to $x^0$ on $\HH_N$. The next invading trait is then found balancing low initial $\mu$-power and high fitness. If one thinks of this 
as an approximation to a stochastic system with individuals, this is only good if $\mu^NK\gg 1$. If we think of $\mu$ 
as $K$-dependent rate, say, $\mu\sim K^{-1/\a}$, then there will be more than one mutant at distance $d$ 
if $\mu^dK=K^{1-d/\a}\gg 1$. At any further distance, there are no mutants at finite times. This can be 
mimicked by putting all populations below a given threshold into a dormant state where they can 
neither die nor reproduce. They can still grow, however, due to incoming mutations.
 This has the effect that mutants can be spread only over a finite distance. 
 We will return to the analysis of the stochastic system with such mutation rates later.

The corresponding differential equations are given by
\bea\label{DE'}
\frac{n^\mu_x(t)}{dt}&=&\left( r_\ell(x)-\sum_{y\in\HH_N}c(x,y)n^\mu_y(t)\right)n^\mu_x(t)\1_{n^\mu_x(t)\geq u^{\ell-1}}
\notag\\
 &+&\mu\sum_{y\sim x}n^\mu_y(t)b(y)m(y,x)\1_{n^\mu_y(t)\geq \mu^{\ell-1}}-\mu n^\mu_x(t)b(x)\1_{n^\mu_x(t)\geq \mu^{\ell-1}}.
\eea

For $\ell\geq k$, we just recover the original scenario of Theorem \ref{det-pes.1}. 

For $\ell=1$, we obtain the greedy adaptive walk of \cite{NoKr15}, where the process always jumps to the fittest direct neighbour of the current resident trait. We keep the assumptions of constant competition and monomorphic initial conditions.

The convergence to a greedy adaptive walk can be stated as follows. \index{adaptive walk!greedy}
\begin{theorem}\label{AW}
Consider the system of differential equations (\ref{DE'}) for $\ell=1$ and assume (A), (C), and (E). Let $n^\mu(0)\in\text{IC}(\{x^0\},\eta,\bar{c})$ such that $\lambda_y\geq 1$ for all $y\sim x^0$, $n^\mu_y(0)=0$,
 for $|y-x^0|\geq2$, and $\eta$ small enough. Define
%
\begin{align}
x^i&:=\arg \max_{y\sim x^{i-1}}r(y),\\
T_i&:=T_{i-1}+\frac{1}{r(x^i)-r(x^{i-1})}.
\end{align}
Set $T_i:=\infty$, as soon as there exists no $y\sim x^{i-1}$ such that $r(y)>r(x^{i-1})$.
Then, for every $t\notin\{T_i,i\geq0\}$,
\begin{align}
\lim_{\mu\to0}\xi^\mu({t\ln({1}/{\mu})}=\sum_{i=0}^\infty \1_{T_i\leq t< T_{i+1}}\delta_{x^i}\bar n_{x^i}(x^i).
\end{align}
\end{theorem}

\begin{remark}
(i) The adaptive walk in Theorem \ref{AW} stops as soon as it reaches a local maximum of the individual fitness $r$ since only direct neighbours of the resident trait can be reached. Local maxima do not need to be strict. However, as in the previous cases, mutants with invasion fitness $0$ cannot invade the resident population.

(ii) It is no longer the case that every step increases the distance to $x^0$. The walk could return to a trait close to $x^0$, which just could not be reached before because one had to go around a valley in the fitness landscape defined by $r$.
\end{remark}

In Chapter \ref{chapter5} we will show that the picture described in the last section arises from the stochastic models 
in the joint limit $K\uparrow \infty$ and $\mu_K\downarrow 0$ for appropriate choices of $\mu_K$.

 \chapter{Stochastic systems with very small mutation rates}\label{chapter5}

\begin{chapquote}
{Charles Darwin, \emph{The Origin of Species}}
%

{As many more individuals of each species are born than can possibly survive; and as, consequently, there is a frequently recurring struggle for existence, it follows that any being, if it vary however slightly in any manner profitable to itself, under the complex and sometimes varying conditions of life, will have a better chance of surviving, and thus be \emph{naturally selected}.}
\end{chapquote}

In this chapter, we look at limiting processes that arise if  \index{carrying capacity}
mutation rates tend to zero simultaneously as the carrying capacity tends to 
infinity. There are, of course, many ways how this can happen, and the 
behaviour of the system will depend on this. 
In this chapter, we impose conditions (first introduced in Champagnat \cite{Cha06}) that ensure that we can be sure that the basic postulate of 
adaptive dynamics, namely that the time scales of \emph{ecology} and 
\emph{evolution} \index{adaptive dynamics}
are well separated, holds in a strong sense. This is the case if the fate of an appearing 
mutant (and in fact of the entire population) is determined before a new mutant appears. If, as we will
also assume here, the evolutionary advantage of a mutant is positive, 
independent of $K$, the time for a single mutant to produce a number of
offspring of order $K$ will be of order $\ln K$, and the competition with the
resident population will lead close to a new equilibrium in time of order $1$, 
finally, an unfit resident will die out in time of order $\ln K$. 
Thus, to satisfy our assumption, the time between consecutive mutants 
must be larger than $\ln K$, which, given that there are $K$ individuals 
around, means that $\mu_K\ll (K\ln K)^{-1}$. The results in this chapter are 
based on the papers \cite{Cha06,C_ME,CM11} by Champagnat and co-workers.

%

\section{Heuristics} Under the conditions above, we can expect that the following picture holds for 
a population that started with an initial condition 
where only a finite number of phenotypes were present.
\begin{itemize}
\item[(i)] For almost all times, the population is very close to an 
ecological equilibrium where only a finite number $d$ of phenotypes are present. They then determine an (invasion)-fitness landscape. 
\item[(ii)] Mutants that are born from such an equilibrium at a phenotype where
the invasion fitness is negative, die out with probability one. 
\item[(iii)] Mutants that are born from such an equilibrium at a phenotype where
the invasion fitness is positive produce $\e K$ (with $1\ll\e>0$) offspring before they die out 
with strictly positive probability. If they produce this number of offspring, 
this takes time $O(\ln K)$.  
\item[(iv)] From the time when the mutant population has reached the level $\e K$, the population stays close to the solution of the
mutation-free deterministic Lotka-Volterra 
system of dimension $d+1$. Under mild hypotheses, this system reaches 
the $\e$-neighbourhood of 
a unique equilibrium with $k\leq d+1$ non-zero components.
\item[(v)] In time of order $\ln K$, the populations corresponding to the 
$d+1-k$ zero-components of this equilibrium die out. 
\end{itemize}

Mathematically, these statements are justified by different methods: 
\begin{itemize}
\item[(i)] To show that the system stays close to the initial equilibrium 
can be shown using Wentzell-Freidlin type large deviation principles
that show that the escape from  an $\e$-neighbourhood of a stable equilibrium takes time of order
 $\exp(K)$. 
Alternatively, at least in the case of monomorphic populations, one can use estimates on 
birth-and-death processes.
\item[(ii, iii)] The initial fate of a mutant can be determined by branching
process theory. Until the time when a mutant population is of order $\e$, 
it affects the resident population only by an $O(\e)$ shift to the position of 
the equilibrium. Its own behaviour can thus be controlled by branching processes.
\item[(iv)] The analysis of the deterministic system in dimension $d+1$ is 
difficult if not impossible. There are, however, results on competitive 
Lotka-Volterra systems that state when we get convergence to a unique new equilibrium. 
\item [(v)] The last step again is branching process theory.
\end{itemize}

The concept outlined above indicates that the population process that will 
emerge can be seen as a jump process between ecological equilibria of
systems of competitive Lotka-Volterra equations of various dimensions. 
An important and difficult question is what the nature of these equilibria will be.

\paragraph{Monomorphism}
We have seen that polymorphic equilibria require specific \index{equilibria! monomorphic}
 \index{equilibria!polymorphic}
conditions on the coefficients, while monomorphic equilibria always exist. If we 
start our population process with a monomorphic population at time zero, 
it will thus approach its ecological equilibrium and reach an 
$\e$-neighbourhood of it in finite time and will stay there with overwhelming 
probability until a mutant appears. If that mutant appears at a type that has
positive invasion fitness, we have seen above that the population will now 
move towards a unique fixed point of the $2$-dimensional system. 
Now there are two possibilities: either this fixed point is monomorphic, or
it is bi-morphic. If it is monomorphic, it must also be stable, so it cannot be
$(\bar n_1,0)$ (this is unstable by assumption), so it must be $(0, \bar 
n_2)$. If we assume that the mutants do not differ much from the residents, the 
ratios of the competition kernels in these equations should be very close 
to one.Then, unless $r(x_1)\approx r(x_2)$, both equations will not hold. 
Thus, the monomorphic fixed point $(0,\bar n_2)$ will be approached and the
population of type $x_1$ will die out. We call this a \emph{trait 
substitution}. This is the generic scenario. The opposite case, when we 
obtain two co-existing types, is called \emph{evolutionary branching}. 
It occurs only if either the two types have almost the same a-priori fitnesses
or it the cross-competition is very weak. 

Starting with a monomorphic initial condition, successive successful 
mutations will thus lead to a sequence of monomorphic populations 
evolving, in some sense, towards higher fitness until
a so-called \emph{evolutionary singularity}. \index{evolutionary singularity}
The precise convergence of the population process towards such a 
\emph{trait substitution sequence} was first derived rigorously by Nicolas 
Champagnat \cite{Cha06}. \index{trait substitution sequence}

There are two types of evolutionary singularities that can be met: either
a trait is reached and a mutation occurs, such that coexistence of the 
resident and mutant trait is possible, i.e. evolutionary branching occurs.
The other possibility is that the traits of all possible mutants have 
negative invasion fitness. In that case, the final monomorphic population 
that is reached represents an \emph{evolutionary stable condition} in the 
\index{evolutionary stable condition}
sense of adaptive dynamics. In that case, evolution appears to come to a 
halt, at least on the time scale of the trait substitution sequence. 

In the case that there is evolutionary branching, a likely scenario is that the
two types under further mutations perform trait substitutions by themselves 
and move away from each other in order to reduce competition. This will 
go on until one of the sub-populations reaches a further evolutionary 
singularity, where the same process is repeated, until, finally, the entire 
polymorphic population has reached an evolutionary stable condition.

\section {Rigorous results}

We now turn to the rigorous statement concerning the PES, following 
\cite{CM11}). \index{polymorphic evolution sequence}
We begin by defining a (strong) notion of coexisting traits. 
To ensure that the process jumps on the evolutionary time scale from one equilibrium to the next, we need an assumption to prevent cycles, unstable equilibria or chaotic dynamics in the  
deterministic system.

\begin{assumption}\label{conv_to_fixedpoint}
For any given traits $(x_1,\ldots ,x_d)\in \XX^d$ 
that coexist and for any mutant trait $y$  
such that $ f(y,\bx)>0$, there exists a neighbourhood $U$ of  $(\bar n(\bx),0)$
 such that all solutions of  $LV(d+1,  (\bx,y)$ with initial condition in 
 $U\cap (0,\infty)^{d+1}$
converge, as $t \uparrow  \infty$, to a unique locally strictly stable
 equilibrium in $\R_+^{d+1}$ denoted by $\bar n^{*}((\mathbf x,y))$.
\end{assumption} 

We write $\bar n^{*}$ and not $\bar n$ to emphasise that 
some components of $\bar n^{*}$ can be zero. Under our assumption, there exists a unique 
equilibrium $LV(d',\bx')$ with $d'\leq d+1$ and $\bx'\subseteq (\bx, y)$ where $\bx'$ are the points 
in $\bx\cup y$ for which $\bar n^{*}_x>0$. 

Assumption \ref{conv_to_fixedpoint}  does not have to hold for all traits, but  only for those traits $y$ which can appear
by mutation.

We now state the main theorem.

\begin{theorem}[Champagnat and Méléard \cite{CM11}]\label{PES.0}
Suppose that Assumption \ref{conv_to_fixedpoint} holds. 
Fix $x_1,\dots, x_d$ coexisting traits and assume that the initial conditions 
converge almost surely to $\bar n (\mathbf x)$. 
 Furthermore, assume  that 
\begin{align}\label{Conv_Cond}
\forall V>0, \qquad \exp(-VK)\ll u_K \ll \frac{1}{K\ln(K)}, \qquad \text{as } K\uparrow \infty.
\end{align}
Then, the sequence of the rescaled processes $(\nu^{K}_{ t/ Ku_K})_{t\geq 0}$ with initial state $\nu_0^K$, converges in the sense of finite-dimensional distributions to the measure-valued pure jump process 
$\Lambda$, which is defined as follows:
$\Lambda_0=\sum_{x\in\bx} \bar n_x (\mathbf x) \delta_{x}$
and the process $\:\Lambda\:$ jumps for all $y\not\in \bx$  from
 \be 
 \sum_{x\in\bx} \bar n_{x}(\bx) \delta_{x}
\quad \text{ to } \quad
  \sum_{x\in \bx\cup y)} \bar n^*_{x}(\bx\cup  y) \delta_{x}
  \ee
  with rate
  \be
  \sum_{x\in \bx} p(x)b(x) \bar n_{x} (\bx)    \frac{ f(y,\bx)_+}{b(y)}m(x, dy).
  \ee
  The process $\L$ is called the \emph{polymorphic evolution sequence (PES)}.
\end{theorem}

\index{jump process!measure valued}

\begin{remark} The PES is absorbed in the set of equilibrium states $LV(n,\bx)$ for which
$f(y,\bx)\leq 0$ for all $y$ that can be reached by a single mutation from $\bx$. These states are 
the \emph{ecologically stable conditions} of adaptive dynamics.
\end{remark}  \index{evolutionary stable condition}

\begin{remark}
The lower bound on $\mu_K$ is not a serious restriction. It is there to ensure that the population will not die out before a first mutation occurs.
\end{remark}
%

\subsection{Control of the resident populations} 
There are several ways to control the stability of the resident population over long periods of time. The simplest and most general one is to use the Wentzell-Freidlin
theory of large deviations and in particular, their results on the problem of exit from an attracting domain \cite{FW84}. This was used in \cite{CM11}. 
We quote their formulation of this result.

\begin{lemma}[Proposition A.2 in \cite{CM11}]\TH(ldp.1)
Fix $\bx \in\XX^d$ and let $\bar n(\bx)$ be the equilibrium.  
 Assume that, for $\e>0$ small enough, 
  $\nu^K_0=\sum_{j=1}^d\tilde{n}^j_0\delta_{x_j}$ such that
  \be\Eq(close.1)
  \left\|\bar n(\bx)- \tilde n(\bx)\right\| <\e/2.
  \ee
     Let $\tau$ be the first mutation time.Then there exist constants $V,c$ such that the exit time of $ \left(\langle \nu^K_t,\1_{x_i}\rangle\right)_{i=1,\dots, d}$ 
     from the $\e$-neighborhood of $\bar n(\bx)$ is larger than $\eee^V\land \t$.
     Moreover, the same conclusion is valid if the death rates of individuals with trait $x_i$ are perturbed by a random process that is uniformly bounded by $c\e$.
     \end{lemma}
%

In the case of a monomorphic equilibrium, one can get a sharper result that will be needed in the next chapter.

\begin{lemma}
  \label{lem:step-1}
  Fix $x,y_1\ldots,y_m\in{\XX}$. Assume that, for $\a\in(0,1/2)$, 
  $\nu^K_0=
  n_0\delta_{x}+\sum_{j=1}^m\tilde{z}^j_0\delta_{y_j}$ such that
  \begin{gather}
   0\leq \tilde{z}^j_0<K^{-\frac 12+\a}\quad\forall j\in\{1,\ldots,m\} \\ \mbox{and}\quad
    |n_0-\bar n_x|<K^{-\frac 12+\a}. \Eq(close.101)
  \end{gather}
   Let $\tau$ be the first mutation time,
  \begin{equation}
    \theta:=\inf\{t\geq 0:\exists
    j\in\{1,\ldots,m\},\langle\nu^K_t,\mathbf{1}_{\{y_j\}}\rangle>K^{-\frac 12+\a}\}
  \end{equation}
  and, for  $M>0$, 
  \begin{equation}
    \theta_M:=\inf\left\{t\geq 0: \left|\langle\nu^K_t,\mathbf{1}_{\{x\}}\rangle
        -\bar n_x\right|>MK^{-\frac 12+\a}\right\}.
  \end{equation}
  Then, there exist constants $C>0$ and $M>0$ that depend only on
  $x$ such that
  \begin{equation}
    \label{eq:lem-1}
    \lim_{K\uparrow \infty}\P\left(\theta_M<\eee^{CK^{2\a}}\wedge\theta\wedge\tau\right)=0.
  \end{equation}
\end{lemma}

\begin{remark} The case $\a=0$ can be included with $K^{-\frac 12+\a}$ replaced by $0<\e\ll 1$.
\end{remark}

\begin{proof}
Define
\begin{equation}
\label{eq:lem-1.5}
X_t\equiv K\langle \nu_t^K,\1_x\rangle- [K \bar n_x].
\end{equation}
Then clearly 
\begin{equation}
\label{eq:lem-1.6}
\theta_M\geq \inf\left\{\tau\geq 0: X_t\geq MK^{1/2+\a}\right\}.
\end{equation}

Let us denote by $\tau_0$ the stopping time
\begin{equation}
\label{eq:lem-1.1}
\tau_0\equiv\inf \{t>0: X_t=0\}
\end{equation}
and
\begin{equation}
\label{eq:lem-1.2}
\tau_1\equiv\inf \{t>0: X_t\geq MK^{1/2+\a}\}.
\end{equation}

We will exploit a natural renewal structure by proving bounds on
probabilities 
\begin{equation}
\label{eq:lem-1.3}
\P_a \left( \tau_1<\tau_0, \tau_1\leq \theta\wedge\tau\right),
\end{equation}
where $\P_a$ denotes the law of the process with any initial
 condition such that $X_0=a$.

We can associate with the continuous time process $X_t$ a
 discrete time (non-Markovian) process $Y_n$ 
which records the sequence of values that $X_t$ takes (this can
be formally defined by introducing  sequences $\sigma_k$ of stopping
times which record the instances when $X_t\neq X_{t^-}$ and 
setting $Y_n=X_{\sigma_n}$. 
Clearly, the probabilities \eqref{eq:lem-1.3} can be computed with
 respect to the stopping times defined for the discrete time process $Y_n$.
Finally, we want to compare these probabilities to those of a Markov
 process.
To this end we first observe that while $t_1<\theta\wedge \tau$, we can
 find a constant, $c>0$, such that, for $0\leq k<K^{1/2+\a}$,
\begin{equation}
\label{eq:lem-1.10}
\P\left(Y_{n+1}=k+1|Y_n=k,\sigma_{n+1} <\theta\wedge\tau \right)\leq \frac
12-c k/ K + C K^{-1/2+\a}\equiv p_+(k),
\end{equation}
for some positive constants $c,C$.
Clearly,
\be
\label{eq:lem-1-4}
\P\left(Y_{n+1}=k-1|Y_n=k,\sigma_{n+1} <\theta\wedge\tau
\right)=1-\P\left(Y_{n+1}=k+1|Y_n=k,\sigma_{n+1} <\theta\wedge\tau
\right).
\ee
Note that the terms $CK^{-1/2+\a}$ takes into account the possible effect
  of the ``small'' populations at the loci $y_j, j=1,\dots, m$.
Now we introduce a coupling, i.e. we define  processes, $Z_n$, with the
following properties:
\begin{itemize}
\item[(i)]
$Z_0=Y_0$;
\item[(ii)] $\P\left(Z_{n+1}=k+1| Y_{n+1}=k+1, Y_n=Z_n=k\right)=1$;\\
\item[(iii)] 
 $\P\left(Z_{n+1}=k+1| Y_{n+1}=k-1, Y_n=Z^+_n=k\right)=$\hfill\break
$\ p_+(k)-
\P\left(Y_{n+1}=k+1|Y_n=k,\sigma_{n+1} <\theta\wedge\tau \right)
$;\\
\item[(iv)] 
$\P\left(Z_{n+1}=k+1| Y_n<Z_n=k\right)
=p_+(k)$;
\item[(v)] $\P\left(Z_{n+1}=k-1| Y_n<Z_n=k\right)
=1-p_+(k)$;
\end{itemize}
Note that by construction $Z_n\geq Y_n$, and so 
\be
\P_a\left(\tau_1^Y<\tau_0^Y\right)\leq 
\P_a\left(\tau_1^Z<\tau_0^Z\right).
\ee
 On the other hand, the marginal
distribution of $Z_n$ is a Markov chain with transition probabilities
\begin{equation}
\label{eq:lem-1.8}
\P(Z_{n+1} =y+1|Z_n=y)=1-\P(Z_{n+1} =y-1|Z_n=y)=p_+(y).
\end{equation}
Standard potential theoretic arguments show then that
\begin{equation}
\label{eq:lem-1.9}
\P_a\left(\tau^Z_{M K^{1/2+\a}} <\tau^Z_0\right)\leq \exp\left(-c\left(M^2
  K^{2\a}-a^2/K\right)\right),
\end{equation}
provided $M$ is large enough.

In the same way, we can construct a coupled process that  
with transition probabilities $\frac 12-c y/ K -m C K^{-1/+\a}\equiv p_-(y)$ stays below $Y_n$.

It is now elementary to see that the number of returns of the process
to $0$ before it reaches $\pm MK^{1/2+\a}$ is geometric with parameter 
$ \exp\left(-c M^2
  K^{2\a}\right)$, and taking into account that the mean of
  $\sigma_n-\sigma_{n-1}$ is  at most of order $K$, the claimed result
  follows easily.
\end{proof}

\subsection{Initial growth of the mutant}
If the resident population is monomorphic and of size $K\bar n_x(1\pm K^{-1/2+
\a})$, the sequence of birth events with mutation is a Poisson process of intensity $K\mu_K b(x) \bar n_x1\pm K^{-1/2+\a})$,
and so the time of the first occurrence of a mutant is an exponential random variable with this parameter.  Then this happens, the resident traits will with overwhelming probability still satisfy 
\eqv(close.1). 
Let the mutant appear at a trait $y$ such that $f(y,x)>0$.  
We now must control what happens until the population at $y$ reaches a size of order $\e K$;
 after that, we can use the LLN to control the evolution of the polymorphic population.
The  following lemma gives a precise value for the probability that the mutant at $y$ reaches
this level before it dies out.

\begin{lemma} \TH(fixation.1)
  Fix $x,y_1\ldots,y_m\in{\XX}$ such that $f(y,x)>0$ and assume that, for $\a\in(0,1/2)$, 
  $\nu^K_0=
  n_0\delta_{x}+K^{-1}\d_y$ such that
  \be\Eq(fixation.2)
    |n_0-\bar n_x|<K^{-\a}. 
  \ee
  Let $\t_z\equiv \inf\{ t>0: \langle \nu^K_t,\1_y\rangle = z/K\}$.
  Then, for any $\e>0$, 
  \be\Eq(fixation.3)
  \lim_{K\uparrow \infty}
  \P(\t_{\e}<\t_0)=\frac{f(y,x)}{b(y)}.
  \ee
%
\end{lemma}

\begin{proof} The proof uses Lemma  \thv(lem:step-1) and is similar to its proof. 
In order to get the sharp estimate in \eqv(fixation.3), we first chose $\a>0$ and show that, 
\be\Eq(fixation.4)
\frac{f(y,x)-MK^{-\a}}{b(y)}\leq
\P(t_{K^{-1/2+\a}}<\t_0)\leq\frac{f(y,x)+MK^{-1/2+\a}}{b(y)}.
\ee
This follows from coupling $\langle\nu^K_t,\1_y\rangle $ to    birth-and death processes,  and using 
the control on the resident population given in Lemma  \thv(lem:step-1), and  \eqv(yule.3). 
Finally, again by  \eqv(yule.3), once the level $K^{-1/2+\a}$ is reached, the probability to reach $\e$ tends to 
one even with the poorer  $O(\e)$ control on the resident. 
\end{proof}

We need one further result that ensures that the time to reach the $\e$ level is not larger than 
$O(\ln K)$. But this follows from   Lemma \thv(prop:BP).
This result ensures that the time of fixation happens surely before the occurrence of the next mutant,
due to the condition $\mu_K\ll 1/(K\ln K)$. 

Finally, we can compute the law of the first mutant that fixates. This depends on both of the 
mutation kernel and the invasion fitness in a very 
intuitive way.

\begin{lemma} \TH(mutant.1)
Let law of the trait of  the first mutant that reaches a level $\e$ is given by
\be\Eq(mutant.2)
\frac{
\frac {f(y,x)_+}{b(y)} m(x,dy)}{\int \frac {f(y,x)_+}{b(y)} m(x,dy)} .
\ee
The time of the first fixation of  a mutant trait (divided by $\mu_K K$) is exponentially distributed with a parameter  that converges to
\be\Eq(mutant.3)
 p(x) b(x)\bar n_x \int \frac {f(y,x)_+}{b(y)} m(x,dy),  \quad K\uparrow\infty.
\ee
\end{lemma}

\begin{proof} 
Clearly $m(x,dy)$ is the probability that the first mutant has trait $y$, while 
$\frac {f(y,x)_+}{b(y)} $ is the probability that a mutant of type $y$ fixates.
Normalising gives the probability that the first fixating trait is of type $y$. This gives \eqv(mutant.2).
Next, mutations happen as a Poisson process with intensity $u_KKp(x) \bar n_x$
and their average success rate is $\int \frac {f(y,x)_+}{b(y)} m(x,dy)$. 
Thus the first successful mutation arises with rate 
$u_KKp(x) \bar n_x\int \frac {f(y,x)_+}{b(y)} m(x,dy)$. Since $\mu_K K\ll 1/\ln (K)$,
the claim \eqv(mutant.3) follows.
\end {proof}

\chapter{Small mutation steps}

\begin{chapquote}
{J.A.J. Metz et al.
\emph{Adaptive Dynamics,
a geometrical study of the consequences of
nearly faithful reproduction}}
{We always think of the world as intrinsically noisy. This not only does away with some
considerable mathematical complications, but it also
has the advantage of being realistic.}
\end{chapquote}

In this chapter, we turn to the third scaling parameter of the model, the size of a mutation step,
or the size of the evolutionary advantage of a mutant.  This assumption was inherent in Darwin's 
thinking of the evolution of continuous traits. More specifically, this assumption was made in the derivation of diffusion models, in particular in Fisher's 
paper from 1918 \cite{fisher18}. In the context of adaptive dynamics, the assumption of a continuous
evolution leads to the idea of the \emph{canonical equation of adaptive dynamics}, which
describes the continuous evolution of a (monomorphic) population in a fitness landscape under 
the influence of mutation and selection. So while this may not cover all possible biological 
situations, it is certainly relevant to analyse this limit in our models.
\index{continuous evolution}

\section{The canonical equation of adaptive dynamics}
\index{canonical equation of adaptive dynamics}

In this section, we discuss the convergence of the process to the so-called 
\emph{canonical equation of adaptive dynamics (CEAD)}.  
We think of a situation where, in the limit $K\uparrow\infty,\mu\downarrow0$, we observe, with
monomorphic initial conditions, convergence to a trait substitution sequence, i.e. the population remains 
monomorphic. If, moreover, the trait space is some nice metric space, say $\R^d$ or a subset thereof, we
can consider a further limit as $\s$, the step-size of a mutation, tends to zero. One then expects  
to obtain, in the limit, a smooth \emph{deterministic} motion of a population that remains monomorphic and follows the gradient of the invasion fitness. This motion will be governed by an integro-differential equation,
called \emph{canonical equation of adaptive dynamics}.  The CEAD drives populations towards singular points where 
the gradient of the invasion fitness vanishes. As we will discuss later, these are points when the population may
become polymorphic, and we say that \emph{evolutionary branching} occurs. After such a branching, 
subpopulations may resume evolving according to a CEAD-type of equation, but this gets more involved. 
\index{evolutionary branching}

Convergence to the CEAD was first
 proven when the limit $\s\downarrow 0$ was taken \emph{after} the limit $K\uparrow\infty$, see
\cite{C_CEAD} and for a trait space $\XX\subseteq \R$. 
%
%
%
%

\begin{theorem} \label{main_thm}
Consider the process generated by $\LL^K$ with mutation rate $\mu_K$ that satisfies the Champagnat 
assumptions and has mutation kernel $M_\s$, where $M$ has bounded support (recall \eqv(mutation-scaled). Assume further that all functions $b,d, c$ 
satisfy Assumption \thv(ass) and that they are in $C^2$.
Assume further that the initial conditions converge, as $K\uparrow \infty$ to $\bar n(x) \d_x$.  
Assume further that $\del_1f(x,x)\neq 0$  on $I\subset \XX$.

Then, for all $T>0$ such that no mass has left $I$ before time $T$, the sequence  of rescaled processes, $\big(\nu^{K}_{t / (K\mu_K \sigma^2)}\big)_{0\leq t\leq T}$, satisfies that
\be\Eq(cead.1)
\lim_{\s\downarrow 0} \lim_{K\uparrow \infty}\nu^{K}_{t / (K\mu_K \sigma^2)}
=\bar n({x_t})\delta_{x_t},
\ee 
where $(x_t)_{0\leq t\leq T}$ is given as a solution of the CEAD, 
\begin{equation}\label{CEAD}
\frac {d x_t}{dt}=\int h\:[h\:p(x_t)\:\bar n(x_t)\:\partial_1 f(x_t,x_t)]_+ m(x_t,dh),
\end{equation}
with initial condition $x_0$, and 
convergence is  in probability with respect to the Skorokhod topology 
on  $\DD([0,T],\mathcal M(\mathcal X))$.
\end{theorem}

\begin{proof} Since $f(x,x)\equiv 0$, $\del_1f(x,x)+\del_2f(x,x)=0$, for all $x\in I$.
Hence, for small enough $\s$, coexistence of $x$ and $x+\s h$ is not possible for any possible mutant. 
Therefore, in the limit $K\uparrow \infty$, we obtain convergence to 
the TSS  for all small enough $\s$ as lang as the process stays in $I]$. 
Thus, the limiting process is a measure-valued Markov process that takes values only in the 
measures of the form $\bar n(x) \d_x$ and is given by
\bea
\Eq(tss-mar.1)
&&(L^{TSS}_\s g)(\bar n(x)\d_x)
\\\nonumber&&=\int p(x)b(x) \bar n(x)\frac{f(x+\s h,x)_+}{b(x+\s h)} \left(g(\bar n({x+\s h})\d_{x+\s h})-
g(\bar n(x)\d_x)\right)m(x,dh).
\eea
The map $\bar n(x)\d_x\to x$ induces a jump process on $\R$ with generator 
\be
\Eq(tss-mar.2)
(\wt L_\s \phi)(x) =\int p(x)b(x) \bar n(x)\frac{f(x+\s h,x)_+}{b(x+\s h)} \left(\phi(x+\s h)-
\phi(x)\right)m(x,dh).
\ee
Simplified to the linear function $\phi(x)=x$, this reduces to 
\be
\Eq(tss-mar.3)
\wt L_\s x =\int p(x)b(x) \bar n(x)\frac{f(x+\s h,x)_+}{b(x+\s h)} \s h
m(x,dh).
\ee
Associated with this, we get the first martingale, 
\bea\Eq(tss-mar.4)
M^1_t&=& X(t)-X(0)\\\nonumber
&&-\int_0^t \int p(X(s))b(X(s)) \bar n(x)\frac{f(X(s)+\s h,X(s))_+}{b(X(s)+\s h)} \s h
m(X(s),dh).
\eea
Proceeding as in Chapter \ref{chapter2}, we can compute the bracket of this martingale,
\bea
\Eq(tss-mar.5)
[M^1]_t &=&\int_0^t  \left((\wt L X^2)(s)-2X(s)(\wh L X)(s)\right)ds\\
\nonumber
&=&
\int_0^t \int p(X(s))b(X(s)) \bar n(x)\frac{f(X(s)+\s h,X(s))_+}{b(X(s)+\s h)}( \s h)^2
m(X(s),dh).
\eea
Note further that, since $f(x,x)=0$, to leading order in $\s$, 
\be
\frac{f(x+\s h,x)_+}{b(x+\s h)} =\s  \frac{[h\del_1f(x,x)]_+}{b(x)} +O(\s^2).
\ee
Thus 
\bea
&&X(t/\s^2)-X(0)\\\nonumber&& -\int_0^t  p(X({s/\s^2}))\bar n(X(s/\s^2))
\int h [h\del_1f(X({s/\s^2}),X({s/\s^2}))]_+ m(X({s/\s^2}),dh) +O(\s)
\eea
is a martingale, and the bracket of this martingale is
\be
\Eq(tss-mar.10)
[M^1]_{t/\s^2} =
\s \int_0^t \int p(X(s/\s^2)) \bar n(X(s/\s^2))[h\del_1f(X(s/\s^2),X(s/\s^2))]_+ h^2
m(X(s/\s^2),dh).
\ee
From this, we get (adding straightforward tightness arguments as in Chapter \ref{chapter2}),
that the rescaled processes 
\be
\Eq(tss-mar.11)
X^\s(t)\equiv X(t/\s^2)
\ee 
converge, as advertised, to the solution of the integral equation
\bea\Eq(tss-mar.12)
&&X^0(t)-X^0(0) \\\nonumber
&&=
\int_0^t  p(X^0(s))\bar n({X^0(s)})
\int h[h \del_1f(X^0(s),X^0(s))]_+ m(X^0(s),dh).
\eea
But this is the integral form of the CEAD.
\end{proof}

The extension of this result to $\XX\subseteq \R^d$ is straightforward. The CEAD in this case reads 
\be\Eq(rd.1)
\frac {d x_t}{dt}=\int h\:[p(x_t)\:\bar n(x_t)\:\langle h,\nabla_1 f(x_t,x_t)]_+ m(x_t,dh).
\end{equation}

We see that the canonical equation stops when it approaches a point where
$\nabla_1 f(x,x)_+ =0$. 

\begin{figure}
\includegraphics[width=0.7\textwidth]{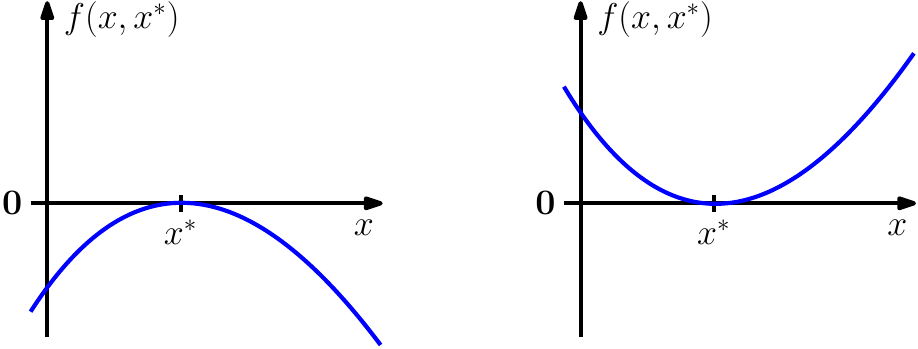}
\caption{Stable and unstable evolutionary singularity.}
\end{figure}

\section{Evolutionary branching} \index{evolutionary branching}
Under the smoothness assumptions on the parameters of the models, 
$\nabla_1 f(x,x)$ is continuous, and so the canonical equation will not reach 
an evolutionary singularity in finite time. Thus, the CEAD runs to a halt when such a point is 
approached. Such points are called \emph{evolutionary singularities}. The future fate of a 
population that approaches such a point $x^*$ depends on the nature of this singularity.
Let us focus on the case $\XX=\R$. \index{evolutionary singularity}

Evolutionary branching occurs at fixpoints of the canonical equation, i.e. when 
$\del_1f(x^*,x^*)=0$. Moreover, such a fixpoint must be unstable, in the sense that
the invasion fitness in both directions is positive, i.e. $f(x,x^*) $ must be convex as a function of
$x$ near $x^*$. We will assume therefore $\del_{11} f(x^*,x^*)>0$. 
Finally, such a fixpoint must be reachable by the CEAD. This means that, for $x$ 
in a neighbourhood of $x^*$, $\del_1 f(x,x)>0$, if $x<x^*$ and  $\del_1 f(x,x)<0$, if $x>x^*$.
Thus, the total derivative of $\del_1f(x,x)$ must be negative at $x^*$, 
which means that 
\be\Eq(seltsam.1)
\del_{11} f(x^*,x^*) +\del_{21} f(x^*,x^*)<0.
\ee
Moreover, since $f(x,x)\equiv 0$, we have the identities 
\bea 
\del_1f(x,x)+\del_2f(x,x)&=&0, 
\\
\del_{11}f(x,x)+2\del_{12}f(x,x) +\del_{22}f(x,x)&=&0.
\eea
This implies that \eqv(seltsam.1) is equivalent to 
\be
\Eq(seltsam.2)
\del_{11}f(x^*,x^*)<\del_{22}f(x^*,x^*).
\ee
In that case, one can 
show that there exists intervals $I^-,\subset (x^*-\e,x^*)$ and  $I^+,\subset (x^*,x^*+\e)$,
such that for  all $x\in I^-$ and $y\in I^+$, $f(x,y)>0$ and $f(y,x)>0$, so that $x, y$ coexist. 
But this means that if for small but finite $\s$ the TSS invades a point in $I^-$, then with positive 
probability, a mutant trait will appear in $I^+$, and fixate to the coexisting bi-morphic population
on $(x,y)$. This phenomenon is called \emph{evolutionary branching}. From that moment on, the 
PES will continue, most likely initially with a bi-morphic population where the two co-existing traits 
move apart. \index{polymorphic evolution sequence}

The case when $x^*$ is an inflexion point is less interesting, since the TSS will just continue at a smaller pace 
towards higher fitness. 

The case when $x^*$ is a local maximum corresponds to $x^*$ being an evolutionary 
stable condition: If we place an equilibrium population at $x^*$, any mutant will be subcritical and
die out almost surely (in the limit $K\uparrow \infty$). Evolution comes to a halt. We will revisit this 
case later. \index{evolutionary stable condition}

\section{All limits simultaneously}  \label{section6.3}

Clearly, taking the limit $\s\downarrow 0$ after the limits $K\uparrow \infty, \mu_K\downarrow 0$ 
have been taken is somewhat unsatisfactory. Ideally, one would like to have conditions 
on $\mu_K$ and $\s_K$ under which a simultaneous limit 
$K\uparrow \infty, \mu_K\downarrow 0,\s_K\downarrow 0$ can be constructed.
Unfortunately, this turns out to be technically far more demanding, and all problems have not been solved so far. One has, however, the following result on the convergence to the CEAD.

\begin{theorem}[\cite{B14}]\TH(cead3)
Let the functions $b,d,c$ satisfy the same conditions as in Theorem \thv(main_thm). 
Assume further that the trait space is a subset of $\R$ and that the 
mutation kernels $m(x,dh)$ are supported on a finite set. Assume that, for some $\a>0$, 
\begin{align}
\label{conv1}  		K^{- \sfrac  1 2 +\alpha}&\ll \sigma_K\ll 1 \qquad \qquad \text{ and }  \\
\label{conv2}	\qquad\qquad	\exp(-K^{\alpha})&\ll \mu_K\ll \frac{\sigma_K^{1+\alpha}}{K\ln K},	\quad \text{ as } \qquad K\rightarrow \infty.
\end{align} 
Fix $x_0\in\mathcal X$ and let $(N^{K}_0)_{K\geq 0}$ be a sequence of $\mathbb N$-valued random variables such that
$N^{K}_0 K^{-1}$ converges in law, as $K\to\infty$,  to the positive constant $\bar n(x_0)$ and
is bounded in $L^p$, for some $p > 1$. \\[0.5em]
For each $K\geq 0$,
let $\nu^{K}_{t}$ be the  process generated by $\LL^K$ with monomorphic initial state ${N^{K}_0} K^{-1} \delta_{\{x_0\}}$.
Then, for all $T>0$, the sequence  of rescaled processes, $\big(\nu^{K}_{t / (K\mu_K \sigma_K{}^2)}\big)_{0\leq t\leq T}$,
converges in probability, as $K\uparrow \infty$,  on the Skorokhod space $\mathbb D([0,T],\mathcal M(\mathcal X))$ to the measure-valued process $\bar n(x_t)\delta_{x_t}$, 
where $(x_t)_{0\leq t\leq T}$ is given as a solution of the CEAD, 
\begin{equation}\label{CEAD.1}
\frac {d x_t}{dt}=\int h\:[h\:p(x_t)\:\bar n(x_t)\:\partial_1 f(x_t,x_t)]_+ m(x_t,dh),
\end{equation}
with initial condition $x_0$.  
\end{theorem}

\begin{remark}\label{remark_main_thm}
\begin{enumerate}[label=(\roman*)]
\setlength{\itemsep}{3pt}
\item If $x_t$ reaches a point where $\partial_1 f(x_t,x_t)=0$  then  $\frac {d\:x_t}{dt}=0$  and the 
 process stops.
\item The condition $ \mu_K\ll \frac{\sigma_K^{1+\alpha}}{K \ln K}$ allows mutation events during an invasion phase of a mutant trait,
 but ensures that there is no "successful" mutational event during this phase. More precisely, the probability that there will occur a mutation during a single 
 invasion step will be $\s_K^\a$, but we need to control $O(\s_K^{-1})$ such steps, so that there can be $O(\s_K^{\a-1})$ mutations. But each mutation is successful only with probability of order $\s_K$. Thus, there will be no successful mutations happening during any of the $\s_K^{-1}$ invations.
 
\item The fluctuations of the resident population are of order $K^{- \sfrac  1 2}$, 
	thus $K^{- \sfrac  1 2 +\alpha}\ll \sigma_K$ ensures that the sign
	of the initial growth rate is not influenced by the fluctuations of the population size.
	If a mutant trait $y$ appears in a monomorphic population with trait $x$,  then
	its initial growth rate is $b(y)-d(y)-c(y,x)\langle \nu_t^K,\1\rangle=f(y,x)+o(\s_K)=(y-x)\partial_1 f(x,x)+o(\s_K)$ since $y-x=O(\s_K)$.
\item  $\exp(K^{\alpha})$ is the time the resident population stays with high probability in an $O(\e\s_K)$-neighbourhood of an attractive domain, see Lemma \thv(lem:step-1). 
	Thus the condition $\exp(-K^{\alpha})\ll \mu_K$ 
	ensures that the resident population is still in this neighbourhood when a mutant occurs.
	The sharper control is necessary since we need it to ensure positive invasion fitness.
\item The time scale is  $(K\mu_K \sigma_K{}^2)^{-1}$ since the expected time for a mutation event is $(K \mu_K)^{-1}$, 
	the probability that a mutant invades
	is of order $\s_K$ and one needs $O(\sigma_K^{-1})$ mutant invasions to see an $O(1)$ change of the resident trait value. 
\end{enumerate}
\end{remark}

\section{Structure of the proof of Theorem \thv(cead3)}\label{The main idea}

The evolution of the population will be described as 
a succession of \emph{mutant invasions} as in the case of the PES. The assumptions ensure that 
the population remains monomorphic, i.e. behaves like a TSS.

We first control a single \emph{invasion step}. Namely, we show that there is a timescale that is long enough for exactly one mutant population to fixate and for the 
 resident trait to die out, but sufficiently short, such that no two successful mutant populations can exist during this time.
We say the mutant trait fixates in the population. Note that this does not prevent the appearance of other 
mutant traits that do not invade.

Second, we consider a much longer time scale on which the single invasion steps aggregate and give rise 
to a \emph{macroscopic evolution} that converges to the CEAD. 
\medskip

\paragraph{Single invasion step:} 
We divide the time until a mutant trait has fixated in the population into two phases

\subparagraph{Phase 1.} The first difficulty arises because $\s_K$ tends to zero with 
$K$ is, that the invasion fitness of a mutant will be only of order $\s_K$. That means that the actual growth 
rate of the mutant is strictly positive only if the size of the resident population does not 
deviate from its equilibrium size by more than $O(\s_K)$. Thus, we need to show that it is highly unlikely 
that such fluctuations occur during a time that is long enough to allow the mutant population to grow sufficiently.
Since the natural scale of fluctuations of the resident is $K^{-1/2}$, we must impose  that  $\s_K\gg 
K^{-1/2}$. It turns out that \eqv(conv1) is a good choice\footnote{It is not certain that a weaker condition might 
be enough, but with present techniques, there seems to be no way to handle such a case.}.

 Here we fix a small $\epsilon>0$ and prove the existence of a constant, $M<\infty$, independent of $\epsilon$, such that, as long as all mutant densities are smaller than $\epsilon\sigma_K$, 
the resident density stays in an $M\epsilon\sigma_K$-neighbourhood of  $\bar n(x)$. Note that, 
because  mutations are rare and the population size is large, 
the monomorphic initial population has  time to stabilise in an $M\epsilon\sigma_K$-neighbourhood of this equilibrium $\bar n(x)$
before the first mutation occurs.
(The time of stabilisation is of order $\ln(K)\sigma_K^{-1}$ and the time where the first mutant occurs is of order $1/K\mu_K$).

This allows us to approximate the density of one mutant trait $y_1$ by a branching process with birth 
rate $b(y_1)$ and death rate $d(y_1)-c(y_1,x)\bar n(x)$ such that we can 
compute the probability that the density of the mutant trait $y_1$ reaches $\epsilon 
\sigma_K$, which is of order $\sigma_K$, as well as the time it takes to reach this level or 
to die out. Therefore, the process  needs $O(\sigma_K^{-1})$ mutation events 
until a mutant subpopulation appears, which reaches a size
$\epsilon \sigma_K$. Such a mutant is called \emph{successful mutant} and its trait will 
be the next resident trait.

We prove in this step that there are never too many 
 different mutants alive at the same time. From all this, we deduce that the  
 subpopulation of the successful mutant reaches the density $\epsilon \sigma_K$, before a 
 different successful mutant appears.
Note that we cannot use large deviation results on our time scale as used in \cite{CM11} to prove this step. Instead, 
we use some standard potential theory and coupling arguments to obtain estimates of moderate deviations needed to prove 
that a successful mutant will appear before the resident density exits an $M \e\s_K$-neighbourhood of its equilibrium. 
  \vspace{2mm}
  
\subparagraph{Phase 2}  We prove that if a mutant population with 
 trait $y_s$  reaches the size 
$\epsilon \sigma_K$, it will increase to an  $M\epsilon\sigma_K$-neighbourhood of its equilibrium density
 $\bar n(y_s)$. This is due to the estimate \eqv(yule.4) and the fact that by our assumption on 
 $\s_K$,  
 \be
\left( \frac 1{1+c\s_K}\right)^{\e K\s_K} \approx \eee^{-c\e \s_K^2 K}
\leq  \eee^{-c\e K^{2\a}}\downarrow0,
\ee
so that a mutant population of size $\e \s_K$ grows with overwhelming probability to size $\e$. 
The time it takes to grow to this size is $\ln K \s_K^{-1}$.  This is relatively similar to the case with 
fixed $\s$, but all estimates on probabilities have to be much more precise.

The main difficulty involves the control of the passage from mutant size $\e$ until the mutant approaches the new equilibrium and the resident dies out. Before this was 
done using the law of large numbers and the deterministic Lotka-Volterra system. 
There is no way that this can work in the case $\s_K\downarrow 0$, for two reasons:
\begin {enumerate}[label=(\roman*)]
\item The $K\downarrow 0$, $\s_K\downarrow 0$, so in the limit, the 
Lotka-Volterra system is degenerate and has a one-dimensional invariant manifold of fixpoints. Thus, the new fixpoint with only the mutant population present will never be reached.
\item Since the drift towards the new fixpoint is of order $\s_K$ only, we must expect that the passage to the new equilibrium takes a time of order $\s_K^{-1}$, which diverges. The LLN never 
gives control over such long periods.
\end{enumerate}

One might hope that the stochastic systems behave very much like the deterministic system 
with the explicit $K$-dependent parameters. This could possibly overcome the problem (i).
But this does not solve problem (ii). In fact, any attempt to control the 
deviation of the stochastic process from the deterministic one using some Gronwall-type 
argument will introduce blow-up factors of order $\exp(C_{Lip}T)$, and the 
Lipschitz constant of the system is of order one. Thus, error estimates will explode in times
of order $\s_K^{-1}$.  

The reason that our expectations are still correct lies in the nature of the 
deterministic 
dynamical system. As we have said, if $\s_K=0$, the deterministic system has 
a \emph{stable} invariant manifold of fixed points with a vector field independent of $\s_K$ pointing towards this manifold. 
Thus, any trajectory is attracted to this manifold and approaches it in finite time. The same holds 
for the stochastic system by the LLN.
Switching  on 
a small 
$\s_K$,  the invariant manifold is slightly perturbed, but persists, and now will have only 
two fixpoints, the unstable and the stable one, with  a very weak drift of order $\s_K$ in the
 manifold (see Fig. \thv(fig.3).  Thus, we  expect the stochastic system to stay close to this 
 invariant manifold and to move 
along it with speed of order $\s_K$.

\begin{figure}  \TH(fig.3)
\includegraphics[width=5cm]{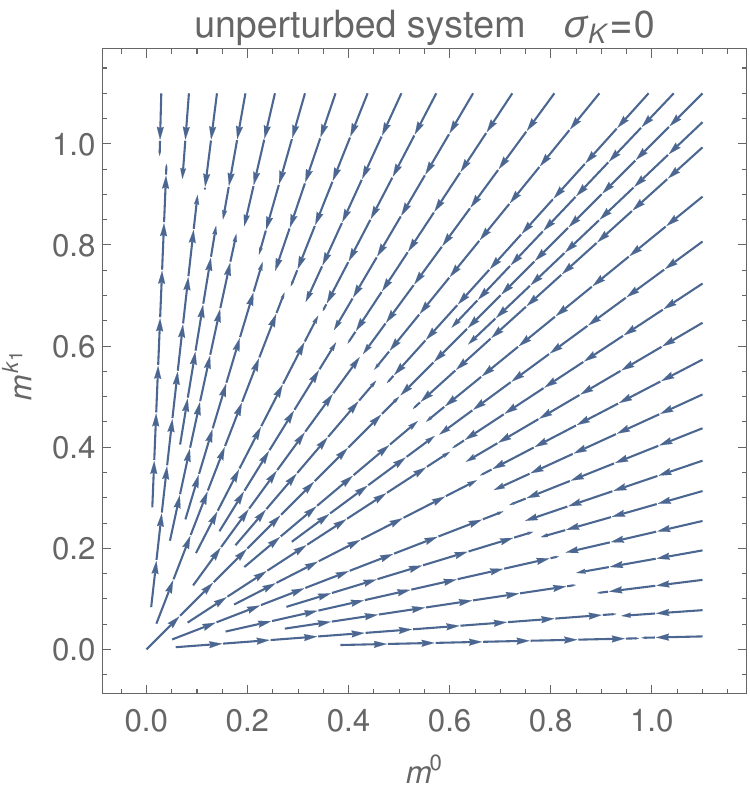} \hspace{4mm}
\includegraphics[width=5cm]{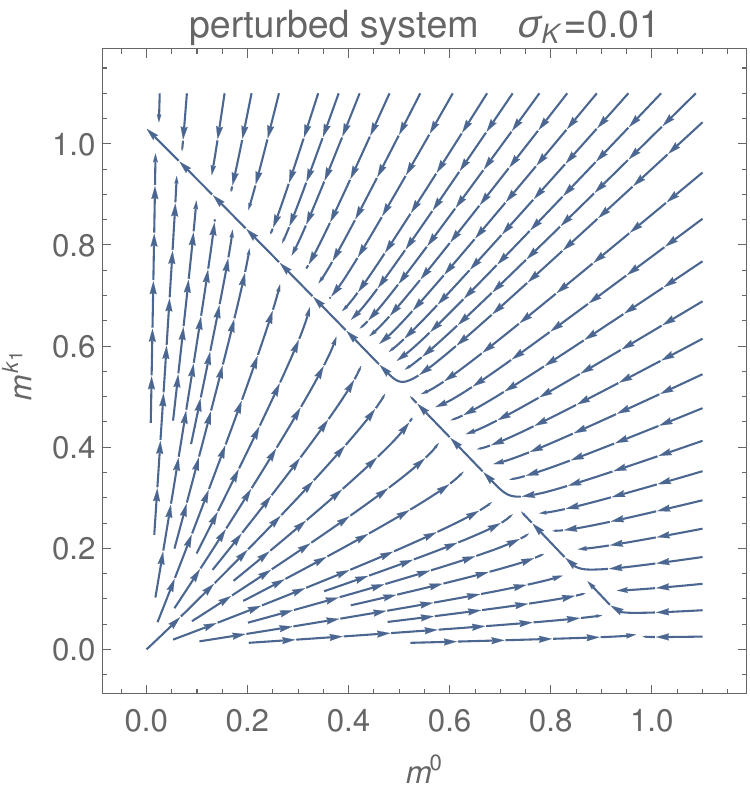}
\caption{
Right: Vector field of the unperturbed system ($\s_K=0$),
Left: Vector field of the perturbed system ($\s_K=0.01$). }
\end{figure}

In \cite{B14}, a rather complicated technique, called 
"rigorous stochastic Euler-Scheme" was developed to prove that this heuristic actually holds. It consists
 of controlling the evolution of the process over small time intervals in a similar spirit via couplings
 as in the first invasion phase. The details are, however, quite painful, and it would be very nice 
 to have a simple proof....

\begin{figure}[ht]
\centering\includegraphics[width=9cm]{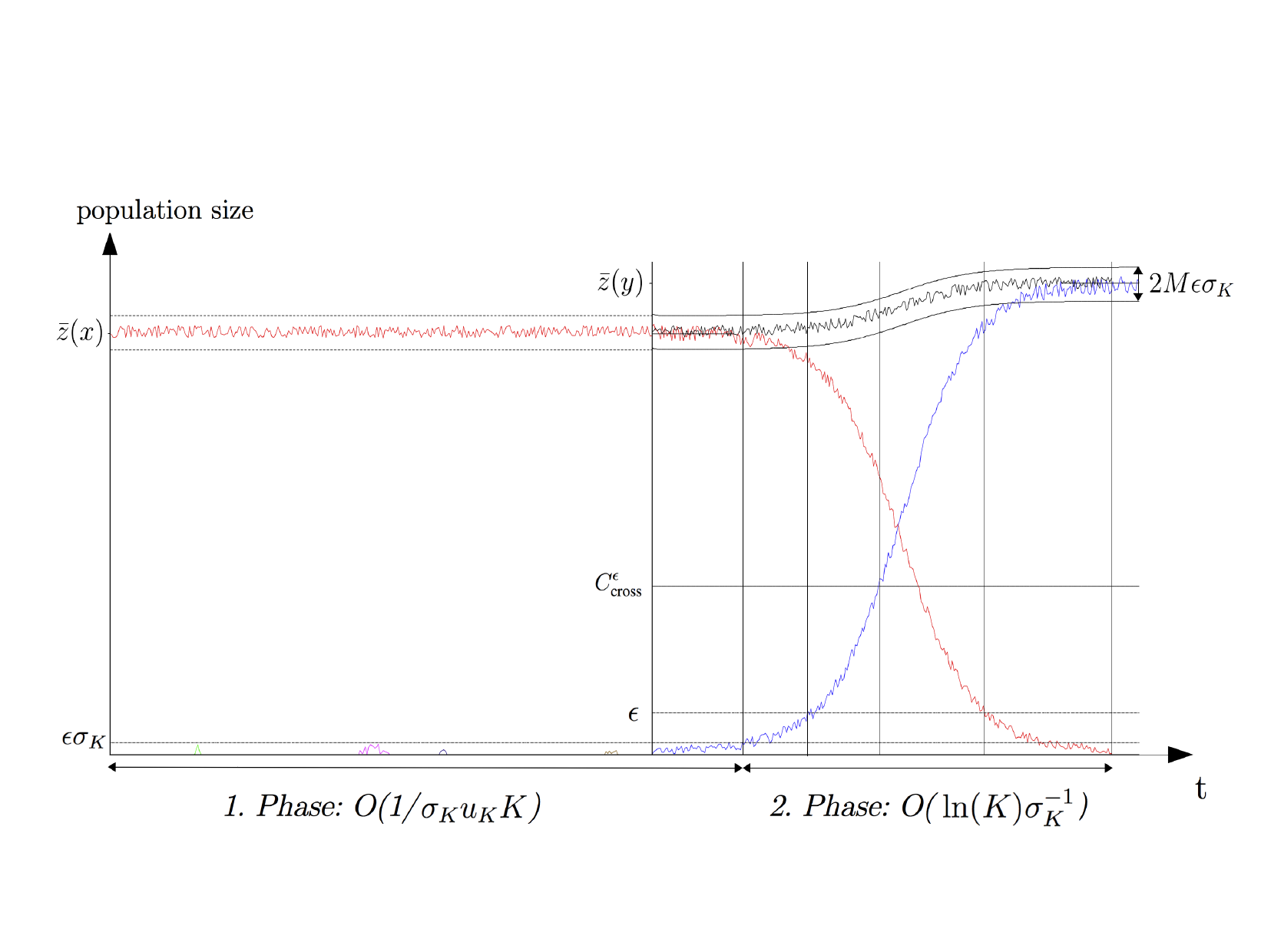}
  \caption{\label{fig}Typical evolution of the population during a mutant invasion. }
\end{figure}
 \begin{figure}[ht]
\centering\includegraphics[width=9cm]{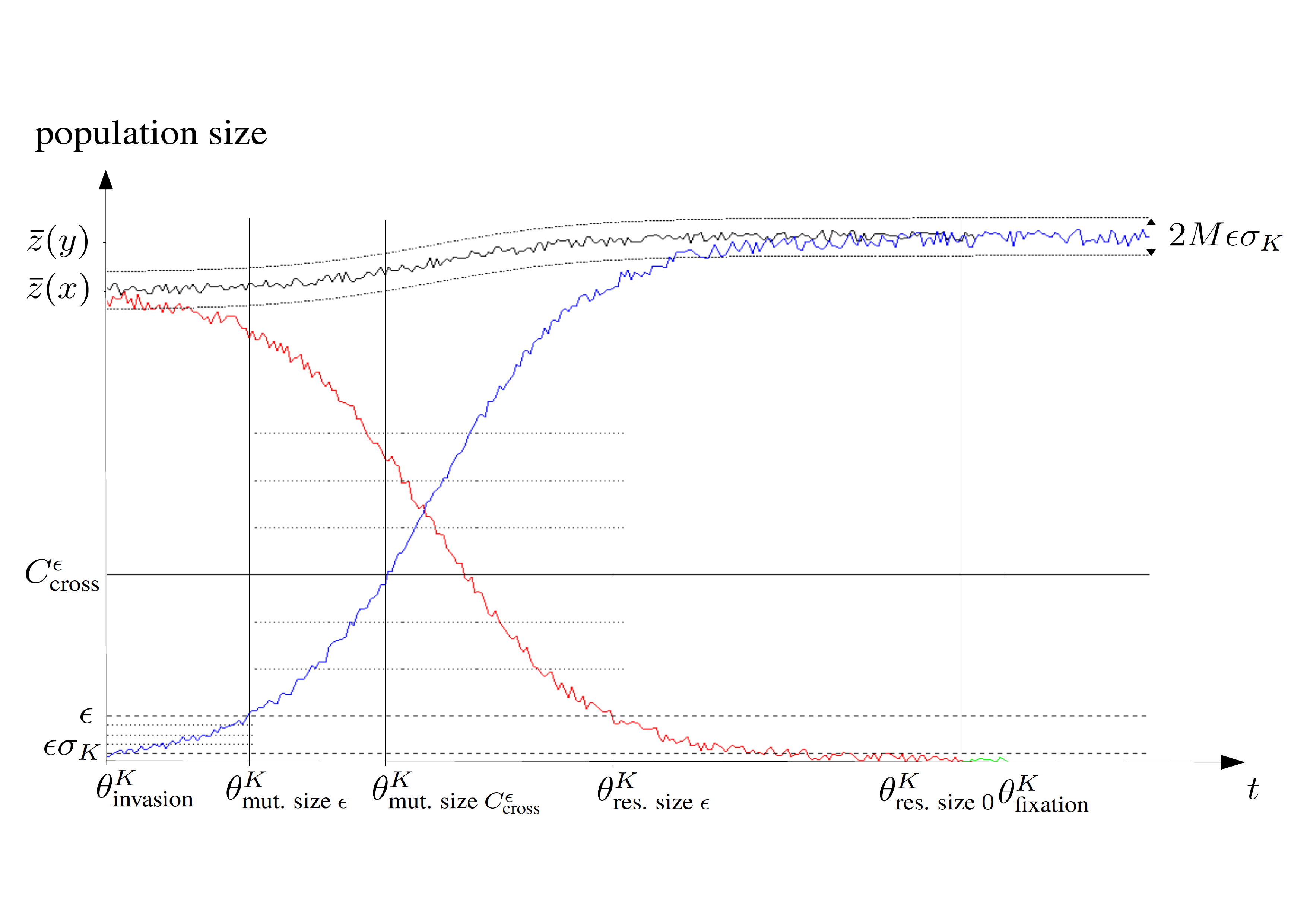}
  \caption{\label{fig2} Evolution of the population after the destiny of the successful mutant has reached the value $\e_K$. }
\end{figure}
With this method, one can prove that the mutant density reaches the $M\e\s_K$-neighbourhood of $\bar n(y_s)$ and
 the resident trait dies out. Note that an unsuccessful mutant may be alive at this time. 
 Therefore, we prove that after the resident trait has died out, there is a time when the population consists only of one trait,
 namely, the one that had fixated before the next successful mutant occurs.


%
\section{Evolutionary branching and beyond}

\index{evolutionary branching!conjectures}

There are rather precise conjectures on how the PES can be described in the joint limit
$K\uparrow \infty, \mu_K\downarrow 0,\s_K\downarrow 0$. 
We have seen that starting in a monomorphic situation, the process converges to the CEAD
on the time scale $1/K\mu_K\s^2_K$. This process comes to a halt at that state when the CEAD 
approaches an evolutionary singularity. Then, there should be a precise longer time scale on which 
evolutionary branching takes place if this fixpoint is unstable. As soon as branching happens,
one should expect that the individual subpopulations again follow, on the faster
time scale, a kind of CEAD, until they all approach a fixpoint. 
Then, branching occurs for all those that are unstable at the larger time scale, and so on. 
Thus, on this second time scale, the process is always at an evolutionary singularity and 
jumps between one such singularity to the next, until a state is reached when 
all populations are in a stable singularity, i.e. when an evolutionary stable condition is reached. 

Unfortunately, all this is not (yet?) proven rigorously. We will, however, 
state below some conjectures that were formulated in an unfinished paper
by N. Champagnat and A. Bovier.

We consider (only) the case when $\XX\subseteq \R$. We also assume the same regularity assumptions as before.
The following is a precise definition of evolutionary branching in this setting.
\index{evolutionary branching!definiton}

\begin{definition}
  \label{dfn:branching-mono}
  Let $\varepsilon>0$ and $x^*$ be a fixpoint  of the
  CEAD, i.e.\ $\partial_1f(x^*,x^*)=0$. We say that
  there is $\varepsilon$-branching at $x^*$ if there exists a time
  $t\geq 0$ such that
  \begin{equation}
    \label{eq:def-br-1}
    \mbox{\textup{supp}}\big(\nu^{K,u,\s_K}_{t}\big)
    \subset[x^*-\varepsilon,x^*+\varepsilon],
  \end{equation}
  and a time $s>t$ such that
  \begin{equation}
    \label{eq:def-br-2}
    \begin{gathered}
      \mbox{\textup{supp}}\big(\nu^{K,u,\s_K}_{t}\big)
      \subset[x^*-2\varepsilon,x^*+2\varepsilon], \\
      \nu^{K,u,\s_K}_s\big((x^*-\varepsilon,x^*+\varepsilon)\big)=0, \\
      \nu^{K,u,\s_K}_s\big([x^*-2\varepsilon,x^*-\varepsilon]\big)>0\quad
      \mbox{\textup{and}} \quad
      \nu^{K,u,\s_K}_s\big([x^*+\varepsilon,x^*+2\varepsilon]\big)>0,
    \end{gathered}
  \end{equation}
  where $\mbox{\textup{supp}}(\nu)$ denotes the support of the
  measure $\nu$.
\end{definition}

\begin{conjecture}
  \label{thm:s2-mono}
  For a fixed $x\in{\XX}$ and $\gamma>0$, assume that
  $\nu^{K,\mu_K,\s_K}_0=N^{K,\mu_K,\s_K}_0\delta_x$ where
  $N^{K,\mu_K,\s_K}_0$ converges in probability to $\gamma\delta_x$
  when $\mu_K\downarrow 0$, $\s_K\downarrow 0$ and
  $K\uparrow+\infty$. Assume that the solution of the
  CEAD  with initial condition $x$ converges when
  time goes to infinity to a steady state $x^*\in{\XX}$. Assume
  also that $\mu_K\downarrow 0$, $\s_K\downarrow 0$ and
  $K\uparrow+\infty$ such that
  \begin{gather}
    K^{-1/2+\alpha} \ll \s_K \ll 1 \quad\mbox{for some $\alpha>0$},
    \label{eq:hyp-s2-mono-1} \\
\exp(-K^\a) \ll \mu_K \ll \frac{\s_K^{1+\a}}{K\log K}.
    \label{eq:hyp-s2-mono-2}
  \end{gather}
\begin{itemize}
\item [(i)]  If
$      \partial_{22}f(x^*,x^*)>\partial_{11}f(x^*,x^*)>0$
       then there is \emph{fast evolutionary branching}: there exists
    $\varepsilon_0>0$ such that for any $\varepsilon<\varepsilon_0$,
    the probability of $\varepsilon$-branching at $x^*$ converges to 1
    when $(K,\mu_K,\s_K)\rightarrow(+\infty,0,0)$.

    Moreover, if $T_\varepsilon$ is the first time $t$ such
    that~(\ref{eq:def-br-1}) holds and $S_\varepsilon$ is the first
    time $s>T_\varepsilon$ such that~(\ref{eq:def-br-2}) holds, then,
    for any $\varepsilon<\varepsilon_0$,
    \be
      \label{eq:cv-time}
      \frac{\mu_K\s_K^2
        K}{\log(1/\s_K)}T^{K,\mu_K,\s_K}_\varepsilon\rightarrow 0
     \ee
     and
     \be
      \label{eq:br-time}
      \frac{\mu_K\s_K^2
        K}{\log(1/\s_K)}S^{K,\mu_K,\s_K}_\varepsilon\rightarrow
      \frac{2}{p(x^*)\bar n(x^*)[\partial_{22}f(x^*,x^*)-\del_{11}f(x^*,x^*)]
        \int m(x^*,dh)h^2},
    \ee
    in probability, as $(K,\mu_K,\s_K)\rightarrow(+\infty,0,0)$.
  \item[(ii)] If
    $  \partial_{22}f(x^*,x^*)>\partial_{11}f(x^*,x^*) $  and 
      $ \partial_{11}f(x^*,x^*)<0$
       evolutionary branching does not occur on the time scale
    $\log(1/\s_K)/\mu_K\s_K^2K$: the process\\
    $$
    \left(\nu^{K,\mu_K,\s_K}_{t\log(1/\s_K)/\mu_K\s_K^2K}\right)_{t\geq 0}
    $$
    converges  to the constant process
    $\bar n(x^*)\delta_{x^*}$.
  \end{itemize}
\end{conjecture}

Why is the timescale $\ln (1/\s_K)/K\mu_K\s^2_K$? For mutants to appear on both sides of $x^*$, the process must approach $x^*$ to a distance of
order $\s_K$. From a distance $\e$, this requites $\e/\s_K$ steps, and in the $n$-th step the distance from $x^*$ is (roughly, assuming all steps 
have length $\s_K$) $\e-n\s_K$. 
Now, for $x$ close to $x^*$, the invasion fitness for a single step
is
\bea\nonumber
f(x+\s_K,x)&\approx& f(x^*+\s_k,x^*)+  (x-x^*) (\del_1f(x^*+\s_K,x^*)+\del_2f(x^*+\s_K,x^*))\\\nonumber
&\approx& \frac{\s_K^2}2\del_{11}f(x^*,x^*)+  (x-x^*) 
\s_K (\del_{11}f(x^*,x^*)+\del_{12}f(x^*,x^*)))\\
&=&\frac12 \s_K^2\del_{11}f(x^*,x*)+ \s_K(x^*-x)[\partial_{22}f(x^*,x^*)-\del_{11}f(x^*,x^*)].\nonumber\\ 
\eea
Then the invasion fitness in the $n$-th step will be roughly 
$\s_K(\e-n\s_K)$, 
and so the time it take to make step number $n$ is of order $1/(K\mu_K\s_K^2) \frac 1{\e/\s_K-n}$.
Summing this over $n$ from one to $\e/\s_K$ gives $\ln(1/\s_K)/K\mu_K\s^2_K$. 
Eq. \eqv(eq:br-time) then follows from the law of large numbers and throwing in the rate at which mutants appear.

Once the population is close enough to $x^*$ that a mutant appears at the other side of $x^*$, 
resident and mutant will coexist (since $f(x^*-h\s,x^*+h'\s)\approx -hh'\s^2\del_{12}f(x^*,x^*)>0$, and
$f(x^*+h'\s,x^*-h'\s)\approx -hh'\s^2\del_{12}f(x^*,x^*)>0$).
 and therefore evolutionary is accomplished. 
 
 The next question is the future fate of the two coexisting traits.  The following conjecture
states that they move apart according to a coupled CEAD.

\begin{conjecture}
  \label{thm:s2-mono-br-state}
  With the assumptions and notation of the previous theorem, in case~(i),
  \begin{equation}
     \label{eq:state-eps-br}
     \nu^{K,\mu_K,\s_K}_{S_\varepsilon}\rightarrow
     \bar n_1(x^\varepsilon_1,x^\varepsilon_2)\delta_{x^\varepsilon_1}+
     \bar n_2(x^\varepsilon_1,x^\varepsilon_2)\delta_{x^\varepsilon_2}
   \end{equation}
   in probability as $(K,\mu_K,\s_K)\rightarrow(+\infty,0,0)$, where
   $x^\varepsilon_1$ and $x^\varepsilon_2$ are such that
   $x^\varepsilon_1\leq x^*-\varepsilon$ and $x^*+\varepsilon\leq
   x^\varepsilon_2$, where at least one inequality is actually an
   equality. They are defined as follows.  Let $\phi(t,x_1,x_2)$
   denote the flow of the ODE in $\R^2$
   \begin{equation}
     \label{eq:CEAD-dimo}
     \begin{cases}
       \frac{dx_1}{dt}=\frac{1}{2}\int hm(x_1,dh)\:
       p(x_1)\bar n_1(x_1,x_2)[h \partial_1 f(x_1;x_1,x_2)]_+, \\
       \frac{dx_2}{dt}=\frac{1}{2}\int hm(x_2,dh)\:
       p(x_2)\bar n_2(x_1,x_2)[h \partial_1 f(x_2;x_1,x_2)]_+,
     \end{cases}
   \end{equation}
Let
   $t_\varepsilon(x_1,x_2)$ be the first time such that
   $\phi(t,x_1,x_2)\not\in(x^*-\varepsilon,x^*+\varepsilon)^2$.  The
   flow $\phi$ is only defined until the first time where
   coexistence of $x_1(t)$ and $x_2(t)$ does not hold; however,
   $t_\varepsilon(x_1,x_2)$ is well defined if
   $(x_1,x_2)\in(x^*-\varepsilon,x^*+\varepsilon)^2$. Then, there
   exists $\varepsilon_1>0$ such that, for any
   $\varepsilon<\varepsilon_1$,
   \begin{equation}
     \label{eq:def-x-eps}
     x^\varepsilon_i=\lim_{\eta\rightarrow
       0}\phi_i(t_\varepsilon(x^*-\eta,x^*+\eta),x^*-\eta,x^*+\eta),
     \quad i=1,2
   \end{equation}
   and
   \begin{equation}
     \label{eq:rapp-x-eps}
     \lim_{\varepsilon\rightarrow
       0}\frac{x^\varepsilon_2-x^*}{x^*-x^\varepsilon_1}
     =\frac{\int_{\R_-}|h|m(x^*,dh)}{\int_{\R_+}|h|m(x^*,dh)}.
   \end{equation}
\end{conjecture}

Equations~(\ref{eq:state-eps-br}),~(\ref{eq:def-x-eps})
and~(\ref{eq:rapp-x-eps}) characterize the state of the population
just after a fast branching. Note that this characterisation involves
the generalization~(\ref{eq:CEAD-dimo}) of the CEAD to a dimorphic
situation. 
This equation should also allow us to describe the evolution of the population
on the first time scale \emph{after} time $S_\varepsilon$.
In fact, we expect that a polymorphic population will evolve according to a multidimensional CEAD 
until it reaches a new singularity on the CEAD timescale.

\begin{conjecture}
  \label{thm:s1-poly}
  For fixed $x_1,\ldots,x_N\in{\XX}$ that coexist,
  $\nu^{K,\mu_K,\s_K}_0$ converges in law to
  \be
  \sum_{i=1}^N\bar n_i(x_1,\ldots,x_n)\delta_{x_i},
  \ee
   as
 $K\rightarrow+\infty$. Assume the usual conditions on $\mu_K$ and $\s_K$.
Then
  \begin{equation}
    \label{eq:cv-s1-poly}
    \nu^{K,\mu_K,\s_K}_{\frac{t}{\s_K^2 \mu_K K}} \Rightarrow
    \sum_{i=1}^{N(t)}\bar n_i(\bx(t))\delta_{x_i(t)},
  \end{equation}
  for the Skorokhod topology on $\DD((0,+\infty),{\MM}({\XX}))$,
  where $N(0)=N$, $x_i(0)=x_i$, for $1\leq i\leq n(0)$, and
  \begin{equation}
    \label{eq:CEAD-poly}
    \frac{dx_i}{dt}=\frac{1}{2}\int_{\XX} hm(x_i,dh)\:
    p(x_i)\bar n_i(\bx(t))[h\partial_1
    f(x_i;\bx(t))]_+,\quad 1\leq i\leq n.
  \end{equation}
%
\end{conjecture}

Note that in the course of time, it is possible that the coexistence of all the $n$ traits gets lost, and some
of the $\bar n_i(\bx(t))$ can be zero. So sub-populations can become extinct, but on the timescale considered, no branching can occur. 

It shows that, on
the first time scale, a branch can go extinct, but no new branch can
appear. Because of the possible extinction of a branch, one has to be
careful in the definition of the limit process, because
$\bar n(x_1,\ldots,x_k)$ is correctly defined only if the traits
$x_1,\ldots,x_k$ coexist. Note that the only way for the number of
branches to decrease is by the extinction of one branch. Two distinct
branches cannot collide on the first time scale. This result also
allows one to describe what happens after the branchings
in Conjecture~\ref{thm:s2-mono}.  For example, in
the case of a fast branching, after time $S_\varepsilon$ (defined in
Conjecture~\ref{thm:s2-mono}), on the first time scale, the evolution of
the population is governed by~(\ref{eq:CEAD-poly}) with $k=2$ and with
initial condition $x_i(0)=x_i^\varepsilon$ for $i=1,2$, defined
by~(\ref{eq:def-x-eps}).

\begin{remark} In a slightly different context, basically a Fisher-Wright model with small mutation steps, Champagnat and Hass  \cite{ChamHass2023} have recently derived the CEAD when only taking the mutation step size, but not the mutation rate to zero with $K$.
\end{remark}
%


\chapter{Stochastic systems with moderately rare mutations}\label{chapter7}

\begin{chapquote}
{Thomas Malthus, \emph{Essay on the Principle of Population}}
{\frakfamily\fraklines {
It is an acknowledged truth in philosophy, \\ 
that a just theory will always be confirmed by experiment.\\
Yet so much friction, and so much minute circumstances\\
 occur in practice, which it is next to impossible for the most\\
  enlarged and penetrating mind to foresee, \\ 
  that on few subjects can any theory be pronounced just, \\ that has not stood the test of experience.}}
\end{chapquote}

In this chapter, we return to the model with a fixed mutation step size but consider slightly higher mutation probabilities than before. In Chapter \ref{chapter5}, the scaling of $\mu_K\ll (K\ln K)^{-1}$ ensured that (with high probability) only one mutant population was present at any given time. In the following, we consider higher probabilities $\mu_K$ that yield multiple mutant populations to be present at the same time, competing to invade the resident population. To nevertheless be able to work with only a finite number of different subpopulations, we restrict ourselves to a discrete trait space, taking the form of a finite directed graph. The limiting dynamics of competing mutant subpopulations and successive resident (possibly coexisting) traits can be captured by an algorithmic description. This chapter is based on the paper \cite{CoqKrautSma21}.

\section{Power law mutation rates}

We already studied a scenario of higher mutation probabilities on a finite trait graph, namely by considering the limit of $\mu\to0$ in the deterministic system. Recall the heuristics from Section \ref{section42}, where a limited radius of mutation was introduced. This idea is now made rigorous for the stochastic model. 

As trait space, consider a finite directed graph $\cG=(V,E)$, where the vertices $v\in V$ are the different attainable traits and directed edges $(v,w)\in E$ mark the possibility of mutation at birth, i.e.\ $m(v,w)>0$ if and only if $(v,w)\in E$. In the following, when we speak of $\ell$th-order neighbours of a trait $v$, we refer to those traits that are reachable by a shortest directed path of length $\ell$ from trait $v$. Moreover, we introduce the directed graph distance $d(v,w)$ as the length of the shortest directed path from $v\in V$ to $w\in V$.

A resident population of size of order $K$ induces neighbouring mutants at a rate of order $K\mu_K$. If, as in Chapter \ref{chapter5}, the scaling of $\mu_K\ll (K\ln K)^{-1}$ applies, this is a rare event, occurring on a time scale slower than $\ln K$, the time scale of mutant growth. However, if this relation is reversed, or even $\mu_K\gg K^{-1}$, then new mutants are induced much faster, and all neighbouring traits of the resident traits are immediately present at a size of order $K\mu_K\gg1$. Iterating this procedure, second-order mutants (at distance 2 to the resident) are produced at a rate $K\mu_K^2$, third-order mutants at a rate $K\mu_K^3$, and so on. In the following, we consider the case where mutants are spread within a radius of $\alpha$ of the resident traits, i.e.\ we choose mutation probabilities of the form $\mu_K=K^{-1/\alpha}$, for some $\alpha\notin\mathbb{N}$, such that $K\mu_K^\ell\gg1$ for $\ell<\alpha$ but $K\mu_K^\ell\ll1$ for $\ell>\alpha$.

Due to the multiple competing mutant traits that can be present at the same time under this scaling of mutation probabilities, we need to study their population size during the exponential growth phase in more detail. To this end, it is convenient to consider the orders of population sizes rather than the absolute population size. With the number of individuals of trait $v$ at time $t$ being equal to $N^K_v(t)=\langle\nu^K_t,\1_v\rangle\cdot K$, let
\be
\beta^K_v(t):=\frac{\ln(1+N^K_v(t\ln K))}{\ln K}\ \Leftrightarrow\ N^K_v(t\ln K)=K^{\beta^K_v(t)}-1.
\ee
Consequently, a single new mutant starts at an order of $\beta^K_v\approx0$ and reaches a size of order $K$ when $\beta^K_v\approx 1$. The time scale of order $\ln K$ is considered since it is the timescale of exponential population growth and exponential growth of $N^K_v$ at rate $f$, i.e.\ $N^K_v(s)=N^K_v(0)\cdot \eee^{fs}$, translates to linear growth of $\beta^K_v$, i.e.\ $\beta^K_v(t)=\beta^K_v(0)+fs$. Note that due to the self-competition of every trait, it is ensured that asymptotically $\lim_{K\to\infty}\beta^K_v(t)\in[0,1]$ for all times $t$.

\section{Results}\label{chapter7results}
The main result of this section is to prove the convergence of the $(\beta^K_v(t))_{v\in V, t\in [0,T]}$ to a limiting $(\beta_v(t))_{v\in V, t\in [0,T]}$, as $K\to\infty$, and to describe this limiting object. Based on this, we can then also state a limit result for the rescaled population sizes $(N^K_v(t\ln K)/K)_{v\in V, t\in [0,T]}$.

The key technique is again to couple the different mutant populations with simpler birth-death processes (with immigration). For those simpler processes, we make use of the following two technical Lemmas.

\begin{lemma}[Champagnat-Méléard-Tran \cite{CMT2021}] \label{lemma_BDPexponents}
Let $Z^K$ be simple birth death processes with birth rate $b$ and death rate $d$, and let $Z^K(0)=\lfloor K^\beta-1\rfloor$, for some $\beta>0$. Then, for all finite times $T>0$,
\be
\left(\frac{\ln(Z^K(t\ln K)+1)}{\ln K}\right)_{t\in[0,T]}\to (\bar{\beta}(t))_{t\in[0,T]}, \text{ as }K\to\infty,
\ee
in probability in $L^\infty([0,T])$, where $\bar{\beta}(t):=(\beta+(b-d)t)\vee 0$.
\end{lemma}

This formalises the exponential growth of $Z^K$ and hence the linear growth of $\beta^K$ at rate $b-d$. The proof of this lemma consists of three main steps. First, the martingale problem for $Z^K$ yields the desired dynamics in expectation, i.e. $\mathbb{E}(Z^K(t))=\eee^{(b-d)t}K^\beta$. Second, the bracket of the martingale induced by the rescaled process $\eee^{-(b-d)t}Z^K(t)=K^\beta+\tilde{M}(t)$ is calculated. Finally, Doob's inequality implies that the probability of diverging from the expected value tends to 0.

\begin{lemma}[Corollary of results in Champagnat-Méléard-Tran [ref] ]\label{lemma_BDPIexponents}
Let $Z^K$ be a birth-death process with immigration, with birth rate $b$, death rate $d$, immigration rate $K^c\eee^{as}$ at time $s\geq0$, $Z^K(0)=\lfloor K^\beta-1\rfloor$, for some $\beta\geq0$. Assume that either $c>0$ or both $c=0$ and $a>0$. Then, for any $\delta>0$ and $T>0$, 
\be
\left(\frac{\ln(Z^K(t\ln K)+1)}{\ln K}\right)_{t\in[\delta,T]}\to (\bar{\beta}(t))_{t\in[\delta,T]}, \text{ as }K\to\infty,
\ee
in probability in $L^\infty([\delta,T])$, where $\bar{\beta}(t):=((\beta\vee c) + (b-d)t) \vee (c + at) \vee 0$.
\end{lemma}

This lemma is proved similarly to the previous one, taking into account the incoming mutation. The different terms represent different sources of population growth. Firstly, the initial exponent $\beta\vee c$ (taking into account initial incoming mutation at rate $K^c$), followed by linear growth of the exponent at intrinsic rate $b-d$. Secondly, the exponent is purely due to incoming mutation at rate $K^{c+at}$ at time $t\ln K$. Lastly, the exponent cannot drop below 0 because this already corresponds to a population size of 0.

These two lemmas now allow us to derive upper and lower bounds for the true $\beta^K$, iteratively taking into account incoming mutants from a wider radius of neighbouring traits. For the purpose of the following heuristics, we use the notation of the limiting processes $\beta_v(t)$ directly. These claims hold for the $\beta^K_v(t)$ and large enough $K$. We begin with the lower bound.

Similar to arguments in previous chapters, as long as coexisting resident traits $\textbf{v}\subset V$ are close to their equilibrium size, the total mutant population has a size smaller than $\varepsilon K$, and mutant trait $w\in V$ is alive at time 0 (in the sense of $\liminf_{K\to\infty}\beta^K_w(0)>0$, we can bound it from below with a coupled pure birth death process, absorbing the competition from the resident population into the death rate. With Lemma \ref{lemma_BDPexponents}, this yields
\be
\beta_w(t)\geq\beta_w(0)+(f(w,\textbf{v})-C\varepsilon)t.
\ee
We apply the same estimate for all living traits $u$ that have $w$ as a direct neighbour (i.e.\ $d(u,w)=1$), such that $w$ receives incoming mutants at a rate of at least $K^{\beta_u(0)-1/\alpha}\eee^{(f(u,\textbf{v})-C\varepsilon)s}$ at time $s>0$. Hence Lemma \ref{lemma_BDPIexponents} implies that
\begin{align}
\beta_w(t)\geq\max_{u:d(u,w)=1}\Big\{&\left((\beta_w(0)\vee(\beta_u(0)-\tfrac{1}{\alpha}))+(f(w,\textbf{v})-C\varepsilon)t\right)\notag\\
&\vee \left(\beta_u(0)-\tfrac{1}{\alpha}+(f(u,\textbf{v})-C\varepsilon)t\right)\vee0\Big\}.
\end{align}
Iterating this procedure, taking traits with increasing distance $d(u,w)$ into account, as well as setting 
\be
\tilde{\beta}_w(0):=\max_{u:\beta_u(0)>0}\left\{\beta_u(0)-\tfrac{d(u,w)}{\alpha}\right\}
\ee
as the starting order of population sizes, immediately after the spreading of initial mutants, yields
\be
\beta_w(t)\geq\max_{u:\tilde{\beta}_u(0)>0}\left[\tilde{\beta}_u(0)-\tfrac{d(u,w)}{\alpha}+(f(u,\textbf{v})-C\varepsilon)t\right]\vee0.
\ee

The argument for the upper bound is similar. Instead of no incoming mutants, we start by assuming the highest possible rate of incoming mutants $CK^{1-1/\alpha}$ from neighbours of a size of order $K$. Lemma \ref{lemma_BDPIexponents} (with $c=1-1/\alpha$ and $a=0$) then implies
\be
\beta_w(t)\leq ((\beta_w(0)\vee(1-\tfrac{1}{\alpha}))+(f(w,\textbf{v})+C\varepsilon)t)\vee(1-\frac{1}{\alpha})\vee0.
\ee
We can now, similarly to before, take incoming mutants from neighbours at increasing distance into account, applying the previous upper bounds to those traits, which will decrease the additional terms from $(1-i/\alpha)$ to $(1-(i+1)/\alpha)$ in each step. We finally obtain
\be
\beta_w(t)\leq\max_{u:\tilde{\beta}_u(0)>0}\left[\tilde{\beta}_u(0)-\tfrac{d(u,w)}{\alpha}+(f(u,\textbf{v})+C\varepsilon)t\right]\vee0.
\ee

Considering arbitrarily small values of $\varepsilon$, the upper and lower bounds are the same and provide us with the exact orders of population sizes. Note that these estimates are only valid as long as no new (possibly fitter) mutant traits arise, and we hence need to restart with a new set of traits to take the maximum over whenever a new trait arises.

Overall, these arguments allow us to algorithmically construct the limiting processes $\beta_w$ as follows:
Let $\textbf{v}_0\subset V$ be the initial set of resident traits. For simplicity, we assume that the initial orders of population sizes converge to $\beta_w(0):=\left(1-\frac{d(\textbf{v}_0,w)}{\alpha}\right)_+$. More general initial conditions can be handled by defining a $\tilde{\beta}(0)$ as above.

The increasing sequence of invasion times is denoted by $(s_k)_{k\geq0}$, where $s_0:=0$ and, for $k\geq1$,
\be
s_k:=\inf\{t>s_{k-1}:\exists\ w\in V\backslash\textbf{v}_{k-1}:\beta_w(t)=1\}.
\ee
By $\textbf{v}_k$, we denote the set of coexisting resident traits that emerges from the Lotka-Volterra dynamics involving the previously coexisting resident traits $\textbf{v}_{k-1}$ and the trait $w\in V\backslash\textbf{v}_{k-1}$ that satisfies $\beta_w(s_k)=1$.

For $s_{k-1}\leq t\leq s_k$, for any $w\in V$, $\beta_w(t)$ is defined by
\be\label{beta}
\beta_w(t):=\max_{\substack{u\in V}}\left[\beta_u(s_{k-1})+(t-t_{u,k}\land t)f(u,\textbf{v}_{k-1})-\frac{d(u,w)}{\alpha}\right]\lor 0,
\ee
where, for any $w \in V$,
\be \label{deftwk}
t_{w,k}:=\begin{cases}\inf\{t\geq s_{k-1}:\exists\ u\in V: d(u,w)=1, \beta_u(t)=\frac{1}{\alpha}\}
&\text{if }\beta_w(s_{k-1})=0,\\s_{k-1}&\text{else},\end{cases}
\ee
is the first time in $[s_{k-1},s_k]$ when this trait arises.

The stopping time $T_0$, which terminates the inductive construction of the limiting trajectories, is set to $s_k$ if
\begin{itemize}
\item[(a)] there is more than one $w\in V\backslash\textbf{v}_{k-1}$ such that $\beta_w(s_k)=1$;
\item[(b)] the mutation-free Lotka-Volterra system associated with $\textbf{v}_{k-1}$ and the trait $w\in V\backslash\textbf{v}_{k-1}$ that satisfies $\beta_w(s_k)=1$ does not have a unique globally attractive stable equilibrium;
\item[(c)] there exists $w\in V\backslash\tilde{\textbf{v}}_{k-1}$ such that $\beta_w(s_k)=0$ and 
$\beta_w(s_k-\varepsilon)>0$ for all $\varepsilon>0$ small enough;
\item[(d)] there exists $w\in V\backslash\textbf{v}_{k-1}$ such that $s_k=t_{w,k}$.
\end{itemize}

With these definitions at hand, we can now formulate the main results of this paper. The first describes the convergence of the orders of the population sizes.

\begin{theorem}\label{ThmConv}
Let $\mathcal{G}=(V,E)$ be a finite graph. Suppose that $\alpha\in\mathbb{R}_+\backslash\mathbb{N}$ and $f(w,\textbf{v})\neq0$ for all coexisting traits $\textbf{v}\subset V$ and traits $w\notin \textbf{v}$.
Let $\textbf{v}_0\subset V$ and assume that, for every $w\in V$ ,
\begin{align}\label{cond-init}
\beta^K_w(0)\rightarrow\left(1-\frac{d(\textbf{v}_0,w)}{\alpha}\right)_+, \quad (K\to\infty) \quad \text{in probability}.
\end{align}
Then, for all $T>0$, as $K\to\infty$, the sequence $((\beta^K_w(t),w\in V),t\in[0,T\land T_0])$ converges in probability in 
$\mathbb{D}([0,T\land T_0],\R_+^V)$ to the deterministic, piecewise affine, continuous function $((\beta_w(t),w\in V),t\in[0,T\land T_0])$, which is defined above.
\end{theorem}

The second result describes the convergence of the rescaled population sizes and precisely quantifies the times of invasions.

\begin{proposition} \label{CorEquilibria}
With the same notation and assumptions as in Theorem \ref{ThmConv},
for all $T>0$, as $K\to\infty$, the sequence $((N^K_w(t\log K)/K,w\in V),t\in[0,T\land T_0])$ 
converges 
in the sense of the finite-dimensional distributions
to a deterministic jump process $((N_w(t),w\in V),t\in[0,T\land T_0])$, which is defined as follows:
\begin{itemize}
\item[(i)] For $t\in[0,T_0]$, $N(t)$ jumps between different Lotka-Volterra equilibria according to
\begin{align*}
N_w(t):=\sum_{k\in\N:s_{k+1}\leq T_0}\mathbf{1}_{s_k\leq t<s_{k+1}}\1_{w\in\textbf{v}_k}\bar{n}_w(\textbf{v}_k).
\end{align*}
\item[(ii)]  The invasion times $s_k$ and the times $t_{w,k}$ when new mutants arise 
can be calculated as follows.
We define the increasing sequence $(\tau_\ell,{\ell\geq0})=\{s_k,{k\geq0}\}\cup\{t_{w,k},{w\in V,k\geq0}\}$ of invasion times or appearance times of new mutants, 
and $(M_\ell, \ell\geq0)$ the sets of living traits in the time interval $(\tau_\ell,\tau_{\ell+1}]$. Initially, $\tau_0=s_0=0$ 
and, according to \eqref{cond-init}, $M_0=\{w\in V:d(\textbf{v}_0,w)<\alpha\}=\{w\in V:\beta_w(0)>0\}$. For $s_{k-1}\leq\tau_{\ell-1}<s_k$, $\tau_\ell$ is defined as
\begin{align*}
\tau_\ell:=s_k\land\min\{t_{w,k}:w\in V, t_{w,k}>\tau_{l-1}\}.
\end{align*}

Given $\tau_\ell$ and $M_{\ell-1}$, we set $M_\ell:=(M_{\ell-1}\backslash\{w\in V:\beta_w(\tau_\ell)=0\})\cup\{w\in V:\tau_\ell=t_{w,k}\}$. $\tau_\ell$ is then given by
\begin{align}\label{Dtau-formula}
\tau_\ell-\tau_{\ell-1}=\min_{\substack{w\in M_{\ell-1}:\\f_{w,\textbf{v}_{\ell-1}}>0}}
\frac{\left(1\land\frac{d(w,V\backslash M_{\ell-1})}{\alpha}\right)-\beta_w(\tau_{\ell-1})}{f(w,\textbf{v}_{\ell-1})}.
\end{align}
\end{itemize}
\end{proposition}

The algorithmic description of the limiting processes can be nicely visualised by so-called \emph{broken line pictures}, which depict the orders of population sizes over time, with slopes corresponding to the respective traits' fitnesses. This is shown for some examples in the next section.

\section{Examples}
In this section, we present three examples of evolutionary trajectories on finite trait graphs that demonstrate certain behaviour that is observable on the $\ln K$ time scale.

\subsection{Arbitrarily large jumps}
A natural question to ask is whether the radius $\alpha$, in which resident traits spread to mutants, restricts the range of the jumps that the macroscopic population process can take on the trait graph to traits which are at a distance less than $\alpha$. The answer is no, as the first example shows.

Consider the graph $\mathcal{G}$ depicted on Figure \ref{fig:lnK_Ex1}, where $V = 
\{0, 1, 2, 3, 4\}$ and $E = \{[0, 1], [1, 2], [2, 3], [3, 4]\}$. Let $3<\alpha<4$, an initial condition $\beta(0)=(\bar{n}(0), 0, . . . , 0)$, and a fitness landscape 
\begin{align}
	f(3,0), f(4,0)>0,\\
	f(1,0), f(2,0), f(0,4), f(1,4), f(2,4), f(3,4)<0,\\
	 \frac{1}{f(4,0)}+\frac{-1+4/\alpha}{f(3,0)}<\frac{3/\alpha}{f(3,0)}.
\end{align}
In this case, the resident trait initially spreads mutants of all traits smaller than or equal to 3 and thus the population of trait 4 vanishes at time 0. However, Proposition \ref{CorEquilibria} implies that the rescaled macroscopic population jumps from trait 0 to trait 4 in time $s_1= \tfrac{1}{f(4,0)}+\tfrac{-1+4/\alpha}{f(3,0)}$ on the $\ln K$ time scale. The exponents are drawn in Figure \ref{fig:lnK_Ex1}. The last assumption on the fitness landscape ensures that trait 4 fixates before trait 3.

It is easy to generalise this example to construct jumps to any distance $L$ larger than $\alpha$, by assuming larger and larger fitnesses beyond the initial negative fitness region. The condition implying the emergence of trait $L$ is then a little more technical to write, since one has to compute the time for the piecewise affine function $\beta_L(t)$ (with multiple slope-breaks) to reach 1 before the other traits.

\begin{figure}[h]
\begin{center}
\includegraphics[width=.6\textwidth]{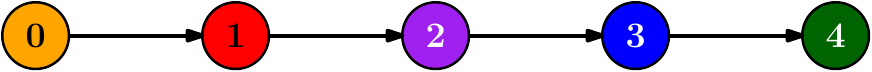}\\[1em]
\includegraphics[width=.8\textwidth]{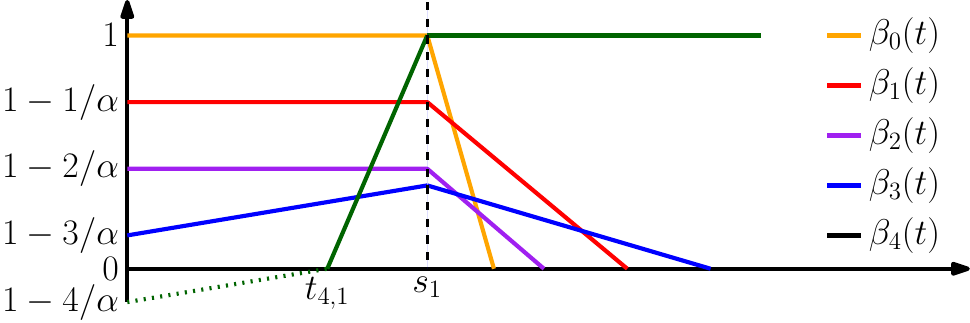}
\end{center}
\caption{Graph $\mathcal{G}$ and exponents $\beta_i(t)$ demonstrating macroscopic jumps on the trait graph of size larger than $\alpha$.}
\label{fig:lnK_Ex1}
\end{figure}

\subsection{Longer path than expected}
If evolution and mutation time scales are separated (i.e.\ in the very small mutation rate regime), mutations occur one at a time, and the number of successive resident traits from the wild type to the type gathering $k$ successively beneficial mutations is $k$. This is not the case if mutations are faster, in which case it is possible to observe even more intermediate resident traits, as the following example shows.

Consider the directed graph $\mathcal{G}$ depicted on Figure \ref{fig:lnK_Ex2}, where $V = \{00, 01, 10, 11\}$ and $E = \{[00, 01], [00, 10], [01, 11], [10, 11]\}$. Let $\alpha > 2$, an initial condition given by $\beta(0)=(\bar{n}(00), 0, 0, 0)$ and a fitness landscape given by
\begin{align}
	&f(01,00), f(10,01) , f(11,10), f(11,01) >0,\\
	&f(00,01), f(00,10), f(00,11), f(01,10), f(01,11), f(10,00), f(10,11), f(11,00) <0,\\
	&f(10,01)>f(11,01).
\end{align}
In this case, in the very small mutation rate regime ($\alpha<1$), the rescaled macroscopic population jumps along 00 - 01 - 11.
In the moderately rare mutations regime (in this case $\alpha>2$, Proposition \ref{CorEquilibria} implies that the rescaled macroscopic population jumps along 00 - 01 - 10 - 11 on the $\ln K$ time scale. More precisely, the exponents are drawn in Figure \ref{fig:lnK_Ex2}. Note that the last condition on the fitness landscape ensures that 11 does not invade 01 before 10 does.

\begin{figure}[h]
\begin{center}
\includegraphics[width=.2\textwidth]{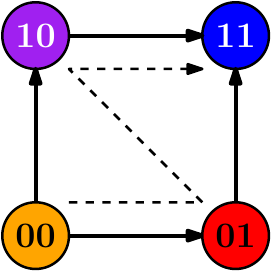}\\[1em]
\includegraphics[width=.8\textwidth]{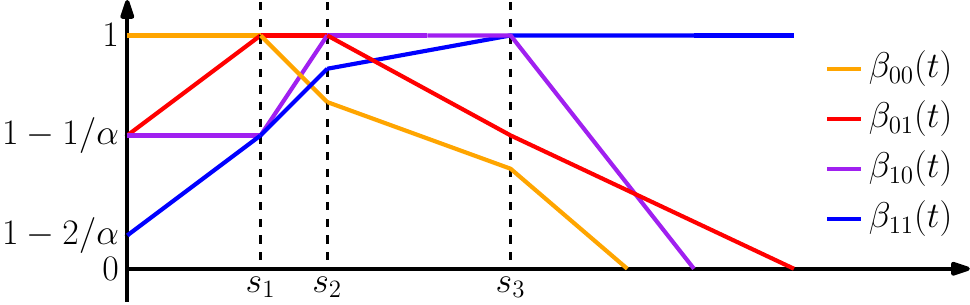}
\end{center}
\caption{Graph $\mathcal{G}$ and exponents $\beta_i(t)$ demonstrating a larger number of intermediate resident traits before reaching 11 than expected.}
\label{fig:lnK_Ex2}
\end{figure}

\subsection{Counter cycle and periodic movement}
Consider the graph $\mathcal{G}$ depicted on Figure \ref{fig:lnK_Ex3}, where $V = \{1, 2, 3\}$ and the edge set is $E = \{[1, 2], [2, 3], [3, 1]\}$. Let $\alpha> 2$, an initial condition $\beta(0)=(\bar{n}(1), 0, 0)$ and a fitness landscape given by
\begin{align}
	f(1,2),f(2,3),f(3,1)>0,\\
	f(1,3),f(2,1),f(3,2)<0.
\end{align}
In this case, Proposition \ref{CorEquilibria} implies that the rescaled macroscopic population jumps along traits 1 - 3 - 2 - 1 (periodic, moving in the clockwise sense) although the mutations are directed counterclockwise. More precisely, the exponents are drawn in Figure \ref{fig:lnK_Ex3}.

If, in addition, the following conditions are satisfied, the period is shorter and shorter, and acceleration takes place, as it is depicted in Figure ?.
\be
f(2,3) > -f(1,3),\quad f(1,2) > -f(3,2),\quad f(3,1) > -f(2,1).
\ee
Note that in the rare mutation regime, with the chosen parameters, there would be no evolution since 2 < 1. Moreover, there are no parameters such that counter-cyclic or accelerating behaviour could arise.

\begin{figure}[h]
\begin{center}
\includegraphics[width=.2\textwidth]{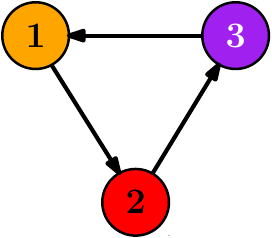}\\[1em]
\includegraphics[width=.8\textwidth]{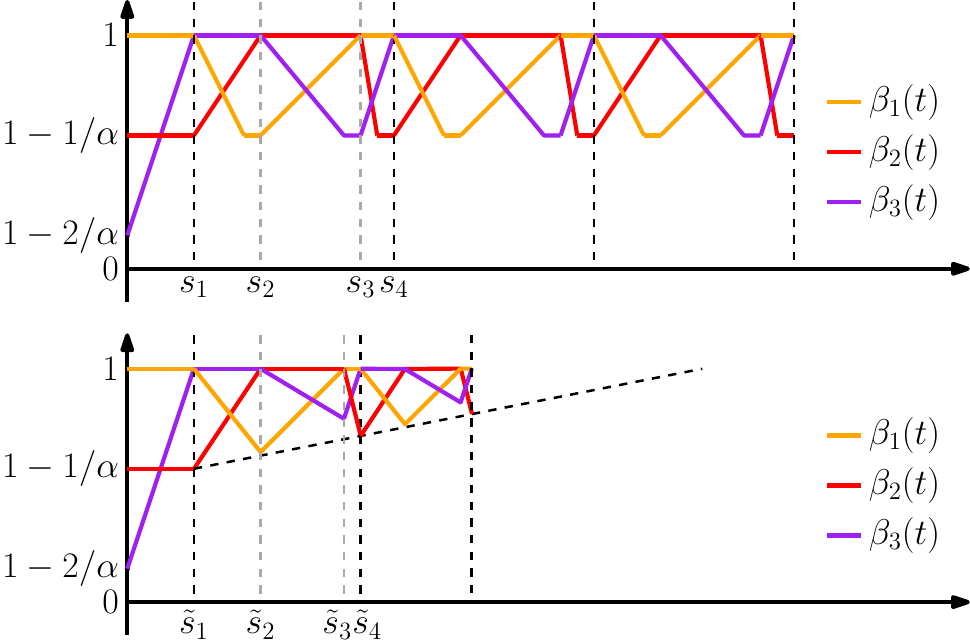}
\end{center}
\caption{Graph $\mathcal{G}$ and exponents $\beta_i(t)$ demonstrating counter cycle movement on the trait space. The second plot satisfies the additional assumptions for accelerated movement.}
\label{fig:lnK_Ex3}
\end{figure}

\newcommand\Ninterval[2]{\llbracket #1,#2\rrbracket}
\newcommand{\xbar}{\bar{z}}
\newcommand{\eps}{\varepsilon}
\newcommand{\ind}[1]{{\mathbf{1}}_{\left[ {#1} \right] }}
\newcommand{\Tinit}{T^{-}}
\newcommand{\Tswap}{ T^{s}}
\newcommand{\xtep}{\tilde x^{\mu_K}}
\def\ep{\varepsilon}
\newcommand{\xep}{x^{\mu_K}}
\newcommand{\xmu}{x^{\mu_K}}
\newcommand{\Dmu}{\Delta^{\mu_K}}
\newcommand{\ymu}{y^{\mu_K}}
\
\chapter{Escape from an evolutionary stable condition 1: a fitness valley}\label{escape}
\begin{chapquote}
{John Maynard Smith, \emph{Natural Selection and the Concept of Protein Space}}
{How often, if ever, has evolution passed through a
non-functional sequence? If so, has this been achieved
by the random walk of genes rendered redundant by
duplication, or by the chance concurrence of two or more
mutations? }
\end{chapquote}

\section{Setting}

In the previous chapters, we have described the evolution of a population 
until the time when it has reached an evolutionary stable condition.
In Chapter \ref{chapter5}, where mutation rates are very small of order $\mu_K\ll1/(K\ln K)$, evolution (on the $1/K\mu_K$ time scale) holds as soon as a resident trait is reached that has no neighbours that can be reached within a single mutation step and has a positive invasion fitness.
In Chapter \ref{chapter7}, where mutation rates are moderately small of order $\mu_K=K^{-1/\alpha}$, the population evolves (on the $\ln K$ time scale) until there are no traits with positive invasion fitness within a distance of $\lfloor\alpha\rfloor$ mutation steps of the - possibly coexisting - resident traits.
However, in both cases, there may be traits with positive invasion fitness at a larger distance that can be reached by a (longer) succession of disadvantageous mutations.
Such scenarios are likely to have occurred many times in 
evolution  \cite{lenski2003evolutionary,cowperthwaite2006bad}. 
 Empirical examples of such phenomena have been found 
in bacteria \cite{schrag1997adaptation,maisnier2002compensatory} and in viruses \cite{o1984vesicular,giachetti1988altered}, for instance.

In this chapter, we give a precise analysis of how an escape from 
an evolutionary stable condition can occur on a longer time scale
in a simple, but illustrative scenario.
This is based on the paper \cite{BovCoqSma2018}. Again, we do not assume small mutation steps in this chapter.

We focus on the scaling of $\mu_K=K^{-1/\alpha}$, which covers the scenario of Chapter \ref{chapter5} for values of $\alpha<1$, and consider the trait space of a linear graph.
First, we demonstrate the implications of Theorem \thv(ThmConv) for evolution to a fit trait within the $\lfloor\alpha\rfloor$ distance, where the $\ln K$ time of mutant growth is the dominating time scale.
Then, as a new result, we consider fit traits that are farther away, but close enough such that mutants occur before the resident trait dies out.
In this case, the fixation time scale is exponentially distributed and dominated by the time needed for the first successful mutant to be born.

\section{A simple model}

We analyse the escape problem in a specific, simple case of our general model, which does, however, capture the key mechanisms.
We chose the trait graph $\mathcal{G}=(V,E)$, where $V=  \{0, 1,\ldots,L\}$.
For simplicity, we will label the traits just by the numbers $j\in V$. 

For simplicity, we assume that
\begin{equation}\label{mut-kernel}
	m(i,j)=\delta_{i+1,j},\quad 0\leq i\leq L-1,\quad m(L,L)=1,
\end{equation}
i.e.\ we allow only mutations to the nearest neighbour to the right, and the final trait $L$ can not mutate.
The results in \cite{BovCoqSma2018} also allow for back-mutations to left neighbours, but we simplify the scenario here, as it does not impact the main ideas.
Biologically, the assumption of nearest neighbour mutations is not unreasonable if we think of mutations that happen on DNA sequences.	
For $i,j\in\N_0$ such that $i\leq j$, we introduce the notation $\llbracket i,j\rrbracket\equiv \{i,i+1,\ldots,j\}$. 

We want to consider the situation when a monomorphic  equilibrium population at $0$ is an evolutionary stable condition, and when $L$ is the closest trait with a positive invasion fitness:

\begin{assumption}[Fitness valley]\label{ass.1}\emph{}\\
\begin{enumerate}
	\item[$\bullet$] All traits are unfit with respect to 0 except $L$:
	\begin{equation} \label{A1} 
		f(i,0)<0, \text{ for } i\in\llbracket 1,L-1\rrbracket  \text{ and } f(L,0)>0.
	\end{equation}
	\item[$\bullet$] All traits are unfit with respect to $L$:
	\begin{equation} \label{A2}
		f(i,L)<0, \text{ for  } i\in\llbracket 0,L-1\rrbracket.
	\end{equation}
	\end{enumerate}
\end{assumption}

\begin{figure}[h!]
	\centering
	\includegraphics[width=.4\textwidth]{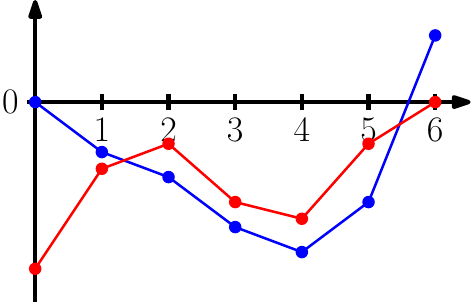}
	\caption {Example of fitness landscape satisfying Assumption \thv(ass.1) with $L=6$. Blue curve: $i\mapsto f(i,0)$, red curve: $i\mapsto f(i,L)$.}
	\label{fitness}
\end{figure}
It is not difficult to show that the model parameters can be chosen so that these assumptions hold.
One possibility is to fix birth and death rates associated with every trait to be $b(i)=1$ and $d(i)=0$, respectively, and set the self-competition of $0$ and $L$ to $c(0,0)=c(L,L)=1$, implying $\bar{n}(0)=\bar{n}(L)=1$.
Assumption \thv(ass.1) will then impose  constraints  on the competition rates $(c(i,0))_{i\in\llbracket 1, L\rrbracket}$ and $(c(i,L))_{i\in\llbracket 0, L-1\rrbracket }$, which must be equal to $(1-f(i,0))_{i\in\llbracket 1, L\rrbracket}$ and $(1-f(i,L))_{i\in\llbracket 0, L-1\rrbracket}$, respectively. 
We can complete the competition matrix by taking symmetric values (except for $c(0,L)$ and $c(L,0)$, which are now fixed, and different) and choosing $c(i,j)=1$ for all pairs $(i,j)\in{\llbracket 1, L-1\rrbracket}^2.$

Under these assumptions on the fitness landscape, as long as $L>\alpha$, all mutants created by the initial population will initially have a negative growth rate and thus tend to die out.
However, if such mutants survive long enough to give rise to further mutants, such that eventually an individual will reach the trait $L$, it will found a population at this trait that, with positive probability, will grow and possibly eliminate the resident population through competition. 
The process that we want to describe in this case can be seen as a manifestation of the phenomenon of \emph{metastability} (see, e.g., the recent monograph \cite {BH15} and references therein).
The initial population appears stable for a long time and makes repeated attempts to send mutants to the trait $L$, which will eventually be reached and take over the entire population. 
As we will see, this leads to several features known from metastable phenomena in other contexts: exponential laws for the transition times, 
a fast realisation of the final ``success run", and the realisation of this run by a ``most likely" realisation.

\section{Results} 

To state the results of this chapter, we again consider the processes $N^K_i(t)=\langle\nu^K_t,\1_i\rangle\cdot K$, i.e.\ the unscaled population sizes of traits $i$ at time $t$.
For $a\geq 0$ and $0\leq i \leq L$,   let  $T_a^{(K,i)}$ denote  the first  time the $i$-population  reaches the  size 
$\lfloor a K\rfloor$, 
\begin{equation}\label{defTepsKM}
 T^{(K,i)}_a := \inf \{ t \geq 0, N^K_i(t)= \lfloor a K \rfloor \}.
\end{equation}

\subsection{Case 1: The LLN regime} \label{res_sto}

The first result concerns the case where one can directly apply Theorem \thv(ThmConv) and Proposition \thv(CorEquilibria), that is, when $L<\alpha$. Adapting the notion of those results, $t_{L,1}=0$, i.e.\ mutants of trait $L$ are present immediately and start growing at an exponential rate $f(L,0)>0$, observable on the $\ln K$ time scale. The time of invasion is in this case
\begin{equation}
	\tau_1=s_1=\frac{L/\alpha}{f(L,0)}.
\end{equation}

The results from Chapter \ref{chapter7}, adapting the perspective from \cite{BovCoqSma2018}, yield the following corollary.

\begin{corollary}
 \label{pro_phase1_mugrand}
Assume that $\mu_K=K^{-1/\alpha}$, $\alpha\notin\mathbb{N}$, and $L < \alpha < \infty$.
 Then there exist two positive constants $\eps_0$ and $c$ such that, for every $0<\eps \leq \eps_0$,
 \begin{equation}\label{eq1th}
  \lim_{K \to \infty} \P \left( (1-c\eps)\frac{1}{\alpha} \frac{L}{f(L,0)}< 
  \frac{T^{(K,L)}_\eps }{\ln K} { <\frac{T^{(K,L)}_{\bar{n}(L)-\eps} }{\ln K}}
 < (1+c\eps)\frac{1}{\alpha} \frac{L}{f(L,0)} \right)=1.
 \end{equation}
 \end{corollary}

\begin{remark} 
Since there is no back-mutation, after the $L$ trait invades the resident population, all other subpopulations are unfit and do not get incoming mutations from a fit trait, and will hence decline. One can compute the time until the $L$ population has driven all other traits to extinction.
Let 
\be\label{timetoextinction}
 t(L,\alpha):= \frac{L}{\alpha} \frac{1}{f(L,0)}+ \sup \left\{\left(1 - \frac{i}{\alpha}\right) \frac{1}{|f(i,L)|}, \; 0 \leq i \leq L-1\right\}, 
 \ee
and the time needed for the   populations at all sites  but $L$ to get extinct,
\begin{equation}\label{defT0ttsaufM}
 T^{(K,\Sigma)}_0 := \inf \Big\{ t \geq 0, \sum_{0\leq i \leq L-1}N^K_i(t)= 0 \Big\}.
\end{equation}
Then 
 \begin{equation}\label{eq2th}
    \frac{T^{(K,\Sigma)}_{0}}{\ln K} \to t(L,\alpha), \quad \text{in probability, as $K\to\infty$.}
 \end{equation}

This is because at the time when trait $L$ becomes the resident trait, all other traits $0\leq i\leq L-1$ have populations of order $K\mu_K^i=K^{1-i/\alpha}$ and then start shrinking exponentially at a rate of $f(i,L)<0$, provided they are not replenished from mutants. Thus, trait $0$ dies out after a time of $\ln K\cdot1/|f(0,L)|$ since it does not get any incoming mutants from other traits If $1/|f(0,L)|<(1-1/\a)/|f(1,L)|$,
the $1$-population is shrinking more slowly and will die out at time $\ln K\cdot(1-1/\a)/|f(1,L)|$, and so on. This gives the formula for $t(\a,L)$.
 \end{remark}

\subsection{Case 2: Metastability}
Next, we consider the case of smaller mutation probabilities or larger fitness valleys, namely $L>\alpha$. In this case, no $L$-mutants
appear within a time of order one, and the fixation of the trait $L$ happens on a much larger time scale.
In this section, we are interested in the case where the mutant $L$ population fixates eventually with a probability close to one. In particular, the first $L$-mutant has to be born before the extinction of the population.

We define, for $0 < \rho <1$,
\be
\lambda(\rho):= \sum_{k=0}^\infty \frac{(2k)!}{{ (k-1)! (k+1)!}} \rho^k \left( 1-\rho \right)^{k+1}.
\ee
For $\lfloor \alpha \rfloor +1 \leq i \leq L-1$, set
\be
\rho(i):= b(i)/(b(i )+ d(i) + c(i,0)\bar{n}(0)).
\ee

\begin{theorem} \label{pro_mupetit}
\begin{itemize}
 \item[(i)] 
 Assume that $\mu_K=K^{-1/\alpha}$, $\alpha \notin \N$, and $\alpha < L$.
 Then there exist {two} positive constants $\eps_0$ and $c$, and exponential random variables $E_-(\eps)$ and $E_+(\eps)$ 
 with  parameters   
 \be\Eq(parameter.1)
  (1\pm c\eps){\frac{ \bar{n}(0)b(0)... b(\lfloor \alpha \rfloor)}{|f(1,0)|...|f(\lfloor \alpha \rfloor, 0)|}} \frac{f(L,0)}{b(L)} \prod_{i=\lfloor \alpha \rfloor +1}^{L-1}\lambda(\rho(i))
  \ee
 such that, for every $0<\eps \leq \eps_0$,
 \begin{equation}
  \liminf_{K \to \infty} \P \left( E_-(\eps)\leq 
  T^{(K,L)}_{\bar{n}(L)-\eps }{K{\mu_K}^L} \leq 
  E_+(\eps) \right) \geq 1-c\eps.
 \end{equation}
 \item[(ii)] There exists a positive constant $V$ such that if ${\mu_K}$ satisfies 
 \be
  \eee^{-VK} \ll K{\mu_K} \ll  1 ,
  \ee
 then the same conclusion holds, choosing any $0<\a<1$ in the definition of the exponential parameters.
 \end{itemize}
 \end{theorem}

In the first case, the typical trajectories of the process are as follows:
Mutant populations of trait $i$, for $1 \leq i \leq \lfloor \alpha \rfloor$, 
reach a size of order $K{\mu_K}^i \gg 1$ in a time of order one due to the high rate of incoming mutations from their left neighbours. Mutant populations of trait $i$, for $\lfloor \alpha \rfloor +1 \leq i \leq L-1$, describe almost surely finite excursions, 
where a proportion of the population of order ${\mu_K}$ produces a mutant of trait $i+1$ during birth events. 
Finally, every $L$-mutant has a probability $f(L,0)/b(L)$ to produce a population which 
outcompetes the resident 0 trait and hence all other populations.

 The term $\lambda(\rho(i))$ is the expected number of birth events before extinction in an excursion of a subcritical
  birth and death process of birthrate $b(i)$ and death rate $d(i)+c(i,0)\bar{n}(0)$.
  Hence ${\mu_K} \lambda(\rho(i))$ is the approximated probability for a trait $i$-population 
  $(\lfloor \alpha \rfloor +1 \leq i \leq L-1)$ to produce a mutant of trait $i+1$,
{and} the overall time scale can be recovered as follows:
  \begin{enumerate}{
\item The last 'large' population is the $\lfloor \alpha \rfloor$-population, which {reaches} a size of order 
{$K{\mu_K}^{\lfloor \alpha \rfloor}$} after a time
   which does not go to infinity with $K$.
   \item The $\lfloor \alpha \rfloor $-population produces an excursion of an  $(\lfloor \alpha \rfloor+1 )$-population at a rate of order 
{$K{\mu_K}^{\lfloor \alpha \rfloor+1}$, which} has a probability of order ${\mu_K}$ to produce an excursion of a  
$(\lfloor \alpha \rfloor+2 )$-population, and so {on,}}
\end{enumerate}
giving a rate of order {$K{\mu_K}^L$} at which mutants of trait $L$ occur.\\

Notice that Theorem \thv(pro_mupetit) implies that, for any mutation rate $\mu_K$ which converges to zero more slowly than 
$\eee^{-VK}/K$, the population will cross 
the fitness valley with a probability tending to 1 as $K\to\infty$. 

The condition  $\alpha \notin \N$ is only a technical assumption which could be dropped, but would 
bring 
more technicalities into the proof. Namely, in this case, 
the population size of trait $\lfloor\alpha \rfloor=\alpha$ would not be divergent with $K$ but would be of order one, and we would have to control its size more carefully.

\begin{remark}
Note that, for $0<\rho<1/2$, one can verify that
\be
\lambda(\rho)= \sum_{k=0}^\infty \frac{(2k)!}{{ (k-1)! (k+1)!}} \rho^k \left( 1-\rho \right)^{k+1}=\rho/(1-2\rho).
\ee
Recalling the definition of $\rho(i)$, this implies
\be
\lambda(\rho(i))=\frac{b(i)}{(b(i)+d(i)+c(i,0)\bar n(0))-2b(i)}=\frac{b(i)}{|f(i,0)|}.
\ee
Plugging this into the exponential parameters in Theorem \thv(pro_mupetit) yields
\bea
&& (1\pm c\eps){\frac{ \bar{n}(0)b(0)... b(\lfloor \alpha \rfloor)}{|f(1,0)|...|f(\lfloor \alpha 
  \rfloor, 0)|}} \frac{f(L,0)}{b(L)} \prod_{i=\lfloor \alpha \rfloor +1}^{L-1}\lambda(\rho(i))
 \\ \nonumber &&=(1\pm c\eps)\frac{\bar n(0)b(0)...b(L-1)}{|f(1,0)|...|f(L-1,0)|}\frac{f(L,0)}{b(L)}=(1\pm 
  c\eps)\bar n(0)b(0)\frac{f(L,0)}{b(L)}\prod_{i=1}^{L-1}\lambda(\rho(i)).
\eea
This expression no longer depends on $\alpha$, which might seem surprising. 
However, there is a simple explanation. Basically, until the fixation of the fit type $L$,
there is an essentially constant population $K \bar n(0)$ while all 
other populations together are very small (essentially finite) at all times.
This means that competition between these is irrelevant, and all the offspring of 
mutants evolve independently. 
So the initial Poisson 
process of incoming type 1 mutants, produced by the resident 0-population at rate 
$K\mu_K\bar n(0)b(0)$, undergoes a thinning with survival probability $\lambda(\rho(i))
\mu_k$, the probability that the offspring of a type $i$ individual produces a type $i+1$ 
mutant before going extinct, in every step, to finally arrive at trait $L$ and fixate with 
probability $f(L,0)/b(L)$. The value of $\alpha$ is most significant in its relation to $L$, 
and determining how the time scale of arriving $L$-mutants, $1/K\mu_K^L=K^{(L/\
\a)-1}$, compares to the time scale of trait $L$ population growth, $\ln K$.
\end{remark}

\section{Proofs}

Corollary \thv(pro_phase1_mugrand), i.e.\ the case of $L<\alpha$, is a consequence
 of the results of Chapter \ref{chapter7}. We turn to the proof of Theorem 
 \thv(pro_mupetit), which covers the case of a large valley, i.e.\ $L>\alpha$.

Similar to the results in previous chapters, a key idea is that, as long as the total 
mutant population size is smaller than $\eps K$, the resident population 
size stays close to its monomorphic equilibrium with a probability close to $1$.
This implies that the mutant populations live in an almost constant environment and 
are subject to an almost constant competitive pressure from other individuals, 
$c(i,0)\bar{n}(0)$. This allows us to couple $i$-population sizes 
($1 \leq i \leq L-1$) with subcritical branching processes with migration $N_i^{(K,-)}$ 
and $N_i^{(K,+)}$ to control their dynamics.

Moreover, after the first growing phase for the $L$-population,
if the sum of the $1$- to $(L-1)$-mutant population sizes stays smaller than $\eps K$ 
whereas the 
$L$-mutant population size exceeds the size 
$\eps K$, the $0$ and $L$ populations will behave as if they were the only ones in competition. 
Consequently, the remaining time needed for the $L$-population
to replace the $0$-population will be of order 1.
Hence, the main work consists of
estimating the time needed for $L$-mutants to arrive and reach the population size $
\lfloor \eps K \rfloor$.
The $i$-populations with $i\leq\alpha$ are of a diverging size due to incoming mutants 
and hence well approximated by the deterministic limit, while for the $i$-populations 
with $i>\a$, the immigration term is not large enough, and the populations are well 
described, at each arrival of a single mutant, by subcritical birth-death processes. Each 
of the iterative excursions of the latter populations has the same probability to 
eventually produce an $L$-mutant, which may generate a large population and invade. 
This can essentially be regarded as a thinning of the Poisson process of mutants 
arriving in $\lfloor\alpha\rfloor+1$ and the waiting time is hence exponential with a 
mean of order $1/(K{\mu_K}^L)$, much larger than the $\ln K$-time it subsequently 
takes the $L$-mutant to grow to a size of $\lfloor \eps K \rfloor$.

\subsection{Population size of traits $i<\alpha$} \label{alpha+1geqM}

To simplify a little bit, we will not formally define the exact couplings but rather take the 
slightly heuristic approach of writing $\pm\eps$ in various places and use generic 
positive finite constants $C$ that might vary from line to line. While the total mutant 
population stays below a threshold of $\lfloor\eps K\rfloor$, the resident $0$-population 
stays close to its equilibrium up to deviations of order $\eps K$. This produces a flow 
of mutants to trait $1$ that is controlled from above and below by two Poisson 
processes with intensities
$\mu_K b(0)K(\bar n(0)\pm \eps C)$. Thus, in time of order one, trait 1 will have a 
population of order $K\mu_K$
and produce a flow of mutants of intensity near $\mu_K^2K$ to trait $2$, and so on, 
until, still in time of order one,   
a flow of order $\mu_K^{\lfloor\alpha\rfloor}K$ reaches trait $\lfloor\alpha\rfloor$. By 
our assumption of $\mu_K=K^{-1/\alpha}$, the intensity of this flow still goes to infinity 
with $K$. To derive the exact transition rate of passing through the valley, we need 
precise control of the population sizes of these traits.

Let us first look at trait 1. Since there is no back mutation and the resident $0$ 
population is essentially of constant size, we can think of the population at $1$ as a 
Markov process with generator
\bea
\Eq(one.1)
(\LL^K h)(N^K_1)&=& [h(N^K_1+1)-h(N^K_1)]\left(b(1)(1-\mu_k) N^K_1+ 
\mu_Kb(0)K(\bar n(0)\pm\eps C)\right)\\\nonumber
&&+[h(N^K_1-1)-h(N^K_1)] N^K_1\left(d(1)+c(1,0)\bar n(0) +c(1,1) N^K_1/K\pm \eps 
C\right).
\eea
With $h(x)=x$, this gives rise, as usual, to the martingale
\bea\Eq(one.2)
M^K_1(t)
&=&N^K_1(t)-N^K_1(0) -\int_0^t(\LL^K\text{Id})(N^K_1(s))\, ds\nonumber\\
&=&N^K_1(t)-N^K_1(0) -\int_0^t\Bigl(f(1,0)N^K_1(s)+b(0)\mu_KK (\bar n(0)\pm\eps C) 
\nonumber\\
&&\quad -(\mu_Kb(1)+c(1,1)N^K_1(s)/K)N^K_1(s)\Bigr)\, ds.
\eea
Since we anticipate that $N^K_1(t)$ will never grow beyond $C \mu_KK$, at least until 
a stopping time, the last line in 
\eqv(one.2) is not relevant. 
Setting $x_1(t)=\E\left[ N^K_1(t)/\mu_KK\right]$, 
we get that 
\be
\Eq(one.3)
x_1(t)-x_1(0)=\int_0^t\Bigl(f(1,0)(1\pm \eps C)  x_1(s)+b(0) (\bar n(0)\pm\eps C)\Bigr)\, 
ds. 
\ee 
This integral equation has the explicit solution 
\be\Eq(one.4)
x_1(t) = \frac {\eee^{f(1,0)(1\pm \eps C) t}-1}{f(1,0)(1\pm\eps C)}  b(0) (\bar 
n(0)\pm\eps C).
\ee
Since $f(1,0)<0$, $\eee^{f(1,0)(1\pm \eps C) t}\leq\eps$ for $t\geq\ln(\eps)/f(1,0)
(1\pm\eps C)$ and hence $x_1(t)$ converges, as $t\uparrow\infty$, to an equilibrium 
\be\Eq(one.5)
\frac { b(0) (\bar n(0)\pm\eps C)}{|f(1,0)(1\pm\eps C)|}.
\ee
Finally, the process stays close to its expectation. To see this, again, we compute the
bracket of $M^K_1$, (recall \eqv(tss-mar.5))
\bea
\Eq(one.6)
[M^K_1]_t&=&
\int_0^t\Bigl( (b(1)+d(1)+c(1,0)\bar n(0) (1\pm \eps C)) N^K_1(s)+b(0)\mu_KK (\bar 
n(0)\pm\eps C) 
\nonumber\\
&&\quad+(\mu_Kb(1)+c(1,1)N^K_1(s)/K)N^K_1(s)\Bigr)\, ds.
\eea
The important fact is that $\E( [M^K_1]_t )\leq C_t\mu_KK $.
To see this, one first derives a bound as in the proof of the law of large numbers, for 
$t\leq T$,
\be
\E\left[\sup_{s\leq t}N^K_1(s)^2\right]\leq ( b(0)\mu_KKT)^2+ CT 
\int_0^tE\left[\sup_{u\leq s}N^K_1(s)^2\right]ds,
\ee
so that by Gronwall's lemma
\be
\E\left[\sup_{s\leq t}(N^K_1(s)/\mu_KK)^2\right]\leq 
b(0)^2 T^2\eee^{CTt}.
\ee
Using this bound and Eq. \eqv(one.6), one obtains the desired bound on the bracket. 
From there, one 
derives, as in the proof of the LLN, that the fluctuations
of the size of the rescaled population of trait one around its expectation, $x_1(t)$,  are 
of order $1/\sqrt {\mu_KK}$, which tends to zero, as long as $\mu_KK\uparrow\infty$, 
i.e.\ as long as the population size at $1$ diverges.

The same argument can now be repeated for traits $2,3,\dots,\lfloor\alpha\rfloor$ 
iteratively. This allows us to control the size of all 
populations that are close to the resident trait 0. We set, for $1 \leq i \leq 
\lfloor\alpha\rfloor$, 
\begin{equation}\label{defspm}
f^\pm(i,0):=(1-{\mu_K})b(i)-d(i)- c(i,0)(\bar n_0\mp\eps C) ,
\end{equation}
\begin{equation}\label{deftepsi}
x_i^{\pm}:=(1\pm\eps C)^{2i}\frac{b(0)... b(i-1) \bar n_0\mu_K^i}{f^\pm(1,0)...f^\pm(i,0)} 
\quad \text{and}
\quad t^{(i)}_\eps := \frac{\ln (\eps)}{f^+(i,0)} .
\end{equation}
Notice that  $f^-(i,0)\leq f^+(i,0)$ and that $f^-(i,0)$ and $f^+(i,0)$ are negative, for $
\eps$ small enough, by Assumption \eqref{A2}.

The following lemma summarises the results.
\begin{lemma} \label{lemme_bounds_ni} 
For every $0 \leq i \leq \lfloor\alpha\rfloor$,
\be\Eq(all.1)
 x_i^-\leq x_i(s), \quad  s\geq t_\eps^{(1)}+...+t_\eps^{(i)} , 
\ee
and
\be\Eq(all.2)
x_i(s) \leq x_i^+, \quad  s\geq 0. 
\ee
\end{lemma}

\subsection{Excursions of traits $\alpha<i<L$}

For $\lfloor \alpha \rfloor +1 \leq i \leq L-1$, the $i$-mutant population sizes are not 
diverging with $k\to\infty$, and the populations will describe finite excursions until a 
successful mutant with trait $L$ is created, i.e.\ a mutant arrives at $L$ and fixates. 
Here again, the key idea will be couplings with birth and death processes without 
competition.

Let us denote by 
$(T^{(k)})_{k\geq1}$  the birth times of mutants of type ($\lfloor \alpha \rfloor +1$), 
descended from an individual of type $\lfloor \alpha \rfloor$. 
From the previous section, we know that the size of the population at $\lfloor\a\rfloor$ 
is very precisely given by
\be
\Eq(alpha.1)
 K \mu_K^{\lfloor\a\rfloor} \frac{\bar n(0)b(0)\dots b(\lfloor\a\rfloor-1)}{|f(1,0)|\dots|
 f(\lfloor\a\rfloor,0)|}
 (1\pm \eps C)\equiv P(\a).
 \ee
Thus the arrival times $T^{(k)}$ of mutants at $\lfloor\a\rfloor+1$ are close to a Poisson 
process
with intensity $P(\a) b(\lfloor\a\rfloor) \mu_K$ and, in particular, the
times between two mutant arrivals are of order $1/(K\mu_K^{\lfloor\a\rfloor+1})\sim 
K^{(\lfloor\a\rfloor+1)/\alpha-1}\uparrow \infty$. 

Once a mutant appears, the population at $\lfloor\a\rfloor+1$ is (up to small errors of 
order $\eps$) extremely well described by a sub-critical binary branching process with 
birth rate $b(\lfloor\a\rfloor+1)$ and death rate
$b(\lfloor\a\rfloor+1)+c(\lfloor\a\rfloor+1,0)\bar n(0)$. We know from Lemma 
\thv(prop:BP)
that the mean time to extinction of this population is finite, so we can be sure that it will 
have disappeared before the next mutant arrives.
To calculate the probability of an $(\lfloor\alpha\rfloor+2)$-mutant being produced 
before extinction, we need a precise estimate of the number of birth events that take 
place during an excursion.

\begin{lemma} \label{lem_taille_excu}
 Let us consider a birth and death process with individual birth rate $b>0$ and 
 individual death rate $d>0$ satisfying $b<d$.
 Let $Z$ denote the total number of births during an excursion of this process initiated 
 with one individual.
 Then, for $k \geq 0$,
 \begin{equation}\Eq(loi_excursion)
\P(Z=k)=\frac{(2k)!}{k! (k+1)!}\left( \frac{b}{d + b} \right)^{k} \left( \frac{d}{d + b} 
\right)^{k+1}.
 \end{equation}
In particular, 
 \begin{equation}\label{esp_excursion}
\E[Z]= \sum_{{k=1}}^\infty \frac{(2k)!}{{(k-1)!}(k+1)!}\left( \frac{b}{d + b} \right)^{k} 
\left( \frac{d}{d + b} \right)^{k+1}.
 \end{equation}
Moreover, $\E[Z]$ is Lipschitz continuous in the parameters $b$ and $d$.
 \end{lemma}

\begin{proof} 
If $Z=k$, then there must be exactly  $k+1$ death events to bring the process to 
extinction.
Thus, to count the number of ways to distribute the $k$ births on the $2k+1$ events, 
we proceed as 
follows: First, the last event must be a death, so the first $2k$ events bring the 
population from $1$ to 
$1$. There are $(2k)!/(k!k!)$ ways to place $k$ births within $2k$ events to achieve 
this. However, in some of these, the 
population goes to zero on the way. The reflection principle tells us how to count these 
cases:
For any choice that led to a premature death, there is a mirrored "path" (mirrored at the 
first intersection with 0, see Figure \ref{fig:Reflection}) that leads to a population size 
$-1$ at the end, i.e.\ that has $k-1$ births and $k+1$ deaths. 
The number of such arrangements is $(2k)!/(k-1)!(k+1)!$. Subtracting this number from 
all the choices
leads to 
\be
\frac {(2k)!}{k!k!}-\frac {(2k)!}{(k-1)!(k+1)!} =\frac {(2k)!}{k!(k+1)!},
\ee
which is our claim. 
\begin{figure}[h]
\begin{center}
\includegraphics[width=.6\textwidth]{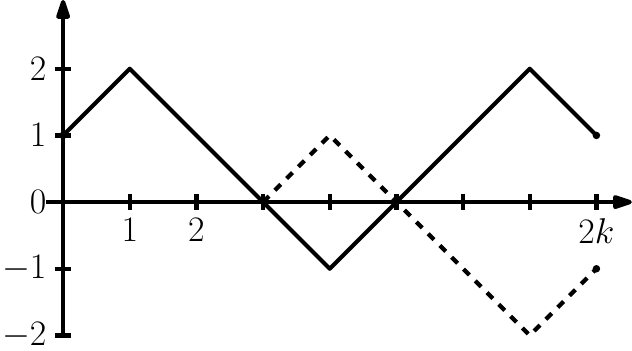}
\end{center}
\caption{Visualisation of the path mirroring for the discrete jump chain of the process 
$Z$.}
\label{fig:Reflection}
\end{figure}

Multiplying by the probabilities of each step gives \eqv(loi_excursion).
Verifying Lipschitz continuity is straightforward.
\end{proof}

It is now easy to  compute the probability that  a mutant that arrives at $\lfloor 
\a\rfloor+1$ produces a
mutant to $\lfloor \a\rfloor+2$ before its offspring dies out: Given that $k$ births 
happen, the probability that exactly one mutant is born
is $k\mu_K(1-\mu_K)^{k-1}$, while the probability that more than one mutant is born is 
of order $\mu_K^2$ and 
can be neglected.
Since the probability that $k$ births happen has been computed above,
the probability of one  mutant being produced is
\be
\mu_K\sum_{k=1}^\infty k  \frac{(2k)!}{k! (k+1)!} (\rho(\lfloor \a\rfloor+1)^k \left(1-
\rho(\lfloor \a\rfloor+1)\right)^{k+1} \equiv \mu_K\l\left(\rho(\lfloor \a\rfloor+1)\right),
\ee
up to a multiplicative error of order $1+O(\mu_K)$,
where 
$\rho(i)=\frac {b(i)}{b(i)+d(i)+c(i,0)\bar n(0)}$.  Conditional on this mutant having 
appeared at $\lfloor\a\rfloor+2$,
there is a again a probability $\mu_K\l\left(\rho(\lfloor \a\rfloor+2)\right)$
that it will give rise to a mutant at $\lfloor\a\rfloor+3$ before the $\lfloor \a\rfloor+2$-
population dies out,
and so on. Finally, 
the probability that from the first mutant at $\lfloor\a\rfloor+1$ a mutant offspring will 
arrive
at $L$ 
is 
\be\Eq(almost-there)
\mu_K^{L-\lfloor\a\rfloor-1}\prod_{i=\lfloor\a\rfloor+1}^{L-1}\l\left(\rho(i)\right).
\ee
This mutant now seeds a supercritical branching process, and this will reach a 
population size of
$\eps K$ before dying out with probability 
$f(L,0) /b(L)$. Thus, any new mutant of type $\lfloor\a\rfloor+1$ leads to the fixation 
of a population at $L$ with probability
\be\Eq(finally-done)
\mu_K^{L-\lfloor\a\rfloor-1}\frac {f(L,0) }{b(L)}  
\prod_{i=\lfloor\a\rfloor+1}^{L-1}\l\left(\rho(i)\right).
\ee
Since the initial flow of mutants from $\lfloor\a\rfloor$ is (very close to) a Poisson 
process with intensity $P(\a) b(\lfloor\a\rfloor) \mu_K$ and the random thinning
is again Poisson, the time of the fixation of the population at $L$ is an exponential 
random variable
with parameter
given by the advertised value  \eqv(parameter.1). 

In the case then $K\mu_K\ll1$, we just have to set $\lfloor \a\rfloor=0$. The condition 
$\mu_K\gg \eee^{-VK}$ ensures that the resident fluctuations are small until the fixation of the mutant.

Once the population at $L$ reaches the level $\eps K$, the deterministic 
approximation sets in and in time 
of order $1$, the new equilibrium is approached to an $\eps$-distance. After an 
additional time of order $\ln K$,
the populations at $0,\dots, L-1$ are extinct.
This concludes the proof of Theorem \thv(pro_mupetit).

\chapter{Escape from an evolutionary stable condition 2: general trait graphs}\label{chapter9}

\begin{chapquote}
{Thomas Malthus, \emph{Essay on the Principle of Population}}
{\frakfamily\fraklines {
I have read some of the speculations on the perfectibility\\
of man and society, with great pleasure. I have been warmed \\
and delighted with the enchanting picture which they hold forth.\\
I ardently wish for such happy improvements.  But I see great, \\
and, to my understanding, unconquerable difficulties in the way to them.}}
\end{chapquote}


\section{Beyond the valley}
In chapter \ref{escape}, we saw how a population can transition through a valley in the fitness landscape on a prolonged time scale $1/K\mu_K^L$, despite the first fit type being at a distance that is initially not reached by mutations. The proposed Assumption \ref{A2} ensured that once the fit mutant trait $L$ is reached, it becomes the new resident type and all other populations go extinct. In this case, on the correct time scale, we would observe the instantaneous transition from trait $0$ to trait $L$ at their respective equilibrium population sizes with no further evolution afterwards. This is what we will refer to as an evolutionary stable condition (ESC), which will be defined precisely below.

Subtle modifications of Assumption \ref{ass.1} can, however, drastically change the observed evolutionary pathway. We give three examples to demonstrate possible evolutionary phenomena before introducing the general setting and results in the following sections.

\subsection{Fast evolution after slow transition}

For the first example, we only slightly change Assumption \ref{ass.1} to make one trait in the valley fit with respect to the new resident trait $L$. We let $2<\alpha<3$, $L=4$, still assume mutation to direct right neighbours only, and impose the fitness landscape
\begin{align}
	f(i,0)<0, \text{ for } i\in\llbracket 1,3\rrbracket  \text{ and } f(4,0)>0,\\
	f(i,4)<0, \text{ for } i\in\{0,1,3\} \text{ and } f(2,4)>0,\\
	f(i,2)<0, \text{ for } i\in\{0,1,3,4\}.
\end{align}

Under these modified assumptions, the dynamics are the same as in Chapter \ref{escape} until trait $L=4$ becomes resident. At this time, all traits other than $4$ and $2$ are unfit and would go extinct. Trait $2$, however, has a positive invasion fitness and initially still has a positive population size of order $K\mu_K^2\gg1$. It can hence grow on the $\ln K$-time scale and reach a macroscopic size to compete with trait $4$. Since every trait is unfit in the presence of $2$, trait $2$ becomes the final resident trait. The populations of $0$ and $1$ go extinct, while $3$ and $4$ survive at a microscopic size due to incoming mutations from $2$. This evolutionary trajectory is summarised as a broken line picture for the orders of population sizes in Figure \ref{fig:BeyondValley}.
\begin{figure}[h]
\begin{center}
\includegraphics[width=.6\textwidth]{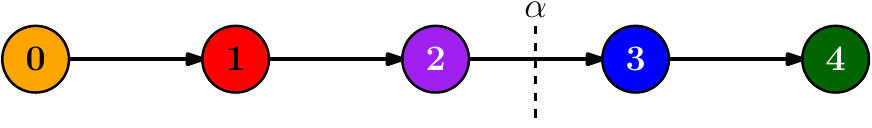}\\[1em]
\includegraphics[width=.8\textwidth]{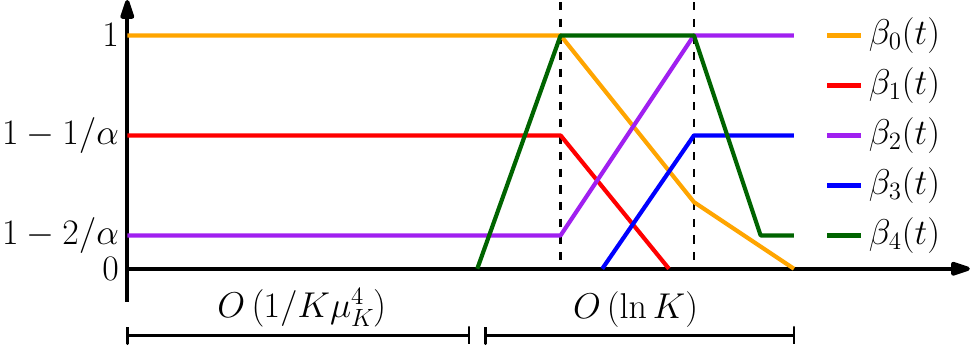}
\end{center}
\caption{Trait graph and broken line picture for orders of population sizes in Example 1.}
\label{fig:BeyondValley}
\end{figure}

Notably, since the transition from resident trait $4$ to $2$ occurs on the $\ln K$-time scale, this time is negligible on the accelerated time scale $1/K\mu_K^4$, on which the transition from $0$ to $4$ occurs. On this time scale, one hence observes a direct transition from the equilibrium states of $0$ to $2$, i.e.\ a jump that is much shorter than the time scale suggests and not expected since $f(2,0)<0$.

\subsection{Longer jumps on shorter time scales}

For the second example, we show that jumps can also be farther than expected on a certain timescale. To achieve this, we add a second trait to the right of the fitness valley. We let $2<\alpha<3$, $L=4$, still assume mutation to direct right neighbours only, and impose the fitness landscape
\begin{align}
	f(i,0)<0, \text{ for } i\in\llbracket 1,3\rrbracket,\\
	f(5,0)>\frac{1}{1-1/\a}f(4,0)>0,\label{fastergrowth}\\
	f(i,5)<0, \text{ for } i\in\llbracket 0,4\rrbracket .
\end{align}

Again, up to the time when mutant trait $L=4$ becomes established, i.e.\ grows to a size that diverges with $K$, the dynamics remain the same as in Chapter \ref{escape}. However, once trait $4$ passes a size of $K^{1/\alpha}$, it starts to foster infinitely many mutants of trait $5$ within a time of order 1, and the trait $5$ population also starts growing on the $\ln K$-time scale. Condition \eqref{fastergrowth} ensures that the time $\ln K/f(5,0)$ it takes trait $5$ to grow to a size of order $K$ is less than the remaining time $\ln K(1-1/\alpha)/f(4,0)$ that it would take trait $4$ to reach the same size. Hence, trait $5$ is the one to compete with trait $0$ and to become the new and final resident trait. The evolutionary trajectory is again summarised as a broken line picture in Figure \ref{fig:FarJump}.
\begin{figure}[h]
\begin{center}
\includegraphics[width=.8\textwidth]{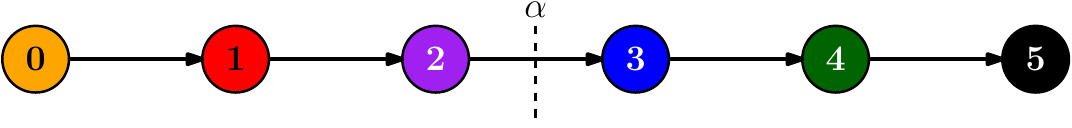}\\[1em]
\includegraphics[width=.8\textwidth]{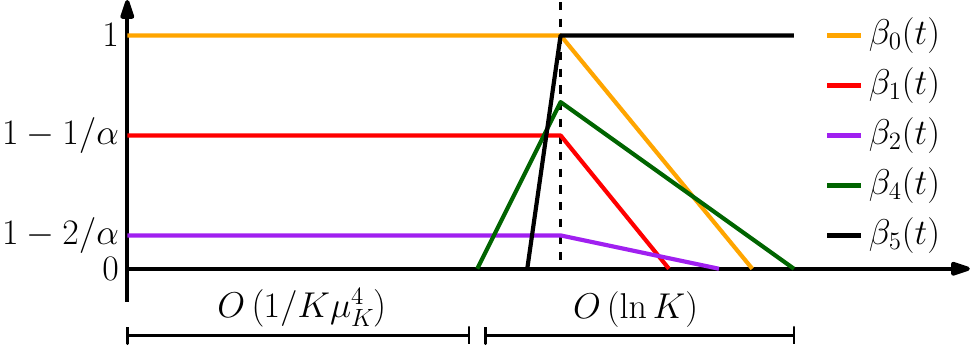}
\end{center}
\caption{Trait graph and broken line picture for orders of population sizes in Example 2.}
\label{fig:FarJump}
\end{figure}

Overall, on the time scale $1/\mu_K^4$, a direct transition between the equilibria of traits $0$ and $5$ is observed.

\subsection{Double valley}

For the final example, we consider the case of a double valley, where two fit traits are accessible at the same distance. We let $1<\alpha<2$, $L=3$, assume mutation to direct right neighbours for positive traits and direct left neighbours for negative traits only (assuming probability 1/2 to mutate left or right from trait $0$), and impose the fitness landscape
\begin{align}
	f(i,0)<0, \text{ for } i\in\{-2,-1,1,2\}  \text{ and } f(\pm3,0)>0,\\
	f(i,\pm3)<0, \text{ for  } i\in\llbracket -2,2\rrbracket.
\end{align}

By similar arguments to Chapter \ref{escape}, on the time scale $1/K\mu_K^3$, successful mutants of trait $3$ arise according to a Poisson process with rate
\begin{equation}
R(3):=\frac{\bar n(0)b(0)b(1)b(2)}{2|f(1,0)||f(2,0)|}\frac{f(3,0)}{b(3)},
\end{equation}
where the factor $1/2$ stems from the mutation kernel of trait $0$. Due to no back mutation, independently of this, successful mutants of trait $-3$ arise according to a Poisson process with rate
\begin{equation}
R(-3):=\frac{\bar n(0)b(0)b(-1)b(-2)}{2|f(-1,0)||f(-2,0)|}\frac{f(-3,0)}{b(-3)}.
\end{equation}

Due to the properties of independent Poisson processes, on the time-scale $1/K\mu_K^3$, a random transition is observed after an exponential time with rate $R(3)+R(-3)$. The new equilibrium state is of trait $\pm3$ with probabilities $R(\pm3)/(R(3)+R(-3))$, respectively. Note that, even if trait $3$ and $-3$ were fit with respect to each other, a transition between them could not be observed on the same time scale since trait $0$, and with it traits $-1$ and $1$, immediately start to decrease in population size and go extinct on the time scale $\ln K$ and a valley of length $6$ is to wide to transition on the considered time scale. Figure \ref{fig:DoubleValley} visualises the corresponding trait graph and two possible evolutionary trajectories as broken line pictures.
\begin{figure}[h]
\begin{center}
\includegraphics[width=.7\textwidth]{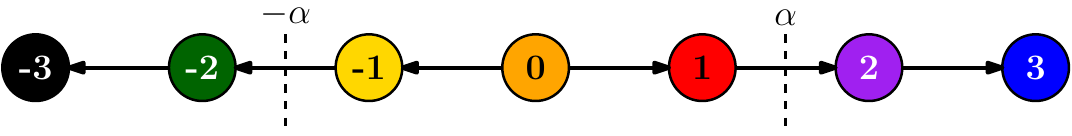}\\[1em]
\includegraphics[width=.49\textwidth]{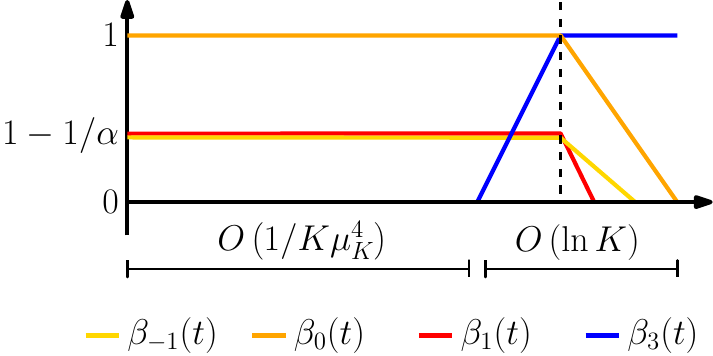}
\includegraphics[width=.49\textwidth]{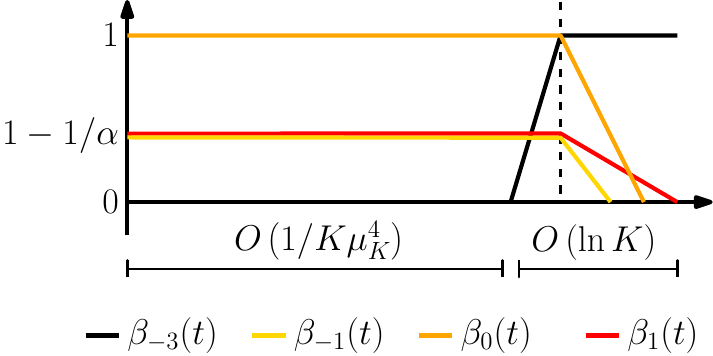}
\end{center}
\caption{Trait graph and broken line pictures for orders of population sizes for two possible evolutionary trajectories in Example 3.}
\label{fig:DoubleValley}
\end{figure}


\section{Setting}
From the previous examples, we have seen that the distance between the resident trait and the closest fit trait, i.e.\ the width of the fitness valley, determines the time scale on which the next transition occurs. However, the next state that is visible on this time scale might not involve the fit trait itself but is rather the final state of the following $\ln K$-dynamics - or even several transitions through smaller fitness valleys. Moreover, if there are multiple fit traits at the same distance, all paths need to be taken into consideration, and the new state is determined at random.

In this chapter, we derive a general description of the evolutionary dynamics on varying time scales for a finite directed trait graph $\cG=(V,E)$. As in Chapter \ref{chapter7}, edges $(v,w)\in E$ mark the possibility of mutation, i.e.\ $m(v,w)>0$, and we consider moderately rare mutations with $\mu_K=K^{-1/\a}$ for some positive $\a\notin\N$. We first introduce some definitions and notation to make precise what is meant by paths on the trait graph, evolutionary stable conditions, and degrees of stability.

\begin{definition}
For a path $\gamma=(\gamma_0,...,\gamma_\ell)$ in $\cG$ (i.e.\ $(\gamma_i,\gamma_{i+1})\in E$, $0\leq i\leq \ell-1$), denote its \textit{length} by $|\gamma|=\ell$. We write $\gamma:\mathbf{v}\to\mathbf{v}'$ as a short notation for all paths $\gamma$ that connect $\mathbf{v}\subset V$ to $\mathbf{v}'\subset V$, i.e.\ that satisfy $\gamma_0\in\mathbf{v}$ and $\gamma_{|\gamma|}\in\mathbf{v}'$.

With this notation, the previously defined graph distance satisfies
\begin{equation}
d(\mathbf{v},\mathbf{v}')=\min_{\gamma:\mathbf{v}\to\mathbf{v}'}|\gamma|.
\end{equation}
Define the \textit{mutation spreading neighbourhood} of a set of vertices $\mathbf{v}$ as $V_\alpha(\mathbf{v})=\{w\in V:d(\mathbf{v},w)<\alpha\}$ and its boundary as $\partial V_\alpha(\mathbf{v})=\{w\in V:d(\mathbf{v},w)=\lfloor\alpha\rfloor\}$.
\end{definition}

Figure \ref{fig:TraitGraph} visualises the mutation spreading neighbourhood (blue) for a set of 
coexisting resident traits $\mathbf{v}$ (pink) and $2<\a<3$. In this case, the closest fit traits are 
at a distance of 4 from the resident traits and marked as $V_{mut}$.

\begin{figure}[h]
\begin{center}
\includegraphics[width=.6\textwidth]{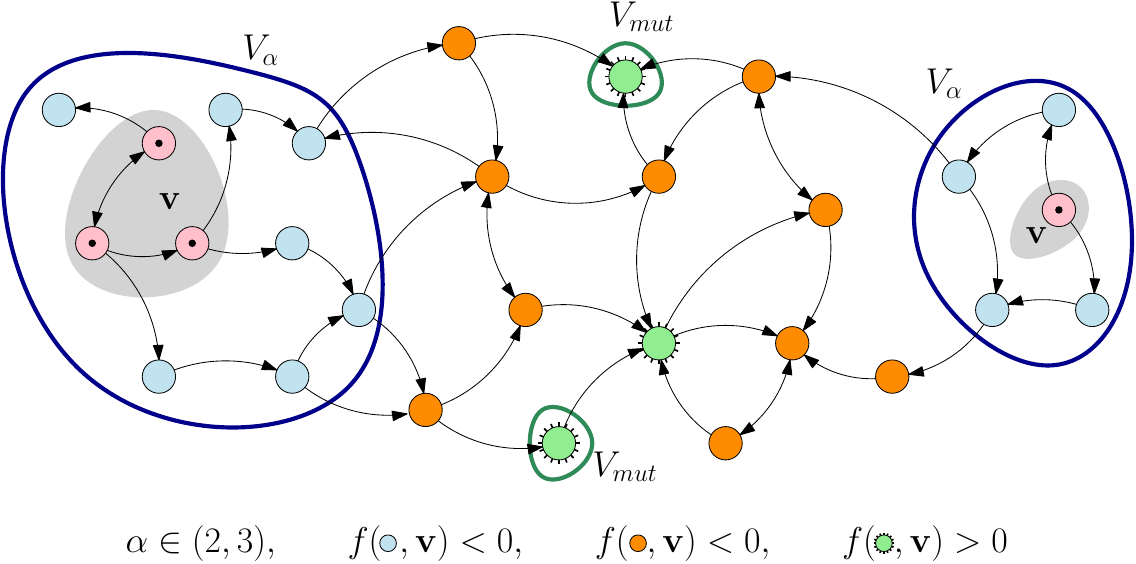}
\end{center}
\caption{Example of a directed trait graph with coexisting resident traits $\mathbf{v}\subset V$ and mutation spreading neighbourhood $V_\a$.}
\label{fig:TraitGraph}
\end{figure}

\begin{definition}
\begin{itemize}
	\item[(i)] A subset $\mathbf{v}\subset V$ and (orders of) population sizes $\beta$ are called an \textit{evolutionary stable condition (ESC)} if the traits $\mathbf{v}$ can coexist at a unique globally asympotically stable equilibrium $\bar{n}(\mathbf{v})$,
	\begin{align}\label{ESCfitness}
		f(w,\mathbf{v})<0,\ \forall w\in V_\alpha(\mathbf{v})\backslash\mathbf{v},
	\end{align}
	and
	\begin{align}\label{ESCsize}
		\beta_w=\left(1-\frac{d(\mathbf{v},w)}{\alpha}\right)_+,\ \forall\ w\in V.
	\end{align}
	\item[(ii)] A subset $\mathbf{v}\subset V$ and population sizes $(\beta^K)_{K\geq0}$ are called an \textit{asymptotic evolutionary stable condition} if the traits $\mathbf{v}$ can coexist at equilibrium $\bar{n}(\mathbf{v})$, \eqref{ESCfitness} is satisfied,
	\begin{align}
	    	\left|\b_w^K-(1-d(\mathbf{v},w)/\alpha)\right|\in O\left(\frac{1}{\ln K}\right), \ \forall w\in V_\a(\mathbf{v}),
	 \end{align}
	 and there exists a $K_0<\infty$ such that $\beta^K_w=0$, for all $K>K_0$ and \mbox{$w\in V\backslash V_\alpha(\mathbf{v})$}.
\end{itemize}
\end{definition}
Note that the last two conditions imply that $N^K_w=K^{\beta^K_w}-1=K^{1-d(\mathbf{v},w)/\a}\cdot O(1)$ for $w\in V_\a(\mathbf{v})$ and $N^K_w=0$ for all $w\in V\backslash V_\alpha(\mathbf{v})$ and $K$ large enough. Essentially, an ESC hence means that traits $\mathbf{v}$ coexist at their equilibrium, all traits that are within their mutation spreading neighbourhood are unfit and at their equilibrium size based on frequently incoming mutants, and no traits outside of the neighbourhood are alive. These are exactly the conditions that describe a final state of evolution on the $\ln K$-time scale, studied in Chapter \ref{chapter7}.

\begin{definition}
	For a subset $\mathbf{v}\subset V$ we define its \textit{stability degree} $L(\mathbf{v})$ by
	\begin{align}
		L(\mathbf{v}):=\begin{cases}\min_{w\in V: f(w,\mathbf{v})>0}d(\mathbf{v},w)&\text{if $\mathbf{v}$ can coexist,}\\
		0&\text{else}.\end{cases}
	\end{align}
	Moreover, set
	\begin{align}
	V_{mut}(\mathbf{v}):=\{w\in V:f(w,\mathbf{v})>0,d(\mathbf{v},w)=L(\mathbf{v})\}.
\end{align}
\end{definition}

A subset $\mathbf{v}$ associated to an ESC satisfies $L(\mathbf{v})>\a$ by definition. On the time scale $1/K\mu_K^L$, one can observe transitions from ESCs of stability degree $L$ to ESCs of stability degree at least $L$. The evolution of the population process reaches a final state, independent of the time scale, once the resident traits satisfy $L(\mathbf{v})=\infty$, i.e.\ there are no fit traits that could invade anymore.


\section{Results}
The main result of this chapter is a general description of the transition time and transition probabilities between ESCs. To describe the time to transition through the fitness valley surrounding an ESC, before triggering the $\ln K$-dynamics and attaining a new ESC, for $L(\mathbf{v})<\infty$, we define the stopping time
\begin{align}
	T^K_\text{fix}:=\inf\left\{t\geq 0:\exists\ w\in V\backslash V_\a(\mathbf{v}):\beta^K_w(t)\geq\frac{1}{\alpha}\right\}.
\end{align}
This is the first time that a new trait reaches a size of order $K^{1/\alpha}$, can thus produce neighbouring mutants within a time of order 1 and influence the subpopulations of other traits. The first part of Theorem \ref{thm:mainESC} below describes the limiting distribution of this stopping time $T^K_\text{fix}$.

For a path $\g:\mathbf{v}\to V_{mut}(\mathbf{v})$ such that $|\g|=L(\mathbf{v})$, i.e.\ a shortest path from $\mathbf{v}$ to one of the closest fit mutants, the rate at which a $w=\g_{L(\mathbf{v})}$ mutant population arises along this path $\g$ and fixates can be derived exactly as in the simplified scenario in Chapter \ref{escape}. With the same notation as in the previous chapter and the generalisation of
\begin{align}
\rho(w,\mathbf{v}):=\frac{b(w)}{b(w)+d(w)+\sum_{v\in\mathbf{v}}c(w,v)\bar{n}_v(\mathbf{v})}\label{def:rho},
\end{align}
for polymorphic resident equilibria, this rate is approximately equal to $R(\mathbf{v},\g)K\mu_K^{L(\mathbf{v})}$, where
\begin{align}
\label{eq:PathRate}
R(\mathbf{v},\g):=&\bar{n}_{\g_0}(\mathbf{v})\left(\prod_{i=1}^{\lfloor\a\rfloor}\frac{b(\g_{i-1})m(\g_{i-1},\g_i)}{|f(\g_i,\mathbf{v})|}\right)
b(\gamma_{\lfloor\a\rfloor})m(\gamma_{\lfloor\a\rfloor},\gamma_{\lfloor\a\rfloor+1})\notag\\
&\times\left(\prod_{j=\lfloor\a\rfloor+1}^{L(\mathbf{v})-1}\lambda(\rho(\g_j,\mathbf{v}))m(\g_j,\g_{j+1})\right)\frac{f(\g_{L(\mathbf{v})},\mathbf{v})}{b(\g_{L(\mathbf{v})})}\notag\\
=&\bar{n}_{\g_0}(\mathbf{v})b(\g_0)m(\g_0,\g_1)\left(\prod_{i=1}^{L(\mathbf{v})-1}\frac{b(\g_{i})m(\g_{i},\g_{i+1})}{|f(\g_i,\mathbf{v})|}\right)
\frac{f(\g_{L(\mathbf{v})},\mathbf{v})}{b(\g_{L(\mathbf{v})})}
\end{align}
Note that, in contrast to the previous chapter, this expression has additional factors of the mutation kernel $m$ since the scenario of this chapter allows for multiple directions of mutations. As a reminder, the first line is the rate at which the first trait in $\g$ outside of $V_\a(\mathbf{v})$ arises, which is related to the equilibrium size of trait $\g_{\lfloor\a\rfloor}$. The first factor in the second line is the probability of producing consecutive mutants during subcritical excursions, and the last factor is the fixation probability of trait $w=\g_{L(\mathbf{v})}$.

Since mutations along different paths correspond to (thinnings of) independent Poisson processes, we can add up the corresponding rates when considering the rate of mutations along several paths. The total rate at which a mutant population of trait $w\in V_{mut}(\mathbf{v})$ arises and fixates collects all shortest paths that end in $w$ and is approximately equal to $R(\mathbf{v},w)\mu_K^{L(\mathbf{v})}$, where
\begin{align}\label{eq:TraitRateFix}
R(\mathbf{v},w):=\sum_{\substack{\gamma:\mathbf{v}\to w\\|\g|=L(\mathbf{v})}}R(\mathbf{v},\g).
\end{align}
Note that only shortest paths of length $L(\mathbf{v})$ need to be considered since the rates along other paths are of order at most $K\mu_K^{L(\mathbf{v})+1}$ and hence negligible in comparison to $K\mu_K^{L(\mathbf{v})}$.

Finally, the total rate at which any mutant population of a trait in $V_{mut}(\mathbf{v})$ arises and fixates, i.e.\ the rate at which the population exits the ESC associated with $\mathbf{v}$, is approximately equal to $R(\mathbf{v})\mu_K^{L(\mathbf{v})}$, where
\begin{align}\label{eq:ExitRate}
R(\mathbf{v}):=\sum_{w\in V_{mut}(\mathbf{v})}R(\mathbf{v},w).
\end{align}
Again by properties of independent Poisson processes, the probability that this population is of trait $w\in V_{mut}(\mathbf{v})$ is proportional to the rate $R(\mathbf{v},w)$.

With these heuristics, we now state the main result of this chapter.
\begin{theorem}\label{thm:mainESC}
	Assume that $\mathbf{v}\subset V$ and $(\beta^K(0))_{K\geq0}$ are an asymptotic ESC.
	Then there exist constants $\eps_0>0$ and $0<c<\infty$ such that, for all $0<\eps<\eps_0$, there exist exponential random variables $E^K_+(\eps)$ and $E^K_-(\eps)$ with parameters $R(\mathbf{v})(1+c\eps)$ and $R(\mathbf{v})(1-c\eps)$, such that
	\begin{align}
	\liminf_{K\to\infty}\P(E^K_-(\eps)\leq T^K_\text{fix} K\mu_K^{L(\mathbf{v})}\leq E^K_+(\eps))\geq 1-c\eps.
	\end{align}
	Moreover, for all $w\in V$, the probability of $w$ being the trait to trigger $T^K_\text{fix}$ is
	\begin{align}\label{FixatingTrait}
	\lim_{K\to\infty}\P\left(\beta^K_w(T^K_\text{fix})=1/\alpha\right)=\begin{cases}R(\mathbf{v},w)/R(\mathbf{v})&\text{if }w\in V_{mut}(\mathbf{v}),\\0&\text{else.}\end{cases}
	\end{align}
\end{theorem}

Once some $w\in V_{mut}(\mathbf{v})$ has reached $\beta^K_w\geq1/\alpha$, the $\ln K$-dynamics evolve as described in Chapter \ref{chapter7}, initiated with $\beta^K_w=1/\alpha$ and $\beta^K_u=(1-d(\mathbf{v},u)/\alpha)_+$, for $u\in V\backslash w$. These dynamics are deterministic, and in case they do not terminate early and if they lead to a new ESC, we denote the associated set of resident traits by $\mathbf{v}_\text{ESC} (\mathbf{v},w)$.

Under the assumption that all traits $w\in V_{mut}(\mathbf{v})$ lead to asymptotic ESCs $\mathbf{v}_\text{ESC} (\mathbf{v},w)$, we define the stopping time at which one of these asymptotic ESCs is obtained by
\begin{align}
	\label{eq:T_ESC}
	T^K_\text{ESC} :=\inf\Bigg\{&t\geq T^K_\text{fix}:\exists\ w\in V_{mut}(\mathbf{v}):\nonumber\\
		&\forall u\in V_\alpha(\mathbf{v}_\text{ESC} (\mathbf{v},w)):\left|\beta^K_u(t)-\left(1-\frac{d(\mathbf{v}_\text{ESC} (\mathbf{v},w),u)}{\alpha}\right)\right|<\eps_K,\notag\\
		&\forall u\notin V_\alpha(\mathbf{v}_\text{ESC} (\mathbf{v},w)):\beta^K_u(t)=0\Bigg\},
\end{align}
	where we pick $\eps_K=C/\ln K$ for a large enough $0<C<\infty$. Then this definition is precisely in line with the definition of an asymptotic ESC.

Since the time $T^K_\text{ESC} -T^K_\text{fix}$ is of order $\ln K$, the asymptotics for $T^K_\text{fix}$ translate to $T^K_\text{ESC}$. Moreover, the transition probabilities from one ESC to another can be expressed in terms of the probabilities of traits $w\in V_{mut}(\mathbf{v})$ fixating in the population. For $\mathbf{w}\subset V$ we define
\begin{align}\label{eq:TransitionRate}
	p(\mathbf{v},\mathbf{w}):=\sum_{\substack{w\in V_{mut}(\mathbf{v}):\\\mathbf{v}_\text{ESC} (\mathbf{v},w)=\mathbf{w}}}\frac{R(\mathbf{v},w)}{R(\mathbf{v})}.
\end{align}

We can now state the result on transitions between ESCs as a direct corollary of Theorem \ref{thm:mainESC}.
\begin{corollary}\label{cor:mainESC}
Suppose the same assumptions as in Theorem \ref{thm:mainESC} are satisfied. 
Moreover, assume that, for every $w\in V_{mut}(\mathbf{v})$, the algorithmic description of the
 $\ln K$-dynamics in Chapter \ref{chapter7}, initiated with
	\begin{align}\label{eq:afterFix}
	\beta_u(0)=\begin{cases}\frac{1}{\alpha}&\text{if }u=w\\
	\left(1-\frac{d(\mathbf{v},u)}{\alpha}\right)_+ &\text{else}\end{cases},
    \end{align}
    does not stop early due to one of its termination criteria and reaches an ESC associated 
    to some traits $\mathbf{v}_\text{ESC} (\mathbf{v},w)$ after finitely many steps. Then, $T^K_\text{ESC} -T^K_\text{fix}\in O(\ln K)$ and therefore, with the same constants $\eps_0$ and $c$ and with the same random variables $E^K_+(\eps)$ and $E^K_-(\eps)$ as in Theorem \ref{thm:mainESC}, 
    \begin{align}
		\liminf_{K\to\infty}\P(E^K_-(\eps)\leq T^K_\text{ESC} K\mu_K^{L(\mathbf{v})}\leq E^K_+(\eps))\geq 1-c\eps.
	\end{align}
	Moreover, for all $\mathbf{w}\subset V$,
	\begin{align}
		\lim_{K\to\infty}\P(\{u\in V:\beta^K_u(T^K_\text{ESC} )> 1-\eps_K\}=\mathbf{w})=p(\mathbf{v},\mathbf{w}).
	\end{align}
\end{corollary}

Based on these results, we can define the so-called \textit{metastability graph}.
\begin{definition}
	As vertices for the general metastability graph \linebreak $\cG_\text{ESC} =\left(\mathcal{V}_\text{ESC} ,\mathcal{E}_\text{ESC} \right)$ we take all sets of resident traits that correspond to an ESC, i.e.\ that have stability degree strictly bigger than $\alpha$, and edges represent possible transitions to other ESCs. More precisely,
	\begin{align}
	\mathcal{V}_\text{ESC} :=&\{\mathbf{v}\subseteq V : L(\mathbf{v})>\alpha\},\\
	\mathcal{E}_\text{ESC} :=&\{(\mathbf{v},\mathbf{w}):\exists w\in V_{mut}(\mathbf{v}) \text{\ s.t.\ } \mathbf{w}=\mathbf{v}_\text{ESC} (\mathbf{v},w)\}.
	\end{align}
\end{definition}

Note that this graph summarises possible transitions between ESCs on various time scales. For any fixed time scale $1/K\mu_K^L$, transitions between ESCs $\mathbf{v}$ and $\mathbf{w}$ with $L(\mathbf{v}>L$ do not occur, while transitions with $L(\mathbf{v}<L$ happen instantaneously. As a consequence, in the latter case transitions out of an ESC with $L(\mathbf{v}=L$ might (visibly) not lead to an ESC corresponding to $\mathbf{v}_\text{ESC}(\mathbf{v},w)$ for some $w\in V_{mut}(\mathbf{v})$ but rather skip ahead through some transitions along shorter valleys until an ESC of stability degree of at least $L$ is reached. This means that, for any fixed time scale $1/K\mu_K^L$, edges in the metastability graph corresponding to transitions out of ESCs of lower stability degree collapse to form the so-called $L$-scale graph $\cG^L$. Again, on this graph, transition rates of multiple paths in $\cG_\text{ESC}$ need to be added up.

An extensive example that demonstrates evolutionary behaviour on varying time scales is provided in the next section.

\begin{itemize}
\item Add more heuristics of the proof/references to techniques from previous chapters?
\end{itemize}


\section{Example}
To simplify, we consider a scenario of constant competition, i.e.\ $c(v,w)\equiv c>0$, for all $v,w\in V$, as in Section \ref{section41}. As a consequence, setting $r(v)=b(v)-d(v)$ as the growth rate without competition, the invasion fitness can be calculated as the difference $f(w,v)=r(w)-r(v)$ and the potential for invasion comes down to whether a mutant has a higher rate $r$ than the current resident trait. We introduce the short notation
\be
F_v(w)=\frac{f(w,v)}{b(w)}
\ee
for the fixation probability of trait $w$ in a resident population of trait $v$.

Figure \ref{fig:GeneralTraitGraph} depicts the directed trait graph and individual growth rates 
$r$ of a specific evolutionary scenario. Fix some $1<\a<2$, an initial condition of 
$N_0^K(0)\approx\bar n(0)K$ and $N_v^K(0)=0$, for all $v\neq0$, and assume that $m(v,
\cdot)$ is uniform between neighbouring traits. In the following, we run through the various 
time scales that can be considered and the evolutionary trajectories that can be observed on 
those time scales.
\begin{figure}[h]
\begin{center}
\includegraphics[width=.9\textwidth]{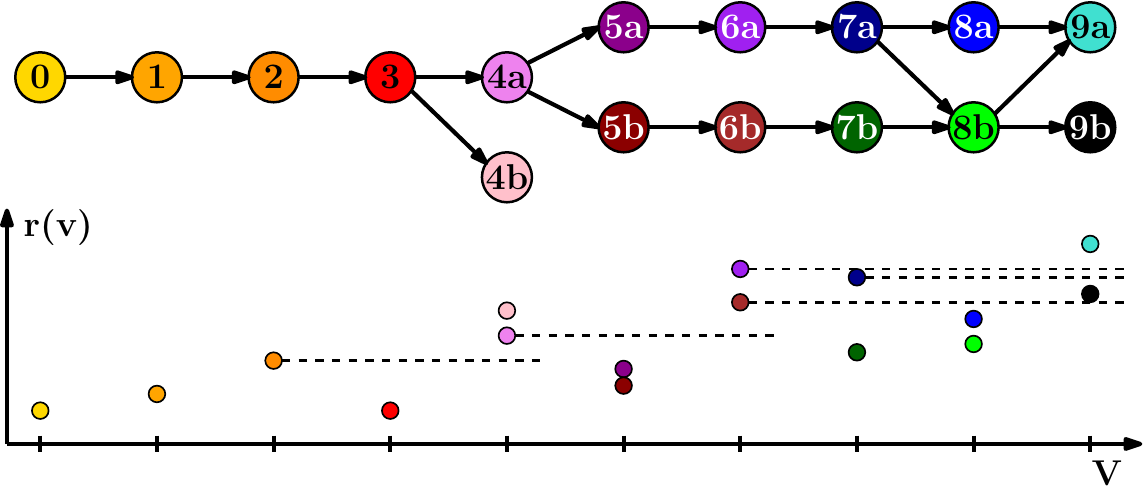}
\end{center}
\caption{...}
\label{fig:GeneralTraitGraph}
\end{figure}

Note that $1/K\mu_K\ll 1$ here, so we do not need to consider this time scale, as mutants of 
trait 1 arrive but no significant growth can be observed.

To start, we hence consider the time scale $\ln K$. Since $\alpha>1$ and $r(0)<r(1)<r(2)$, we 
have $f(1,0),f(2,0),f(2,1)>0$ and the algorithmic construction in Chapter \ref{chapter7} yields 
a trajectory of either
\be
0\to1\to2\ \text{ or }\ 0\to2
\ee
of resident traits, depending on whether $(1-1/\a)/f(1,0)<1/f(2,0)$ or not (i.e.\ whether trait 2 
mutants surpass the trait 1 mutant population before it can invade). Once trait 2 is resident, 
no further evolution can be observed on this time scale since $r(3)<r(2)$ and the next fit 
mutant is at a distance of $4-2>\alpha$. Hence, $\mathbf{v}=\{2\}$ with $\b_2=1$, $
\b_3=1-1/\a$ and all other $\b_v=0$ is an ESC.

Similarly to this, since all fitter traits are at least $2>\alpha$ mutation steps away, $\{4a\}$, $
\{4b\}$, $\{6a\}$, $\{6b\}$, $\{7a\}$, $\{9a\}$, and $\{9b\}$ - together with the appropriate $
\b_v$s, make up the ESCs of this trait graph. Note that, since no two traits apart from 0 and 
2, which have a fitter neighbour 1, have the same growth rate $r$, there are no coexistence 
equilibria that make up an ESC in this case. The resulting metastability graph $
\cG_\text{ESC}$ with all possible transitions is shown in Figure \ref{fig:GeneralESCGraph}.

The next time scale to consider is $1/K\mu_K^2\gg\ln K$. Here, all $\ln K$ dynamics occur 
instantaneously, and the dominating phenomena are transitions through fitness valleys of 
width 2, i.e.\ out of ESCs with stability degree $L(\mathbf{v})=2$. Starting with resident trait 0, 
the population therefore immediately transitions to trait 2 at equilibrium. From there, two 
possible transitions can occur, either to $4a$ or $4b$, at rates
\be
R(\{2\},4a)=\frac{\bar n(2)b(2)b(3)f(4a,2)}{2|f(3,2)|b(4a)}\ \text{ or }\ R(\{2\},4b)=\frac{\bar n(2)b(2)b(3)f(4b,2)}{2|f(3,2)|b(4b)},
\ee
respectively, where the factor 2 in the denominator comes from $m(3,4a)=m(3,4b)=1/2$. Note 
that in this case the fit mutant trait and resulting ESC traits are the same, i.e.\
 $\mathbf{v}_\text{ESC}(w)=w$. The transition probabilities are hence 
\be
p(\{2\},\{4a\})=\frac{R(\{2\},4a)}{R(\{2\})}=\frac{F_2(4a)}{F_2(4a)+F_2(4b)}\ 
\text{ and }\ p(\{2\},\{4b\})=\frac{F_2(4b)}{F_2(4a)+F_2(4b)}.
\ee
In the first case, there are again two options of transitioning to the equilibria of $6a$ 
or $6b$ respectively, with rates
\be
R(\{4a\},6a)=\frac{\bar n(4a)b(4a)b(5a)f(6a,4a)}{2|f(5a,4a)|b(6a)}\ \text{ and }\ 
R(\{4a\},6b)=\frac{\bar n(4a)b(4a)b(5b)f(6b,4a)}{2|f(5b,4a)|b(6b)}
\ee
and probabilities
\begin{align}
&p(\{4a\},\{6a\})=\frac{F_{4a}(6a)/|F_{4a}(5a)|}{F_{4a}(6a)/|F_{4a}(5a)|+F_{4a}(6b)/|F_{4a}
(5b)|}\notag\\
&\text{ and }\ p(\{4a\},\{6b\})=\frac{F_{4a}(6b)/|F_{4a}(5b)|}{F_{4a}(6a)/|F_{4a}(5a)|+F_{4a}
(6b)/|F_{4a}(5b)|}.
\end{align}
Note that the probabilities are slightly more complicated than the previous ones, since here branching in the trait graph occurs in the first mutation step instead of the second, and hence fewer common terms cancel.
The evolutionary trajectory stops once $4b$, $6a$ or $6b$ are reached since all of these $ESC$ have a stability degree larger than 2. Notably, while not connected to this trajectory, $\{7a\}$, $\{9a\}$, and $\{9b\}$ are ESCs of stability degree at least 2 as well, and a transition from $7a$ to $9a$ could be observed on this timescale. While there is also a path on the trait graph from $7a$ to $9b$, this transition will not be made since $r(7a)>r(9b)$ and hence $p(\{7a\},\{9a\})=1$. The possible transitions and stable states for the time scale $1/K\mu_K^2$ are summarised in graph $\cG^2$ in Figure \ref{fig:GeneralESCGraph}.

The final time scale to consider is $1/K\mu_K^3$. Again, all transitions on shorter time scales are instantaneous and hence the population immediately jumps to the equilibrium state of $4b$ with probability $p(\{2\},\{4b\})$, $6a$ with probability $p(\{2\},\{4a\})\cdot p(\{4a\},\{6a\})$, or $6b$ with probability $p(\{2\},\{4a\})\cdot p(\{4a\},\{6b\})$. In the first case, no further evolution can be observed. In the other two cases, the population can still transition to either $9a$ or $9b$. The corresponding transition rates, taking into account that there are two possible paths from $6a$ to $9a$, through $8a$ or $8b$, are
\begin{align}
R(\{6a\},9a)&=\frac{\bar n(6a)b(6a)b(7a)b(8a)f(9a,6a)}{4|f(7a,6a)||f(8a,6a)|b(9a)}
+\frac{\bar n(6a)b(6a)b(7a)b(8b)f(9a,6a)}{4|f(7a,6a)||f(8b,6a)|b(9a)},\\
R(\{6b\},9a)&=\frac{\bar n(6b)b(6b)b(7b)b(8b)f(9a,6b)}{2|f(7b,6b)||f(8b,6b)|b(9a)},\\
R(\{6b\},9b)&=\frac{\bar n(6b)b(6b)b(7b)b(8b)f(9b,6b)}{2|f(7b,6b)||f(8b,6b)|b(9b)}.
\end{align}
The corresponding transition probabilities are $p(\{6a\},\{9a\})=1$ and $p(\{6b\},\{9a\})=R(\{6b\},9a)/R(\{6b\})$, $p(\{6b\},\{9b\})=R(\{6b\},9b)/R(\{6b\})$, where $R(\{6b\})=R(\{6b\},9a)+R(\{6b\},9b)$. The possible transitions and stable states for the time scale $1/K\mu_K^3$ are summarised in graph $\cG^3$ in Figure \ref{fig:GeneralESCGraph}.
\begin{figure}[h]
\begin{center}
\includegraphics[width=.6\textwidth]{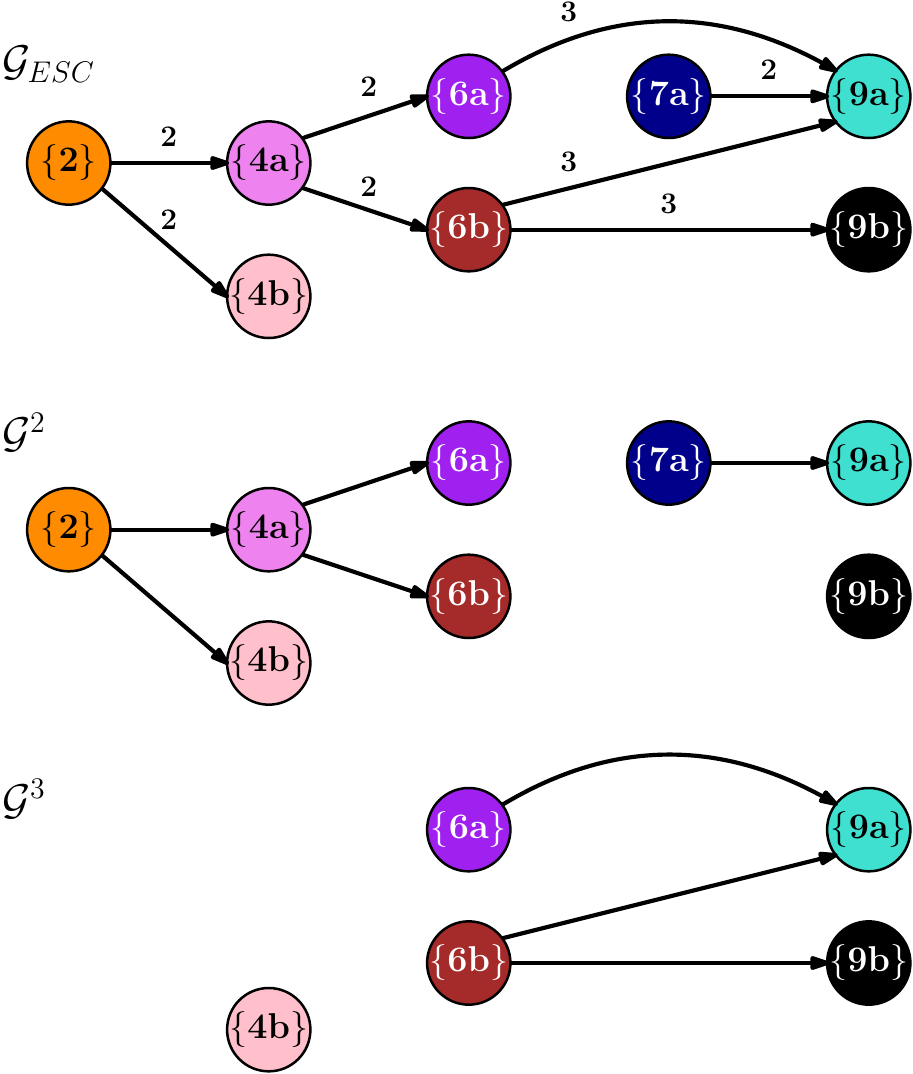}
\end{center}
\caption{Metastability graph $\cG_\text{ESC}$ for the trait graph in Figure \ref{fig:GeneralTraitGraph} and $1<\alpha<2$. Edges correspond to possible metastable transitions and are annotated with the widths of the corresponding fitness valleys, i.e.\ stability degrees of the traits the arrows are rooted in. Graphs $\cG^2$ and $\cG^3$ visualise stable states and transitions observed on time scales $1/K\mu_K^2$ and $1/K\mu_K^3$, respectively.}
\label{fig:GeneralESCGraph}
\end{figure}

\part{Variations on the basic models}
\chapter{Diploid models and genetic diversity}
\def\ff{b}

\begin{chapquote}
{Charles Darwin, \emph{The Origin of Species}}
{The laws governing inheritance are quite unknown; no one can say why the same particularity in different individuals of the same species, and in individuals of different species, is sometimes inherited and sometimes not so; why the child often reverts in certain characters to its grandfather or grandmother or other much more remote ancestor; why a peculiarity is often transmitted form one sex to both sexes, or to one sex alone, more commonly but not exclusively to the like sex.}
\end{chapquote}

\index{diploid!model}\index{diploid!reproduction}
In the previous chapters, we have looked at population models with haploid reproduction schemes, meaning that a single parent gives birth to clonal children (unless a mutation occurs).  In contrast, in \emph{diploidal}
reproduction, children have two parents who contribute genetic material to the children's genotype.
In fact, the bulk of work in the context of adaptive dynamics concerns haploid models. This is one of the
points ofter criticised by biologists who assert that clonal reproduction is insufficient to explain evolution 
and that sexual reproduction, or more generally, the exchange of genetic information 
between individuals, is of crucial importaÏnce to explain the observed genetic variability and the complexity
of populations.
In fact, we have seen so far that in haploid models monomorphism is somewhat the rule and 
genetic diversity is produced only in  exceptional situations, at least under the assumption of very small 
mutation rates. 
 In this chapter, which is based on the papers \cite{BovNeu16} and \cite{BovCoqNeu18},
we explain a way in which the introduction of diploid reproduction schemes can indeed favour diversity.
Mendelian diploid models have been studied for over a century in the context of \emph{population genetics} (see, e.g., \cite{yule06}, \cite{fisher18}, 
 \cite{wright31}, \cite{haldane24a}, and \cite{haldane24b}). Textbook expositions of population genetics are given in,  e.g.,
 \cite{crowkimura}, \cite{nagylaki92}, \cite{ewens04}, and \cite{buerger2000}.  
 Diploid models have been considered in adaptive dynamics already by 
 \cite{KG99} and \cite{M12}. 
 In the context of individual-based models, the first time (to our knowledge) that diploid models were studied, is the
paper by Collet, Méléard, and Metz \cite{CMM13}, where the TSS is derived in a Mendelian diploid model under certain assumptions (that we will discuss below). Some more recent papers are 
\cite{coron}, \cite{coron2}, and \cite{coron3}.

\section{The Mendelian diploid model}

We will consider only a very simple situation with a single locus. \index{heterozygotous}
Each individual $i$ is characterised by a single \emph{heterozygotous} gene that carries two alleles, $u_1^iu_2^i$, from some allele space $\mathcal{U}\in\R$.  These two 
alleles define the genotype of individual $i$, which in turn defines its phenotype, $\phi(u_1^iu_2^i)$, through a function 
$\phi:\mathcal{U}^2\rightarrow\R$. We suppress parental effects, thus $\phi(u_1^iu_2^i)=\phi(u_2^iu_1^i)$. 
We consider a hermaphroditic population where any two individuals can be the 
parents of a 
child. Reproductive events are now more complicated than in the haploid case, 
as they involve 
the selection of two individuals from the population with genotypes $(u_1,u_2)$ 
and $(v_1,v_2)$
and the assignment of the children's genotype. 

The selection of the parents is driven by a parameter $\ff_{u_1u_2}$, called 
\emph{fertility}. For the distribution \index{fertility}
of the parental genotypes on the children, we consider the simple case of equal 
probability of all
mixtures, i.e. the four possible cases $(u_1,v_1), (u_1,v_2), (u_2,v_1), (u_2,v_2)$ have equal probability. 

 The population at time $t$ is represented as a point measure $\nu_t$ on 
 $\mathcal{U}^2$.
%
As before we denote by $\MM^K$ the space of all finite point measures of this type.

The following demographic parameters depend all on the phenotype, but we suppress this from the notation. 
Let us define
\vspace{-6pt}
\begin{center}
	\begin{tabular}{ll}
		  $\ff_{u_1u_2}\in\R_+$ 	& the rate of birth (fertility) of an individual with genotype $u_1u_2$.\\
 		 $d_{u_1u_2}\in\R_+$ 	& the rate of natural death of an individual with genotype $u_1u_2$.\\
		  $\frac{c_{u_1u_2,v_1v_2}}{K}\in\R_+$ & \parbox[t]{1.0\textwidth}{the competition effect felt by an individual with genotype $u_1u_2$\\
		  						 from an individual with genotype $v_1v_2$.}\\
		$m(u,dh)	$				&\parbox[t]{1.0\textwidth}{mutation law of a mutant allelic trait $u+h\in\mathcal{U}$, born from an\\individual with allelic trait $u$.} 
 \end{tabular}
 \end{center}		  
 \vspace{3pt}
%
%
The generator of the process is defined as in \cite{CMM13}.
First we define, for the genotypes $u_1u_2,v_1v_2$ and a point measure $\nu$, the Mendelian reproduction operator:
\bea
&&(A_{u_1u_2,v_1v_2}F)(\nu)\\
&&=\frac{1}{4}\left[F\left(\nu+\frac{\d_{u_1v_1}}{K}\right)+F\left(\nu+\frac{\d_{u_1v_2}}{K}\right)+F\left(\nu+\frac{\d_{u_2v_1}}{K}\right)+F\left(\nu+\frac{\d_{u_2v_2}}{K}\right)\right]-F(\nu),\nonumber
\eea
and the Mendelian reproduction-cum-mutation operator:
\bea
(M_{u_1u_2,v_1v_2}F)(\nu)
&=&\frac{1}{8}\int\Bigg[\left(F\left(\nu+\frac{\d_{hv_1}}{K}\right)+F\left(\nu+\frac{\d_{hv_2}}{K}\right)\right)
m(u_1,dh)\nonumber\\
&+&\left(F\left(\nu+\frac{\d_{hv_1}}{K}\right)+F\left(\nu+\frac{\d_{hv_2}}{K}\right)\right)
m(u_2,dh)\nonumber\\
&+&\left(F\left(\nu+\frac{\d_{u_1h}}{K}\right)+F\left(\nu+\frac{\d_{u_2h}}{K}\right)\right)
m(v_1,dh)\nonumber\\
&+&\left(F\left(\nu+\frac{\d_{u_1h}}{K}\right)+F\left(\nu+\frac{\d_{u_2h}}{K}\right)\right)
m(v_2,dh)\Bigg]
-F(\nu).
\eea
The process $(\nu_t)_{t\geq0}$ is then 
 a $\mathcal{M}^K$-valued Markov process with  generator $L^K$, given for any bounded measurable
  function $F:\mathcal{M}^K\rightarrow\R$ and $\nu\in\mathcal{M}^K$ by:
\bea
&&(L^KF)(\nu)\nonumber\\
&&=\int_{\mathcal U^2}\left(d_{u_1u_2}+\int_{\mathcal U^2}c_{u_1u_2,v_1v_2}\nu(d(v_1v_2))\right)\left(F\left(\nu-\frac{\d_{u_1u_2}}{K}\right)-F(\nu)\right))
K\nu(d(u_1u_2))\nonumber\\
&&+\int_{\mathcal U^2}(1-\mu_K)\ff_{u_1u_2}\left(\int_{\mathcal U^2}\frac{\ff_{v_1v_2}}{\langle\nu,b\rangle}(A_{u_1u_2,v_1v_2}F)(\nu)\nu(d(v_1v_2))\right)K\nu(d(u_1u_2))\nonumber\\
&&+\int_{\mathcal U^2}\mu_K\ff_{u_1u_2}\left(\int_{\mathcal U^2}\frac{\ff_{v_1v_2}}{\langle\nu,b\rangle}(M_{u_1u_2,v_1v_2}F)(\nu)\nu(d(v_1v_2))\right)K\nu(d(u_1u_2)).
\eea
The first non-linear term describes the competition between individuals, just as in the haploid case.
The second and last linear terms describe the birth without and with mutation. There, 
$\ff_{u_1u_2}\frac{\ff_{v_1v_2}}{K\langle\nu,b\rangle}$ is the reproduction rate of an individual with genotype $u_1u_2$ with 
an individual with genotype $v_1v_2$. The interpretation of this term is that the individual $(u_1u_2)$ 
reproduces with rate $b_{u_1u_2}$ and then chooses a mate from the population  that has type $(v_1v_2)$  
with probability $\frac{\ff_{v_1v_2}}{K\langle\nu,b\rangle}$.

Measurability and boundedness assumptions on the parameters, as in the haploid model 
ensure the existence of the process as well as the convergence to a deterministic limit as $K\uparrow\infty$.

Champagnat, Méléard, and Metz \cite{CMM13} have shown that in the limit of large population and rare mutations, 
and under a co-dominance assumption of alleles,  the suitably time-rescaled process converges to the TSS model of adaptive dynamics, 
essentially as in the haploid case. Their assumptions exclude the degenerate, but biologically not excluded 
possibility that the fitter allele is dominant, i.e. that the phenotype of the heterozygotic and homozygotic genotype containing this allele are the same. 
In the papers
\cite{BovNeu16} and \cite{BovCoqNeu18}, precisely this case was studied and it was shown that this leads to interesting new features.
In particular, this setting favours the emergence of coexisting traits.
 
We consider a limited setting where the allele space is the finite set $\UU=\{a,A,B\}$. 
We will see that under rather natural conditions and without any fine-tuning, an initial monomorphic $(aa)$-
type population can evolve under Mendelian reproduction to a bi-morphic population containing both the
$(aa)$ and the $(BB)$ types through mutations $a\to A$ and $A\to B$. 

Let us first look at the first invasion step when an $A$ allele appears in a monomorphic 
$aa$-population. 

We write as before
set $n_{uv}(t)\equiv \langle \nu_t^K,\1_{uv}\rangle$.
The \emph{equilibrium size} of a monomorphic $uu$-population is the non-trivial 
 fixed point of  a
1-dim Lotka-Volterra 
equation and is given by
\be
\label{equi}
\bar  n(u)=\frac{\ff_{uu}-d_{uu}}{c_{uu,uu}}.
\ee
The probability that a mutant allele can fixate is governed by the growth rate of
a single individual $uv$ in an equilibrium $uu$ population. 
To compute this, we need to compute the transition rate in the occupation number representation of our process. We denote the number of individuals of type $uv$ by $X_{uv}$. 
We do this just for the mutation-free system with the two alleles $a,A$. 
Consider the rate of increase by one of the $aa$-population, $r(X_{aa}\to X_{aa}+1)$.
There are four ways for this to happen: an $aa$-individual that chooses to give birth selects an $aa$-individual to mate. This happens at a rate 
\be
X_{aa} b_{aa} b_{aa} X_{aa}/(b_{aa} X_{aa}+b_{aA}X_{aA}+b_{AA}X_{AA}),
\ee
 or,
the $aa$-individual choses an $aA$-mate and then with probability $1/2$ gives birth to a $aa$ child. 
This happens at rate  
\be\frac 12X_{aa} b_{aa} b_{aA} X_{aA}/(b_{aa} X_{aa}+b_{aA}X_{aA}+b_{AA}X_{AA}).
\ee
Alternatively, an  $aA$-individual chose to give birth and selects an $aa$ mate and with probability 
$1/2$ produces an $aa$ child, or it selects an $aA$-mate and produces an $aa$ child with probability $1/4$.
This happens with rates  
\be\frac 12X_{aA} b_{aA} b_{aa} X_{aa}/(b_{aa} X_{aa}+b_{aA}X_{aA}+b_{AA}X_{AA}),
\ee
respectively  
\be
\frac 14X_{aA} b_{aA} b_{aA} X_{aA}/(b_{aa} X_{aa}+b_{aA}X_{aA}+b_{AA}X_{AA}).
\ee
Consequently, 
\be
r(X_{aa}\to X_{aa}+1)
=\frac{\left(\ff_{aa}X_{aa}+\frac{1}{2}\ff_{aA}X_{aA}\right)^2}{b_{aa}X_{aa}+b_{aA}X_{aA}+b_{AA}X_{AA}}.
\ee
Similar reasoning shows that 
\be
r(X_{aA}\to X_{aA}+1)
=\frac{2\left(\ff_{aa}X_{aa}+\frac{1}{2}\ff_{aA}X_{aA}\right)\left(\ff_{AA} X_{AA}+\frac{1}{2}\ff_{aA}X_{aA}\right)}{b_{aa}X_{aa}+b_{aA}X_{aA}+b_{AA}X_{AA}},
\ee
and 
\be\Eq(birth-rates)
r(X_{AA}\to X_{AA}+1)=
\frac{\left(\ff_{AA}X_{AA}+\frac{1}{2}\ff_{aA}X_{aA}\right)^2}{b_{aa}X_{aa}+b_{aA}X_{aA}+b_{AA}X_{AA}}.
\ee
The corresponding death rates are simpler and just as in the haploid case:
\bea\Eq(death-rates)
r(X_{aa}\to X_{aa}-1)
&=&X_{aa}\left(d_{aa}+c_{aa,aa}n_{aa}+c_{aa,aA}n_{aA}+c_{aa,AA} n_{AA}\right),\nonumber\\
r(X_{aA}\to X_{aA}-1)
&=&X_{aA}\left(d_{aA}+c_{aA,aa}n_{aa}+c_{aA,aA}n_{aA}+c_{aA,AA} n_{AA}\right),\nonumber\\
r(X_{AA}\to X_{AA}-1)
&=&X_{AA}\left(d_{AA}+c_{AA,aa}n_{aa}+c_{AA,aA}n_{aA}+c_{AA,AA} n_{AA}\right).
\eea
We see that, if $X_{aA}=1$ and $X_{aa}=K \bar n(a)$, then 
\be
r(X_{aA}\to X_{aA}+1)= \frac {b_{aA}b_{aa} K \bar n(a)+\frac 12 b_{aA}^2}{Kb_{aa} \bar n(a)+b_{aA}}
=b_{aA} +O(1/K),
\ee
while 
\be
r(X_{aA}\to X_{aA}-1)=d_{aA}+c_{aA,aa} \bar n(a) +O(1/K).
\ee
This implies that we can define  invasion fitness as in the haploid model as
\be
f_{uv,uu}=\ff_{uv}-d_{uv}-c_{uv,uu}\bar n(u)
\ee

Note that before the number of individuals with type $aA$ has reached a level $O(K)$, the
probability that two $aA$ individuals mate is $o(K)/K$ and so no $AA$ individuals are 
born. However, once the $aA$  population has reached a threshold $\e K$, $AA$ 
individuals arrive, and their properties are important for the ultimate fate of the 
population. The interesting case for us is when the phenotypes of $aA$ and $AA$ 
coincide. We will see that in this case (in the absence of further mutations), 
in the deterministic limit the approach to the stable fix-point $(0,0,\bar n(A))$ will be 
only algebraic rather then exponential, and, as a consequence, the
$a$ allele survives in the stochastic model for a time of order $K^{1/2}$. 
 


To make the analysis more transparent, we make some simplifying assumptions.

 We assume that the  dominant $A$-allele defines the phenotype of an individual, i.e. 
$AA$ and $Aa$ individuals have the same phenotype.
In particular,  the fertility and the natural death rates are the 
same for $aA$- and $AA$-individuals. We also assume that the competition rates are the same 
for all three genotypes.
To sum up, we make the following Assumptions \textbf{(B)} on the rates: 
\begin{description}
\item[\textbf{(B1)}] $\ff_{aa}= \ff_{aA}= \ff_{AA}\equiv\ff$,\vspace{3pt}	 
\item[\textbf{(B2)}] $d_{AA}= d_{aA}\equiv d$\hspace{9pt} \emph{but} \hspace{9pt}$d_{aa}= d+\D$, $(b-d)>\D>0$,\vspace{3pt}
\item[\textbf{(B3)}]$c_{u_1u_2,v_1v_2}= c,$ \hspace{9pt}  $\forall u_1u_2,v_1v_2\in\{aa,aA,AA\}$.
\end{description}

One checks that, under Assumptions  \textbf{B},
\bea
\Eq(init)
f_{aA,aa}&=&f_{AA,aa}=\D,\\
f_{aa,aA}&=&f_{aa,AA}=-\D.
\eea
Therefore, the $aA$-individuals are as fit as the $AA$-individuals and both are fitter than the $aa$-individuals.
%
Under Assumption \textbf{B}, the birth and death rates simplify to 
\bea
r(X_{aa}\to X_{aa}+1)
&=&\frac{\ff\left(X_{aa}+\frac{1}{2}X_{aA}\right)^2}{X_{aa}+X_{aA}+X_{AA}}\\
r(X_{aA}\to X_{aA}+1)
&=&\frac{2\ff\left(X_{aa}+\frac{1}{2}X_{aA}\right)\left( X_{AA}+\frac{1}{2}X_{aA}\right)}{X_{aa}+X_{aA}+X_{AA}},
\\
r(X_{AA}\to X_{AA}+1)
&=&
\frac{\ff\left(X_{AA}+\frac{1}{2}X_{aA}\right)^2}{X_{aa}+X_{aA}+X_{AA}}.
\eea
The corresponding death rates are
\bea
r(X_{aa}\to X_{aa}-1)
&=&X_{aa}\left(d+\D+c\left(n_{aa}+n_{aA}+ n_{AA}\right)\right),\nonumber\\
r(X_{aA}\to X_{aA}-1)
&=&X_{aA}\left(d+c\left(n_{aa}+n_{aA}+n_{AA}\right)\right),\nonumber\\
r(X_{AA}\to X_{AA}-1)
&=&X_{AA}\left(d+c\left(n_{aa}+n_{aA}+ n_{AA}\right)\right).
\eea

In the sequel, the \emph{sum process}, $\Sigma(t)$,  defined by
\be
\label{sumprocess}
\Sigma(t)=X_{aa}(t)+X_{aA}(t)+X_{AA}(t),
\ee
plays an important role. 
A simple calculation shows that the sum process jumps   with rates
\bea
r(\S\to\S+1)&=&\ff \Sigma,\nonumber\\
r(\S\to\S-1)&=&d\Sigma + c\S^2/K+\D X_{aa}..
\eea

\section{Prolonged survival of the unfit}
\label{sectionmain}
We denote by $(\t_n)_{n\in\N}$, the sequence of times when a mutation occurs in the population. 
We assume $\t_0=0$
and make the Assumption \textbf{(C)} on the mutation rate $\mu_K$:
\be
\label{mu}
K^{-3/2+\a}\ll \mu_K\ll \frac 1{K\ln(K)}.
\ee
The lower bound on the mutation rate is made to allow that a new mutation can occur before the $a$-allele is extinct.

As in the haploid system, the system remains close to an equilibrium before a mutation happens. 
\begin{proposition}[Proposition D.2 in \cite{CMM13}]
\label{propD}
Let $\supp(\nu_0^K)=\{aa\}$ and let $\t_1$ denote the first mutation time. For any sufficiently small $\d>0$, if 
$\langle\nu_0^K,\1_{aa}\rangle$ belongs to the $\d/2$-neighbourhood of $\bar n(a)$ then the time of exit of 
$\langle\nu_t^K,\1_{aa}\rangle$ from the $\d$-neighbourhood of $\bar n(a)$ is bigger than $\eee^{VK}\land\t_1$, for $V>0$, 
with probability converging to 1.
Moreover, there exists a constant $M$, such that, for any sufficiently small $\d>0$, this remains true, if the death rate of an 
individual with genotype $aa$,
is perturbed by an additional random process that is uniformly bounded by $M\d$.
\end{proposition}

We start the population process when $n_{aa}$ is in a 
$\d/2$-neighbourhood of its equilibrium, $\bar n(a)$, and  there is one individual with genotype $aA$. 
The first theorem says that there is a positive probability that the mutant population fixates in the resident $aa$-population, and 
that under Assumption \textbf{B}, the $a$ allele survives much longer than usual. 
 Define
\bea
\label{stopmut}
\t_{\d}^{mut}&\equiv&\inf\{t>0: n_{AA}(t)+n_{aA}(t)\geq \delta\}, \\
\t_{0}^{mut}&\equiv&\inf\{t>0: n_{AA}(t)+n_{aA}(t)=0\}\\
\t^{hit}_\eta&\equiv&\inf\{t\geq\t_\d^{mut}: n_{aA}(t)\leq \eta\},\\
\t_{sur}&\equiv&\inf\{t\geq\t_\d^{mut}: n_{aA}+x_{aa}(t)(t)=0\}.
\eea
%

\begin{theorem}[\cite{BovNeu16}]
\TH(main1)
Consider the model verifying Assumption \textbf{B}. 
For any $\a>0$ and $1\gg\d>\e>0$, 
 conditional on survival of the mutant, i.e., on  the event $\{\tau_\d^{mut}<\tau_0^{mut}\}$, with probability converging to one as $K\uparrow \infty$, the following statements hold:
\begin{itemize}
\item[(i)] \quad$\t_{\e}^{hit}=\OO(\ln K)$, and
\item[(ii)] \quad$\t_{sur}\geq\OO(K^{1/2-\a})$. 
\end{itemize}
\end{theorem}
\begin{remark}
As long as there are $aA$-individuals in the population, the smaller $aa$-population does not die out, since the $aA$-
population always gives birth to $aa$-
individuals. For smaller values of the power $\frac12-\a$ in (ii), the natural fluctuations 
of the big $AA$-population are too high: the death rate of $n_{aa}(t)$ can be too large due to the competition felt by 
$n_{AA}(t)$ and could induce the death of the $aa$-population and hence also of the $aA$-population.
\end{remark}
%

\section{The deterministic system}
\label{pardet}

The source of the special behaviour of the diploid model can be traced to the limiting Lotka-Volterra system.
As in the haploid model, we get convergence to these equations in the large population limit.

\begin{proposition}[Proposition 3.2 in \cite{CMM13}]
\label{LPA}
Let $T>0$ and $C\subset\R_+^3$ be compact.
Assume that the initial condition $\frac{1}{K}(X_{aa}^0,X_{aA}^0,X_{AA}^0)$ converges almost surely to a deterministic vector $(x_0,y_0,z_0)\in C$ when $K$ goes to infinity. 
Let $(x(t),y(t),z(t))=\phi(t;(x_0,y_0,z_0))$ denote the solution to
\be
\label{det}
\dot\phi(t;(x_0,y_0,z_0))=\left(
\begin{array}{c}
\tilde b_{aa}(x(t),y(t),z(t))-\tilde d_{aa}(x(t),y(t),z(t))\\
\tilde b_{aA}(x(t),y(t),z(t))-\tilde d_{aA}(x(t),y(t),z(t))\\
\tilde b_{AA}(x(t),y(t),z(t))-\tilde d_{AA}(x(t),y(t),z(t))
\end{array}
\right)=:X(x(t),y(t),z(t)),
\ee
where 
\bea
\tilde b_{aa}(x(t),y(t),z(t))&=&\frac{(\ff_{aa}x(t)+\sfrac{1}{2}\ff_{aA}y(t))^2}{(\ff_{aa}x(t)+\ff_{aA}y(t)+\ff_{AA}z(t))},\\
\tilde d_{aa}(x(t),y(t),z(t))&=&x(t)(d_{aa}+c_{aa,aa}x(t)+c_{aa,aA}y(t)+c_{aa,AA}z(t)),
\eea
and similar expression for the $aA$- and $AA$-type, that can be read off the
rates given in \eqv(birth-rates) and \eqv(death-rates).  Then, for all $T>0$,
\be
\label{inva}
\lim_{K\rightarrow\infty}\sup_{t\in[0,T]}|z_{uv}(t)-\phi_{uv}(t;(x_0,y_0,z_0))|=0, \quad a.s.,
\ee
for all $uv\in\{aa, aA, AA\}$.
\end{proposition}
The first step to understand the behaviour of the process is to analyse the deterministic system \eqref{det}. 
The vector field \eqref{det} under Assumptions B  is given by
\be
\label{XD}
X(x,y,z)=X_\D(x,y,z)=
\begin{pmatrix}
\ff\frac{(x+\frac{1}{2}y)^2}{x+y+z}-(d+\D+c(x+y+z))x\\
2\ff\frac{(x+\frac{1}{2}y)(z+\frac{1}{2}y)}{x+y+z}-(d+c(x+y+z))y\\
\ff\frac{(z+\frac{1}{2}y)^2}{x+y+z}-(d+c(x+y+z))z
\end{pmatrix},
\ee
which has some particular properties:
\begin{theorem}
\label{dettheo}
Let $\D>0$. Then: 
\begin{itemize}
\item[(i)] the vector field \eqref{XD} has the unstable fixed point $\mathfrak z_{aa}\equiv(\bar n(a),0,0)$ and the stable fixed point $\mathfrak z_{AA}\equiv(0,0,\bar n(A))$,
\item[(ii)] the Jacobian matrix at the unstable fixed point $\mathfrak z_{aa}$, $DX_\D(\mathfrak z_{aa})$, has two negative and one positive eigenvalues,
\item[(iii)] the Jacobian matrix at the stable fixed point $\mathfrak z_{AA}$, $DX_\D(\mathfrak z_{AA})$, has two negative and one zero eigenvalues,
\item[(iv)]  Any solution with initial conditions $|x(0)-\bar x(a)|\leq \e$ and $\d<y(o)\leq \e$, for some $0<\d<\e\ll 1$ will satisfy,  as $t\uparrow \infty$, 
\be
y(t)= \frac 1t\frac {{2(b-d)}{(b+\D)}}{ {bc}{\D}}(1+O(\ln t/t)).
\ee
\end{itemize}
\end{theorem}

\begin{remark}
The main peculiarity is the slow $1/t$ decay of the  $aA$ population (and hence of the $a$-allele) as the population approaches the stable fixpoint. This is caused by the zero-eigenvalue of the Jacobian at this fixpoint
and is a result of the dominance assumption of the advantageous $A$ allele. Naturally, if $y(t)\sim 1/t$, then 
$x(t)\sim 1/t^2$, which just reflects that as long as there are $aA$ individuals, they will give birth to $aa$ children. For the stochastic system, this suggests that the $aA$ population will have size $K^{1/2+\a}$ by time $K^{1/2-\a}$. The threat of extinction of the $a$-allele comes from upward fluctuations of the total population.
These are typically of order $K^{1/2}$, and excert an extra death rate of order $K^{-1/2}$ on the $aA$ population. Since the duration of such a fluctuation is $O(K^{-1/2})$, this is not sufficient to drive 
a population of size $K^{1/2+\a}$ to extinction (which would require a time $K^{-1/2 +\a}$). 
Thus the $a$ allele will survive (in the absence of any further mutations) almost a time $K^{1/2}$.
\end{remark}

Before proving the theorem, we consider the special case $\D=0$. 

\begin{lemma} \TH(delta.0)
Consider the vector field $X$ given in \eqv(XD) and let $\D=0$. 
Let $\s^*=\frac {b-d}c$. 
Then $X$ vanished on the lines 
\bea \Eq(solutions.1)
x
&=&\frac{\s^*-y}2+\sqrt {\frac{\s*}2\left(\frac {\s^*}2  -y\right)}\\
z
&=&\frac{\s^*-y}2-\sqrt {\frac{\s*}2\left(\frac {\s^*}2  -y\right)},
\eea
and 
\bea\Eq(solutions.2)
x
&=&\frac{\s^*-y}2-\sqrt {\frac{\s*}2\left(\frac {\s^*}2  -y\right)},\\
z
&=&\frac{\s^*-y}2+\sqrt {\frac{\s*}2\left(\frac {\s^*}2  -y\right)},
\eea
for $y\in [0,\s^*/2]$. 
In particular, this line of fixpoints, $\G_0$, connects the monomorphic fixpoints 
$(\s^*,0,0)$ and $(0,0,\s^*)$. At $y=\s^*/2$,  $x=z=\s^*/4$. 
Moreover, $\G_0$ is attractive, and the Jacobian matrix on each point of the line has the eigenvalues
$-(b-d), 0$, and $-b$. 
\end{lemma}

\begin{proof} First, by summing the three equations and setting $x=y+z=\s$, we 
get the fixpoint equation 
\be
(b-d)\s-c\s^2=0,
\ee
hence either $\s=0$ (which is the trivial fixpoint $(0,0,0)$) or $\s=\frac {b-d}c\equiv\s^*$. 
The trivial fixpoint is of no interest to us. Inserting $\s=\s^*$, the equations involving the $x$ and $z$ components 
become
\bea
x^2+xy+\frac 14 y^2-x\s^*&=&0,
\\
z^2+zy+\frac 14 y^2-z\s^*&=&0,
\eea
Both equations for $x$ and $z$ are the same, and they admit the two solutions 
given in \eqv(solutions.1) and \eqv(solutions.2).  They are real and non-negative only for $y\in[0,\s^*/2]$. 
If we choose for $x$ and $z$ the same solution, we 
would get
\be
x+z= \s^*-y+ \sqrt {\frac {\s^*}2\left(\frac{\s^*}2-y\right)},
\ee 
which is compatible with $x+y+z=\s^*$ only if $y=\s^*/2$, 
which is the case when both solutions collapse. 
Hence, we must choose one solution for $x$ and the other for $y$. 
Since 
\be
(x+y/2)(z+y/2) =\frac{ (\s^*)^2}4- {\frac{\s*}2\left(\frac {\s^*}2  -y\right)
=\s^*y/2},
\ee
we see 
 that the third equation (for the $y$-component) is then also verified. 
 A rather tedious computation shows that the remaining statements of the lemma.
 
The line of fixpoints can be re-paramtrised conveniently using $v=z-x=-2\sqrt {\frac{\s*}2\left(\frac {\s^*}2  -y\right)}$. The fixpoint line then is given by
\begin{equation}
\Gamma_0(v)=
\begin{pmatrix}
\frac{(v-\s^*)^2}{4\s^*}\\
-\frac{v^2-(\s^*)^2}{2\s^*}\\
\frac{(v+\s^*)^2}{4\s^*}
\end{pmatrix}, \quad v\in [-\s^*,\s^*].
\end{equation}

The Jacobian  of the vector field $X_0$ at each point of the curve $\Gamma_0$, $DX_0(\Gamma_0(v))$, has the three eigenvectors (see \cite{CMM13})
\bea\nonumber
&&e_1(v)=\G_0(v)=
\begin{pmatrix}
\frac{(v-\s^*)^2}{4\s^*}\\
-\frac{v^2-(\s^*)^2}{2\s^*}\\
\frac{(v+\s^*)^2}{4\s^*}
\end{pmatrix},\hspace{2mm}
e_2(v)=\frac {d\G_0(v)}{dv}=
\begin{pmatrix}
\frac{v-\s^*}{2\s^*}\\
-\frac{v}{\s^*}\\
\frac{v+\s^*}{2\s^*}
\end{pmatrix},\hspace{2mm}
\\
&&e_3(v)=\frac {d^2\G_0(v)}{dv^2}=
\frac{1}{2\s^*}
\begin{pmatrix}
1\\
-2\\
1
\end{pmatrix},
\eea
with respective eigenvalues $-(\ff-d)<0$, $0$ and $-\ff<0$.  $DX_0(\Gamma_0(v))^t$ has the three eigenvalues, $-\ff+D, 0$, and $-\ff$, with corresponding eigenvectors
\be
\b_1(v)=\frac{1}{\s^*}
\begin{pmatrix}
1\\1\\1
\end{pmatrix},\hspace{2mm}
\b_2(v)=
\begin{pmatrix}
-\frac{v+\s^*}{\s^*}\\
-\frac{v}{\s^*}\\
-\frac{v-\s^*}{\s^*}
\end{pmatrix},\hspace{2mm}
\b_3(v)=
\begin{pmatrix}
\frac{(v+\s^*)^2}{2\s^*}\\
\frac{v^2-(\s^*)^2}{2\s^*}\\
\frac{(v-\s^*)^2}{2\s^*}
\end{pmatrix},
\ee
which satisfy, for any $i,j\in\{1,2,3\}$ and any $v$,
\be
\label{EVb}
\langle\b_i(v),e_j(v)\rangle=\d_{i,j}.
\ee

This proves the lemma.
\end{proof}

\begin{proof}[Proof of Theorem \ref{dettheo}]
As soon as $\D>0$, the line of fixpoints from the $\D=0$ case collapses to
the two obvious monomorphic fixpoints
\be
\mathfrak {z}_{aa}\equiv (\bar n(a),0,0), \quad \bar n(a) =\frac {b-d-\D}c,
\ee
and 
\be
\mathfrak {z}_{AA}\equiv (0,0,\bar n(A)), \quad \bar n(A) =\frac {b-d}c.
\ee
Note that there is no monomorphic fixpoint for the mixed type, since $x=0$ or $z=0$ implies $y=0$ and $z=0$. 
First, we see that the equation for the total mass is no longer autonomous, but reads
\be
(b-d) \s- c\s^2 -\D x=0,
\ee
so that we get the total mass parametrised by $x$ as
\be
\s^*(x)=\frac {\s^*}2+\sqrt{ {\s^*}^2/4-\D x/c}.
\ee
Note that this has the correct values at the monomorphic fixpoints when $x=0$ resp. $x=\frac{b-d-\D}c$.
One can show that these are, apart from the trivial fixpoint zero, the only fixpoints in the positive quadrant.
The line of fixpoints is then transformed into an invariant line connecting the two fixpoints. This line is 
still attractive and thus a neighbourhood of the $aa$-fixpoint will flow towards the stable $AA$ fixpoint.



The Jacobian matrix at the fixed point $\mathfrak z_{aa}$ is given by
\be
\label{DXaa}
DX_\D((\bar n(a),0,0))=
\begin{pmatrix}
-\ff+d+\D & -\ff+d+\D & -2\ff+d+\D\\
0 & \D & 2\ff\\
0 &  0 & -\ff+\D
\end{pmatrix}.
\ee
The matrix has the three eigenvalues $\lambda_1=-(\ff-d-\D)$, $\lambda_2=\D$ and $\lambda_3=-(\ff-\D)$. 
As long as $\bar n(a)>0$, namely $(b-d-\D)>0$,  we have  $\lambda_1,\lambda_3<0$, whereas $\lambda_2>0$. Thus the fixed point $\mathfrak z_{aa}$ is unstable.\\
The Jacobian matrix at the fixed point $\mathfrak z_{AA}$ is given by
\be
DX_\D((0,0,\bar n(A)))=
\begin{pmatrix}
 -\ff-\D &  0 &  0\\
2\ff & 0 & 0\\
-2\ff+d & -\ff+d & -\ff+d
\end{pmatrix},
\ee
with eigenvalues $\lambda_1=-\ff-\D<0$, $\lambda_2=0$ and $\lambda_3=-(\ff-d)<0$. 
The fact that one of the eigenvalues is zero is due to the fact that the phenotypes of $aA$ and $AA$ 
are equal. If the fitness of $AA$ is assumed larger than that of $aA$, as in \cite{CMM13}, all three eigenvalues would be negative.
Because of the zero eigenvalue, $\mathfrak z_{AA}$ is a non-hyperbolic equilibrium point of the system and linearisation fails to determine its stability properties. 

We will follow a pedestrian approach and expand the vector fields to quadratic order around the fixpoint. 
For this it will be very convenient to chose the coordinates $ x, y$, and $\s\equiv (x+y+z)$. 
The fixpoint in these coordinates reads as before $(0,0,\bar n(A)$. However, the vector field simplifies to 
\be
\label{sigma.1}
\wt X(x,y,\s)=
\begin{pmatrix}
\ff\frac{(x+\frac{1}{2}y)^2}{\s}-(d+\D+c\s) x\\
2\ff\frac{(x+\frac{1}{2}y)(\s-x-\frac{1}{2}y)}{\s}-(d+c\s) y\\
(\ff-d)\s -\D x-c\s^2
\end{pmatrix}.
\ee
Calculating the derivatives and second derivatives of the three components,
we arrive at the following approximations of the 
three differential equations near the fixpoint (we set  $\d_\s=\s-\bar n(A)$).
\bea\Eq(sigma.2)
x'(t)&\sim&-(b+\D)x+\frac{bc}{b-d}x^2+\frac{bc}{4(b-d)}y^2+\frac {bc}{b-d} xy-cx\d_\s,\\
y'(t)&\sim&2bx-\frac{2bc}{b-d}x^2-\frac {bc}{2(b-d)}y^2-\frac{2bc}{b-d}xy-cy\d_\s,\\
\d_\s'(t)&\sim&-(b-d)\d_\s-\D x-c\d_\s^2.
\eea
Looking at the equation for $x$, we see that the linear term $-(b-d) x$ dominates all 
other terms except possibly the term proportional to $y^2$.
Next,  we see that $\d_\s(t)$  approaches $\frac {\D}{b-d}x(t)$ (in fact, given $x$, 
$\d_s(t)\sim \eee^{-(b-d)t}\int_0^t \D x(s)\eee^{(b-d)s}ds$). 
Therefore, also in the second equation, the
leading order terms are those proportional to $x$ and $y^2$.

This implies that the terms $x^2, xy, y\d_\s, x\d_\s$ in these equations can be neglected, leaving us with 
\bea
x'(t)&\sim&-(b+\D)x+\frac{bc}{4(b-d)}y^2,\\\Eq(sigma.2.1)
y'(t)&\sim&2bx-\frac {bc}{2(b-d)}y^2.
\eea
as effective equations for the asymptotic of $x$ and $y$.
From the first equation we see that $x$ cannot drop below 
$ \frac{bc}{4(b-d)(b+\D)}y^2$. Hence 
\be
\Eq(sigma.5)
y'(t)>\frac{b^2c}{2(b-d)(b+\D)}y^2 -\frac {bc}{2(b-d)}y^2
=-\frac{bc}{2(b-d)}\frac{\D}{b+\D} y^2.
\ee
This is readily solved and yields (provided $\D>0$)
\be
\Eq(sigma.6)
y(t)\geq \frac {1}{t \frac{bc}{2(b-d)}\frac{\D}{b+\D}+1/y_0}.
\ee

Using that Eq. \eqv(sigma.2) has the explicit solution 
\be\Eq(sigma.7)
x(t)=\frac{bc}{4(b-d)}\eee^{-(b+\D)t}\int_0^t \eee^{s(b+\D)s} y^2(s) ds +C\eee^{-(b+\D)t},
\ee
and the bound \eqv(sigma.6), we see that indeed  
\be\Eq(sigma.7.1)
x(t)\sim \frac{bc}{4(b-d)(b+\D)} y^2(t).
\ee 
Thus, the differential inequality  \eqv(sigma.5) is  effectively an  (asymptotic) equality, and hence the same for 
 \eqv(sigma.6) .
Moreover $x(t)\sim t^{-2}$. 
\end{proof}

There is also a biological explanation for the behaviour of $z_{aA}(t)$ described in Theorem \ref{dettheo} (iv). 
Since the $A$-allele is the fittest and dominant one, and because of the phenotypic viewpoint, the $aA$-population is as fit 
as the $AA$-population and both die with the same rate. The $aA$-population only decreases because of the 
disadvantage in reproduction due to the less fit, decreasing $aa$-population.
Observe that Theorem \ref{dettheo} (i)+(ii) also holds in the model of \cite{CMM13} (cf. Proposition 3.3 therein) but the 
Jacobian matrix of their fixed point $\mathfrak z_{AA}$ has three negative eigenvalues and thus they get the exponential
decay of $z_{aA}(t)$.\\

\begin{figure}[t]
	\centering
\includegraphics[width=0.5\textwidth]{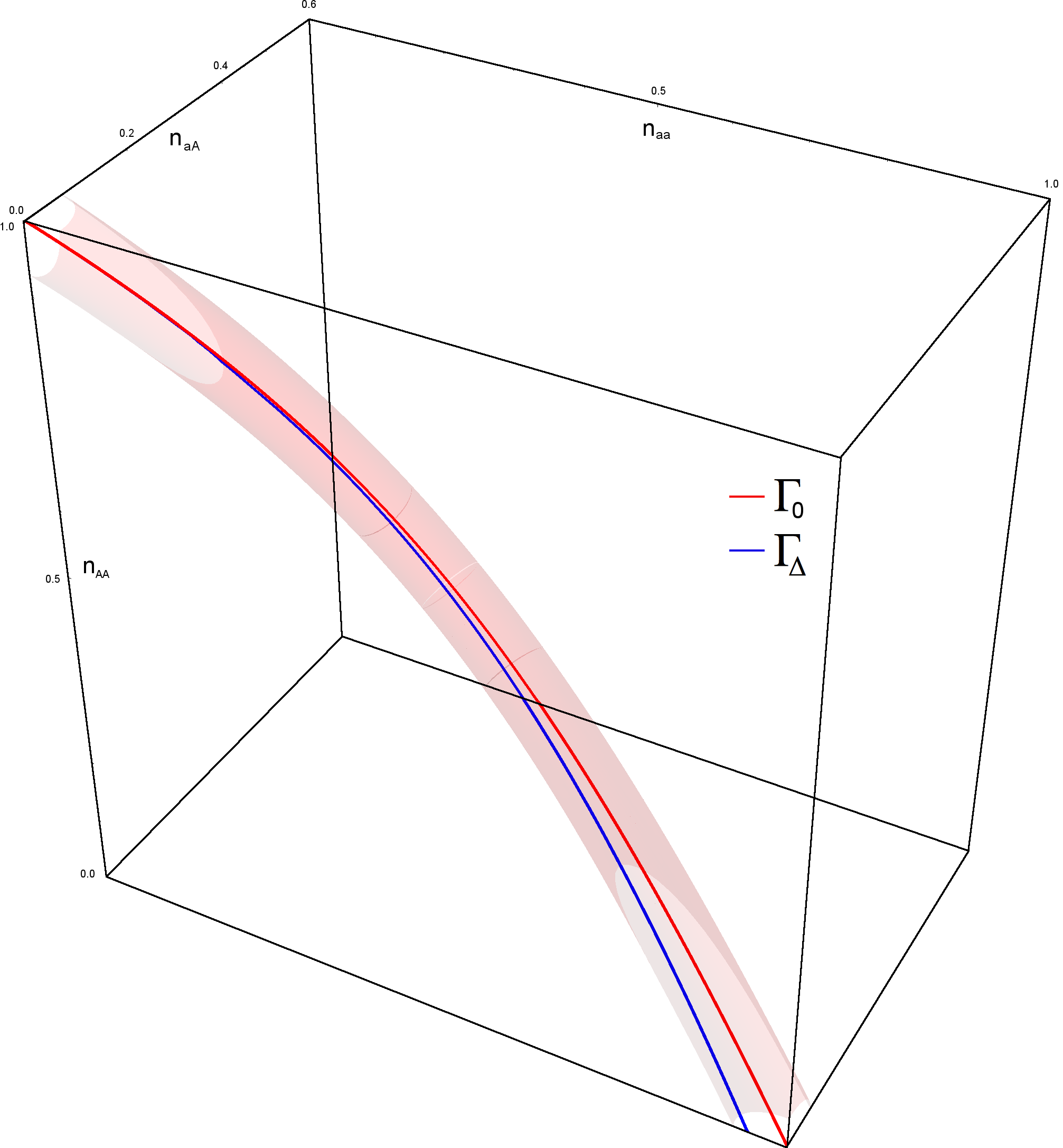}
	\caption{The curves $\Gamma_0$, $\Gamma_\D$ and the tube $\mathcal{V}$ in the perturbed vector field $X_\D$ \tiny{(Simulation by Loren Coquille)}}
	\label{fig3}
\end{figure}

\subsection{Outline of the proof}
The analysis of the stochastic system follows the scheme used in the haploid 
model. After the appearance of the first mutant $A$-allele, the $Aa$ population can be controlled as
a supercritical branching process until it reaches a level $\e$. At that time, the $AA$ population has also increased to a level $\e^2$. From that time on, the system follows the deterministic system until, in finite time, 
an $\e$-neighbourhood of the $AA$-fixpoint is reached. The main novelty -- and difficulty -- 
is to show that the final phase, i.e. the time to extinction of the $a$-allele, is 
much longer than the usual $\ln K$. In fact, if the stochastic system follows the deterministic one further on, 
we should expect that the $aA$ population takes a time of order $K^{1/2}$ to drop to the size
$K^{1/2}$. As argued earlier, 
once the population is down to that size, we may expect that it will die out due to upward fluctuations of the total 
population size.
Namely, the natural fluctuations of the total population are of order $K^{1/2}$ and have a duration of order $1$.  During such a fluctuation, the $aA$ proves can be approximated as a subcritical 
branching process with death rate $d\sim b+O(K^{-1/2})$. Using the computation of the mean time
to extinction in Chapter 2, to wit \eqv(ej.4) with $k=K$, $j=K^{1/2+\a}$ and $d=b+\sqrt K$, 
\be
e^K_{K^{1/2+\a}} \sim dK^{1/2+\a}.\ee

 Showing that, up to that time, the heuristics is correct, is technically very difficult and we will 
not give the details here. The methods used in \cite{BovNeu16} are essentially the stochastic Euler scheme
 developed in \cite{B14}, see  Section \ref{section6.3}.

\begin{figure}[h]
\begin{minipage}[t]{0.49\textwidth}
	\centering
  \includegraphics[width=1\textwidth]{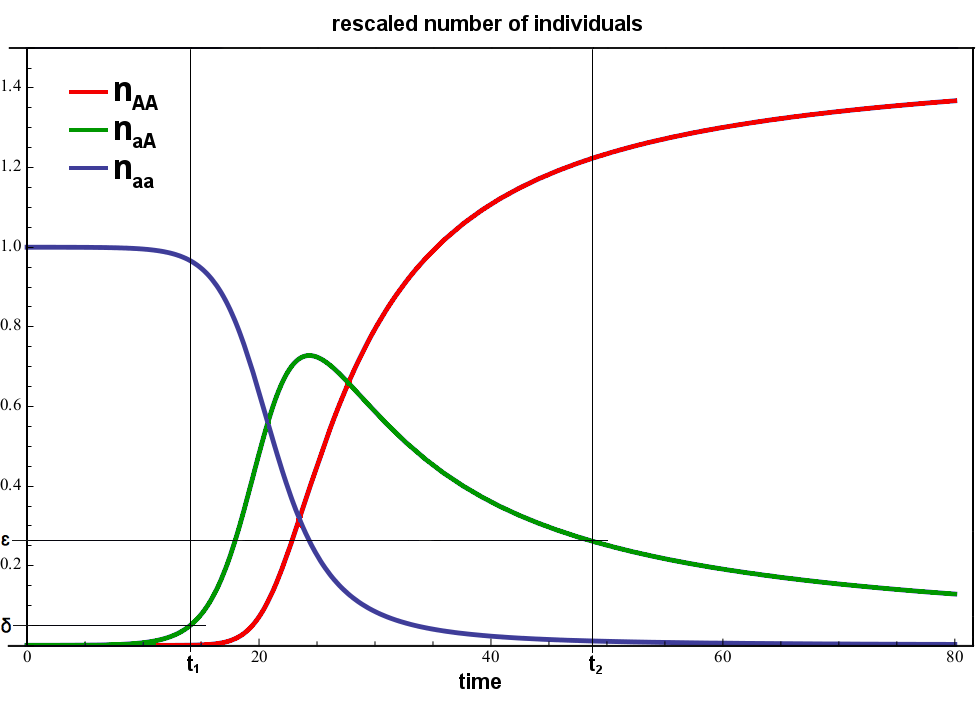}
	\caption{ $A$ fittest type and dominant}
	\label{fig1}
\end{minipage}
\begin{minipage}[t]{0.5\textwidth}
	\centering
  \includegraphics[width=1\textwidth]{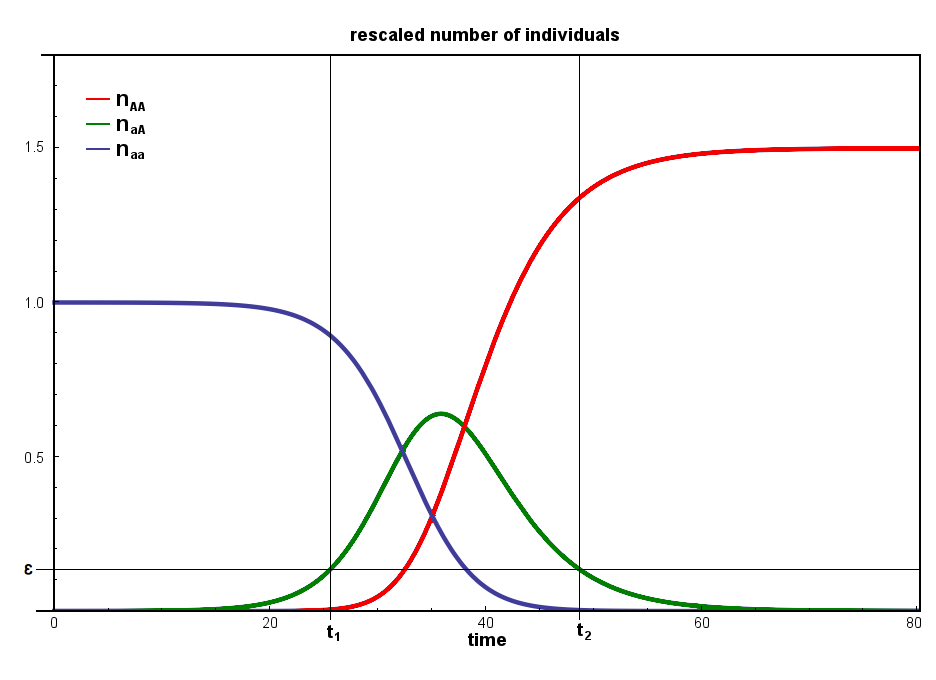}
	\caption{$a$ and $A$ co-dominant}
	\label{fig5}
\end{minipage}
\end{figure}

\section{Recovery of the original allele}

We will now show that the prolonged survival of
the $a$ allele after the invasion of the $A$ allele can indeed lead to a recovery of the $aa$-type. 
To do this, we assume that a new mutant allele (on the same gene), $B$, that on the one hand has 
a higher fitness than the $AA$-phenotype but that (for simplicity) has no competition with the
$aa$-type. This should lead to a suppression of the $AA$-phenotype and hence allow the $aa$-type to recover. 

The possible genotypes after this mutation are  $aa,aA,AA,aB,AB$, and $BB$, so that 
even for the deterministic system we have now to deal with a $6$-dimensional dynamical system 
whose analysis is far from simple. 

Following \cite{BovCoqNeu18}, we look at the following particular scenario.

\subsection{Model setup}\label{sec-model}

Let $\mathcal G=\{aa,aA,AA,aB,AB,BB\}$ be the genotype space.
We assume that the $B$ allele is the fittest and the most dominant one. That means the ascending order of dominance (in the Mendelian sense) is given by $a<A<B$, i.e.
\begin{enumerate}
	\item phenotype $a$ consists of the genotype $aa$,
	\item phenotype $A$ consists of the genotypes $aA,AA$,
	\item phenotype $B$ consists of the genotypes $aB,AB,BB$.\\
\end{enumerate}	

 Fertilities are the same for all genotypes and natural death rates are the same 
within the three different phenotypes. Moreover, we assume that there can be no reproduction between $a$ and $B$ phenotypes.\\
To summarise, we make the following Assumptions \textbf{D} on the rates: 

\begin{itemize}
	\item[\textbf{(D1)}] \emph{Fertilities.} For all  $ i\in\mathcal G$, and some $\ff>0$
	 \begin{align}\ff_{i}\equiv  \ff.\end{align}
	 
	\item[\textbf{(D2)}] \emph{Natural death rates.} The difference in fitness of the three phenotypes is realised by choosing 
	 a slightly higher natural death-rate of the $a$-phenotype and a slightly lower death-rate for the $B$-phenotype. For some $0<\D<\min(d, b-d)$,
	\begin{align}
	d_{aa}&= d+\D,\\
	d_{AA}\equiv  d_{aA}&=d,\\
	d_{aB}\equiv  d_{AB} \equiv  d_{BB}&=d-\D.
	\end{align}
	
	\item[\textbf{(D3)}] \emph{Competition rates.} We require that phenotypes $a$ and $B$ do not compete with each other. 
	Moreover, we introduce a parameter $\eta\geq0$ which lowers the competition between $BB$ and $aA$.
	For some $0\leq\eta<c$, 
	\begin{equation*}\left(c_{i,j}\right)_{\{i,j\}\in\mathcal G\times\mathcal G}=:
	\begin{tabular}{c|c|c|c}
	& $aa$ & $aA \quad AA$&$aB\quad AB\quad BB$\\
	\hline
	$aa$ & $c$ &$c\quad c$&$0 \quad 0 \quad 0$\\
	\hline
	$aA$ & $c$ &$c\quad c$&$c \quad c \quad c-\eta$\\
	$AA$ & $c$ &$c\quad c$&$c \quad c \quad c$\\
	\hline
	$aB$ & $0$ &$c\quad c$&$c \quad c \quad c$\\
	$AB$ & $0$ &$c\quad c$&$c\quad c \quad c$\\
	$BB$ & $0$ &$c-\eta \quad c$&$c \quad c\quad c$\\
	\hline
	\end{tabular}
	\end{equation*}
	\begin{remark} If $\eta>0$, the competition does not depend only on the phenotype, and can be interpreted as a refinement of a phenotypic competition for resources: the strength (or ability to get resources of an individual not only depends on its phenotype but also on its genotype. Genetically, it makes sense to assume that (positive) competition rates are decreasing in the "genetic distance" between two individuals. 
	\end{remark}
	
		\item[\textbf{(D4)}] \emph{Reproductive compatibility.} We require that phenotypes $a$ and $B$ do not reproduce with each other. 
	\begin{equation*}\left(R_{i}(j)\right)_{\{i,j\}\in\mathcal G\times\mathcal G}\equiv  
	\begin{tabular}{c|c|c|c}
	& $aa$ & $aA \quad AA$&$aB\quad AB\quad BB$\\
	\hline
	$aa$ & $1$ &$1\quad 1$&$0 \quad 0 \quad 0$\\
	\hline
	$aA$ & $1$ &$1\quad 1$&$1 \quad 1 \quad 1$\\
	$AA$ & $1$ &$1\quad 1$&$1 \quad 1 \quad 1$\\
	\hline
	$aB$ & $0$ &$1\quad 1$&$1 \quad 1 \quad 1$\\
	$AB$ & $0$ &$1\quad 1$&$1\quad 1 \quad 1$\\
	$BB$ & $0$ &$1 \quad 1$&$1 \quad 1\quad 1$\\
	\hline
	\end{tabular}
	\end{equation*}
	
\end{itemize}

Under Assumptions  \textbf{D}, the invasion fitnesses satisfy
\begin{align}
\label{unfit.1}
f_{AB,AA}&=\ff-(d-\D)-c\bar z_{AA}=\ff-d+\Delta-c\frac{\ff-d}{c}=\D,\\
f_{aa,BB}&=\ff-d-\D,\\
c\bar z_B&=\ff-d+\Delta
\end{align}
Therefore,  the mutant $AB$ has a  positive invasion fitness in the population $AA$, as well as $aa$ in the $BB$ population (due to the absence of competition between them).

Since we assume that there is no recombination between phenotypes $a$ and $B$,
\begin{enumerate}
	\item the pool of possible partners for the phenotype $a$ consists of phenotypes $a$ and $A$; 
	the total population of this pool is denoted by 
	\begin{equation}
	\Sigma_3:=X_{aa}+X_{aA}+X_{AA},
	\end{equation}
	\item the pool of possible partners for the phenotype $A$ consists of all phenotypes $a$;
	the total population of this pool is denoted 
	\begin{equation}
	\Sigma_6:=X_{aa}+X_{aA}+X_{AA}+X_{aB}+X_{AB}+X_{BB},
	\end{equation}
	\item the pool of possible partners for the phenotype $B$ consists of phenotypes $A$ and $B$;
	the total population of this pool is denoted by 
	\begin{equation}
	\Sigma_5:=X_{aA}+X_{AA}+X_{aB}+X_{AB}+X_{BB}.
	\end{equation}
\end{enumerate} 
Computing the reproduction rates with the Mendelian rules as described in \cite{BovNeu16} 
leads to the following (time-dependent) birth-rates (we abbreviate $b_i=r(X_i\to X_i+1)$): 

\begin{align}\label{birthratesaa}
b_{aa}=&\ff\frac{X_{aa} \left(X_{aa}+\frac12X_{aA}\right)}{\S_3}+\ff\frac{\frac12X_{aB} \left(\frac12X_{aA}+\frac12X_{aB}\right)}{\S_5}\nonumber\\
&
+\ff\frac{\frac12X_{aA} \left(X_{aa}+\frac12X_{aA}+\frac12X_{aB}\right)}{\S_6},\\[0.5cm]\label{birthratesaA}
b_{aA}=&\ff\frac{X_{aa} \left(\sfrac12X_{aA}+X_{AA}\right)}{\S_3}+\ff\frac{\frac12X_{aA} \left(\frac12X_{aB}+\frac12X_{AB}\right)+\frac12X_{aB} \left(X_{AA}+X_{AB}\right)}{\S_5}\nonumber\\[0.2cm]
&+\ff\frac{\left(\frac12X_{aA}+X_{AA}\right) \left(X_{aa}+X_{aA}+\frac12X_{aB}\right)+\frac14X_{aA} X_{AB}}{\S_6},\\[0.5cm]\label{birthratesAA}
b_{AA}=&\ff\frac{\frac12X_{AB} \left(\frac12X_{aA}+X_{AA}+\frac12X_{AB}\right)}{\S_5}+\ff\frac{\left(\frac12X_{aA}+X_{AA}\right) \left(\frac12X_{aA}+X_{AA}+\frac12X_{AB}\right)}{\S_6},\\[0.5cm]\label{birthratesaB}
b_{aB}=&\ff\frac{\left(\frac12X_{aA}+X_{aB}\right) \left(\frac12X_{aB}+\frac12X_{AB}+X_{BB}\right)}{\S_5}+\ff\frac{\frac12X_{aA} \left(\frac12X_{aB}+\frac12X_{AB}+X_{BB}\right)}{\S_6},\\[0.5cm]\label{birthratesAB}
b_{AB}=&\ff\frac{\left(\frac12X_{aA}+X_{AA}+X_{AB}\right) \left(\frac12X_{aB}+\frac12X_{AB}+X_{BB}\right)}{\S_5}\nonumber\\
&
+\ff\frac{\left(\frac12X_{aA}+X_{AA}\right) \left(\frac12X_{aB}+\frac12X_{AB}+ X_{BB}\right)}{\S_6},\\[0.5cm]\label{birthratesBB}
b_{BB}=&\ff\frac{\frac14\left(X_{aB}+X_{AB}+2 X_{BB}\right)^2}{\S_5}.
\end{align}

The death rates are the sum of the natural and competition death rates ($d_i=r(X_i\to X_i-1)$):
\begin{align}\label{death-ratesaa}
d_{aa}&=X_{aa}(d+\D+c\S_3),\\[0cm]\label{death-ratesaA}
d_{aA}&=X_{aA}(d+c \S_6-\eta X_{BB}),\\[0cm]\label{death-ratesAA}
d_{AA}&=X_{AA}(d+c\S_6),\\[0cm]\label{death-ratesaB}
d_{aB}&=X_{aB}(d-\D+c\S_5),\\[0cm]\label{death-ratesAB}
d_{AB}&=X_{AB}(d-\D+c\S_5),\\[0cm]\label{death-ratesBB}
d_{BB}&=X_{BB}(d-\D+(c-\eta)X_{aA}+c(X_{AA}+X_{aB}+X_{AB}+X_{BB})).
\end{align}

\begin{figure}[h!]
	\includegraphics[width=.7\textwidth]{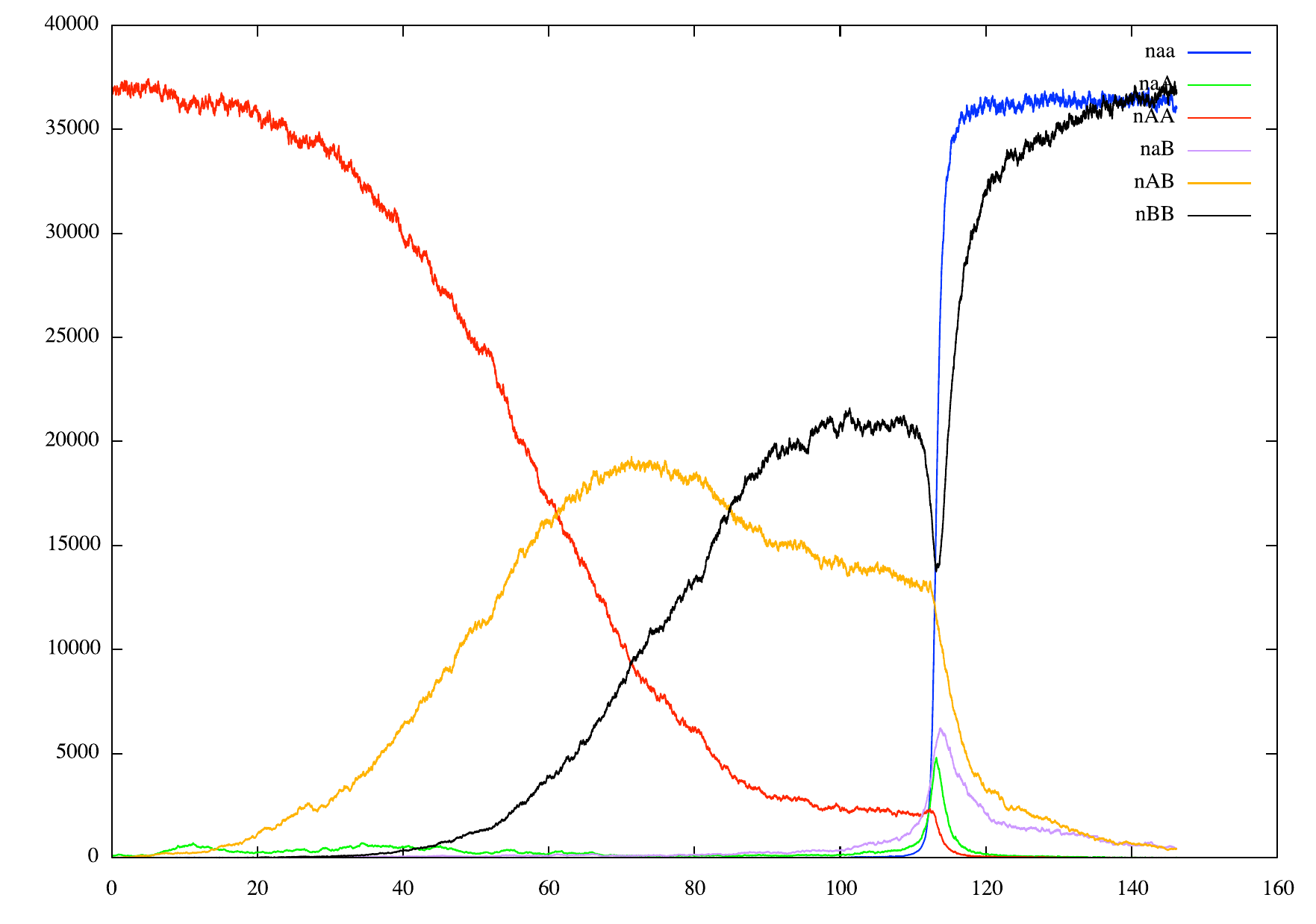}
	\caption{Simulation of the stochastic system for $\ff=6,D=0.7$, $\Delta=0.1$, $c=1, \eta=0.02$, $\e=0.014$ and $K=7000$.}
	\label{pic-stoch}
\end{figure}

The deterministic limit of this process is now a six-dimensional dynamical system
\be\Eq(dyn-syst)
\frac d{dt}x(t)=X(x(t)),
\ee
 driven by the vector field
\begin{align}\Eq(six.1)
&X(x_1,\dots,x_6)=\\\nonumber
&\left(\begin{array}{l}
\ff\frac{x_1(x_1+\sfrac 12 x_2)}{\s_3}    +\ff\frac{\sfrac 14x_4(x_2+x_4)}{\s_5}  +
\ff\frac{\sfrac 12x_2(x_1+\sfrac 12 x_2+\sfrac 12 x_4)}{\s_6}  
-x_1(d+\D+c \s_3)\\
\ff\frac{x_1(x_3+\sfrac 12 x_2)}{\s_3}    +\ff\frac{(\sfrac 14x_2(x_4+x_5)+\sfrac 12x_4(x_3+x_5)}{\s_5}  +
\ff\frac{(\sfrac 12x_2+x_3)(x_1+x_2+\sfrac 12 x_4)+\sfrac 12 x_2x_5)}{\s_6}  
-x_2(d+c \s_6-\eta x_6)\\
\ff\frac{\sfrac 12x_5(\sfrac 12x_2+x_3+\sfrac 12 x_5)}{\s_5}    +\ff\frac{(\sfrac 12x_2+x_3)(\sfrac12 x_2+x_3+\sfrac12 x_5)}{\s_6}
-x_3(d+c \s_6)\\
\ff\frac{\sfrac 12(x_2+x_4)(\sfrac 12 x_4+\sfrac 12 x_5+x_6)}{\s_5}    +\ff\frac{\sfrac 12x_2(\sfrac 12x_4+\sfrac 12x_5+x_6)}{\s_5}  +
\ff\frac{(\sfrac 12x_2+x_3)(x_1+x_2+\sfrac 12 x_4)+\sfrac 12 x_2x_5)}{\s_6}  
-x_4(d-\D+c \s_5)
\\
\ff\frac{(\sfrac 12x_2+x_1+x_4)(\sfrac 12x_4+\sfrac 12 x_5+x_6)}{\s_5}    +\ff\frac{(\sfrac 12x_2+x_3)(\sfrac 12x_4+x_5)+\sfrac 12x_5+x_6)}{\s_6}  
-x_5(d-\D+c \s_5)\\
\ff\frac{\sfrac 14(x_4+x_5+2x_6)^2}{\s_5}   
-x_6(d-\D+c \s_5-\eta x_2)\\
\end{array}
\right),
\end{align}

We  consider the dynamical system \eqref{dyn-syst} starting with the initial condition
\be
\label{init-condAA}
x_1(0)=O(\e^2), x_2(0)=)(\e), x_3(0) \sim \bar z(A), x_4(0)=0,x_5(0)=O(\e^3), x_6(0)=0,
\ee
for small $\e>0$. This is the state of the system when a $B$-mutant has just fixated. 

Clearly, analysing the behaviour of the six-dimensional non-linear system is far from easy. 
Numerical simulations for particular choices of the parameters are, on the other hand, easy.
Figure \thv(bild.1) shows a rather typical example that confirms our expectations.

%

%
%

\begin{figure}[t]\label{bild.1}
	\begin{center}\raisebox{1.5cm}{\includegraphics[width=.495\textwidth]{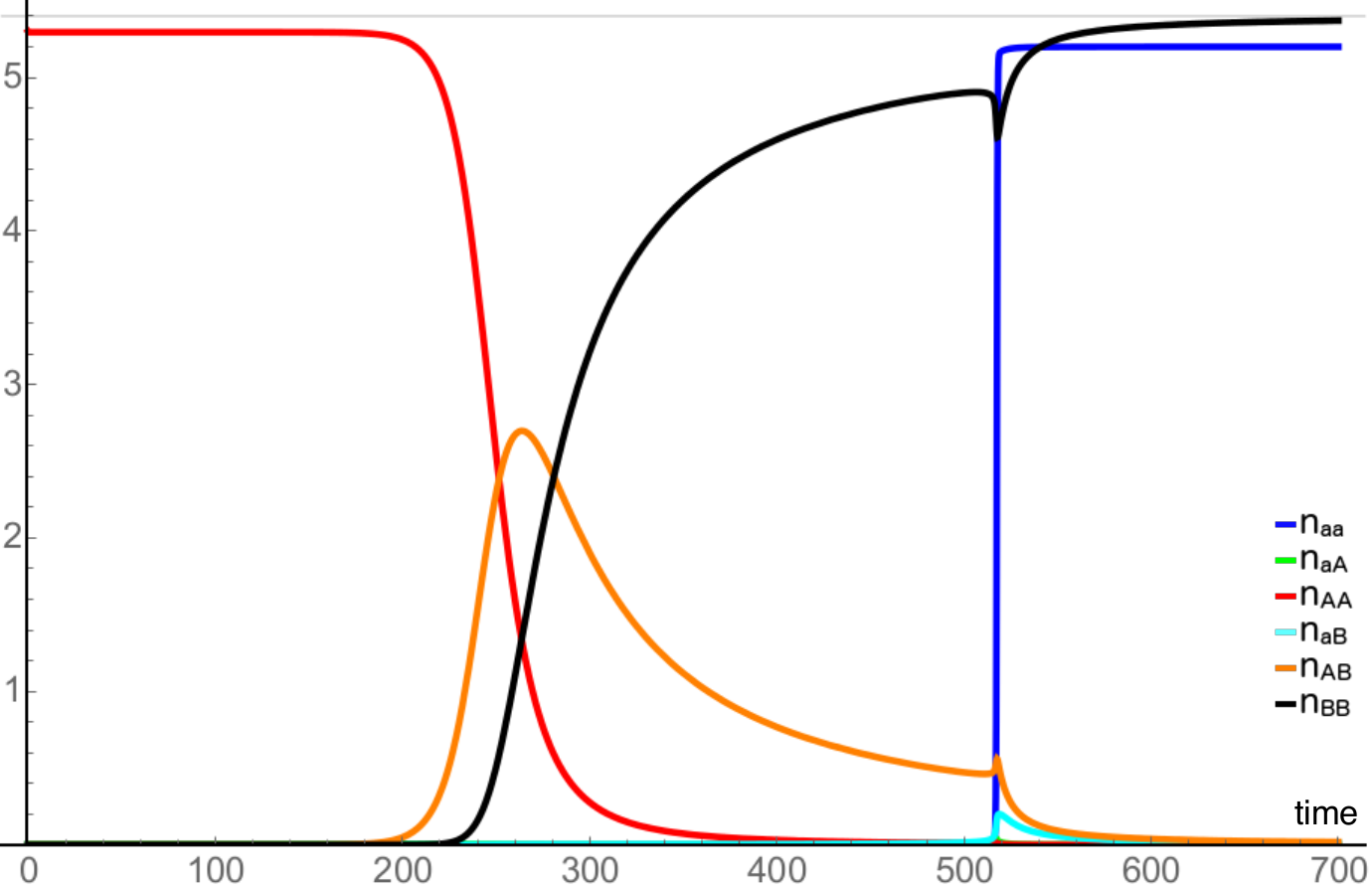}}\hspace{4mm}
		\includegraphics[width=0.45\textwidth]{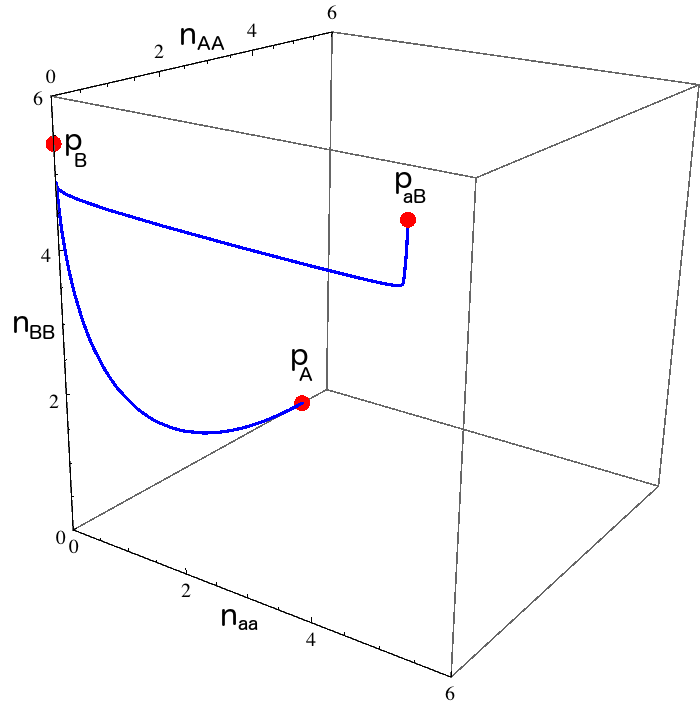}
		\caption{General qualitative behaviour of $\{x(t)\}$ and projection of the dynamical system on the coordinates $aa, AA$ and $BB$.  $\eta=0.02$ for both pictures. }
		\label{pic-general}
	\end{center}
\end{figure}

The first step in the analysis of the dynamical system is to find the equilibrium points. In dimension six, this is far from trivial, and in fact we do not know all of them. However, one  can easily identify the three relevant fixpoints:
\begin{align}
	p_A&=(0,0,\bar z_A,0,0,0),\\
	p_B&=(0,0,0,0,0,\bar z_B),\\
	p_{aB}&=(\bar z_a,0,0,0,0,\bar z_B),
\end{align}
The fact that the bi-morphic fixpoint $p_{aB}$ is easily identified is due to the assumption that the competition 
between the $a$ and the $B$  phenotypes is zero and that there is no mating between the $a$ and the $B$ phenotypes.This is clearly a simplifying assumption.

The main result is that starting with initial conditions \eqv(init-condAA), that is close to $p_A$ (with small coordinates in directions $aa$, $aA$ and $AB$), and under minimal assumptions on the parameters, the system gets very close to  $p_B$ before finally converging to $p_{aB}$, see Figure \ref{pic-general}.

\begin{theorem}\label{main-thm}
	Consider the dynamical system \eqref{dyn-syst} started with initial conditions \eqref{init-condAA}. Suppose the following Assumptions \textbf{E} on the parameters hold:
	\begin{enumerate}
		\item[\textbf{(E1)}] 	 $\Delta$ is sufficiently small, 
		\item[\textbf{(E2)}] 	 $\ff$ is sufficiently large, 
		\item[\textbf{(E3)}]   $0\leq\eta<c/2$.
	\end{enumerate}
	Then the system converges to the fixed point $p_{aB}$.
	More precisely, for any fixed $\delta>0$, as $\e\to0$, 	it reaches a $\delta$-neighbourhood of $p_{aB}$ in a time of order $\OO(\e ^{-1/(1+\eta\bar z_B-\D)})$. \\
	Moreover:  
	\begin{enumerate}
		\item \label{eta0decay} for $\eta=0$,  the amount of  allele $a$ in the population decays to $\OO(\e^{1+\D/(1+\D)})$ before reaching $\OO(1)$,
	
		\item \label{etalarge} for $\eta>\frac{4\D}{\bar z_B}$, the amount of $a$ allele in the population is  bounded below by $\OO(\e)$ for all $t>0$.
	
	\end{enumerate}
\end{theorem}

\begin{remark}
		For $\eta$ large, we prove that the fixed point $p_{aB}$ is unstable.
		We observe numerically that the system is attracted to a fixed point where all the 6 populations coexist, but we do not prove this.
\end{remark}

\subsection{Heuristics of the proof of Theorem \thv(main-thm)}\label{subsec-heuristics}

Let us now briefly discuss the linear stability of the relevant fixed points and explain the 
 heuristics of the proof of Theorem \ref{main-thm}.

The Jacobian matrix  $DX$ of the vector field $X$ defined in \eqv(six.1) can be explicitly computed at $p_A$ and $p_{aB}$ and the situation is as follows:
\begin{itemize}
	\item [(i)] \hspace{1mm}The eigenvalues of $DX(p_A)$ are $0, \Delta>0$ and $-(\ff-d), -(\ff+\Delta), -(\ff-\Delta)$ (double), which are all strictly negative under Assumptions (C). The fixed point $p_A$ is thus unstable.
	\item[(ii)] \hspace{1mm} The eigenvalues of $DX(p_{aB})$ are $0$ (double), and $ -(2\ff-d), -(\ff-d+\Delta), 
	-(\ff-d-\Delta), -((\ff-d)(5\ff-4d)+\ff\Delta)/(4(\ff-d)+\eta\bar z_B)$ which are strictly negative under 
	Assumptions (C).  Analysing the quadratic approximation of the vector field exhibits that it is, in fact, stable. 
	\end{itemize}
	
	It turns out that $DX(p_B)$ is singular but as the invasion fitness  $f_{aa,BB}>0$ (see \eqref{unfit.1}), this implies that a small perturbation in the first coordinate will be amplified, and thus implies the instability of the fixed point $p_B$. Thus both   fixpoints $p_a$ and $p_B$ are unstable.

A numerical solution of the system is provided on Figure \ref{detsys}. 

	\begin{enumerate}
				\item[\textbf{Phase 1.}] Time period:  until $n_{AB}=\e_0$.\\
				  The mutant population, consisting of all individuals of phenotype $B$, first grows up to $\e_0$ exponentially fast {with rate $\Delta$} without perturbing	the behaviour of the 3-system $(aa,aA,AA)$. The rate of growth corresponds to the invasion fitness of $AB$ in the resident population $AA$, see \eqref{unfit.1}. Following \cite{BovNeu16}, $AA$ stays close to $\bar z_A$, while $aA$ and $aa$ continue to decay like $1/t$ and $1/t^2$ respectively.
				The duration $T_1$ of this phase is such that $\OO(\e^3)e^{t\D}=\OO(1)\Leftrightarrow T_1=\OO(|\log\e|)$.\\
				
				\item[\textbf{Phase 2.}] Time period:  until $n_{aA}=\OO(n_{AA})$.\\
				The evolution  is a perturbation of an effective 3-system $(AA,AB,BB)$ which behaves exactly the $(aa,aA,AA)$ system, since the parameters satisfy the same hypotheses (slightly lower death rate for phenotype $B$  than for phenotype $A$, and constant competition parameters). 
				A comparison result  
				 shows that this 3-system is almost unperturbed until $n_{aA}=\OO(n_{AA})$. 	If that happens in a time $T_2$ diverging with $\e$ (which we ensure throughout the calculation), we thus know that $BB$ approaches $\bar z_B$, while 	$n_{AB}\propto1/t$ and
				$n_{AA}\propto1/t^2$.\\
				The important fact in this phase is that the amount of  allele $a$ in the population decays 
				for $\eta$ small while it increases for large enough $\eta$. Indeed, let us derive some bounds on $\Sigma_{aA,aB}=n_{aA}+n_{aB}$.
				The population $\Sigma_{aA,aB}$ reproduces by taking the dominant allele in a population of order $\OO(1)$ and the allele $a$ in itself. Thus its birth rate satisfies $b_{\Sigma_{aA,aB}}
				\approx b \Sigma_{aA,aB}$. We can compute its death rate exactly and use that $n_{BB}\approx\Sigma_5\approx\bar z_B$: 
				\begin{align}\nonumber
				d_{\Sigma_{aA,aB}}&=\Sigma_{aA,aB}(D-\D+c\S_5)-\eta n_{aA}n_{BB}+\D n_{aA}\\
				&\approx \ff\Sigma_{aA,aB}-n_{aA}(\eta\bar z_{B}-\D),\\ \nonumber
				\dot{\Sigma}_{aA,aB}&\approx n_{aA}(\eta\bar z_{B}-\D)\\
				&=\OO({\Sigma}_{aA,aB}\cdot n_{AB})(\eta\bar z_{B}-\D). \label{sigmadot}
				\end{align}
				The last equality comes from the fact that $aA$ newborns have mainly their $a$ allele coming from ${\Sigma}_{aA,aB}$ and their $A$ allele coming from $AB$. 
			   Using the $1/t$ decay of $AB$ we get:
			   \begin{align}
			   	\dot{\Sigma}_{aA,aB}\approx \frac{\OO({\Sigma}_{aA,aB})}{\OO(1)+\OO(1)t}(\eta\bar z_{B}-\D)
				\end{align} 
			    As ${\Sigma}_{aA,aB}(T_1)=\OO(\e)$ we deduce that
			   ${\Sigma}_{aA,aB}(t)={\OO(\e)}(\OO(1)+\OO(1)t)^{\OO(\eta\bar z_{B}-\D)}$,
			   and thus
			   $n_{aA}=\OO(n_{AB}\cdot{\Sigma}_{aA,aB} )={{\OO(\e)}(\OO(1)+\OO(1)t)^{\OO(\eta\bar z_{B}-\D)}}/{(\OO(1)+\OO(1)t)}$.
			   By solving $n_{aA}=\OO(n_{AA})=\OO(n_{AB}^2)$ we get the order of magnitude of $T_2=\OO\left(\e^{-1/(1+\eta\bar z_{B}-\D)}\right)$.
			   Note that for $\eta=0$,  ${\Sigma}_{aA,aB}(T_2)=\OO(\e ^{1+{\D}/{(1-\D)}})$. 
			   Moreover, \eqref{sigmadot} implies that for $\eta>\Delta/\bar z_B$, we have $\dot{\Sigma}_{aA,aB}>0$, which proves points \ref{eta0decay} and \ref{etalarge} of Theorem \ref{main-thm}.\\

				\item[\textbf{Phase 3.}] Time period:  until ${aa}$ reaches equilibrium.\\
				The fact that $n_{aA}=\OO(n_{AA})$ has a crucial effect on the birth rate of $aa$ 
				(see \eqref{birthratesaa}) since the term
				 $(n_{aa}+\frac12n_{aA})/({n_{aa}+n_{aA}+n_{AA}})$ becomes of order 
				 $\OO(1)$. As long as $AA$ stays smaller than $\OO(\e)$, we get a lower bound
				  on $n_{aa}$ which  grows exponentially fast  since $\ff$ is chosen large enough 
				  (Assumption C2):
				\begin{align}
				b_{aa}&\geq bn_{aa}\OO(1),\\
				d_{aa}&\leq n_{aa}(d+\D+\OO(\e)),\\
				\dot{n}_{aa}&\geq n_{aa}(\ff\OO(1)-d-\D-\OO(\e)).
				\end{align} 
				As $aa$ grows, it makes ${\Sigma}_{aA,aB}$ grow, and thus $AA$ and $AB$ as 
				well. We have to show  that this could not prevent $aa$ from reaching 
				equilibrium. We do not give a detailed argument here, but essentially, the 
				presence of the macroscopic $BB$ population prevents all the non-$aa$ 
				populations to grow too much. 
				Note that if $\eta$ is too large, then $aA$ could get a positive fitness and grow to
				 a macroscopic level. That is why we have to impose Assumption C3, which will 
				 become clearer heuristically in the next phase.
				We recall that $aa$ does not compete with $BB$ and thus it grows exponentially 
				fast  with rate ${\ff-(d+\Delta)}$ until an $\e_0$-neighbourhood of the  fixed point 
				where $aa$ and $BB$ coexist. The rate of growth corresponds to the invasion 
				fitness of $aa$ in the resident population $BB$, see \eqv(unfit.1). Note that, due to 
				Assumption C2, this rate is much larger than the invasion rate of $BB$ into $AA$. 
				That is why the fourth phase looks very steep on Figure \ref{detsys}, see the 
				stretched version on Figure \ref{zoom-in}. This phase lasts a time 
				$T_3=\OO(|\log\e |)$.\\
				
				\item[\textbf{Phase 4.}] 
				The Jacobian matrix of the field \eqref{dyn-syst} at the fixed point $p_{aB}$ has two zero, and 4 negative eigenvalues.
				 $p_{aB}$ is thus a non-hyperbolic equilibrium point of the system and linearisation fails to determine its stability properties. As in the case of the fixation of $AA$, one needs to analyse the quadratic 
				 approximation to show stability of this fixpoint. This is shown in 
				 \cite{BovCoqNeu18} to be true for $\eta $ small and $b$ large. 			 
				For higher values of $\eta$, numerical solutions show that the system converges to a fixed point where the 6 populations co-exist, but we do not prove this.
	\end{enumerate}

\begin{figure}[t]
	\includegraphics[width=.9\textwidth]{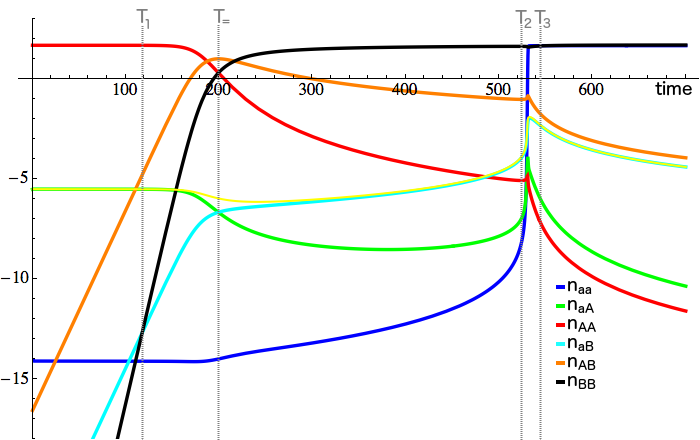}
	\caption{Numerical solution of the deterministic system for $\eta=0.02$, logplot. 
	}
	\label{detsys}
\end{figure}

The details of the proof are quite tedious and will not be given here. They can be found in 
the original paper \cite{BovCoqNeu18}.

\chapter{Phenotypic plasticity}
\newcommand{\frakn}{{n}}

\begin{chapquote}
{Thomas Malthus, \emph{Essay on the Principle of Population}}
{\frakfamily\fraklines {
The theory, on which the truth of this position depends,\\
appears to me so extremely clear; that I feel at a loss \\
to conjecture what part of it can be denied.}}
\end{chapquote}

\begin{chapquote}
{James D. Watson, \emph{Avoid boring people: Lessons from a Life in Science}}
{Avoid boring people!}
\end{chapquote}

So far, we have considered models where 
individuals are described by their phenotype, and where mutations occur at birth only. This means that we essentially equated phenotype and genotype. However, this is not always appropriate, since phenotypes may change without changes of the phenotype 
during the life cycle of an individual. 
 Phenotypic switches without mutations are certainly relevant in many if not most, biological systems (see e.g. \cite{bethedging,plasticity} and references therein). An interesting and relevant example that drew our attention to this are phenotype switches of cancer cells under treatment that lead to therapy escape. See, e.g. \cite{Baar2016stochastic,Holzel:2013ys} and references therein. A somewhat extreme form of phenotypic switches goes by the name of \emph{dormancy} and has received considerable attention recently, see \cite{Blath-Lennon} for a recent review.
This chapter is based on the paper \cite{BaaBov2018} with Martina Baar. \index{dormancy}

\section{The microscopic  model}\label{sec-micro-model}
Individuals 
are no longer only described by their phenotype, but by both 
their \emph{genotype} and their phenotype. Moreover, an individual of a given 
phenotype can express several phenotypes, and it can change 
its phenotype during its lifetime. The term \emph{phenotypic plasticity} is often 
used to describe a situation where the expression of the phenotype depends on the
environment in a deterministic way, whereas \emph{bet hedging} describes stochastic switches 
that "hedge" against future changes in the environment. 
\index{phenotypic plasticity}
\index{bet hedging}
\index{genotype}

Here, we take a broader look at a large class of models. By expanding the techniques of
 \cite{CM11} we prove  that  the microscopic process converges  on the evolutionary time scale to a generalisation of the Polymorphic Evolution Sequences (PES) 



Let   $\mathcal X$ be a finite set of the form  $\mathcal 
X=\mathcal G \times \mathcal P$, where $\mathcal G$ is the set of  
genotypes and
 $\mathcal P$ is the set of phenotypes. We call $\mathcal X$ 
the \emph{trait space} of the population. As usual, we introduce the \emph{carrying capacity} $K$. 
We model the time evolution of a population as an $\mathcal 
M^K(\mathcal X)$-valued, 
continuous time Markov  process $(\nu^K_t)_{t\geq 0}$. 
To account for the basic mechanisms of evolution and 
the phenotypic plasticity, we introduce the following parameters:
\begin{enumerate}
\setlength{\itemsep}{6pt}
\item $b(p)\in\mathbb R_+$ is the \textit{rate of birth} of an individual with phenotype $p\in\mathcal P$.
\item $d(p)\in\mathbb R_+$ is the \textit{rate of natural death} of an individual with with phenotype $p\in\mathcal P$.
\item$c(p,\tilde p)K^{-1}\in\mathbb R_+$ is the \textit{competition kernel}
		which models the competitive pressure an individual with phenotype  $p\in\mathcal P$ feels from an individual with phenotype  $\tilde p\in\mathcal P$ and is inversely proportional to the carrying capacity $K$.
\item $ s_{\text{nat.}}^g(p,\tilde p)\in\R_+$ 
	is the \emph{natural switch kernel} which models the natural switching
	 from  phenotype $p$  to $\tilde p$ of individuals with genotype $g$. 
\item $ s_{\text{ind.}}^g(p, \tilde p) (\hat p)K^{-1}\in\mathbb R_+$ is the \emph{induced switch kernel} which models the switching
	 from  phenotype $p$  to $\tilde p$ of individuals with genotype $g$ induced by an individual with phenotype $\hat p$.
\item$\mu_K m(g)$ is the \textit{probability that a mutation occurs at birth} from an individual with genotype $g\in\mathcal G$, where $\mu_K$ is a scaling parameter.
\item$M((g,p),(\tilde g, \tilde p))$ is the \textit{mutation law},  i.e.\ if a mutant is born from  an individual 
	 with trait $(g,p)$, then the mutant's trait is $(\tilde g,\tilde p)$ with probability  $M((g,p),(\tilde g,\tilde p))$. 
\end{enumerate}  
Note that most of the parameters depend on the phenotype only, and that we explicitly allow that individuals with different genotypes can express the same phenotype and conversely that individuals with the same genotype can express different phenotypes.

We assume that phenotypic switches can be enhanced by the environment, i.e. induced switching
 can only strengthen natural switching, and that phenotypic switches happen at a rate of order one, 
 i.e. they occur 
with positive probability during the lifetime of an individual. One might also be interested in cases where these rates are much smaller, i.e
tending to zero with $K$, but we do not consider this here.

\begin{assumption} \label{switch}  We assume that $ s_{\text{ind.}}^g(p, \tilde p) (\hat p)K^{-1}=0$, for all $\hat p\in \mathcal X$, whenever $ s_{\text{nat.}}^g(p,\tilde p)=0$, i.e. depending on the environment, the total switching rate can be larger or smaller but \emph{not} zero or non-zero.   
\end{assumption}
%
%
The generator  $\LL^K$ of the  population process $ \nu^K$ is a then given by 
\bea\nonumber
\left({\LL^K}\phi\right)(\nu )&=&\sum_{(g,p)\in\mathcal G\times \mathcal P}
		\left(\phi\left(\nu +\tfrac{\delta_{(g,p)}}K\right)-\phi(\nu )\right) 
			(1-\mu_K m(g))b(p)K\nu (g,p)   \\\nonumber
&+& \sum_{(g,p)\in\mathcal G\times \mathcal P} \sum_{(\tilde g, \tilde p)\in\mathcal G\times \mathcal P}
		\left(\phi\left(\nu +\tfrac{\delta_{(\tilde g,\tilde  p)}}K\right)-\phi(\nu )\right)  
		\mu_K m(g) M\big((g,p),  (\tilde g,\tilde p)\big) b(p) K\nu (g,p)\\\nonumber
&+& \sum_{(g,p)\in\mathcal G\times \mathcal P}
		\left(\phi\left(\nu -\tfrac{\delta_{(g,p)}}K\right)-\phi(\nu )\right) 
			\biggl(d(p)+\sum_{\tilde p\in\mathcal P} c(p,\tilde p)\nu  (\tilde p) 
					\biggr)K
						\nu(g,p)\\\nonumber
&+& \sum_{(g,p)\in\mathcal G\times \mathcal P}\:\sum_{\tilde p\in\mathcal P}
	\left(\phi\left(\nu +\tfrac{\delta_{(g,\tilde p)}}K-\tfrac{\delta_{(g,p)}}K\right)-\phi(\nu )\right)\nonumber\\
	&&\quad\quad \times
	\biggl(s_{\text{nat.}}^g(p,\tilde p)+\sum_{\hat p\in \mathcal P} s^g_{\text{ind.}}(p, \tilde p)(\hat p)\nu(\hat p)\biggr)\:
	 K\nu (g,p).
\eea
The only difference from the standard model is the presence of the fourth term that corresponds to the phenotypic switches. However, this term changes the dynamics substantially. In particular, the system of differential equations which arises in the large population limit without mutation ($\mu_K=0$) is not a generalised Lotka-Volterra system anymore, i.e.\ has \emph{not} the form $\dot{\frakn}= \frakn f(\frakn)$, 
where $f$ is linear in $\frakn$ (cf.\ Thm.\ \ref{det-limit} and Def.\ \ref{LVS}). 
%

\subsection{The Law of Large Numbers.}
If the  mutation rate is independent of $K$ and the initial conditions converge to a deterministic limit, then the sequence of rescaled processes, $(\nu_{}^K)_{K\geq 1}$, converges {in probability} as 
$K\uparrow  \infty$ to 
the solution of a system of ODEs. This is just as in the standard model.

\begin{theorem}\label{det-limit} Let  $\mu_K\equiv  {0}$.
	Suppose that the initial conditions converge 
	in probability to a deterministic
	limit, i.e.\ $\nu^K_0\xrightarrow{p}\nu_0$ as $K\uparrow \infty$ , where $\nu^{}_0$ is a finite measure on $\mathcal X$. Assume further that $\lim_{K\uparrow\infty}=0$. 
	Then, for every $T>0$, exists a deterministic function 
	$\xi\in C([0,T],\mathcal M_F(\mathcal X))$ such that for all $\eta>0$, 
	\be
	\lim_{K\uparrow \infty}\mathbb P\Big[\sup_{t\in[0,T]}\big|\big| \nu^K_t-\xi_t  
	\big|\big|^{}_{\text{TV}}>\eta\Big]=0 ,\qquad 
	\ee
	where $||\,.\,||{}_{\text{TV}}$ is the total variation norm.
	Moreover, let $\frakn$ be the
	unique solution to the  dynamical system 
\begin{align}
	\quad\dot{\frakn}_{(g,p)}(t)
	\:=&\;
	\frakn_{(g,p)}(t)
      	  \:\Biggl(\bigl(b(p)-d(p)
		 -\sum_{(\tilde g, \tilde p)\in \mathcal G\times\mathcal P}  c(p,\tilde p) \frakn_{ (\tilde g,\tilde p)}(t)\\\nonumber&\quad\quad
		-\sum_{\tilde p\in\mathcal P} \biggl( s_{\text{nat.}}^g(p,\tilde p)+\sum_{(\hat g,\hat p)\in\mathcal G\times \mathcal P} s^g_{\text{ind.}}(p, \tilde p)(\hat p) \frakn_{ (\hat g,\hat p)}(t)\biggr)
		\Biggr) \nonumber\\\nonumber
	&+\quad\sum_{\tilde p\in\mathcal P}  \;\;\frakn_{(g,\tilde p)}(t)\:\biggl( s_{\text{nat.}}^g(\tilde p, p)+\sum_{(\hat g,\hat p)\in\mathcal G\times \mathcal P} s^g_{\text{ind.}}(\tilde p,  p)(\hat p) \frakn_{ (\hat g,\hat p)}(t)\biggr)
	\end{align}
 with initial condition $\frakn_x(0)=\nu_0(x)$,  for all $x\in \mathcal X$.\\[0.5em]			
Then, $\xi$ is given as $\xi_t=\sum_{x\in \mathcal X}\frakn_x(t)\delta_x$.
\end{theorem}

\section{The interplay between rare mutations and fast switches.}\label{sec-PES}
 As in Chapter \ref{chapter5}, 
we place ourselves under the assumptions 
\begin{equation}\label{conditions}
\forall V>0, \qquad \exp(-VK)\ll \mu_K \ll \frac{1}{K\ln K}, \qquad \text{as } K\uparrow\infty.
\end{equation}

The key element in the proof of the convergence to the PES 
 is a precise analysis of how a mutant population fixates. 
A crucial assumption  is that  the 
Lotka-Volterra systems that describe the large population limit 
always have a unique stable 
fixed point $ \bar \frakn$. 

  Given a population in a stable equilibrium  that populates a certain set
  of traits, say $M\subset \mathcal X$, the invasion fitness $f(x,M)$ is the growth rate of a population consisting of a single individual 
  with trait $x\not\in M$ in the presence of the equilibrium population $ \bar\frakn$ on $M$. 
%
Since switches between phenotypes associated with the same genotype 
happen at times of order one, the growth rate of the initial mutant 
phenotype alone does not determine the probability of fixation. 
We have to adapt the notion of invasion fitness to the case where fast phenotypic switches are present.

\subsection{ Lotka-Volterra systems with phenotypic plasticity.}

Assume that the initial condition is supported on $d$ traits,
\be
(\mathbf g,\mathbf p)=((g_1,p_1),\ldots, 
(g_d,p_d))\in(\mathcal G\times\mathcal P)^d.
\ee 
Since $\mu\equiv 0$, no new genotype can appear in the population process $\nu^K$. 
Moreover, not every genotype can express every phenotype.

For all $g\in \mathcal G$, let $X^g$ be  a  stationary discrete-time Markov chain with state space  
$\mathcal P$ and  transition probabilities  (Here we may take $s^g=s_{nat}^g$)
\be
\P\left(^g_i=\tilde p \:|\: X^g_{i-1}=p\right)=\frac{s^g(p,\tilde p)}{\sum_{\hat p \in \PP}s^g(p, \hat p)}, \quad\text{ if }\sum_{\hat p \in \mathcal P}s^g(p, \hat p)>0
\ee
and 
\be
\P\left(X^g_i= p \:|\: X^g_{i-1}=p\right)=1, \quad\text{ if }\sum_{\hat p \in \PP}s^g(p, \hat p)=0.
\ee
The Markov chains $\{X^g, g\in\mathcal G\}$ contain only partial 
information on the switching behaviour of the process $\nu^K$, but this is the key information needed later.

In the sequel, we work under the following simplifying assumption:
\begin{assumption}\label{recurrence} For all $g\in \mathcal G$, all communicating classes of $X^g$ are recurrent. 
\end{assumption}

 We denote the communicating class 
associated with $(g,p)\in\mathcal G\times \mathcal P$ by $[p]_g$. This is 
the communicating class of $X^g$ which contains $p$, i.e.\ $p$ can be 
seen as a representative of the class.
\begin{figure}[h!]
	\centering
	 \includegraphics[width=0.5\textwidth]{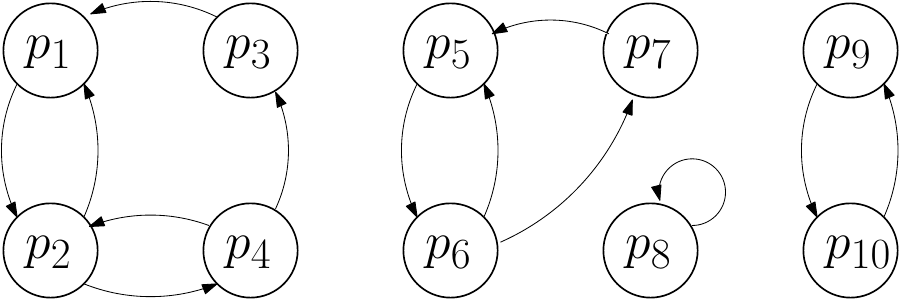}
	 \caption{\small Example of a Markov chain $X^g$. Here, $\mathcal 
	 P=\{p_1,\ldots,p_{10}\}$ and 
	 $X^g$ has four communicating classes:
	 $\{p_1,p_2,p_3,p_4\},\{p_5,p_6,p_7\},\{p_8\},\{p_9,p_{10}\}$. The 
	 class $\{p_8\}$ has only one element, i.e.\ $\sum_{i=1}^{10}
	 s^g(p_8,p_i)=0$ in this example. }	  
\end{figure}
By Assumption \ref{switch},
this ensures that if we start with a large enough population consisting only 
of individuals carrying the same trait $(g,p)$, then, after a time of order $1$, all 
phenotypes in the class $[p]_g$ will be present in the population, but 
none of the other classes. 
Thus, $[p]_{g}$ is the set of phenotypes which are reachable in the 
Markov chain $X^g$ with $X^g_0=p$ and
\emph{the set of traits which can appear in the population process $
\nu^K$} is given by
\be
\mathcal X_{(\mathbf g,\mathbf p)}\equiv \bigcup_{i=1}^d\{g_i\}\times 
[p_i]_{g_i}.
\ee
With this notation, $\xi$ is given by $\xi(t)=\sum_{x\in \mathcal X_{(\mathbf g,\mathbf p)}}  {\frakn}_{x} (t)\delta_{x}$, where $ {\frakn}$ is the solution of the \emph{competitive Lotka-Volterra system with phenotypic plasticity} defined below. 
\index{Lotka-Volterra system with phenotypic plasticity}

\begin{definition} \label{LVS}
For any $ (\mathbf g,\mathbf p)\in(\mathcal G\times\mathcal P)^d$, we denote by $LVS(d,  (\mathbf g,\mathbf p))$
the  \emph{competitive Lotka-Volterra system with phenotypic plasticity}. This is an $|\mathcal X_{(\mathbf g,\mathbf p)}|$-dimensional system of ODEs given 
 by  
\begin{align}\label{deterministic-system-initial}
\dot{  {\frakn}}_{(g,p)}
=\:& \frakn_{(g,p)} 
	\biggl( b(p)-d(p)
	-\!\!\sum_{(\tilde g,\tilde p)\in \mathcal X_{(\mathbf g,\mathbf p)}}\!\!c(p,\tilde p)\frakn_{(\tilde g,\tilde p)}
\\
&	-\!\! \sum_{\tilde p \in[p]_{g}}
		\biggl( s_{\text{nat.}}^g(p,\tilde p)+\!\!\sum_{(\hat g,\hat p)\in\mathcal X_{(\mathbf g,\mathbf p)}} 
				\!\!s^g_{\text{ind.}}(p, \tilde p)(\hat p) \frakn_{ (\hat g,\hat p)}(t)
		\biggr)
	 \biggr)\nonumber\\\nonumber
    &+\!\! \sum_{\tilde p \in[p]_{g}}\frakn_{(g,\tilde p)}
    		\biggl( s_{\text{nat.}}^g(\tilde p, p) +\!\!\sum_{(\hat g,\hat p)\in\mathcal X_{(\mathbf g,\mathbf p)}}
				 \!\!s^g_{\text{ind.}}(\tilde p,  p)(\hat p) \frakn_{ (\hat g,\hat p)}(t)
		\biggr), \quad (g, p)\in\mathcal X_{(\mathbf g,\mathbf p)}.
\end{align}
\end{definition} 

We introduce the notation of coexisting traits in this context. \index{coexisting traits}
\begin{definition} 
For any $d\geq 2$, we say that the distinct traits  $(g_1,p_1),\ldots ,(g_d, p_d)$ \emph{coexist} if the system $LVS(d,  (\mathbf g,\mathbf p))$ has a unique non-trivial equilibrium $\bar{\frakn} (\mathbf g,\mathbf p)\in (0,\infty)^{|\mathcal X_{(\mathbf g,\mathbf p)}|}$ which is locally strictly stable,  that is, all eigenvalues of the Jacobian matrix of the system $LVS(d,  (\mathbf g,\mathbf p))$ at $\bar{\frakn} (\mathbf g,\mathbf p)$ have strictly negative real parts. 
\end{definition}

Note that if $(g_1,p_1),\ldots ,(g_d, p_d)$ coexist, then all traits of $\XX_{(\mathbf g,\mathbf p)}$ coexist and the equilibrium $\bar{\frakn} (\mathbf g,\mathbf p)$ is asymptotically stable.
We will prove later that if these traits  coexist, then the \emph{invasion probability} of a mutant trait $(\tilde g, \tilde p)$ which appears in the resident population $\XX_{(\mathbf g,\mathbf p)}$ close to  $\bar{\frakn} (\mathbf g,\mathbf p)$ is given by the function
\be\label{invasion probability} 
1-q_{{(\mathbf g,\mathbf p)}}(\tilde g,\tilde p),
\ee
 where $q_{(\mathbf g,\mathbf p)}(\tilde g,\tilde p)$ is given as  follows: 
Let us denote the elements of $[\tilde p]_{\tilde g}$ by $\tilde p_1,\tilde p_2,\ldots,\tilde p_{ |[\tilde p]_{\tilde g}|}$. 
Then,  $q_{(\mathbf g,\mathbf p)}(\tilde g,\tilde p)$ is the  $\tilde p$-component of the smallest nonnegative solution of 
\be\mathbf 
 u(\mathbf y)=\mathbf 0,
 \ee
 where $\mathbf  u$ is a map from $\R^{ |[\tilde p]_{\tilde g}|}$ to $\R^{ |[\tilde p]_{\tilde g}|} $ defined, for all $i\in \{1,\ldots ,|[\tilde p]_{\tilde g}|\}$, by
\begin{align}
u_i (\mathbf y)\equiv& 
	\; b(\tilde p_i)\:y_i^2
	+\sum_{j=1}^{ |[\tilde p]_{\tilde g}|}
		 \biggl( s_{\text{nat.}}^{\tilde g}(\tilde p_i, \tilde p_j) +\!\!\sum_{(g, p)\in\mathcal X_{(\mathbf g,\mathbf p)}}
				 \!\!s^{\tilde g}_{\text{ind.}}(\tilde p_i, \tilde p_j)( p) \bar \frakn_{ ( g, p)}
		\biggr)\: y_j
	\\\nonumber&
	+d(\tilde p_i) 
	+ \sum_{(g,p)\in\XX_{(\mathbf g,\mathbf p)}} 
		c(\tilde p_i, p)\bar{\frakn}_{(g,p)}(\mathbf g,\mathbf p) \\\nonumber
	&-\:\Bigg( b(\tilde p_i)+\sum_{j=1}^{ |[\tilde p]_{\tilde g}|}
		 \biggl( s_{\text{nat.}}^{\tilde g}(\tilde p_i, \tilde p_j) +\!\!\sum_{(g, p)\in\mathcal X_{(\mathbf g,\mathbf p)}}
				 \!\!s^{\tilde g}_{\text{ind.}}(\tilde p_i, \tilde p_j)( p) \bar \frakn_{ (g, p)}
		\biggr)\\\nonumber
		&+d(\tilde p_i) +\!\! \sum_{(g,p)\in\XX_{(\mathbf g,\mathbf p)}} \!\!c(\tilde p_i, p)\bar{\frakn}_{(g,p)}(\mathbf g,\mathbf p) \Bigg)\:y_i.
\end{align}
In fact, $(1-q_{{(\mathbf g,\mathbf p)}}(\tilde g,\tilde p))$ is the (asymptotic as $K\uparrow\infty$) probability that a single mutant reaches a population size of order $K$ in a resident population with traits $\mathcal X_{(\mathbf g,\mathbf p)}$. This is shown by approximating the mutant population with multi-type branching processes.
 The function $(1-q_{(\mathbf g,\mathbf p)}(\tilde g,\tilde p))$ plays the same role as the function
$[f(y;\mathbf x)]_+ / b(y)$ in the standard case.

Note that $\mathbf u$ is the analog to the function $g(s)-s$ in the standard branching process theory 
(see Chapter 2, Section 6), for a multi-state branching process.

The behaviour of the LV system can be very complex. For simplicity, we make assumptions that exclude 
  cycles, unstable equilibria, or chaotic dynamics in the  
deterministic system.

\begin{assumption}\label{conv_to_fixedpoint-2}
For any given traits $(g_1,p_1),\ldots ,(g_d, p_d)\in \mathcal G\times \mathcal P$ 
that coexist and, for any mutant trait $(\tilde g,\tilde p)\in\mathcal X\setminus \mathcal X_{(\mathbf g,\mathbf p)}$  
such that $ 1-q_{{(\mathbf g,\mathbf p)}}(\tilde g,\tilde p)>0$, there exists a neighbourhood $U\subset \mathbb R^{|\mathcal X_{(\mathbf g,\mathbf p)}|+ |[\tilde p]_{\tilde g}|}$ of  $(\bar{\frakn}(\mathbf g,\mathbf p),0,\ldots,0)$
 such that all solutions of  $LVS(d+1,  ((\mathbf g,\mathbf p),(\tilde g,\tilde p)))$ 
with initial condition in $U\cap (0,\infty)^{|\mathcal X_{(\mathbf g,\mathbf p)}|+ |[\tilde p]_{\tilde g}|}$
converge as $t \uparrow  \infty$ to a unique locally strictly stable
 equilibrium in $\mathbb R^{|\mathcal X_{(\mathbf g,\mathbf p)}|+ |[\tilde p]_{\tilde g}|}$ denoted by $\frakn^{*}((\mathbf g,\mathbf p),(\tilde g,\tilde p))$.
\end{assumption}

We write $\frakn^{*}$ and not $\bar{\frakn}$ to emphasise that 
some components of $\frakn^{*}$ can be zero. We use the 
shorthand notation
$((\mathbf g,\mathbf p),(\tilde g,\tilde p))$, for $((g_1,p_1),\ldots ,(g_d, p_d),(\tilde g,\tilde p))$. 
Assumption \ref{conv_to_fixedpoint}  does not have to hold for all traits in $\mathcal X\setminus \mathcal X_{(\mathbf g,\mathbf p)}$, but  only for those traits $(\tilde g,\tilde p)$ which can appear in the resident population by mutation, i.e.\ only if $\sum_{(g,p)\in\mathcal X_{(\mathbf g,\mathbf p)}}m(g)M((g,p),(\tilde g,\tilde p))$ is positive.

\subsection{Convergence to the generalised PES. } \label{sec-Thm}
The following theorem establishes the convergence to a PES. 

\begin{theorem}\label{PESwP} Suppose that Assumptions \ref{switch}, \ref{recurrence} and  \ref{conv_to_fixedpoint-2} hold  and that the mutation rates satisfy \eqv(conditions).
Assume that $(g_1,p_1),\ldots, (g_d,p_d)$ are coexisting traits and that the initial conditions have support $\mathcal X_{(\mathbf g,\mathbf p)}$ and 
converge almost surely to $\bar{\frakn} (\mathbf g,\mathbf p)$. 
Then, the sequence of the rescaled processes $(\nu^{K}_{ t/ K\mu_K})_{t\geq 0}$, generated by $ {\LL}^K$ with 
initial state $\nu_0^K$, converges in the sense of finite-dimensional distributions to the measure-valued pure 
jump process $\Lambda$ defined as follows:
$\Lambda_0=\sum_{(g,p)\in\mathcal X_{(\mathbf g,\mathbf p)}} \bar{\frakn}_{(g,p)} (\mathbf g,\mathbf p) \delta_{(g,p)}$
 and the process $\:\Lambda\:$ jumps, for all $(\hat g,\hat p)\in \mathcal X_{(\mathbf g,\mathbf p)}$,  from
 \be 
 \sum_{(g,p)\in\mathcal X_{(\mathbf g,\mathbf p)}} \bar{\frakn}_{(g,p)} (\mathbf g,\mathbf p) \delta_{(g,p)}
\quad \text{ to } \quad
  \sum_{(g,p)\in\mathcal X_{((\mathbf g,\mathbf p),(\tilde g,\tilde p))}} \frakn^*_{(g,p)} ((\mathbf g,\mathbf p),(\tilde g,\tilde p)) \delta_{(g,p)}
  \ee
  with an infinitesimal rate
  \be
  m(\hat g)b(\hat p) \bar{\frakn}_{(\hat g,\hat p)} (\mathbf g,\mathbf p) (1-q_{{(\mathbf g,\mathbf p)}}(\tilde g,\tilde p))M((\hat g,\hat p), (\tilde g,\tilde p)).
  \ee
\end{theorem}

 We call $\Lambda$  a \emph{Polymorphic Evolution Sequence with Phenotypic Plasticity (PESP)}.

The strategy of the proof is similar to 
the one given in \cite{CM11}.
The population is either in a stable phase or in an invasion phase. Until the first mutant appears, the 
population is in a stable phase, i.e.\ the population stays close to a given equilibrium.
From the first mutation event until the population reaches again a stable state, the population is in an
 invasion phase. In fact, the mutant either survives and the population quickly reaches a new stable state (where 
 the mutant trait is present) or the mutant goes extinct and the population is again in the old stable state. After 
 this, the population is again in a stable phase until the next mutation, etc..

Note that the invasion phases are short $(O(\ln(K)))$ compared to the stable phase $(O(1/\mu_K K))$. Since we study the process on the time scale $1/K\mu_K$, the limit process is a pure jump process which jumps from one stable state to another.
\\[0.5em]
\emph{The stable phase:} Fix $\epsilon> 0$.  Let $\mathcal X_{(\mathbf g,\mathbf p)}$ be the support of the initial conditions. 
For large $K$, the population process 
 $\nu^K$ is, with high probability, still in a small neighbourhood of the equilibrium $\bar{\frakn}(\mathbf g,\mathbf 
 p)$ when the first mutant appears. 
In fact, using large deviation results on the problem of exit from a domain (cf.\ \cite{FW84}), we obtain that 
there exists a constant $M>0$ such that the first time  $\nu^K$ leaves the $M\epsilon$-neighbourhood of 
$\bar{\frakn}(\mathbf g,\mathbf p)$  is bigger than $\exp(V K)$, for some $V > 0$, with high probability.
Thus, until this stopping time, mutations born from individuals with trait $x\in\mathcal X_{(\mathbf g,\mathbf p)}$  appear with a rate which is close to 
\be
 \mu_K m(x)b(x) K\bar{\frakn}_x(\mathbf g,\mathbf p).
\ee 
The 
condition $ 1/( {K\mu_K})\ll \exp (V K)$, for all $V>0$, in  (\ref{Conv_Cond})
  ensures that the first mutation appears before this exit time.
%

\emph{The invasion phase:} We divide the invasion of a given mutant trait 
$(\tilde g,\tilde p)$ into \emph{three steps}.
In the first step, from a mutation event until the mutant population goes extinct or the mutant density reaches the value $\e$, the number of mutant individuals is
small 
As in Chapter \ref{chapter5}, 
 the resident population stays close to its equilibrium density $\bar{\frakn}(\mathbf g,\mathbf p)$ during this step. Using similar arguments as  Champagnat et al.\ \cite{Cha06,CM11}, we prove that
 the mutant population  is well approximated  by a $|[\tilde p]_{\tilde g}|$-type
branching process $Z$, as long as the mutant population 
has less than $\epsilon K$ individuals. More precisely, let  us denote the elements of $[\tilde p]_{\tilde g}$ by $\tilde p_1,\ldots,\tilde p_{|[\tilde p]_{\tilde g}|} $, then,
for each  $1\leq i\leq |[\tilde p]_{\tilde g}|$, each individual in $Z$ (carrying trait $(\tilde g,\tilde p_i)$) undergoes
\begin{enumerate}
\setlength{\itemsep}{6pt}
 \item {birth (without mutation) with rate $b(\tilde p_i)$,}
\item {death with rate $ d(\tilde p_i)  + \sum_{(g,p)\in\mathcal X_{(\mathbf g,\mathbf p)}}c(\tilde p_i, p) 
\bar{\frakn}_{(g,p)}(\mathbf g,\mathbf p)$ and}
\item{ switch to  $\tilde p_j$ with rate $s_{\text{nat.}}^{\tilde g}(\tilde p_i, \tilde p_j) +\!\!\sum_{( g, p)\in\mathcal X_{(\mathbf g,\mathbf p)}}
				 \!\!s^{\tilde g}_{\text{ind.}}(\tilde p_i, \tilde p_j)( p) \bar \frakn_{ ( g, p)}$, for all $1\leq j\leq |[\tilde p]_{\tilde g}|$.}
\end{enumerate} 
This continuous-time multi-type branching process is supercritical if and only if the largest eigenvalue of its infinitesimal generator
 $\lambda_{\text{max}}$  is larger than zero. 
Hence, 
the mutant invades with positive probability if and only if $\lambda_{\text{max}}>0$. Moreover, the probability that the density of the mutant's genotype, $\nu^K(\tilde g)$, reaches $\epsilon$ at some time $t_1$ is close to the probability that the multi-type branching process reaches the total mass $\epsilon K$, which converges as $K\uparrow \infty$ to $(1-q_{{(\mathbf g,\mathbf p)}}(\tilde g,\tilde p))$. 

In the second step, we obtain as a consequence of  Theorem \ref{det-limit} that once the mutant density has reached $\epsilon$, for large $K$, the stochastic process $\nu^K$ can be approximated on any finite time interval by the solution of $LVS(d+1,  ((\mathbf g,\mathbf p),(\tilde g,\tilde p)))$
with a given initial state. By Assumption \ref{conv_to_fixedpoint}, this solution reaches the $\epsilon$-neighbourhood of its new equilibrium $\frakn^*((\mathbf g,\mathbf p),(\tilde g,\tilde p))$ in finite time.  Therefore, for large $K$, the stochastic process $\nu^K$ 
also reaches with high probability
the $\epsilon$-neighbourhood of $\frakn^*((\mathbf g,\mathbf p),(\tilde g,\tilde p))$ at some finite ($K$ independent) time $t_2$.

In the third step, one uses similar arguments as in the first. Since $\frakn^*((\mathbf g,\mathbf p),(\tilde g,\tilde p))$ is a
strongly locally stable equilibrium (Ass. \ref{conv_to_fixedpoint}), the 
stochastic process $\nu^K_t$ stays close $\frakn^*((\mathbf g,\mathbf p),
(\tilde g,\tilde p))$ and we can approximate the densities of the traits $
(g,p)\in\mathcal X_{((\mathbf g,\mathbf p),(\tilde g ,\tilde p))}$
 with $n_{(g,p)}^*((\mathbf g,\mathbf p),(\tilde g,\tilde p)) = 0$ by $|[p]_g|$-
 type branching 
 processes which are subcritical and therefore become extinct, a.s.. 
 \\[0.5em]
The duration of the first and third steps is 
proportional to $\ln(K)$, whereas 
the time of the second step is bounded.
Thus, the second inequality in \eqv(Conv_Cond) guarantees that, with 
high probability, 
the three steps of invasion are completed before a new mutation occurs.
 After the last step, the process is again back in a stable phase, but with a 
 possibly different resident population, until the next mutation happens.

 \section{Multi-type branching processes and the invasion of the mutant}
 
 The main novel aspect in the proof is the treatment of the invasion phase. Due to the presence of phenotypic 
 switches, the invading mutant has to be approximated by a multi-type branching process. 
   We collect some properties about multi-type continuous-time branching 
processes. Most of these can be found in \cite{AN} or \cite{SEW}. The limit 
theorems we need in the sequel were first obtained  by Kesten and Stigum 
\cite{KesSte1,KesSte2,KesSte3} in the discrete-time case and by Athreya \cite{A_MTBP} in the 
continuous-time case. 

Let $Z(t)$ be a  $k$-dimensional continuous-time branching process. Assume that $Z(t)$ is 
non-singular and that the first moments exist.
Then, the  so-called \emph{mean matrix}  $\mathbf M(t)$ of  $Z(t)$ is the $k\times k$ matrix with elements 
\be
 m_{ij}(t)\equiv \mathbb E[Z_j(t)|Z(0)=\mathbf e_i ], \quad 1\leq i,j\leq k,
\ee
where $\mathbf e_i$ is the $i$-th unit vector in $\mathbb R^k$.
It is well known (cf.\  \cite{AN} p.\ 202) that there exists a matrix $\mathbf A$, called the infinitesimal generator of the semigroup, such that for all $t\in\R_+$,
\be\label{exp_dar}
\mathbf M(t)\equiv \exp(\mathbf A t).
\ee
If we denote by $\M$ the transition matrix of the corresponding discrete-time processes and by $\mathbf R$ the diagonal matrix 
whose diagonal entries are the branching rates of the different types, we can represent the generator $\mathbf A$ as
(see \eqv(disc-gen) in Chapter \ref{chapter2})
\be\Eq(inf.gen)
\mathbf A= \mathbf R ( \mathbf M- \mathbf I).
\ee

Under the basic assumption of \emph{positive regularity}, i.e.  that there exists a time $t_0$ such that 
$\mathbf M(t_0)$ has strictly  positive entries, the Perron-Frobenius theory asserts that
\begin{enumerate}
\setlength{\itemsep}{6pt}
\item the largest eigenvalue of $\mathbf M(t_0)$ is real-valued and  strictly positive,
\item the algebraic and geometric multiplicities of this eigenvalue are both one, and 
\item the corresponding eigenvector has strictly positive components.
\end{enumerate}
 By (\ref{exp_dar}), the eigenvalues of $\mathbf M(t)$ are given by $
 \exp({\lambda_i t})$, where 
  $\{\l_i; i=1,\ldots, k\}$ are the eigenvalues of  $\mathbf A$, and both 
  matrices have the same eigenvectors, which implies that the left and right 
  eigenvectors  
 $\mathbf u$ and $\mathbf v$ of  $\l_{\text{max}}(\mathbf A)$ can be chosen with strictly positive 
 components and satisfying 
 \be
\textstyle \sum_{i=1}^k v_iu_i=1\quad \text{and} \quad\sum_{i=1}^k u_i=1.
\ee
 The process $Z$ is called supercritical, critical, or subcritical according as $\l_{\text{max}}(\mathbf A)$ is larger, equal, or smaller than zero. 
 
Observe that the following properties are equivalent  (cf.\  \cite{SEW} p. 95-99 and \cite{DISS_CTMTBP}):\\[0.5em]
$Z$ is irreducible  \; $\Leftrightarrow$\; $\mathbf M$ is irreducible \; $\Leftrightarrow$\; 
 $\mathbf A$ is irreducible \; $\Leftrightarrow$\; 
$\mathbf M(t)$ is irreducible for all $t > 0$ \; $\Leftrightarrow$\; 
 $\mathbf M(t)>0$ for all $t>0$.\\[0.5em]
In particular, irreducible implies positive regular. 
Note that a matrix is irreducible if it is not similar via a permutation to a block upper triangular matrix and that a Markov chain is irreducible if and only if the mean matrix is irreducible.

The following lemma is an extension of Theorem 4 of \cite{Cha06} for 
multi-type branching processes.
\begin{lemma}\label{expansion of Thm 4}
Let $(Z(t))_{t\geq0}$ be a non-singular and irreducible $k$-dimensional 
continuous-time Markov branching process and $\mathbf q$ the 
extinction vector of $Z$, i.e.
\be 
q_i\equiv \P\left(Z(t)=0, \text{ for some }t\geq0 | Z(0)=\mathbf e_i\right), 
\quad \text{ for $1\leq i\leq k$}.\ee Furthermore, let $(t_N)_{N\geq 1}$ be 
a sequence of positive numbers such that $\ln(N)\ll t_N$, define $T_{\rho}
\equiv\inf \{t\geq 0: \sum_{i=1}^k Z_i(t)=\rho\}$ and
assume that, for all $i,j\in\{1,\ldots, k\}$ and $t\in[0,\infty)$,  
\be
\mathbb E\left[Z_j(t) \ln(Z_j(t))| Z(0)=\mathbf e_i\right]<\infty.
\ee
 \begin{enumerate}
 \item{If $Z$  is \emph{subcritical}, i.e.\ $\l_{\text{max}}(\mathbf A)<0$, 
 then, for any $\epsilon>0$,
 \be\label{suredeath}
 \lim_{N\uparrow \infty}\mathbb P\left(T^{}_0\leq t_N\wedge T^{}_{\lceil \e 
 N\rceil }\:\middle |\: Z (0)=\mathbf e_i \right )=1, \qquad \text{ for all }i \in 
 \{1, \ldots, k\},
 \ee
 and 
  \be\label{(i)2}
 \lim_{N\uparrow \infty}\,\inf_{ \mathbf x\in \partial B_{\e N}  }\mathbb P\left(T^{}_0\leq 
 t_N\,\middle|\, Z (0)= \mathbf x\right)=1,\quad\text { where }  \partial B_{\e N}
 \equiv\{ \mathbf x\in \mathbb N_0^k:  \textstyle\sum_{i=1}^{k} x_i =\lceil \e N\rceil \}.
 \ee
Moreover, for 
$\bar u=  \frac{\max_{1\leq i\leq k}u_i}{\min_{1\leq j\leq k}u_j}$ and, for any $\e>0$,
\be\label{(i)3}
\quad\lim_{N\uparrow \infty }
\sup_{\mathbf x \in B_{\e^2 N} }\mathbb P\left( T^{}_{\lceil \e N\rceil }\leq T^{}_0\:\middle|\:  Z (0)= \mathbf x\right)\leq\bar u \e,\quad \text {where }  B_{\e^2 N}\equiv\{ \mathbf x\in \mathbb N_0^k: \textstyle\sum_{i=1}^{k} x_i \leq\lceil \e^2N\rceil \}.\ee
}
  \item{If $Z$  is \emph{supercritical}, i.e.\ $\l_{\text{max}}(\mathbf A)>0$, then, for any $\epsilon>0$ (small enough),
 \be\label{(ii)1}
 \lim_{N\uparrow \infty}\mathbb P\left(T^{}_0\leq t_N\wedge T^{}_{\lceil \e N\rceil }\:\middle|\: Z (0)=\mathbf e_i \right)=q_i, \qquad \text{ for all }i \in \{1, \ldots, k\}
 \ee
 and 
  \be\label{(ii)2}
 \lim_{N\uparrow \infty}\mathbb P\left(T^{}_{\lceil \e N\rceil }\leq t_N\:\middle|\: Z (0)=\mathbf e_i \right)=1-q_i,\qquad\text{ for all }i \in \{1, \ldots, k\}.
 \ee
Moreover, conditionally on survival, the proportions of the different types present in the population converge
almost surely,  as $t \uparrow  \infty$, to the corresponding ratios of the components
of the eigenvector:  for  all $ i=1,\ldots,k,$,
\begin{equation}
\label{limit-ratio}
\lim_{t \uparrow  \infty}\frac{Z_{i}(t)}{\sum_{j=1}^{k}Z_j(t)}=\frac{v_{i}}{\sum_{j=1}^{k}v_j},
\quad \text{a.s. on } \{T_0=\infty\}.
\end{equation}
 }\end{enumerate}
\end{lemma}
\begin{remark} In the critical case, when $\l_{\text{max}}(\mathbf A)=0$, 
\eqref{suredeath} still holds, i.e. a mutant will not fixate. But Eq. \eqref{(i)2} does not hold, 
since the time to extinction is now much larger. This leads to complications that 
would obfuscate the picture.  Assumption 3 on the existence of a unique strictly stable equilibrium  
excludes the occurrence of these critical situations.
\end{remark}
\begin{proof}
We start with the proof of
(i).  Since $Z(t)$ is in this case a subcritical irreducible continuous-time branching process and $\mathbb E[Z_j(t) \ln(Z_j(t))| Z(0)=\mathbf e_i]<\infty$, we obtain by applying Satz 6.2.7 of \cite{SEW} the existence of a constant $C>0$ such that
 \be\label{Lim_Extinct}
 \lim_{t \uparrow  \infty} \frac{1-q_i(t)}{\eee^{\l_{\text{max}}(\mathbf A) t}}=C  u_i,
 \ee 
 where $q_i(t)\equiv\mathbb P\left(Z(t)=0 \: |\: Z(0)=\mathbf e_i\right)$.  Moreover, we have a non-explosion condition. Thus, for all $\epsilon>0$, either $T^{}_{\lceil \e N\rceil }$ equals infinity or it converges to infinity as $N \uparrow \infty$. Putting both together, there exists a sequence $s_N$ with $\lim_{N \uparrow  \infty} s_N=+ \infty$ such that  
 \be
 \lim_{N\uparrow \infty}\mathbb P\left(T^{}_0\leq t_N\wedge T^{}_{\lceil \e N\rceil }\middle | Z (0)=\mathbf e_i \right )\geq 
  \lim_{N\uparrow \infty}\mathbb P\left(T^{}_0\leq s_N\:\middle |\: Z (0)=\mathbf e_i \right )= \lim_{N\uparrow \infty} q_i(s_N)=1.
 \ee
The branching
property implies that,  for all $\mathbf x\in \mathbb N^k$, $\mathbb P(Z(t) = 0\:|\: Z(0)= \mathbf x)= \prod_{i=1}^{k}(q_i(t))^{x_i}$ (cf.\ \cite{DISS_CTMTBP} p.\ 25).
So, we get
  \be\label{for (i) 2} 
  \inf_{\mathbf x\in\partial B_{\e N} }\!\mathbb P\left(T^{}_0\leq t_N\:\middle |\: Z (0)=\mathbf x \right )= 
  \inf_{\mathbf x\in \partial B_{\e N} }\!\mathbb P\left(Z(t_N)=0\:\middle |\: Z (0)=\mathbf x \right )=
  \inf_{\mathbf x\in \partial B_{\e N} } \prod_{i=1}^{k}(q_i(t_N))^{x_i}.
\ee
For all $i\in\{1,\ldots, k\}$, $1\geq (q_i(t_N))^{x_i}\geq(q_i(t_N))^{\lceil\e N\rceil} $ and by (\ref{Lim_Extinct}) we have $1-q_i(t_N)=O(\eee^{{\l_{\text{max}}(\mathbf A) t_N}})$. Moreover,  
for any sequence $ (w_N)_{N\geq 1}$ such that  $\lim_{N\uparrow \infty}w_N=0$, 
\be 
\lim_{N\uparrow \infty} \{1+\frac{w_N}{N} \}^N=1.
\ee 
This implies that, for all $t_N$ with $t_N\gg\ln(N)$ and $C>0$, since  $
\lim_{N\uparrow \infty} C \eee^{\l_{\text{max}}(\mathbf A) t_k}\lceil\e 
N\rceil=0$,
\be\lim_{N\uparrow \infty }(1-C\eee^{\l_{\text{max}}(\mathbf A) 
t_k})^{\lceil\e N\rceil}=1.\ee 
Thus, taking the limit $N\uparrow  \infty$ in (\ref{for (i) 2}), we obtain the 
desired equation (\ref{(i)2}). To prove the  inequality (\ref{(i)3}) we use the 
fact that 
$(\sum_{i=1}^k u_iZ_i(t))\eee^{-\lambda_{\text{max}}t}$ is a martingale (cf.\ \cite{A_MTBP}, Prop.\  2).  
By applying Doob's stopping theorem to the stopping time $T_{\lceil\e N\rceil}\wedge T_0$ we obtain, for all $\mathbf x\in B_{\e^2 N}$, that
\be\textstyle
\mathbb E\left[\left(\sum_{i=1}^k u_iZ_i(T_{\lceil\e N\rceil})\right)\eee^{-\lambda_{\text{max}}(\mathbf A)T_{\lceil\e N\rceil}}\1_{\{T_{\lceil\e N\rceil}< T_0\}}^{}\middle| Z(0)=\mathbf x\right]=\sum_{i=1}^ku_i x_i.
\ee
Therefore, since $\lambda_{\text{max}}(\mathbf A)<0$ in the subcritical case, 
\be
\mathbb E\left[\min_{1\leq i\leq k} u_i\lceil\e N\rceil \1_{\{T_{\lceil\e N\rceil}< T_0\}}^{}\middle| Z(0)=\mathbf x\right]\leq\max_{1\leq i\leq k} u_i\lceil \e^2 N\rceil, \quad \text{  for all $\mathbf x\in B_{\e^2 N}$},
\ee
which implies (\ref{(i)3}).

Let us continue by proving
(ii). Since $Z(t)$ is supercritical in this case,
applying Theorem 5.7.2 of \cite{AN} yields that
\be\label{Thm 5.7.2} 
\lim_{t \uparrow  \infty} Z(t) \eee^{-\l_{\text{max}}(\mathbf A) t}=W \mathbf v, \hbox{\rm a.s.}, 
\ee
 where $W$ is a nonnegative random variable. Since we
assume that,
for all $i\in\{1,\ldots, k\}$, \\
 $\mathbb E[Z_j(t) \ln(Z_j(t))| Z(0)=\mathbf e_i]<\infty$, we get that 
 \be
\mathbb P\left(W=0| Z(0)=\mathbf e_i \right)=q_i, \quad \mathbb E[W| Z(0)=\mathbf e_i ]=u_i,
\ee
and $W$ has an absolutely continuous distribution on $(0,\infty)$.
 All components of $\mathbf v$ are strictly positive and $W>0$, a.s., on 
 the event $\{\o:T_0(\o)=\infty\}$. Hence, we have 
\be 
Z(t)=O \left(\eee^{\l_{\text{max}}(\mathbf A) t} \right) \quad\text{ a.s.\quad  on }
\{T_0=\infty\}.
\ee
This implies that  $\lim_{N\uparrow\infty}\mathbb P( Z(t_N)< \lceil \e N\rceil, 
T_0=\infty)=0$ and thus
\be
\lim_{N\uparrow \infty}\mathbb P\left(T_0=\infty,\: T_{\lceil \e N\rceil }\geq 
t_N\right)=0.
\ee
Note that we used that $t_N\gg\ln(N)$.
Since $\mathbb P\left(T_0=\infty| Z(0)=\mathbf e_i \right)=1-q_i $, we deduce (\ref{(ii)2}).
On the other hand, there exist two sequences $s^1_N$ and $s^2_N$, which  converge to infinity as $N\uparrow  \infty$,  such that,
 $\lim_{N\uparrow\infty}\mathbb P(s_N^1\leq t_N\wedge T_{\lceil \e N\rceil }\leq s_N^2)=1$. This yields (\ref{(ii)1}), because, for all $i\in\{1,\ldots k\}$ and $l=1,2$, we have that  
$\lim_{N\uparrow \infty}\mathbb P(T_0<s_N^l|Z_0=\mathbf e_i)=q_i$. 
Note that equation (\ref{limit-ratio}) is a simple consequence of (\ref{Thm 5.7.2}).
\end{proof}
Using these properties about multi-type branching processes, we can compute the invasion probability. 

\noindent{\bf The first invasion step.} Let us introduce the following stopping times
\bea
\theta^{K, M\epsilon}_{\text{exit}}&=&\inf\Big\{t\geq 0: ||\nu_t^K-\textstyle\sum_{ x\in\mathcal X_{(\mathbf g,\mathbf p)}}\bar{\frakn}_{x}(\mathbf g,\mathbf p)\delta_{x}||_{TV} >M \epsilon\Big\}\\
\tilde \theta^{K}_\epsilon&=&\inf\Big\{t\geq 0: \:\nu_t^K(\tilde g) \geq \epsilon\:\Big\}\\
\tilde\theta^{K}_{0}&=&\inf\Big\{t\geq 0: \:\nu_t^K(\tilde g)= 0\:\Big\}
\eea
Until $\tilde\theta^{K}_{\epsilon}$ the mutant population $\nu_t^K(\tilde g)$ influences only the death and switching rates of the resident population, and this perturbation is uniformly bounded by $(\bar c + \bar s_{\text{ind.}})\epsilon$. Thus, by applying Theorem \ref{det-limit} (ii), we obtain
\be\label{exit_domain_apply}
\lim_{K\uparrow  \infty}
\mathbb P\left [  \theta^{K, M \epsilon}_{\text{exit}} <\eee^{KV}\!\wedge \tau^{}_{\text{mut.}}\wedge \tilde\theta^{K}_{\epsilon}\right]=0.
\ee
On the time interval  $[0,\theta^{K,M \epsilon}_{\text{exit}} \wedge \tau^{}_{\text{mut.}}\wedge \tilde\theta^{K}_{\epsilon}]$, the resident population can be approximated by $\textstyle\sum_{ x\in\mathcal X_{(\mathbf g,\mathbf p)}}\bar{\frakn}_{x}(\mathbf g,\mathbf p)\delta_{x}$ and no further mutant appears. This allows us to approximate $\nu_t^K(\tilde g)$ by multi-type branching processes.

For a rigorous analysis, one needs to construct two coupled processes that control the populations from above 
and below. Here, we consider just one approximating process.

Let $k\equiv |[\tilde p]_{\tilde g}|$. We construct the  $(\mathbb N_0)^{k}$- valued  multi type process $X(t)$.

For each  $1\leq i\leq k$, each individual in $X(t)$, with trait $(\tilde g,\tilde p_i)$ undergoes

\begin{enumerate}
 \item birth (without mutation) with rate \:$b(\tilde p_i)$.
\item death with rate \:$D_{(\mathbf g,\mathbf p)}(\tilde p_i)$, where
 \be
 D_{(\mathbf g,\mathbf p)}(\tilde p_i)\equiv  d(\tilde p_i) + \sum_{(g,p)\in\mathcal X_{(\mathbf g,\mathbf p)}} c(\tilde p_i, p) 
\bar{\frakn}_{(g,p)}(\mathbf g,\mathbf p),
\ee, 
\item{ switch to $\tilde p_j$ with rate $S_{(\mathbf g,\mathbf p)}(\tilde p_i,\tilde p_j)$,  for all $j\neq i$,\\
where $ 
S_{(\mathbf g,\mathbf p)}(\tilde p_i,\tilde p_j)\equiv s_{\text{nat.}}^{\tilde g}(\tilde p_i, \tilde p_j) +\sum_{( g, p)\in\mathcal X_{(\mathbf g,\mathbf p)}}
				 s^{\tilde g}_{\text{ind.}}(\tilde p_i, \tilde p_j)( p) \bar \frakn_{ ( g, p)}$. }
\end{enumerate} 
Moreover, 
There exists a $K_0>1$ such that,  for all ${\tilde p}_i\in [\tilde p]_{\tilde g}$ and for all $K\geq K_0$,
\be\label{compare}
\forall \:0\:\leq t\:\leq \theta^{K,\epsilon}_{\text{exit}} \wedge \tau^{}_{\text{mut.}}\wedge \tilde\theta^{K}_{\epsilon}:\qquad
 X_i(t) = \nu_t^K(\tilde g,\tilde p_i) K +O(\e). 
\ee
Fixation of the mutant will be determined by the behaviour of the process $X$.

Define, for $ i=1,\ldots,k$,
\be\label{apparent-fitness} 
f_{(\mathbf g, \mathbf p)}	{(\tilde g, \tilde p_i)}	\equiv b(\tilde p_i) -D_{(\mathbf g,\mathbf p)}(\tilde p_i) -\textstyle\sum_{j\neq i}S_{(\mathbf g,\mathbf p)}(\tilde p_i,\tilde p_j).
\ee 

($f_{(\mathbf g, \mathbf p)}{(\tilde g, \tilde p_i)}$ would be the invasion fitness of phenotype $\tilde p_i$ if there was no switch back from the
other phenotypes to $\tilde p_i$.)
 Then, by Equation (\ref{inf.gen}), 
 the generators are given by the following matrices:  

\be \mathbf A ({X^{l,\epsilon}}) =\begin{pmatrix}
f^{l,\e}_{(\mathbf g, \mathbf p)}(\tilde g, \tilde p_1)&	S_{(\mathbf g,\mathbf p)}(\tilde p_1,\tilde p_2)\!-\! \bar s_{\text{ind.}} M \epsilon&	\ldots	&	S_{(\mathbf g,\mathbf p)}(\tilde p_1,\tilde p_k)\!- \!\bar s_{\text{ind.}} M \epsilon\\[0.3em]
S_{(\mathbf g,\mathbf p)}(\tilde p_2,\tilde p_1)\!-\! \bar s_{\text{ind.}} M \epsilon	&	f^{l,\e}_{(\mathbf g, \mathbf p)}(\tilde g, \tilde p_2)&			&		\\[0.3em]
\vdots	&			&	\ddots	&		\vdots		\\[0.3em]
S_{(\mathbf g,\mathbf p)}(\tilde p_k,\tilde p_1)\!-\! \bar s_{\text{ind.}} M \epsilon&		&		\ldots	&f^{l,\e}_{(\mathbf g, \mathbf p)}(\tilde g, \tilde p_k)
\end{pmatrix} \ee
for $l\in\{1,2\}$, where 
$
f^{1,\e}_{(\mathbf g, \mathbf p)}(\tilde g, \tilde p_i)
\equiv  f_{(\mathbf g, \mathbf p)}	{(\tilde g, \tilde p_i)}	- \epsilon(1+ \bar c M + (k-1)\bar s_{\text{ind.}} M )$ and $
f^{2,\e}_{(\mathbf g, \mathbf p)}(\tilde g, \tilde p_i)\equiv   f_{(\mathbf g, \mathbf p)}	{(\tilde g, \tilde p_i)}	+ \epsilon(1+ \bar c M+3(k-1)\bar s_{\text{ind.}} M )$.

We prove in the following that 
the number of mutant individuals grow with positive probability to $\epsilon K$ before dying out if and only if $\l_{\text{max}}$ of $\mathbf A_{(\tilde g,\tilde p)}$ is strictly positive. Thus, $\l_{\text{max}}(\mathbf A_{(\tilde g,\tilde p)})$ is an appropriate generalisation of the invasion fitness of the class $[\tilde p]_{\tilde g}$:
\begin{equation}
\label{invasion-fitness-case1}
F_{[\tilde p]_{\tilde g}}(\mathbf g, \mathbf p)\equiv \l_{\text{max}}(\mathbf A_{(\tilde g,\tilde p)}).
\end{equation}
 The process $X$ is non-singular and irreducible and satisfies the conditions of Lemma \ref{expansion of Thm 4}.
Let $\mathbf q (X)$  denote the extinction probability vector  of  $ X$, i.e. 
\be 
\mathbf q (X)\equiv(q_1 (X),\ldots, q_k (X)),  \quad\text{where } q_i(X)\equiv\mathbb P\left(X(t)=0, \text{ for some }t\; \middle| X(0)=\mathbf e_i\right).
 \ee
Observe  that  $\mathbf q (X)=(1,\ldots, 1)$ if $X$ is not supercritical. To characterise  $\mathbf q (X^{l,\e})$ in the supercritical case,  one computes the infinitesimal generating function
 \be
  \mathbf u:[0,1]^k \to \mathbb R^k, 
 \ee 
 defined, for all $1\leq i\leq k$, by
 \bea\mathbf u_i(\mathbf y )\!&\equiv\!&\nonumber
b(\tilde p_i)y_i^2+\sum_{j\neq i} \{ S_{(\mathbf g,\mathbf p)}(\tilde p_i,\tilde p_j) \}\: y_j+D_{(\mathbf g,\mathbf p)}(\tilde p_i)\\&&
-\Big( b(\tilde p_i) +\sum_{j\neq i} S_{(\mathbf g,\mathbf p)}(\tilde p_i,\tilde p_j)+D_{(\mathbf g,\mathbf p)}(\tilde p_i)\Big)\:y_i.
\eea  
Moreover,
 the extinction vector of a multi-type branching process is  given as the unique root of the
  generating function in the unit cube (cf.\ \cite{AN} p.\ 205 or \cite{SEW} Chap.\ 5). Thus,  in the supercritical case,
 $\mathbf q (X)$  is   the unique solution of 
\be \label{eq_extinct1}
\mathbf u(\mathbf y )=0 \qquad \text{for $\mathbf y \in[0,1)^k$ }
\ee
From this we can read of the fixation probabilities in the theorem and this explains the result.

\chapter{Evolution in a changing environment}

\begin{chapquote}
{J.A.J. Metz, S.A.H. Geritz, G. Meszéna, F.J.A. Jacobs, J.S. van Heerwaarden, \emph{
Adaptive Dynamics,
a geometrical study of the consequences of
nearly faithful reproduction}}
{That there is no sign yet that evolution on this earth is going to freeze has two causes.
The easy one is that the physical configuration of the world keeps changing. But it usually does so
relatively slowly.}
\end{chapquote}

The models considered in the previous chapters have in common that they study conditions of constant environmental parameters. There are, however, many applications where this assumption is broken, such as parameters changing with seasons, where competition for fewer resources is higher in winter and birth rates are higher in spring. Another scenario in applications related to disease modelling is treatment that is applied in phases, where drug administration is followed by a treatment break that gives both healthy cells and pathogens a chance to recover.
In this chapter, we consider a variation of the previous model that takes into account such changes in the population's environment.

\section{Setting}

As in Chapters \ref{chapter7}, \ref{escape}, and \ref{chapter9}, we consider a finite trait space given by a directed graph $G=(V,E)$ with vertices representing the traits and individual can carry and edges marking possible transitions through single mutations. We consider moderately rare mutations with mutation probability $\mu_K=K^{-1/\alpha}$, for some $\alpha\in\R_+\backslash\N$. We introduce a periodically changing environment with a finite number $\ell\in\N$ of phases. For each phase $i=1,\cdots,\ell$ and all traits $v,w\in V$, the system parameters impacting a population's growth rate may vary and are now denoted by 
\begin{align}
b^i(v),\ d^i(v), \text{ and } c^i(v,w).
\end{align}
For simplicity, we assume that the mutation probability $\mu_K$ and mutant law $m(v,w)$ do not vary, although this would not qualitatively impact the results as long as orders of magnitude are preserved.

From the previous chapters, we know that for a new mutant, reaching a macroscopic population size through exponential growth takes a time of order $\ln K$.
For an established resident population, it takes a time of order $1$ to reach a small neighbourhood of its new equilibrium after an environmental change. 
In order for environmental changes to happen slowly enough that the resident populations can adapt, but fast enough that they influence the growth of mutants, we choose
\begin{align}
	1\ll\lambda_K\ll\ln K
\end{align}
as an intermediate time-scale for the length of the $\ell$ phases. For each $i=1,\ldots,\ell$, we assume that the $i$-th phase has length $T_i\lambda_K$, where $T_i>0$. To refer to the endpoints of the phases, we define $T^\Sigma_j:=\sum_{i=1}^{j}T_i$.

Building on this, we define the time-dependent birth-, death-, and competition rates as the periodic extension of
\begin{align}
	b_t^K(v)=\sum_{i=1}^{\ell}\1_{t\in [T^\Sigma_{i-1}\lambda_K,T^\Sigma_i\lambda_K)}b^i(v),
\end{align}
and analogously for $d_t^K(v)$ and $c^K_t(v,w)$. The dynamics of the Markov process $(N^K_v(t))_{v\in V, t\geq0}$ are then summarised by its infinitesimal generator
\begin{align}\label{eq:time_dep_generator}
	\left(\cL_t^K\phi\right)(N)=&\sum_{v\in V}(\phi(N+\d_v)-\phi(N))\left(N_vb_t^K(v)(1-\mu_K)+\sum_{w\in V}N_wb_t^K(w)\mu_Km(w,v)\right)\notag\\
	&+\sum_{v\in V}(\phi(N-\d_v)-\phi(N))N_v\left(d_t^K(v)+\sum_{w\in V}\frac{c^K_t(v,w)}{K}N_w\right),
\end{align}
where $\phi:\N_0^V\to\R$ is measurable and bounded and $\d_v$ denotes the unit vector at $v\in V$.

In this chapter, we only consider the case of monomorphic resident populations, denoting their (rescaled) equilibrium sizes and corresponding invasion fitnesses during an $i$-phase by
\begin{align}
\bar{n}^i(v):=\frac{b^i(v)-d^i(v)}{c^i(v,v)}\ \text{ and }\ f^i(w,v):=b^i(w)-d^i(w)-c^i(w,v)\bar{n}^i(v).
\end{align}
Moreover, we define the time-dependent equilibrium size and fitness by the periodic extensions of
\begin{align}
	\bar{n}^K_t(v):=\sum_{i=1}^{\ell}\1_{t\in [T^\Sigma_{i-1}\lambda_K,T^\Sigma_i\lambda_K)}\bar{n}^i(v)\quad\text{and}\quad f^K_t(w,v):=\sum_{i=1}^{\ell}\1_{t\in [T^\Sigma_{i-1}\lambda_K,T^\Sigma_i\lambda_K)}f^i(w,v),
\end{align}
and the average fitness by
\begin{align}
	f^{\text{av}}(w,v):=\frac{\sum_{i=1}^\ell T_i f^i(w,v)}{T^\S_\ell}.
\end{align}

With this notation, we can now study the evolution of the population and derive scaling limit results corresponding to those of Chapter \ref{chapter7} - describing the dynamics on the $\ln K$-time scale for reaching traits within the mutation spreading neighbourhood - and Chapter \ref{escape} - describing the transition of wider valleys in the fitness landscape on a $1/K\mu_K^L$-time scale.


\section{Within the mutation neighbourhood}

As it was the case in the previous chapters, the analysis of the limiting dynamics is split into approximations for the resident and the mutant populations. Larger sections of the arguments in Chapter \ref{chapter7} for the dynamics on the $\ln K$-time scale transfer directly to the scenario of a changing environment, so we focus on pointing out the differences. We refer to \cite{EsserKraut2} for details.

Previously, to ensure that - as long as the mutant populations stay below a certain small $\eps K$-threshold - the resident population also only deviates from its equilibrium state by an amount of order $\eps K$, a large deviation result was applied, which guarantees for the stochastic process to stay close to an attractive equilibrium for an exponential time in $K$. In this case, to bound the probability of failure (i.e.\ deviating too far from the equilibrium), we need to concatenate these results for an order of $\ln K/\l_K$ many $\lambda_K$-phases that are necessary to observe mutant growth on the $\ln K$-time-scale. During each phase, the equilibrium size changes and a short re-equilibration time of order 1 is needed before the resident population is again close to its equilibrium, and a stability result can be applied.

By conditioning on not deviating too much during the previous phases, we can write the overall probability of failure as the sum of the probabilities to deviate during specific phases. We hence need the latter probabilities to converge to 0 faster than $\l_K/\ln K$, such that the sum of $\ln K/\l_K$ such terms still converges to 0.
Neither Lemma \ref{ldp.1} (a large deviation principle for coexisting resident traits) nor Lemma \ref{lem:step-1} (using a potential theoretic approach for monomorphic resident populations) specify the exact speed of convergence. However, a careful revision of the arguments for Lemma \ref{lem:step-1} does yield the desired rate of convergence of $o(\lambda_K/\ln K)$ for a single parameter phase, once the resident population has gotten sufficiently close to its equilibrium. One can combine this with a revised version of the standard convergence result to the deterministic system of \cite[Ch.\ 11.2]{EthKur1986} to address the short re-equilibriation times. This revised version proves convergence in probability instead of almost surely, but can again quantify the convergence speed and bound the probability of failure in $o(\lambda_K/\ln K)$ in return. Overall, concatenating these two results for $O(\ln K/\lambda_K)$ many phases yields a vanishing probability for the resident population to stray from its respective equilibria.

With these bounds on the resident population, one can couple the mutant populations to simpler birth-death processes (with immigration) to estimate their growth. In \cite{CMT2021}, one finds a collection of general results on the growth of birth-death processes (with immigration), which are used to study bounding processes coupled to the processes $N^K_V(t)$. These results, however, only cover processes with constant parameters. One can argue that equivalent results apply for changing parameters, with the time-average fitness $f^\text{av}_{w,v}$ as the mutants' growth rate since $\lambda_K\ll\ln K$, i.e.\ the parameter fluctuations occur on a faster time-scale than the growth of the mutants. This requires a careful rerun of the proofs in \cite{CMT2021} to keep track of the error stemming from this averaging approximation and to ensure that the changing parameters do not disrupt the extinction of old and emergence of new types.

Finally, based on these results for the coupled processes, one can derive the limiting piecewise affine growth of the orders of population sizes
\begin{align}
\beta^K_v(s)=\frac{\ln(N_v^K(s\ln K)+1)}{\ln K}
\end{align}
of the mutant populations as in Chapter \ref{chapter7}. The algorithmic description of the limiting processes $(\beta_v)_{v\in V}$ is exactly the same as in Section \ref{chapter7results}, only replacing $f(w,v)$ by $f^{\text{av}}(w,v)$ and replacing termination criterion (b) with criteria
\begin{itemize}
	\item[(b1)] there exists an $i\in\{1,\ldots,\ell\}$ such that either $f^i(v_{k-1},v_k)\geq 0$ or $f^i(v_k,v_{k-1})=0$;
	\item[(b2)] there exists an $i\in\{1,\ldots,\ell\}$ such that  $b^i(v_k)-d^i(v_k)\leq 0$.
\end{itemize}
The former criterion (b) ensured that there had to be a unique globally attractive stable equilibrium to the Lotka-Volterra dynamics of the former resident and newly macroscopic mutant traits. The first part of criterion (b1) is a sufficient condition to ensure the principle of \textit{invasion implies fixation}, i.e.\ any mutant trait that reaches a macroscopic population size replaces the former resident trait and there is no coexistence. The condition is not necessary, and there are other possible (more complicated) scenarios where the invading mutant replaces the resident population. The second part ensures that the $\ve K$-threshold is reached by the mutant at a time when invasion can take place in finite time, i.e.\ the comparison to the deterministic Lotka-Volterra system is possible. Criterion (b2) ensures that the new resident possesses a strictly positive monomorphic equilibrium $\bar{n}^i(v_k)$ in all phases.

The first result on the convergence of the orders of population sizes is directly equivalent to Theorem \ref{ThmConv}.
	
	\begin{theorem}[Convergence of $\b$, Theorem 2.4 in \cite{EsserKraut2}]
		\label{thm:conv_beta}
		Let a finite graph $G=(V,E)$ and $\a\in\R_{>0}\setminus\N$ be given and consider the model defined by \eqref{eq:time_dep_generator}. Let $v_0\in V$ and
		assume that, for every $w\in V$, 
		\begin{align}
		\b^K_w(0)\to \b_w(0)=\left(1-\frac{d(v_0,w)}{\alpha}\right)_+
		\end{align}
		in probability, as $K\to\infty$.
		Then, for all fixed $0\leq T\leq T_0$, the following convergence holds in probability, with respect to the $L^\infty([0,T],\R^V_{\geq 0})$ norm
		\begin{align}
			(\b_w^K(t), w\in V)_{t\in [0,T]}\overset{K\to\infty}{\longrightarrow}(\b_w(t), w\in V)_{t\in [0,T]},
		\end{align}
		 where $\b_w$ are the deterministic, piecewise affine, continuous functions defined above.
	\end{theorem}

	The second result on the ``visible'' evolution of the population process, i.e.\ the progression of macroscopic traits that dominate the whole system, needs to be phrased slightly differently than Proposition \ref{CorEquilibria}. Since the equilibrium size of the resident population changes with every $\lambda_K$-phase, one cannot expect convergence of the absolute population sizes on the $\ln K$-time scale. Instead, the result is phrased in terms of the support of the macroscopic population.
	
	\begin{corollary}[Sequence of resident traits, Corollary 2.5 in \cite{EsserKraut2}]
		\label{cor:seq_residents}
		Let
		\begin{align}
			\omega^K_\ve(t):=\sum_{w\in V: N^K_w(t)>\ve K}\d_w\quad\text{ and }\quad\omega(s):=\sum_{k\in\N_0} \1_{s_k\leq s<s_{k+1}}\d_{v_k}	
		\end{align}
		denote point measures having support on the macroscopic traits. Then, under the assumptions of Theorem \ref{thm:conv_beta}, there exists an $\ve_0>0$ such that, for all $0<\ve<\ve_0$ and all $p\in[1,\infty)$, the following convergence holds in probability, with respect to the $L^p([0,T],\mathcal{M}(V))$ norm
		\begin{align}
			\left(\omega^K_\ve(s\ln K)\right)_{s\in [0,T]}\overset{K\to\infty}{\longrightarrow} 	\left(\omega(s)\right)_{s\in [0,T]},
		\end{align}
		where $\mathcal{M}(V)$ denotes the set of finite, non-negative point measures on $V$ equipped with the weak topology.
	\end{corollary}


\section{Crossing a fitness valley}

As before, the evolution on the $\ln K$-time scale stops once an evolutionary stable condition is achieved. As in Chapter \ref{escape}, we study the crossing of a fitness valley of length $L$. By this, we mean to start initially with a monomorphic wild-type population of trait $0$, near its equilibrium $\bar{n}^1(0) K$, and wait until mutants have transitioned through several unfit intermediate traits to eventually produce a mutant of trait $L$
that forms a subpopulation of macroscopic order $K$ and replaces the wild-typ as the resident trait. To depict this situation, we fix the initial condition as follows.
\begin{assumption}[Initial condition]
      \label{Ass:InitialCond}
    \begin{enumerate}[label=(\roman*)]
        \item $N^K_0(0)=\lfloor\bar{n}^1(0)K\rfloor$,
        \item $N_v^K(0)=0$ , for all $v\in\{1,...,L\}$.
    \end{enumerate}
\end{assumption}
Moreover, we introduce the following stopping time that marks the time when the $L$-trait has taken over the population,
\begin{align}
    T^{(K,\varepsilon)}_{\text{inv}}
    =\inf\left\{t\geq0: \left|\frac{N^K_L(t)}{K}-\bar{n}_t^K(L)\right|<\varepsilon
    \text{\ \ and\ \ }\frac{1}{K}\sum_{j=0}^{L-1}N^K_j(t)<\varepsilon\right\}.
\end{align}

To ensure that an $L$-mutant subpopulation can fixate and invade in a phase when it is fit with respect to the resident 0-trait, we make the following assumptions.
\begin{assumption}[Guaranteed invasion]
      \label{Ass:InvFix}
    \begin{enumerate}[label=(\roman*)]
        \item $f^i(0,L)<0$, whenever $f^i(L,0)>0$,
        \item $f^i(L,0)\neq 0$, for all $i=1,\ldots,\ell$.
    \end{enumerate}
\end{assumption}
Note that while the first part of the assumption prevents coexistence, the second part is only technical and avoids the situation of critical branching process approximations.

To define a fitness valley, one might now simply require the average fitness to be negative for all intermediate traits in $\{1,L-1\}$. However, a negative average fitness only prevents long-term growth on the $\ln K$-time scale, as studied in the previous section. On the $\lambda_K$-time scale of environmental changes, there might still be phases $i$ of positive invasion fitness $f^i(v,0)>0$, for some trait $v\in\{1,L-1\}$, which would allow for temporary growth to a mesoscopic size of this mutant subpopulation. Such a short-term growth significantly complicates the study of a fitness valley transition since the intermediate traits no longer only exhibit subcritical excursions. For simplicity, we therefore restrict ourselves to the scenario of a \emph{strict fitness valley} in the sense that the traits within the valley are unfit in every phase (cf.\ Assumption \ref{Ass1:strictFV}). An extension, allowing exactly one trait to have positive fitness in one phase, is presented in \cite{EsserKraut3}. These conditions are referred to as a $\emph{pit stop}$.

\begin{assumption}[Strict fitness valley]
    \label{Ass1:strictFV}
    \begin{enumerate}[label=(\roman*)]
        \item $\bar{n}^i(0)>0$, for all $i=1,\ldots,\ell$,
        \item $f^i(w,0)<0$, for all $w\in\{1,L-1\}$ and all $i=1,\ldots,\ell$,
        \item $f^\text{av}(L,0)>0$.
    \end{enumerate}
\end{assumption}

Since the subcritical excursions of intermediate unfit traits only last for a time of order 1, as elaborated on in Chapter \ref{escape}, any crossing of the fitness valley must occur within the same parameter phase. This makes it a very rare but itself a fast event, and we can hence treat it phase by phase and define the phase-dependent crossing rates, for $i=1\ldots\ell$,
\begin{align}\label{eq:crossrate_i}
    R^i_L:=\bar{n}^i(0) \left(\prod_{v=1}^{\lfloor\alpha\rfloor}\frac{b^i(v-1)}{|f^i(v,0)|}\right) b^i(\lfloor\alpha\rfloor) \left(\prod_{w=\lfloor\alpha\rfloor+1}^{L-1}\lambda(\rho^i(w))\right) \frac{\left(f^i(L,0)\right)_+}{b^i(L)},
\end{align}
where
\begin{align}
    \rho^i(w):=\frac{b^i(w)}{b^i(w)+d^i(w)+c^i(w,0)\bar{n}^i(0)}
    \quad\text{and}\quad
    \lambda(\rho^i(w))=\frac{\rho^i(w)}{1-2\rho^i(w)}=\frac{b^i(w)}{|f^i(w,0)|}.
\end{align}
These are the same crossing rates as in Chapter \ref{escape}, simply with phase-dependent instead of constant parameters. We hence refer to this chapter for a derivation of these expressions.

Notably, $R^i_L$ describes the rate at which $L$-mutants occur in phase $i$ and foster a population that initially grows during this $i$-phase. It does however not account for the chance of the population going extinct (due to negative fitness) in a successive phase.
To conclude the effective rate at which an $L$-mutant occurs and not only initially survives but invades the population - i.e.\ reaches a size of order $K$ and out-competes the current resident trait - we need to consider the growth dynamics of an $L$-population over the course of many phases. During an $i$-phase, the $L$-population grows approximately at exponential rate $f^i(L,0)$. Hence, starting with a size of order 1 at time $T\lambda_K$, after a time $S\lambda_K$ the population would have grown to a size of order
\begin{align}\label{eq:wsize}
    \eee^{\int_{T\lambda_K}^{(T+S)\lambda_K}f_t(L,0)\dd t}=\eee^{\lambda_K\int_{T}^{(T+S)}f_{u\lambda_K}(L,0)\dd u}.
\end{align}
To guarantee survival, this order of the population size needs to stay larger than 1 (and in fact almost sure extinction can be shown in the case where it drops below 1), i.e.\ one needs
\begin{align}
    \int_{T}^{(T+S)}f_{u\lambda_K}(L,0)\dd u>0.
\end{align}
Since  by assumption $f^\text{av}_{L,0}>0$, this can only fail within the first cycle of phases. We therefore introduce the set of possible arrival times of successful $L$-mutants of
\begin{align}
    A:=\left\{t\geq 0:\int_t^{t+s} f^K_{u\lambda_K}(L,0)\dd u>0\ \forall s\in(0,T^\Sigma_\ell]\right\}.
\end{align}
Finally, the \textit{effective crossing rate}, i.e.\ the rate at which $L$-mutants occur, initially survive, and grow to a population size of order $K$, can be calculated by averaging the phase-dependent rates over a full cycle of phases and taking the above set $A$ into account, which yields a rate of $K\mu_K^LR^\text{eff}_L$, where
\begin{align}
   R^\text{eff}_L:=
    \frac{1}{T^\S_\ell}\int_0^{T^\Sigma_\ell}\left(\sum_{i=1}^\ell R^i_L\1_{t\in[T^\Sigma_{i-1},T^\Sigma_i)}\right) \1_{t\in A} \dd t.
\end{align}
Since this is an exponential rate of order $K\mu_K^L$, the crossing event itself occurs on a time scale of order $1/K\mu_K^L$. The exponential growth of the $L$-mutant from a population size of order 1 to a size of order $K$ occurs within a $\ln K$-time, and the Lotka-Volterra dynamics of the $L$-mutant taking over the resident population plays out in a time of order 1 once both populations are of the same order. Both of these events are negligible on the $1/K\mu_K^L$ time scale.

Using the above notation, we can describe the crossing times of a strict fitness valley as follows.
\begin{theorem}
    \label{Thm:Main_1}
    Suppose that Assumptions \ref{Ass:InitialCond}, \ref{Ass:InvFix}, and \ref{Ass1:strictFV} are satisfied. Then there exist $\ve_0>0$ and $c\in(0,\infty)$ such that, for all $0<\ve<\ve_0$, there are exponential random variables $E^{(K,\pm)}(\ve)$ with parameters $(1\pm c\ve)R^\text{eff}_L$ such that
    \begin{align}
        \liminf_{K\to\infty}\P\left(E^{(K,-)}(\ve)\leq T^{(K,\ve)}_\text{inv} K\mu_K^L\leq E^{(K,+)}(\ve)\right)\geq 1-c\ve.
    \end{align}
\end{theorem}


\backmatter
  \bibliography{Library-Biomaths}\label{abbrv}

\begin{thebibliography}{100}

\bibitem{Aldous78}
{\sc D.~Aldous}, {\em Stopping times and tightness}, Ann. Probability, 6
  (1978), pp.~335--340.

\bibitem{arnold}
{\sc V.~I. Arnold}, {\em Ordinary differential equations}, Springer Textbook,
  Springer-Verlag, Berlin, 1992.

\bibitem{A_MTBP}
{\sc K.~B. Athreya}, {\em Some results on multitype continuous time {M}arkov
  branching processes}, Ann. Math. Statist., 39 (1968), pp.~347--357.

\bibitem{AN}
{\sc K.~B. Athreya and P.~E. Ney}, {\em Branching processes}, Springer-Verlag,
  New York-Heidelberg, 1972.
\newblock Die Grundlehren der mathematischen Wissenschaften, Band 196.

\bibitem{BaaBov2018}
{\sc M.~Baar and A.~Bovier}, {\em The polymorphic evolution sequence for
  populations with phenotypic plasticity}, Electron. J. Probab., 23 (2018),
  pp.~1--27.

\bibitem{B14}
{\sc M.~Baar, A.~Bovier, and N.~Champagnat}, {\em {From stochastic,
  individual-based models to the canonical equation of adaptive dynamics in one
  step}}, Ann. Appl. Probab., 27 (2017), pp.~1093--1170.

\bibitem{Baar2016stochastic}
{\sc M.~Baar, L.~Coquille, H.~Mayer, M.~H{\"{o}}lzel, M.~Rogava,
  T.~T{\"{u}}ting, and A.~Bovier}, {\em {A stochastic model for immunotherapy
  of cancer}}, Scientific Reports, 6 (2016), p.~24169.

\bibitem{BansayeMel2015}
{\sc V.~Bansaye and S.~M\'{e}l\'{e}ard}, {\em Stochastic models for structured
  populations}, vol.~1 of Mathematical Biosciences Institute Lecture Series.
  Stochastics in Biological Systems, Springer, Cham; MBI Mathematical
  Biosciences Institute, Ohio State University, Columbus, OH, 2015.
\newblock Scaling limits and long time behavior.

\bibitem{bethedging}
{\sc H.~J.~E. Beaumont, J.~Gallie, C.~Kost, G.~C. Ferguson, and P.~B. Rainey},
  {\em Experimental evolution of bet hedging}, Nature, 462 (2009), pp.~90--93.

\bibitem{BeBruShi16}
{\sc J.~Berestycki, E.~Brunet, and Z.~Shi}, {\em The number of accessible paths
  in the hypercube}, Bernoulli, 22 (2016), pp.~653--680.

\bibitem{BeBruShi17}
\leavevmode\vrule height 2pt depth -1.6pt width 23pt, {\em Accessibility
  percolation with backsteps}, ALEA Lat. Am. J. Probab. Math. Stat., 14 (2017),
  pp.~45--62.

\bibitem{BolPac1}
{\sc B.~Bolker and S.~W. Pacala}, {\em Using moment equations to understand
  stochastically driven spatial pattern formation in ecological systems},
  Theor. Popul. Biol., 52 (1997), pp.~179--197.

\bibitem{BolPac2}
{\sc B.~M. Bolker and S.~W. Pacala}, {\em {Spatial moment equations for plant
  competition: understanding spatial strategies and the advantages of short
  dispersal}}, Am. Nat., 153 (1999), pp.~575--602.

\bibitem{BovCoqNeu18}
{\sc A.~Bovier, L.~Coquille, and R.~Neukirch}, {\em The recovery of a recessive
  allele in a {M}endelian diploid model}, J. Math. Biol., 77 (2018),
  pp.~971--1033.

\bibitem{BovCoqSma2018}
{\sc A.~Bovier, L.~Coquille, and C.~Smadi}, {\em Crossing a fitness valley as a
  metastable transition in a stochastic population model}, Ann. Appl. Probab.,
  29 (2019), pp.~3541--3589.

\bibitem{BH15}
{\sc A.~Bovier and F.~den Hollander}, {\em Metastability: A potential-theoretic
  approach}, vol.~351 of Grundlehren der Mathematischen Wissenschaften
  [Fundamental Principles of Mathematical Sciences], Springer, Cham, 2015.

\bibitem{BovWang2013}
{\sc A.~Bovier and S.-D. Wang}, {\em Trait substitution trees on two time
  scales analysis}, Markov. Proc. Rel. Fields, 19 (2013), pp.~607--642.

\bibitem{bramson_monograph}
{\sc M.~Bramson}, {\em Convergence of solutions of the {K}olmogorov equation to
  travelling waves}, Mem. Amer. Math. Soc., 44 (1983), pp.~iv+190.

\bibitem{buerger2000}
{\sc R.~B\"{u}rger}, {\em The mathematical theory of selection, recombination,
  and mutation}, Wiley Series in Mathematical and Computational Biology, John
  Wiley \& Sons, Ltd., Chichester, 2000.

\bibitem{Cha06}
{\sc N.~Champagnat}, {\em {A microscopic interpretation for adaptive dynamics
  trait substitution sequence models}}, Stochastic Processes Appl., 116 (2006),
  pp.~1127--1160.

\bibitem{C_CEAD}
{\sc N.~Champagnat, R.~Ferri{\`{e}}re, and G.~{Ben Arous}}, {\em The canonical
  equation of adaptive dynamics: A mathematical view}, Selection, 2 (2001),
  pp.~73--83.

\bibitem{C_ME}
{\sc N.~Champagnat, R.~Ferri{\`{e}}re, and S.~M{\'{e}}l{\'{e}}ard}, {\em {From
  individual stochastic processes to macroscopic models in adaptive
  evolution}}, Stochastic Models, 24 (2008), pp.~2--44.

\bibitem{ChamHass2023}
{\sc N.~Champagnat and V.~Hass}, {\em Convergence of population processes with
  small and frequent mutations to the canonical equation of adaptive dynamics},
  Ann. Appl. Probab., 35 (2025), pp.~1--63.

\bibitem{CM11}
{\sc N.~Champagnat and S.~M{\'{e}}l{\'{e}}ard}, {\em {Polymorphic evolution
  sequence and evolutionary branching}}, Probab. Theor. Rel. Fields, 151
  (2011), pp.~45--94.

\bibitem{CMT2021}
{\sc N.~Champagnat, S.~M\'{e}l\'{e}ard, and V.~C. Tran}, {\em Stochastic
  analysis of emergence of evolutionary cyclic behavior in population dynamics
  with transfer}, Ann. Appl. Probab., 31 (2021), pp.~1820--1867.

\bibitem{CMM13}
{\sc P.~Collet, S.~M{\'{e}}l{\'{e}}ard, and J.~A.~J. Metz}, {\em {A rigorous
  model study of the adaptive dynamics of Mendelian diploids}}, J. Math. Biol.,
  67 (2013), pp.~569--607.

\bibitem{CoqKrautSma21}
{\sc L.~Coquille, A.~Kraut, and C.~Smadi}, {\em Stochastic individual-based
  models with power law mutation rate on a general finite trait space},
  Electron. J. Probab., 26 (2021), pp.~Paper No. 123, 37.

\bibitem{coron2}
{\sc C.~Coron}, {\em Stochastic modeling of density-dependent diploid
  populations and the extinction vortex}, Adv. Appl. Probab., 46 (2014),
  pp.~446--477.

\bibitem{coron3}
\leavevmode\vrule height 2pt depth -1.6pt width 23pt, {\em {Slow-fast
  stochastic diffusion dynamics and quasi-stationarity for diploid populations
  with varying size}}, J. Math. Biol., 72(1-2) (2016), pp.~171--202.

\bibitem{coron}
{\sc C.~Coron, S.~M{\'e}l{\'e}ard, E.~Porcher, and A.~Robert}, {\em Quantifying
  the mutational meltdown in diploid populations}, Am. Nat., 181 (2013),
  pp.~623--636.

\bibitem{cowperthwaite2006bad}
{\sc M.~C. Cowperthwaite, J.~J. Bull, and L.~A. Meyers}, {\em From bad to good:
  Fitness reversals and the ascent of deleterious mutations}, PLoS Comp. Biol.,
  2 (2006), p.~e141.

\bibitem{crowkimura}
{\sc J.~F. Crow, M.~Kimura, and Others}, {\em {An introduction to population
  genetics theory}}, New York, Evanston and London: Harper \& Row, Publishers,
  1970.

\bibitem{Darwin1859}
{\sc C.~Darwin}, {\em The Origin of Species}, John Murray, London, 1859.

\bibitem{erasmus1804}
{\sc E.~Darwin}, {\em Temple of Nature; or, The Origin of Society}, John W.
  Butler, and Bonsal and Niles, Baltimore, 1804.

\bibitem{Dawson2017}
{\sc D.~A. Dawson}, {\em Introductory Lectures on Stochastic Population
  Systems}, arXiv 1705.03781, 2017.

\bibitem{deVries1901}
{\sc H.~de~Vries}, {\em {Die Mutationstheorie. Versuche und Beobachtungen über
  die Entstehung von Arten im Pflanzenreich}}, Veit \& Comp., Leipzig, 1901-03.

\bibitem{DeAngelisGrimm}
{\sc D.~L. DeAngelis and V.~Grimm}, {\em Individual-based models in ecology
  after four decades}, F1000Prime Reports, 6 (2014).

\bibitem{DL96}
{\sc U.~Dieckmann and R.~Law}, {\em {The dynamical theory of coevolution: a
  derivation from stochastic ecological processes}}, J. Math. Biol., 34 (1996),
  pp.~579--612.

\bibitem{DieLaw}
\leavevmode\vrule height 2pt depth -1.6pt width 23pt, {\em {Moment
  approximations of individual-based models}}, in The geometry of ecological
  interactions: simplifying spatial complexity, U.~Dieckmann, R.~Law, and
  J.~A.~J. Metz, eds., Cambridge University Press, 2000, pp.~252--270.

\bibitem{Doering2003}
{\sc C.~R. Doering, C.~Mueller, and P.~Smereka}, {\em Interacting particles,
  the stochastic {F}isher-{K}olmogorov-{P}etrovsky-{P}iscounov equation, and
  duality}, Phys. A, 325 (2003), pp.~243--259.
\newblock Stochastic systems: from randomness to complexity (Erice, 2002).

\bibitem{DuEll}
{\sc P.~Dupuis and R.~S. Ellis}, {\em A weak convergence approach to the theory
  of large deviations}, Wiley Series in Probability and Statistics: Probability
  and Statistics, John Wiley \& Sons, Inc., New York, 1997.
\newblock A Wiley-Interscience Publication.

\bibitem{EsserKraut3}
{\sc M.~Esser and A.~Kraut}, {\em Crossing a fitness valley in a changing
  environment: With and without pit stop}.
\newblock arXiv:2503.19766, 2025.

\bibitem{EsserKraut2}
{\sc M.~Esser and A.~Kraut}, {\em Effective growth rates in a periodically
  changing environment: from mutation to invasion}, Stochastic Process. Appl.,
  184 (2025), pp.~Paper No. 104598, 36.

\bibitem{Etheridge2004}
{\sc A.~M. Etheridge}, {\em Survival and extinction in a locally regulated
  population}, Ann. Appl. Probab., 14 (2004), pp.~188--214.

\bibitem{EthKur1986}
{\sc S.~N. Ethier and T.~G. Kurtz}, {\em Markov processes}, Wiley Series in
  Probability and Mathematical Statistics: Probability and Mathematical
  Statistics, John Wiley \& Sons, Inc., New York, 1986.

\bibitem{EthNor77}
{\sc S.~N. Ethier and M.~F. Norman}, {\em Error estimate for the diffusion
  approximation of the {Wright--Fisher model}}, Proc. Natl. Acad. Sci., 74
  (1977), pp.~5096--8.

\bibitem{ewens04}
{\sc W.~J. Ewens}, {\em Mathematical Population Genetics. {I} {Theoretical
  Introduction}}, vol.~27 of Interdisciplinary Applied Mathematics,
  Springer-Verlag, New York, second~ed., 2004.

\bibitem{fisher18}
{\sc R.~A. Fisher}, {\em The correlation between relatives on the supposition
  of {M}endelian inheritance}, Trans. Roy. Soc. Edinb., 42 (1918),
  pp.~399--433.

\bibitem{fisher1937}
\leavevmode\vrule height 2pt depth -1.6pt width 23pt, {\em The wave of advance
  of an advantageous gene}, Ann. Eugenics, 7 (1937), pp.~355--369.

\bibitem{FM04}
{\sc N.~Fournier and S.~M{\'{e}}l{\'{e}}ard}, {\em {A microscopic probabilistic
  description of a locally regulated population and macroscopic
  approximations}}, Ann. Appl. Probab., 14 (2004), pp.~1880--1919.

\bibitem{franklin1953}
{\sc R.~E. Franklin and G.~R.G.}, {\em Molecular configuration in sodium
  thymonucleate}, Nature, 171 (1953), pp.~740--741.

\bibitem{FW84}
{\sc M.~I. Freidlin and A.~D. Wentzell}, {\em Random perturbations of dynamical
  systems}, vol.~260 of Grundlehren der Mathematischen Wissenschaften
  [Fundamental Principles of Mathematical Sciences], Springer, Heidelberg,
  third~ed., 2012.

\bibitem{plasticity}
{\sc G.~Fusco and A.~Minelli}, {\em Phenotypic plasticity in development and
  evolution: facts and concepts}, Phil. Trans. Royal Soc. London B: Biol. Sci.,
  365 (2010), pp.~547--556.

\bibitem{giachetti1988altered}
{\sc C.~Giachetti and J.~J. Holland}, {\em Altered replicase specificity is
  responsible for resistance to defective interfering particle interference of
  an sdi-mutant of vesicular stomatitis virus}, Jour. Virol., 62 (1988),
  pp.~3614--3621.

\bibitem{Gillespie1976}
{\sc D.~T. Gillespie}, {\em A general method for numerically simulating the
  stochastic time evolution of coupled chemical reactions}, J. Computational
  Phys., 22 (1976), pp.~403--434.

\bibitem{gillespie1984molecular}
{\sc J.~H. Gillespie}, {\em Molecular evolution over the mutational landscape},
  Evolution, 38 (1984), pp.~1116--1129.

\bibitem{Grimm2005}
{\sc V.~Grimm and S.~F. Railsback}, {\em Individual-based modeling and
  ecology}, Princeton Series in Theoretical and Computational Biology,
  Princeton University Press, Princeton, NJ, 2005.

\bibitem{haldane24a}
{\sc J.~B. Haldane}, {\em {A mathematical theory of natural and artificial
  selection. {P}art {I}}}, Trans. Camb. Phil. Soc., 23 (1924), pp.~19--41.

\bibitem{haldane24b}
\leavevmode\vrule height 2pt depth -1.6pt width 23pt, {\em {A mathematical
  theory of natural and artificial selection. {P}art {II}}}, Trans. Camb. Phil.
  Soc., Biol, Sci., 1 (1924), pp.~158--163.

\bibitem{Hardy1908}
{\sc G.~H. Hardy}, {\em Mendelian proportions in a mixed population}, Science,
  28 (1908), pp.~49--50.

\bibitem{HS90}
{\sc J.~Hofbauer and K.~Sigmund}, {\em {Adaptive dynamics and evolutionary
  stability}}, Appl Math Lett, 3 (1990), pp.~75--79.

\bibitem{hofbauersigmund}
\leavevmode\vrule height 2pt depth -1.6pt width 23pt, {\em Evolutionary Games
  and Population Dynamics}, Cambridge University Press, Cambridge, 1998.

\bibitem{Holzel:2013ys}
{\sc M.~H{\"{o}}lzel, A.~Bovier, and T.~T{\"{u}}ting}, {\em {Plasticity of
  tumour and immune cells: a source of heterogeneity and a cause for therapy
  resistance?}}, Nat. Rev. Cancer, 13 (2013), pp.~365--376.

\bibitem{houchmandzadeh}
{\sc B.~Houchmandzadeh and M.~Vallade}, {\em Fisher waves: An individual-based
  stochastic model}, Phys Rev E., 96 (2017), p.~012414.

\bibitem{JabinRaoul}
{\sc P.-E. Jabin and G.~Raoul}, {\em On selection dynamics for competitive
  interactions}, J. Math. Biol., 63 (2011), pp.~493--517.

\bibitem{Jain07}
{\sc K.~Jain}, {\em Evolutionary dynamics of the most populated genotype on
  rugged fitness landscapes}, Phys. Rev. E, 76 (2007), p.~031922.

\bibitem{JainKrug05}
{\sc K.~Jain and J.~Krug}, {\em Evolutionary trajectories in rugged fitness
  landscapes}, J. Statist. Mech., P04008 (2005).

\bibitem{JainKrug07}
\leavevmode\vrule height 2pt depth -1.6pt width 23pt, {\em Deterministic and
  stochastic regimes of asexual evolution on rugged fitness landscapes},
  Genetics, 175 (2007), pp.~1275--1288.

\bibitem{lamarck}
{\sc C.~d.~L. Jean-Baptiste Pierre Antoine~de Monet}, {\em Philosophie
  Zoologique}, Dentu, Paris, 1809.

\bibitem{KaLe87}
{\sc S.~Kauffman and S.~Levin}, {\em Towards a general theory of adaptive walks
  on rugged landscapes}, J. Theoret. Biol., 128 (1987), pp.~11--45.

\bibitem{Kau92}
{\sc S.~A. Kauffman}, {\em The origins of order: Self-organization and
  selection in evolution}, in Spin glasses and biology, World Scientific, 1992,
  pp.~61--100.

\bibitem{KesSte2}
{\sc H.~Kesten and B.~P. Stigum}, {\em Additional limit theorems for
  indecomposable multidimensional {G}alton-{W}atson processes}, Ann. Math.
  Statist., 37 (1966), pp.~1463--1481.

\bibitem{KesSte1}
\leavevmode\vrule height 2pt depth -1.6pt width 23pt, {\em A limit theorem for
  multidimensional {G}alton-{W}atson processes}, Ann. Math. Statist., 37
  (1966), pp.~1211--1223.

\bibitem{KesSte3}
\leavevmode\vrule height 2pt depth -1.6pt width 23pt, {\em {Limit theorems for
  decomposable multi-dimensional {G}alton-{W}atson processes}}, J. Math. Anal.
  Appl., 17 (1967), pp.~309--338.

\bibitem{khalil}
{\sc H.~K. Khalil}, {\em Nonlinear systems}, Macmillan Publishing Company, New
  York, 1992.

\bibitem{Kimura55}
{\sc M.~Kimura}, {\em Solution of a process of random genetic drift with a
  continuous model}, Proc. Natl. Acad. Sci., 41 (1955), pp.~144--150.

\bibitem{KG99}
{\sc É.~Kisdi and S.~A.~H. Geritz}, {\em Adaptive dynamics in allele space:
  Evolution of genetic polymorphism by small mutations in a heterogeneous
  environment}, Evolution, 53 (1999), pp.~993--1008.

\bibitem{kpp}
{\sc A.~N. Kolmogorov, I.~G. Petrovsky, and N.~S. Piscounov}, {\em Etude de
  l'\'equation de la diffusion avec croissance de la quantit\'e de mati\`ere et
  son application \`a un probl\`eme biologique}, Moscow Univ. Math. Bull., 1
  (1937), pp.~1--25.

\bibitem{kraut2018}
{\sc A.~Kraut and A.~Bovier}, {\em From adaptive dynamics to adaptive walks},
  J. Math. Biol., 79 (2019), pp.~1699--1747.

\bibitem{Krug2021}
{\sc J.~Krug}, {\em Accessibility percolation in random fitness landscapes}, in
  {Probabilistic Structures in Evolution}, E.~Baake and A.~Wakolbinger, eds.,
  EMS Press, Berlin, 2021, pp.~1--22.

\bibitem{KrugKarl03}
{\sc J.~Krug and C.~Karl}, {\em Punctuated evolution for the quasispecies
  model}, Physica A, 318 (2003), pp.~137--143.

\bibitem{Kurtz1970}
{\sc T.~G. Kurtz}, {\em Solutions of ordinary differential equations as limits
  of pure jump {M}arkov processes}, J. Appl. Probability, 7 (1970), pp.~49--58.

\bibitem{Kurtz1971}
{\sc T.~G. Kurtz}, {\em Limit theorems for sequences of jump {M}arkov processes
  approximating ordinary differential processes}, J. Appl. Probability, 8
  (1971), pp.~344--356.

\bibitem{kurtz1978}
{\sc T.~G. Kurtz}, {\em Strong approximation theorems for density dependent
  {M}arkov chains}, Stochastic Process. Appl., 6 (1977/78), pp.~223--240.

\bibitem{Blath-Lennon}
{\sc J.~Lennon, F.~den Hollander, M.~Wilke-Berenguer, and J.~Blath}, {\em
  Principles of seed banks and the emergence of complexity from dormancy}, Nat.
  Comm., 12 (2021), p.~4807.

\bibitem{lenski2003evolutionary}
{\sc R.~E. Lenski, C.~Ofria, R.~T. Pennock, and C.~Adami}, {\em The
  evolutionary origin of complex features}, Nature, 423 (2003), pp.~139--144.

\bibitem{lotka1912}
{\sc A.~J. Lotka}, {\em Quantitative studies in epidemiology}, Nature, 88
  (1912), pp.~497--498.

\bibitem{maisnier2002compensatory}
{\sc S.~Maisnier-Patin, O.~G. Berg, L.~Liljas, and D.~I. Andersson}, {\em
  Compensatory adaptation to the deleterious effect of antibiotic resistance in
  salmonella typhimurium}, Mol. Microbiol., 46 (2002), pp.~355--366.

\bibitem{malthus1798}
{\sc T.~Malthus}, {\em An Essay on the Principle of Population as it Affects
  the Future Improvement of Society, with Remarks on the Speculations of Mr.
  Goodwin, M. Condorcet and Other Writers}, J. Johnson in St Paul's
  Church-yard, London, 1798.

\bibitem{Maynard62}
{\sc J.~Maynard~Smith}, in {The Scientist Speculates: an Anthology of
  Partly-Baked Ideas}, I.~Good, ed., Basic Books, New York, 1962, pp.~252--256.

\bibitem{Maynard70}
\leavevmode\vrule height 2pt depth -1.6pt width 23pt, {\em Natural selection
  and the concept of a protein space}, Nature, 225 (1970), pp.~563--564.

\bibitem{Meleard2016}
{\sc S.~M\'{e}l\'{e}ard}, {\em Mod\`eles al\'{e}atoires en ecologie et
  evolution}, vol.~77 of Math\'{e}matiques \& Applications (Berlin)
  [Mathematics \& Applications], Springer-Verlag, Berlin, 2016.

\bibitem{mendel1865}
{\sc G.~Mendel}, {\em {Versuche \"uber Pflanzen-Hybriden}}, Verhandlungen des
  naturforschenden Vereines in Br\"unn, IV (1865), pp.~3--47.

\bibitem{mendel1869}
\leavevmode\vrule height 2pt depth -1.6pt width 23pt, {\em {\"Uber einige aus
  k\"unstlicher Befruchtung gewonnenen Hieracium-Bastarde}}, Verhandlungen des
  naturforschenden Vereines in Br\"unn, VIII (1869), pp.~26--31.

\bibitem{metz2012}
{\sc J.~A. Metz}, {\em Adaptive dynamics}, in Encyclopedia of Theoretical
  Ecology, Cambridge University Press, Cambridge, 2012.

\bibitem{M12}
{\sc J.~A.~J. Metz}, {\em {Invasion fitness, canonical equations, and global
  invasion criteria for Mendelian populations}}, in Elements of Adaptive
  Dynamics, U.~Dieckmann and J.~A.~J. Metz, eds., Cambridge University Press,
  2012.

\bibitem{MG96}
{\sc J.~A.~J. Metz, S.~A.~H. Geritz, G.~Mesz\'{e}na, F.~J.~A. Jacobs, and J.~S.
  van Heerwaarden}, {\em Adaptive dynamics, a geometrical study of the
  consequences of nearly faithful reproduction}, in Stochastic and spatial
  structures of dynamical systems ({A}msterdam, 1995), vol.~45 of Konink.
  Nederl. Akad. Wetensch. Verh. Afd. Natuurk. Eerste Reeks, North-Holland,
  Amsterdam, 1996, pp.~183--231.

\bibitem{Moran1958}
{\sc P.~A.~P. Moran}, {\em Random processes in genetics}, Math. Proc. Cambridge
  Phil. Soc., 54 (1958), pp.~60--71.

\bibitem{mukharjee2016}
{\sc S.~Mukherjee}, {\em The Gene}, Bodley, London, 2016.

\bibitem{nagylaki92}
{\sc T.~Nagylaki}, {\em Introduction to theoretical population genetics},
  vol.~21 of Biomathematics, Springer-Verlag, Berlin, 1992.

\bibitem{NeidKrug11}
{\sc J.~Neidhart and J.~Krug}, {\em Adaptive walks and extreme value theory},
  Phys. Rev. Lett., 107 (2011), p.~178102.

\bibitem{BovNeu16}
{\sc R.~Neukirch and A.~Bovier}, {\em {Survival of a recessive allele in a
  Mendelian diploid model}}, J. Math. Biol., 75 (2017), pp.~145--198.

\bibitem{NoKr15}
{\sc S.~Nowak and J.~Krug}, {\em Analysis of adaptive walks on {NK} fitness
  landscapes with different interaction schemes}, J. Stat. Mech. Theory Exp.,
  2015 (2015), p.~P06014.

\bibitem{o1984vesicular}
{\sc P.~J. O'Hara, S.~T. Nichol, F.~M. Horodyski, and J.~J. Holland}, {\em
  Vesicular stomatitis virus defective interfering particles can contain
  extensive genomic sequence rearrangements and base substitutions}, Cell, 36
  (1984), pp.~915--924.

\bibitem{Orr03}
{\sc H.~A. Orr}, {\em A minimum on the mean number of steps taken in adaptive
  walks}, J. Theoret. Biol., 220 (2003), pp.~241--247.

\bibitem{DISS_CTMTBP}
{\sc S.~P{\'e}nisson}, {\em Conditional limit theorems for multitype branching
  processes and illustration in epidemiological risk analysis}, PhD thesis,
  Universit\"at Potsdam, Potsdam (Germany), 2010.

\bibitem{Rogers1}
{\sc L.~C.~G. Rogers and D.~Williams}, {\em Diffusions, {M}arkov processes, and
  martingales. {V}ol. 1}, Cambridge Mathematical Library, Cambridge University
  Press, Cambridge, 2000.
\newblock Foundations, Reprint of the second (1994) edition.

\bibitem{Rogers2}
\leavevmode\vrule height 2pt depth -1.6pt width 23pt, {\em Diffusions, {M}arkov
  processes, and martingales. {V}ol. 2}, Cambridge Mathematical Library,
  Cambridge University Press, Cambridge, 2000.
\newblock It\^{o} calculus, Reprint of the second (1994) edition.

\bibitem{SchKr14}
{\sc B.~Schmiegelt and J.~Krug}, {\em Evolutionary accessibility of modular
  fitness landscapes}, J. Stat. Phys., 154 (2014), pp.~334--355.

\bibitem{schrag1997adaptation}
{\sc S.~J. Schrag, V.~Perrot, and B.~R. Levin}, {\em Adaptation to the fitness
  costs of antibiotic resistance in escherichia coli}, Proc. Royal Soc. London
  B: Biological Sciences, 264 (1997), pp.~1287--1291.

\bibitem{SEW}
{\sc B.~A. Sewastjanow}, {\em Verzweigungsprozesse}, R. Oldenbourg Verlag,
  Munich-Vienna, 1975.

\bibitem{smale76}
{\sc S.~J. Smale}, {\em On the differential equations of species in
  competition}, J. Math. Biol., 3 (1976), pp.~5--7.

\bibitem{uchiyama}
{\sc K.~Uchiyama}, {\em The behavior of solutions of some nonlinear diffusion
  equations for large time}, J. Math. Kyoto Univ., 18 (1978), pp.~453--508.

\bibitem{volterra1928}
{\sc V.~Volterra}, {\em Variations and fluctuations of the number of
  individuals in animal species living together}, J. Conseil Int. Explor. Mer.,
  3 (1928), pp.~1--51.

\bibitem{Wang-Theses}
{\sc S.-D. Wang}, {\em Multi-scale analysis of adaptive population dynamics},
  ph.d. thesis, Bonn University, 2011.

\bibitem{crickwatson1953}
{\sc J.~D. Watson and F.~H.~C. Crick}, {\em A structure for deoxyribose nucleic
  acid}, Nature, 171 (1953), pp.~737--738.

\bibitem{Weinberg1908}
{\sc W.~Weinberg}, {\em {Über den Nachweis der Vererbung beim Menschen}},
  {Jahreshefte des Vereins für vaterländische Naturkunde in Württemberg}, 64
  (1908), pp.~368--382.

\bibitem{wilkins1953}
{\sc M.~H. Wilkins, A.~R. Stokes, and H.~R. Wilson}, {\em Molecular structure
  of deoxypentose nucleic acids}, Nature, 171 (1953), pp.~738--40.

\bibitem{wright31}
{\sc S.~G. Wright}, {\em {Evolution in {M}endelian populations}}, Genetics, 16
  (1931), pp.~97--157.

\bibitem{yule06}
{\sc G.~U. Yule}, {\em {On the theory of inheritance of quantitative compound
  characters on the basis of Mendel's laws: a preliminary note}}, Spottiswoode
  \& Company, Limited, 1907.

\bibitem{Zee1993}
{\sc M.~L. Zeeman}, {\em {Hopf bifurcations in competitive three-dimensional
  {L}otka-{V}olterra systems}}, Dynam. Stability Systems, 8 (1993),
  pp.~189--217.

\end{thebibliography}
\bibliographystyle{siam}

\printindex
\end{document}